\PassOptionsToPackage{noadjust}{cite}
\documentclass{guthesis}
\usepackage{bm}
\usepackage{chngcntr}
\usepackage{multicol}
\usepackage{amsfonts, amsmath, amsthm, amssymb}
\usepackage{ifthen}
\usepackage[all,cmtip]{xy}
\usepackage[pdftex]{graphicx}
\usepackage{multirow} 
\usepackage{stackrel}
\usepackage{appendix}
\usepackage{fancyhdr}
\usepackage[refpage]{nomencl}

\DeclareMathAlphabet{\mathscr}{OT1}{pzc}{m}{it}
\newcommand{\bC}{\mathbb C}
\newcommand{\bR}{\mathbb R}
\newcommand{\bZ}{\mathbb Z}

\newcommand{\n}{\textup}
\newcommand{\f}{\mathfrak}
\newcommand{\ca}{\mathcal}

\newcommand{\newfn}{\DeclareMathOperator}
\newfn{\Tr}{Tr}
\newfn{\End}{End}
\newfn{\Hom}{Hom}
\newfn{\Dom}{Dom}
\newfn{\Sp}{Sp}
\newfn{\qdet}{qdet}
\newfn{\sgn}{sgn}
\newfn{\Res}{Res}
\newfn{\ii}{i}
\newfn{\e}{e}
\def\cprime{$'$}

\renewcommand{\Im}{\mathrm {Im} \, }
\renewcommand{\L}{\mathrm L}
\newcommand{\dd}{\, \mathrm d}
\newcommand{\sd}[2][]{\frac{\mathrm d^{#1}}{\mathrm d #2^{#1}}}

\newcommand{\Vac}{\Psi_{\emptyset}}
\newcommand{\set}[2]{\left\{ \, {#1} : {#2}  \, \right\} }
\newcommand{\inner}[2]{\left\langle {#1}, {#2} \right \rangle }
\newcommand{\innerrnd}[2]{\left({#1}, {#2}\right) }

\newtheorem{thm}{Theorem}[section]
\newtheorem{cor}[thm]{Corollary}
\newtheorem{lem}[thm]{Lemma}
\newtheorem{prop}[thm]{Proposition}
\newtheorem{conj}[thm]{Conjecture}
\newtheorem{rem}[thm]{Remark}
\newtheorem{defn}[thm]{Definition}
\newtheorem{exm}[thm]{Example}
\newtheorem*{notn}{Notation}

\newcommand{\boxedenv}[1]{ \fbox{ \addtolength{\linewidth}{-6mm} \begin{minipage}{\linewidth}  #1 \end{minipage} } }
\newcommand{\boxedalign}[1]{\[ \fbox{ \addtolength{\linewidth}{-2\fboxsep}\addtolength{\linewidth}{-2\fboxrule} \begin{minipage}{\linewidth}  \vspace{-6mm}  \begin{align}#1 \end{align} \end{minipage} } \nonumber \] }
\newcommand{\boxedgather}[1]{\[ \fbox{ \addtolength{\linewidth}{-2\fboxsep}\addtolength{\linewidth}{-2\fboxrule} \begin{minipage}{\linewidth}  \vspace{-6mm}  \begin{gather}#1 \end{gather} \end{minipage} } \nonumber \] }

\newcommand{\rf}[1]{(\ref{#1})}
\newcommand{\rfeqn}[1]{Eqn. (\ref{#1})}
\newcommand{\rfeqnser}[2]{Eqns. (\ref{#1}-\ref{#2})}
\newcommand{\rft}[1]{Thm. \ref{#1}}

\newcommand{\rfc}[1]{Cor. \ref{#1}}
\newcommand{\rfl}[1]{Lemma \ref{#1}}
\newcommand{\rfls}[2]{Lemmas \ref{#1} and \ref{#2}}
\newcommand{\rflser}[2]{Lemmas \ref{#1}-\ref{#2}}
\newcommand{\rfp}[1]{Prop. \ref{#1}}
\newcommand{\rfcn}[1]{Conj. \ref{#1}}
\newcommand{\rfd}[1]{Defn. \ref{#1}}
\newcommand{\rfr}[1]{Rem. \ref{#1}}
\newcommand{\rfex}[1]{Example \ref{#1}}
\newcommand{\rff}[1]{Figure \ref{#1}}

\counterwithin{equation}{section}
\counterwithin{figure}{section}
\makenomenclature
\newcommand{\nc}[3]{\nomenclature[#1]{#2}{#3 \dotfill}}

\renewcommand{\nomgroup}[1]{ \ifthenelse{\equal{#1}{R}}{\item[\textbf{Roman letters}]}{\ifthenelse{\equal{#1}{G}}{\item[\textbf{Greek letters}]}{}}}
\renewcommand{\nompreamble}{See below for the list of symbols used in this thesis and the page where they are first defined. We also refer the reader to the various boxes headed ``Notation'' on pages \pageref{notn1}, \pageref{notn2}, \pageref{notn3}, \pageref{notn5}, \pageref{notn6} and \pageref{notn7}.
}

\title{A non-symmetric Yang-Baxter Algebra for the Quantum Nonlinear Schr\"odinger Model}
\author{Bart Vlaar}
\copyrightto{Bart Vlaar}
\submissionmonth{September}
\submissionyear{2011}

\begin{document}

\frontmatter

\maketitle

\newpage

\chapter*{Abstract}

We study certain non-symmetric wavefunctions associated to the quantum nonlinear Schr\"o-dinger (QNLS) model, introduced by Komori and Hikami using representations of the degenerate affine Hecke algebra. In particular, they can be generated using a vertex operator formalism analogous to the recursion that defines the symmetric QNLS wavefunction in the quantum inverse scattering method. Furthermore, some of the commutation relations encoded in the Yang-Baxter equation are generalized to the non-symmetric case.

\newpage

\chapter*{Acknowledgements}

\begin{quote}
\emph{``Feeling gratitude and not expressing it is like wrapping a present and not giving it.''} \hfill William Arthur Ward
\end{quote}

There are several people to whom thanks are due. I am greatly indebted to my supervisor Dr. Christian Korff for introducing the topic to me, for his insights and always constructive comments and for his support, especially when things did not go according to plan. It has been a pleasure to work with him. I would like to thank the ISMP group at Glasgow for allowing me to present some of my own work and learn of theirs and others'.
More generally, I thank the (now) School of Mathematics and Statistics for allowing me to study in Glasgow and providing a stimulating environment for research. I also acknowledge the financial support from EPSRC.\\

The 522 office, and the postgraduate students at large, have been a great bunch of people. I have enjoyed working alongside them and have enjoyed with them the occasional distraction, such as the odd hill-walking trip, the quotidian lunchtime puzzle and too many interesting discussions on a plethora of topics, including, but in no way restricted to, mathematical ones. \\

I doubt that this thesis would have been possible without Lindy, my wife. I am also immensely grateful to her for raising our son Hamish with me, who has been the greatest diversion from and strongest motivation for my work. My parents (of both consanguine and affine type) also receive a special mention for their interest and support, and providing me with homes away from home. I am grateful to all my family and friends for putting up with me and helping me in their way.

\newpage

\chapter*{Author's declaration}

This thesis is submitted in accordance with the regulations for the degree of Doctor of Philosophy at the University of Glasgow. \\

Chapters \ref{chIntro} and \ref{chQISM} (except Section \ref{altformulae}) and Sections \ref{secsymmgroup} and \ref{secdAHA} cover some background and review existing theory.
The rest of this thesis is the author's original work unless explicitly stated otherwise.

\newpage

\tableofcontents

\newpage

\mainmatter

\numberwithin{equation}{chapter}

\chapter{Introduction}
\label{chIntro}

The quantum nonlinear Schr\"odinger (QNLS) model for a 1-dimensional bosonic gas was introduced by Lieb and Liniger \cite{LiebLi} in 1963 and has been studied extensively since, e.g. \cite{YangYang, Gaudin1983, Dorlas, Korepin, Sklyanin1989, KomoriHikami, Hikami, Gutkin1988, Gutkin1982, Gutkin1985, HeckmanOpdam1997, Gaudin1971-1,Gaudin1971-2,Gaudin1971-3,Emsiz,EmsizOS}. It describes a system of $N$
\nc{rnc}{$N$}{Number of particles}%
particles moving along a circle or an infinite line with pairwise contact interaction whose strength is determined by a constant $\gamma \in \bR$%
\nc{gcl}{$\gamma$}{Coupling constant}%
. Most of the theory deals with the repulsive case ($\gamma>0$). \\

The QNLS model was introduced \cite{LiebLi} as the first example of a parameter-dependent boson gas for which eigenstates and eigenvalues of the quantum Hamiltonian can be calculated exactly. Earlier, Girardeau \cite{Girardeau} studied a related system which does not contain a (nontrivial) parameter but which can be obtained from the Lieb-Liniger system in the limit $\gamma \to \infty$. We also remark that a free system of bosons is obtained in the limit $\gamma \to 0$, which is an important test case for results on the QNLS model. \\

Assume the particle coordinates are given by $\bm x = (x_1,\ldots,x_N) \in J^N$ 
\nc{rxl}{$\bm x = (x_1,\ldots,x_N)$}{Particle locations}%
for some closed interval $J \subset \bR$%
\nc{rjc}{$J$}{Closed interval containing permitted particle locations}%
.
In units where Planck's constant $\hbar$ equals 1 and the mass of each particle $\frac{1}{2}$, the Hamiltonian $H_\gamma$ for the QNLS model is formally given by
\begin{equation} \label{Hamiltonian}
H_\gamma = - \Delta + 2\gamma \sum_{1 \leq j < k \leq N} \delta(x_j-x_k),
\end{equation}%
\nc{rhca}{$H_{\gamma}$}{Hamiltonian for the QNLS model \nomrefeqpage}%
\nc{gdc}{$\Delta = \sum_{j=1}^N \partial_j^2$}{Laplacian operator}%
\nc{gdl}{$\delta(x)$}{Dirac delta}%
with associated eigenvalue problem
\[ H_\gamma \Psi = E \Psi, \]%
\nc{rec}{$E$}{Eigenvalue of Hamiltonian}%
for some $E \in \bR$, where $\Psi$ is an element of a yet-to-be-determined function space.
We have written $\Delta = \sum_{j=1}^N \partial_j^2$ for the Laplacian, where $\partial_j = \frac{\partial}{\partial x_j}$.
\nc{rdlva}{$\partial_{j} = \frac{\partial}{\partial x_j}$}{$j$-th partial derivative}
We emphasize that the definition \rfeqn{Hamiltonian} is entirely formal, i.e. we have not specified the domain of $H_\gamma$. We will address this in due course. \\

Despite having been studied for a long time, the QNLS model still has open questions attached to it. 
Dealing with these issues is all the more important since the QNLS model is in many ways a prototypical integrable model; it is often chosen as a test case for (new) methods.
With this thesis we hope to make some advances in the theoretical understanding of the QNLS model.

\section{Function spaces and the symmetric group}

To place the Hamiltonian \rfeqn{Hamiltonian} on a more rigorous footing, we introduce the following standard terminology for function spaces.
Given $U$ a subset of $\bR^N$ or $\bC^N$ which has a nonempty interior, let $\ca F(U)$
\nc{rfc}{$\ca F(U)$}{Vector space of functions: $U \to \bC$} 
denote the vector space of complex-valued functions on $U$ and consider its subspace
\[ \ca C(U) = \set{f \in \ca F(U)}{f \n{ is continuous}}. \]%
\nc{rccxa}{$\ca C(U)$}{Vector space of continuous functions: $U \to \bC$}%
If $U$ is open, we have the further subspaces
\begin{align*}
\ca C^k(U) &= \set{f \in \ca F(U)}{f \n{ has derivatives up to order }k\n{ which are continuous}}, \quad k \in \bZ_{> 0}, \\
\ca C^\infty(U) &= \set{f \in \ca F(U)}{f \n{ is smooth}}, \\
\ca C^\omega(U) &= \set{f \in \ca F(U)}{f \n{ is real-analytic}}, \\
\ca P(U) &= \set{f \in \ca F(U)}{f \n{ is polynomial}}.
\end{align*}%
\nc{rccxb}{$\ca C^k(U)$}{Vector space of $k$-times continuously differentiable \\ functions: $U \to \bC$}%
\nc{rccxc}{$\ca C^\infty(U)$}{Vector space of smooth functions: $U \to \bC$}%
\nc{rccxd}{$\ca C^\omega(U)$}{Vector space of real-analytic functions: $U \to \bC$}%
\nc{rpcw}{$\ca P(U)$}{Vector space of polynomial functions: $U \to \bC$}%
\vspace{-5mm}

Given that each particle's location is restricted to be in the closed interval $J \subset \bR$, 
note that the term in \rf{Hamiltonian} proportional to $\gamma$ is linked to an arrangement of hyperplanes 
\[ V_{j \, k} = \set{\bm x \in J^N}{x_j=x_k}, \qquad \n{for } j<k. \]%
\nc{rvc}{$V_{j\, k}$}{Hyperplane $\set{\bm x \in J^N}{x_j = x_k}$}%
The corresponding set of \emph{regular vectors} is given by
\[ J^N_\n{reg} = J^N \setminus \bigcup_{1 \leq j<k \leq N} V_{j\,k} = \set{\bm x \in J^N}{x_j \ne x_k \n{ if } j \ne k}. \]%
\nc{rjc}{$J^N_\n{reg}$}{Set of regular vectors}%
The \emph{alcoves} are the connected components of $J^N_\n{reg}$.
The \emph{fundamental alcove} is given by
\[ J^N_+ = \set{\bm x \in J^N}{x_1>\ldots>x_N}. \]%
\nc{rjc}{$J^N_{+}$}{Fundamental alcove}%
\vspace{-7mm}

We recall some basic facts and definitions related to the symmetric group $S_N$. %
\nc{rsca}{$S_N$}{Symmetric group}%
The transposition of two elements $j,k \in \{1,\ldots,N\}$, $i \ne j$ is denoted $s_{j \, k}$.
\nc{rsl}{$s_{j \, k}$}{Transposition swapping $j$ and $k$}
For ease of notation we will sometimes write $1 \in S_N$ as $s_{j \, j}$ for some $j=1,\ldots,N$. 
For $j=1,\ldots,N-1$ the transposition $s_{j \, j+1}$ is written $s_j$
\nc{rsl}{$s_j = s_{j \, j+1}$}{Simple transposition swapping $j$ and $j+1$} and called \emph{simple}.
$S_N$ is generated by the simple transpositions $\set{s_j}{1 \leq j \leq N-1}$ and has the corresponding \emph{presentation}
\[ S_N = \left\langle s_1,\ldots,s_{N-1} | s_j^2=1, \, s_js_{j+1}s_j=s_{j+1}s_js_{j+1}, \, s_j s_k = s_k s_j \, \n{ if }  |j-k|>1 \right\rangle. \]
$S_N$ has an obvious (left) action on $J^N$ defined by $w (x_1,\ldots, x_N) = (x_{w 1},\ldots, x_{w N})$, where $(x_1,\ldots,x_N) \in J^N$ and $w \in S_N$. %
\nc{rwl}{$w$}{Arbitrary element of $S_N$}%
In fact $S_N$ acts as a Weyl group associated to the collection of hyperplanes $V_{j \, k}$, i.e. the transpositions $s_{j \, k}$ act as reflections in the hyperplane $V_{j \,k}$, which are isometries with respect to the standard Euclidean inner product defined by
\[ \inner{\bm x}{\bm y} := \sum_{j=1}^N x_j y_j \qquad \n{ for } \bm x,\bm y \in J^N, \]
in other words $\inner{w \bm x}{\bm y} = \inner{\bm x}{w^{-1} \bm y}$ for all $w \in S_N$ and all $\bm x,\bm y \in J^N$.
Furthermore, $J^N_\n{reg}$ is an invariant subset under the above action, and $S_N$ also acts on the collection of hyperplanes and on the collection of alcoves. The latter action is transitive, so that $J^N_\n{reg} = \bigcup_{w \in S_N} w J^N_+$. 
Also, there is a left action of $S_N$ on $\ca F(J^N)$, defined by 
\[ (w f)(\bm x) = f(w^{-1} \bm x), \] 
for $w \in S_N$, $f \in \ca F(J^N) \to \bC$, and $\bm x \in J^N$.
The vector space $\ca C(J^N)$ is an invariant subset under this action, and in the case $J=\bR$, so are
$\ca C^k(\bR^N)$, $\ca C^\infty(\bR^N)$, $\ca C^\omega(\bR^N)$ and $\ca P(\bR^N)$. 
We note that $S_N$ acts on the sets obtained from the above by replacing $J^N$ by $J^N_\n{reg}$ or $\bC^N$.
Given $z \in \bC$, we denote by $\bar z$ its complex conjugate.
The corresponding complex Euclidean inner product is defined by
\[ \inner{\bm z}{\bm z'} := \sum_{j=1}^N z_j \bar z'_j \qquad \n{ for } \bm z,\bm z' \in \bC^N, \]
and satisfies $\inner{w \bm z}{\bm z'} = \inner{\bm z}{w^{-1} \bm z'}$ for all $w \in S_N$ and $\bm z,\bm z' \in \bC^N$.
If $X$ is a set acted upon by $S_N$, then $X^{S_N}$ denotes the subset of elements of $X$ that are left fixed by $S_N$.

\section{The QNLS problem revisited; the coordinate Bethe ansatz}

It is well-known \cite{LiebLi,Dorlas,Gutkin1988} that the eigenvalue problem of the Hamiltonian $\rf{Hamiltonian}$ should be interpreted as the system of equations
\begin{align}
-\Delta \Psi|_{J^N_\n{reg}} &= E \Psi|_{J^N_\n{reg}}, \label{QNLS1} \\
\left(\partial_j-\partial_k\right)\Psi|_{V_{j \, k}^+}-\left(\partial_j-\partial_k\right)\Psi|_{V_{j \, k}^-} &=  2 \gamma \Psi|_{V_{j \, k}}, \qquad \n{for } 1 \leq j<k \leq N, \label{QNLS2} 
\end{align}
for a $\Psi \in \ca C(J^N)^{S_N}$ whose restriction to $J^N_\n{reg}$ is twice continuously differentiable.
Here we have used the notation
\begin{equation} \label{Vjkplusmin} f|_{V_{j \, k}^\pm} = \lim_{x_k \to x_j \atop x_j \gtrless x_k} f \end{equation}%
\nc{rvcz}{$V_{j \, k}^\pm$}{Indicates a limit $x_k \to x_j$ is taken with $x_j \gtrless x_k$ \nomrefeqpage}%
for $f \in \ca C(J^N)$.
The equations \rf{QNLS2} are called the \emph{derivative jump conditions}. \\

Furthermore, if $J$ is bounded without loss of generality we may assume that $J=[-L/2,L/2]$, for some $L \in \bR_{> 0}$. %
\nc{rlca}{$L$}{Length of bounded interval $J$}%
In this case for the Hamiltonian to be (formally) self-adjoint it is necessary to apply the following boundary conditions to $\Psi$ and its derivative:
\begin{equation} \label{QNLSperiodic}
\begin{aligned}
\Psi(\bm x)|_{x_j=-L/2} &= \Psi(\bm x)|_{x_j=L/2}, \\
\partial_j \Psi(\bm x)|_{x_j=-L/2} &= \partial_j \Psi(\bm x)|_{x_j=L/2},
\end{aligned}
\qquad \n{for } j =1,\ldots,N.
\end{equation}

Because $\Psi$ is $S_N$-invariant, it is sufficient to impose the conditions
\begin{align}
-\Delta \Psi|_{J^N_+} &= E \Psi|_{J^N_+}, \label{QNLS1'} \\
\left(\partial_j-\partial_k\right)\Psi|_{V_{j \, j+1}^+} &= \gamma \Psi|_{V_{j \, j+1}}, \qquad \n{for } j=1,\ldots, N-1, \label{QNLS2'} 
\end{align}
where we may now take $\Psi \in \ca C(\overline{J^N_+})$ with twice continuously differentiable restriction to $J^N_+$. The conditions \rf{QNLSperiodic} can also be simplified:
\begin{equation} \label{QNLSperiodic'}
\begin{aligned}
\Psi(L/2,x_1,\ldots,x_{N-1}) &= \Psi(x_1,\ldots,x_{N-1},-L/2), \\
\partial_x \Psi(x,x_1,\ldots,x_{N-1})|_{x=L/2} &= \partial_x \Psi(x_1,\ldots,x_{N-1},x)|_{x=-L/2}.
\end{aligned}
\end{equation}

Lieb and Liniger \cite{LiebLi} solved this system by modifying an approach which Bethe used to analyse the one-dimensional Heisenberg model \cite{Bethe}. This method is now known as the \emph{(coordinate) Bethe ansatz}. 
Write $\e^{\ii \bm \lambda} \in \ca C^\omega(J^N)$ for the \emph{plane wave} with wavenumbers given by $\bm \lambda = (\lambda_1,\ldots,\lambda_N) \in \bC^N$\nc{gll}{$\bm \lambda = (\lambda_1,\ldots,\lambda_N)$}{Vector of wavenumbers}, i.e. the function defined by
\begin{equation}
\e^{\ii \bm \lambda}(\bm x) = \e^{\ii \inner{\bm \lambda}{\bm x}} = \e^{\ii \sum_{j=1}^N \lambda_j x_j}.
\end{equation}%
\nc{rel}{$\e^{\ii \bm \lambda}$}{Plane wave with wavenumbers given by $\bm \lambda$ \nomrefeqpage}%
This Bethe ansatz results in the statement that the \emph{Bethe wavefunction} $\Psi_{\lambda_1,\ldots,\lambda_N} \in \ca C(J^N)$ given by
\begin{equation} \label{Bethewavefn} \Psi_{\lambda_1,\ldots,\lambda_N}(\bm x) = \frac{1}{N!} \sum_{w \in S_N} \left( \prod_{1 \leq j<k \leq N} \frac{\lambda_{w j}-\lambda_{w k}-\ii \gamma}{\lambda_{w j}-\lambda_{w k}} \right) \e^{\ii \sum_{j=1}^N \lambda_{w j}x_j}, \qquad \n{for } \bm x \in J^N_+,  
\end{equation}%
\nc{gyc}{$\Psi_{\bm \lambda}$}{Bethe wavefunction \nomrefeqpage}%
solves the QNLS problem, i.e. it satisfies \rfeqnser{QNLS1}{QNLS2} with $E=\sum_{j=1}^N \lambda_j^2$.
Furthermore, if $\lambda_1,\ldots,\lambda_N$ are distinct real numbers satisfying the \emph{Bethe ansatz equations} (BAEs), viz.
\begin{equation} \label{BAEintro} \e^{\ii \lambda_j L} = \prod_{k=1 \atop k \ne j}^N \frac{\lambda_j-\lambda_k + \ii \gamma}{\lambda_j-\lambda_k - \ii \gamma}, \qquad \n{for } j =1,\ldots, N, 
\end{equation}
then in addition $\Psi_{\bm \lambda}$ satisfies Eqns. \rf{QNLSperiodic}. 

\begin{exm}[$N=2$] \label{Bethewavefn2}
We present the explicit expression of the Bethe wavefunction for $N=2$ and $J=[-L/2,L/2]$ for some $L>0$.
On the fundamental alcove $\set{(x_1,x_2) \in J^2}{x_1>x_2}$ we have
\[ \Psi_{\lambda_1,\lambda_2}(x_1,x_2) = \frac{1}{2} \left( \frac{\lambda_1\! - \!\lambda_2\! - \!\ii \gamma}{\lambda_1\! - \!\lambda_2} \e^{\ii (\lambda_1 x_1+\lambda_2 x_2)} + \frac{\lambda_1\! - \!\lambda_2+ \ii \gamma}{\lambda_1\! - \!\lambda_2} \e^{\ii (\lambda_2 x_1+\lambda_1 x_2)}\right); \]
the reader may verify that this leads to the following expressions for arbitrary $(x_1,x_2) \in J^2$:
\begin{align*} 
\Psi_{\lambda_1,\lambda_2}(x_1,x_2) &= \frac{1}{2} \left( \frac{\lambda_1\! - \!\lambda_2\! - \! \sgn(x_1\! - \!x_2) \ii \gamma}{\lambda_1\! - \!\lambda_2} \e^{\ii (\lambda_1 x_1+\lambda_2 x_2)} + \right. \\
& \hspace{10mm} \left. + \frac{\lambda_1\! - \!\lambda_2+ \sgn(x_1\! - \!x_2) \ii \gamma}{\lambda_1\! - \!\lambda_2} \e^{\ii (\lambda_2 x_1+\lambda_1 x_2)}\right).
\end{align*}
This function satisfies the system \rfeqnser{QNLS1}{QNLS2}, which in this case reads:
\begin{align*}
-(\partial_1^2+\partial_2^2)\Psi_{\lambda_1,\lambda_2}(x_1,x_2) = (\lambda_1^2+\lambda_2^2) \Psi_{\lambda_1,\lambda_2}(x_1,x_2), \qquad & \n{if } -L/2 < x_1 \ne x_2 < L/2, \\
\lefteqn{\hspace{-100mm} \lim_{x_2 \to x_1 \atop x_1 > x_2} (\partial_1-\partial_2) \Psi_{\lambda_1,\lambda_2}(x_1,x_2)-
\lim_{x_2 \to x_1 \atop x_1 < x_2} (\partial_1-\partial_2) \Psi_{\lambda_1,\lambda_2}(x_1,x_2) =
2 \gamma \Psi_{\lambda_1,\lambda_2}(x_1,x_1), }\\[-4mm]
& \n{if } -L/2<x_1<L/2.
\end{align*}
Furthermore, if $\frac{\lambda_1-\lambda_2+\ii \gamma}{\lambda_1-\lambda_2-\ii \gamma} = \e^{\ii \lambda_1 L} = \e^{-\ii \lambda_2 L}$, then we have
\begin{align*} 
\Psi_{\lambda_1,\lambda_2}(L/2,x)&=\Psi_{\lambda_1,\lambda_2}(-L/2,x) & 
\Psi_{\lambda_1,\lambda_2}(x,L/2)&=\Psi_{\lambda_1,\lambda_2}(x,-L/2) \\
(\partial_1 \Psi_{\lambda_1,\lambda_2})(L/2,x) &= (\partial_1 \Psi_{\lambda_1,\lambda_2})(-L/2,x) &
(\partial_2 \Psi_{\lambda_1,\lambda_2})(x,L/2) &= (\partial_2 \Psi_{\lambda_1,\lambda_2})(x,-L/2).
 \end{align*}
We will keep returning to the $N=2$ case throughout the thesis.
\end{exm}

\section{Quantum inverse scattering method}

An important solution technique for the QNLS model, and an example of a method using the QNLS model as a test case, has been the \emph{quantum inverse scattering method} (QISM).
The QISM is a quantized version of the classical inverse scattering method, a technique used for solving certain nonlinear partial differential equations, and was developed by the Faddeev school \cite{Faddeev1995,Sklyanin1982,SklyaninFaddeev,SklyaninTakhtajanFaddeev,KorepinBI} after Baxter's pioneering work on exactly solvable models in statistical mechanics and his method of commuting transfer matrices; see \cite{Baxter1989} for a text book account and references therein. \\

The application of the QISM to the QNLS model will be reviewed in Chapter \ref{chQISM}.
We make some general remarks here.
The QISM revolves around the so-called \emph{monodromy matrix} $\ca T_\lambda = {\tiny \begin{pmatrix} A_\lambda & \gamma B_\lambda  \\ C_\lambda & D_\lambda  \end{pmatrix}}$, a parameter-dependent 2x2-matrix whose entries are operators on the relevant state space (in the case of the QNLS model, the Fock space, which is the direct sum of all spaces $\L^2(J^N)^{S_N}$ of symmetrized square-integrable functions on the finite box $J^N$). 
$\ca T$ satisfies the Yang-Baxter equation related to the Yangian of $\f{gl}_2$ whence commutation relations are obtained for its entries.
The algebra generated by $A_\lambda, B_\lambda , C_\lambda, D_\lambda $ is called the \emph{Yang-Baxter algebra}.
The relevance of the monodromy matrix to physical models lies in the fact that the \emph{transfer matrices} $T_\lambda=A_\lambda+D_\lambda $ form a self-adjoint commuting family and commute with (in fact are generating functions for) the integrals of motion, including the Hamiltonian. It can be shown that under certain conditions on the $\lambda_j$ (also called \emph{Bethe ansatz equations}), the functions 
\begin{equation} \label{PsiQISMintro} \Psi_{\lambda_1,\ldots,\lambda_N} := B_{\lambda_N} \ldots B_{\lambda_1} \Vac \end{equation}
are eigenfunctions of the transfer matrix, and hence of the Hamiltonian, where $\Vac$ is a reference state (in the QNLS case it is $1 \in \L^2(J^0) \cong \bC$). 

\section{Root systems and affine Hecke algebras}

An important contribution by Gaudin \cite{Gaudin1971-3} was the realization that the Lieb-Liniger system can be naturally generalized in terms of (classical) crystallographic root systems. These generalizations can also be solved by the Bethe ansatz approach and have been the subject of further study \cite{Gutkin1982,GutkinSutherland,HeckmanOpdam1997,EmsizOS,Emsiz}. It has been highlighted by Heckman and Opdam \cite{HeckmanOpdam1997} that representations of certain degenerations of affine Hecke algebras play an essential role, providing the second method for solving the QNLS problem. \\

In particular, a \emph{non-symmetric} function $\psi_{\lambda_1,\ldots,\lambda_N}$ can be constructed that solves the QNLS eigenvalue problem \rfeqnser{QNLS1}{QNLS2}. Upon symmetrization one recovers the Bethe wavefunction $\Psi_{\lambda_1,\ldots,\lambda_N}$; hence we will refer to $\psi_{\lambda_1,\ldots,\lambda_N}$ as the \emph{pre-wavefunction}.
This function $\psi_{\lambda_1,\ldots,\lambda_N}$ was introduced into the theoretical picture of the QNLS model by Komori and Hikami \cite{KomoriHikami,Hikami} in analogy to the non-symmetric Jack polynomials in the Calogero-Sutherland-Moser model \cite{Polychronakos,BernardGHP,Hikami1996}. Furthermore, Hikami \cite{Hikami} has made clear the connection with Gutkin's \emph{propagation operator} or \emph{intertwiner} \cite{Gutkin1982, EmsizOS}, which intertwines two representations of the relevant degeneration of the affine Hecke algebra. 
We will review this approach in Chapter \ref{chdAHA} in more detail. 

\section{Norm formulae, completeness and integrability}

In general, the norms of quantum-mechanical wavefunctions are important for the calculation of probabilities because of the following reason. 
If $\Psi$ is a wavefunction of a quantum-mechanical system whose $\L^2$-norm $\| \Psi \|$ equals 1, then $| \Psi(\bm x) |^2$ can be interpreted as the probability density of finding the quantum system at location $\bm x$. 
If $\Psi$ is an arbitrary wavefunction, then $\tilde \Psi = \Psi / \| \Psi \|$ is normalized and the aforementioned notion of probability density can be assigned to $| \tilde \Psi(\bm x) |^2 = | \Psi(\bm x)|^2 / \| \Psi \|^2$. \\

As for the Bethe wavefunctions, their norms were conjectured by Gaudin \cite{Gaudin1983} to be given by so-called determinantal formulae. 
This was proven by Korepin \cite{Korepin} using QISM techniques. Emsiz \cite{Emsiz} has conjectured similar formulae for the norms of the eigenfunctions associated with more general crystallographic root systems which have been verified for systems of small rank \cite{BustamanteVDDlM}. \\

An important and difficult problem has been to determine the completeness of the Bethe wavefunctions. The spectrum of the Hamiltonian changes drastically depending on the type of interaction (repulsive or interactive) and the geometry (line or circle). 
In the attractive case, where $\gamma < 0$, bound states occur (due to multi-particle binding), leading to a mixed spectrum. For the system on the line, completeness of the Bethe eigenstates was shown by Oxford \cite{Oxford}. For the corresponding system on the circle this question is still an open problem. \\

For the repulsive case, the completeness of the Bethe wavefunctions in $\L^2(\bR^N)^{S_N}$, where the Hamiltonian has a purely continuous spectrum, was proven by Gaudin \cite{Gaudin1971-1,Gaudin1971-2, Gaudin1983}. 
For the corresponding problem on a bounded interval $J$ (i.e. the system of quantum particles on a circle), completeness (in $\L^2(J^N)^{S_N}$) and orthogonality of the set of eigenfunctions of the QNLS Hamiltonian, viz. $\set{\Psi_{\bm \lambda}}{\bm \lambda \n{ satisfies the BAEs } \rf{BAEintro}}$, were proven by Dorlas \cite{Dorlas} using completeness of the plane waves, a continuity argument at $\gamma=0$ and QISM techniques. To do this, it is crucial to specify the right domain of $H_\gamma$, so that it becomes essentially self-adjoint (i.e. its closure is self-adjoint).\\

The QNLS model is a \emph{quantum integrable system}, by which we mean that infinitely many integrals of motion (conserved charges) exist: operators on $\L^2(J^N)^{S_N}$ which are simultaneously diagonalized by the solutions of \rfeqnser{QNLS1}{QNLS2}, i.e. the Bethe wavefunctions. Because of completeness of the Bethe wavefunctions, these operators mutually commute.

\section{Experimental construction of the QNLS model}

So far we have discussed the rich mathematical structures of the QNLS model. Its physical significance has been demonstrated by recent experiments \cite{Amerongen,AmerongenEWKD} where systems described by the QNLS model, consisting of magnetically confined ultra-cold rubidium atoms, have been manufactured. These atoms are trapped using the magnetic field generated by electrical currents on a microchip; because of the low temperature, the movement of the atoms is reduced and an effectively one-dimensional system is created. Theoretical advances in the QNLS model may be of importance to such experiments and any new technology that arises out of them: by virtue of integrability it may be possible to obtain exact data which can help to calibrate equipment.  

\section{Present work}

The general problem that this thesis aims to address is the disparateness of the QISM and the Hecke algebra approach. Each has their own advantages; the QISM yields recursive relations for the Bethe wavefunctions whereas the degenerate affine Hecke algebra can be immediately generalized to other reflection groups. 
However, since both methods solve the QNLS problem there should be connections, and we will highlight some of them. \\

In particular, we will focus on the pre-wavefunctions $\psi_{\lambda_1,\ldots,\lambda_N}$ and demonstrate that they are more important to the theory of the QNLS model than previously thought. 
In \rft{brecursion} we will prove that the pre-wavefunctions can be generated by operators $b^\pm_\mu: \ca C(\bR^N) \to \ca C(\bR^{N+1})$ (for which we will give explicit formulae) as follows:
\[ \psi_{\lambda_1,\ldots,\lambda_N} = b^-_{\lambda_N}  \ldots b^-_{\lambda_1} \Vac = b^+_{\lambda_1}  \ldots b^+_{\lambda_N}  \Vac, \]
where $(\lambda_1, \ldots, \lambda_N) \in \bC^N$.
In particular, we remark that the pre-wavefunctions are defined using the affine Hecke algebra method, but satisfy a QISM-type relation (cf. \rfeqn{PsiQISMintro}); this ties these two solution methods for the QNLS model more closely together. 
Furthermore, we obtain the relation 
\[ B_\lambda  \ca S^{(N)} = \ca S^{(N+1)} b^\pm_\lambda , \]
where $\ca S^{(N)} = \frac{1}{N!} \sum_{w \in S_N} w$, leading to a new proof that $B_{\lambda_N} \ldots B_{\lambda_1} \Vac$ indeed is the Bethe wavefunction.
These operators $b^\pm_\mu$ can also be seen as operators densely defined on non-symmetric\footnote{Or rather, ``not-necessarily-symmetric''.} Fock space, the direct sum of all the spaces $\L^2(J^N)$ for a bounded interval $J$, on which we can define further operators $a_\mu, c^\pm_\mu , d_{\mu}$. These, together with $\ca S^{(N+1)} b^\pm_\mu$, restrict to the QISM operators $A_\mu, B_\mu , C_\mu , D_\mu $ and we will highlight similar commutation relations that they satisfy, prompting the concept of a \emph{non-symmetric Yang-Baxter algebra}.

\section{Outline of thesis}

In Chapters \ref{chQISM} and \ref{chdAHA} we will review the main theories in existence that solve the QNLS model. In both chapters we present some original work.
Chapter \ref{chInterplay} is a short chapter highlighting some connections between these two solution methods, which is largely original work. The main body of original work, the theory surrounding the operators $a_\mu, b^\pm_\mu, c^\pm_\mu , d_{\mu}$, is found in Chapter \ref{ch5}. 
Finally, in Chapter \ref{chSummary} we provide some concluding remarks.\\

There are two appendices with detailed calculations, to which will be referred in Chapters 3, 4 and 5 where needed. This is followed by a list of symbols on page \pageref{listofsymbols} and a list of references on page \pageref{references}.


\newpage

\numberwithin{equation}{section}

\chapter[The quantum inverse scattering method]{The quantum inverse scattering method (QISM)} \label{chQISM}

In this chapter we recall that the QNLS system of $N$ spinless quantum particles on a line, a line segment or a circle formally corresponds to the $N$-particle sector of a bosonic nonrelativistic quantum field theory which can be studied with the aid of the \emph{quantum inverse scattering method} (QISM), also known as the \emph{algebraic Bethe ansatz} (ABA). This is a quantum version of the (classical) inverse scattering methods and was introduced and developed by the (then) Leningrad school led by Faddeev \cite{Faddeev1995, Sklyanin1982, SklyaninFaddeev, SklyaninTakhtajanFaddeev,KorepinBI}, which was preceded by Baxter's seminal work \cite{Baxter1989} on exactly solvable models in statistical mechanics. It also contributed to the development of quantum groups. 
Besides the work of the Faddeev school specific to the QNLS model \cite{Sklyanin1989, Korepin} we should also mention Gutkin's exposition \cite{Gutkin1988}.  
For a mathematical background, we refer the reader to \cite{ReedSimon,ReedSimon2}. \\

We will start off by defining various Hilbert spaces and then briefly review the quantum field-theoretic context of the QNLS Hamiltonian and the quantum inverse scattering method, centred around the so-called monodromy matrix, Its entries are the generators of a spectrum generating algebra (Yang-Baxter algebra) of the Hamiltonian and their commutation relations are encoded in the famous Yang-Baxter equation.
We will discuss the use of the Yang-Baxter algebra in constructing eigenfunctions of the QNLS Hamiltonian and highlight some further properties of the Yang-Baxter algebra, providing a connection to Yangians. Having reviewed the existing theory, we present new integral formulae for the entries of the monodromy matrix which we believe have computational advantages.

\section{Hilbert spaces}

We will consider systems of quantum particles whose movement is restricted to a one-dimensional set (an infinite line, a finite line segment or a circle). 
In quantum mechanics, the possible states of these systems are described by wavefunctions, which are elements of a Hilbert space. 
Let $J \subset \bR$ be a closed interval. The reader should keep in mind two cases: $\bR$ itself and an interval of some finite length $L>0$.
In any event, for $N$ a nonnegative integer, consider 
\[ \L^2(J^N) = \set{f: J^N \to \bC}{\int_{J^N} \dd^N \bm x |f(\bm x)|^2 < \infty}, \]%
\nc{rlcw}{$\L^2(J^N)$}{Hilbert space of square-integrable functions: $J^N \to \bC$}%
the set of square-integrable functions on $J^N$, which is a Hilbert space with respect to the inner product\footnote{For $\ca H_N(J)$ it is also common to use the inner product defined by $\frac{1}{N!} \int_{J^N} \dd^N \bm x \, f(\bm x) \overline{g(\bm x)}$. However, by not including the $1/N!$ factor we may express certain adjointness relations more easily.}
\begin{equation} \label{Hilbertinnerprod} \innerrnd{f}{g}_N = \int_{J^N} \dd^N \bm x \, f(\bm x) \overline{g(\bm x)}, \end{equation}
for $f,g \in \L^2(J^N)$.
The corresponding norm is denoted by $\| \cdot \|_N$. 

\begin{defn}
Let $N$ be a nonnegative integer and $J \subset \bR$ a closed interval.
The \emph{(non-symmetric) $N$-particle sector} and the \emph{symmetric $N$-particle sector} are the two Hilbert spaces 
\[ \f h_N(J) =  \L^2(J^N), \qquad \ca H_N(J) = \f h_N(J)^{S_N} \subset \f h_N(J),\]%
\nc{rhl}{$\f h_{N} = \f h_N(J)$}{Non-symmetric $N$-particle sector $\L^2(J^N)$}%
\nc{rhcx}{$\ca H_{N} = \ca H_N(J)$}{Symmetric $N$-particle sector $\L^2(J^N)^{S_N}$}%
respectively. Note that $\f h_0(J) = \ca H_0(J) = \bC$, which is spanned by the constant $1\in \f h_0(J) = \ca H_0(J)$, i.e. the constant function: $J^0 \to \bC$ with value 1, which we will denote by $\Vac$ and refer to as the \emph{pseudo-vacuum}.
\nc{gyc}{$\Vac=1 \in \f h_0(J)$}{Pseudo-vacuum}
\end{defn}

\begin{rem}
The established QISM for the QNLS is formulated in terms of the symmetric $N$-particle sectors defined above. However, since the non-symmetric $N$-particle sectors can be seen as an intermediate step in the construction of the $N$-particle sectors, we discuss them here as well. In Chapter \ref{ch5} they will play a more central role.
\end{rem}

\boxedenv{
\begin{notn}[Symmetric and non-symmetric objects] \label{notn1}
In the rest of this thesis, whenever there is a pair of objects (i.e. sets, functions on those sets, operators acting on such functions) of which one is the symmetrized counterpart of the other, we will write the symmetric object with a capital letter and the more general, not necessarily symmetric, object with a lower-case letter, as we have done already for $\f h_N(J)$ and $\ca H_N(J)$.\end{notn} 
}

\begin{rem}
The $N$-particle sectors are state spaces for several particles moving along $J$.
In particular, the symmetric $N$-particle sector $\ca H_N(J)$ is the state space of a system of indistinguishable particles or bosons, i.e. quantum particles whose wavefunction is invariant under exchange of coordinates 
(or more generally, invariant under exchange of ``quantum numbers'', i.e. eigenvalues of ``observables'', certain formally self-adjoint operators on the state space).
This concept is the reason for the use of the symmetric group in the definition of $\ca H_N(J)$. \end{rem}

Let $w \in S_N$. Note that $\innerrnd{wf}{g}_N = \innerrnd{f}{w^{-1}g}_N$, for all $f, g \in \f h_N(J)$.
Hence, the inner product on $\ca H_N(J)$ satisfies
\[ \innerrnd{F}{G}_N = N! \int_{J^N_+} \dd^N \bm x \, F(\bm x) \overline{G(\bm x)}, \]
where we recall the fundamental alcove $J^N_+= \set{\bm x \in J^N}{x_1 > \ldots > x_N}$.\\

Consider the direct sum of all non-symmetric $N$-particle sectors $\bigoplus_{N \geq 0} \f h_N(J)$. 
Formally define an inner product on this space as follows
\[ \innerrnd{f_0 \oplus f_1 \oplus f_2 \oplus \ldots}{g_0 \oplus g_1 \oplus g_2 \oplus \ldots} := \innerrnd{f_0}{g_0}_0 + \innerrnd{f_1}{g_1}_1 + \innerrnd{f_2}{g_2}_2 + \ldots, \]
where for each $N \in \bZ_{\geq 0}$, $f_N,g_N \in \f h_N(J)$.
We emphasize that this does not properly define an inner product as the infinite sum may not converge.
As usual we may also (formally) define the associated norm $\| f \| = \sqrt{\innerrnd{f}{f}}$ for $f \in \bigoplus_{N \geq 0} \f h_N(J)$; we will also formally consider the inner product given by this formula and the corresponding norm on the direct sum of symmetric $N$-particle sectors, which is a subspace. 
Finally, we note that the restriction of the formal inner product $\innerrnd{}{}$ to any $N$-particle sector yields the existing inner product $\innerrnd{}{}_N$, so that we may denote all these inner products simply by $\innerrnd{}{}$ if convenient.
We are now ready for

\begin{defn}
The \emph{(non-symmetric) Fock space} and the \emph{symmetric Fock space} are given by
\[ \f h(J) = \set{f \in \bigoplus_{N \geq 0} \f h_N(J)}{\| f \| < \infty}, \qquad \ca H(J) = \set{F \in \bigoplus_{N \geq 0} \ca H_N(J)}{\| F \| < \infty} \subset \f h(J), \]%
\nc{rhl}{$\f h = \f h(J)$}{Non-symmetric Fock space $\oplus_{N \geq 0} \f h_N(J)$}%
\nc{rhcx}{$\ca H = \ca H(J)$}{Symmetric Fock space $\oplus_{N \geq 0} \ca H_N(J)$}%
respectively.
These are Hilbert spaces with respect to the inner product $\innerrnd{}{}$ defined above and they contain the dense subspaces of \emph{finite vectors} (cf. \cite[Prop. 6.2.2]{Gutkin1988})
\[ \f h_\n{fin}(J) = \set{f \in \f h(J)}{\exists M \in \bZ_{\geq 0}: \, f \in \bigoplus_{N=0}^M \f h_N(J)}, \qquad \ca H_\n{fin}(J) = \f h_\n{fin}(J) \cap \ca H. \] 
\nc{rhl}{$\f h_{\n{fin}} = \f h_\n{fin}(J)$}{Dense subspace of $\f h(J)$ of finite vectors}%
\nc{rhcx}{$\ca H_{\n{fin}} = \ca H_\n{fin}(J)$}{Dense subspace of $\ca H(J)$ of finite vectors}%
\end{defn}

\vspace{-5mm}
\begin{rem}
The Fock spaces allows for linear combinations of elements from different $N$-particle sectors (i.e. superpositions of states containing a different number of particles).
\end{rem}

We note that the $N$-particle sectors $\f h_N(J)$ and $\ca H_N(J)$, and also the Fock spaces $\f h(J)$ and $\ca H(J)$ are \emph{separable} Hilbert spaces, i.e. admit countable orthonormal bases. 

\begin{rem}
Essentially all Hilbert spaces used in physical theories are separable.
\end{rem}

\subsection{Test functions}

In our case, the $N$-particle sectors $\f h_N(J)$ and $\ca H_N(J)$ contain dense subspaces 
\[ \f d_N(J) := \ca C_\n{cpt}^\infty(J^N), \qquad \ca D_N(J) = \ca C_\n{cpt}^\infty(J^N)^{S_N} = \f d_N(J) \cap \ca H_N,\]%
\nc{rdlx}{$\f d_{N} = \f d_N(J)$}{Vector space of non-symmetric test functions $\ca C_\n{cpt}^\infty(J^N)$}%
\nc{rdcx}{$\ca D_{N} = \ca D_N(J)$}{Vector space of symmetric test functions $\ca C_\n{cpt}^\infty(J^N)^{S_N}$}%
respectively, 
consisting of \emph{test functions} on $J^N$, i.e. smooth functions from $J^N$ to $\bC$ with compact support.
We form the spaces
\[ \f d(J) := \bigoplus_{N \geq 0} \f d_N(J), \qquad \ca D(J) = \bigoplus_{N \geq 0} \ca D_N(J) = \f d(J) \cap \ca H(J),  \]%
\nc{rdlx}{$\f d = \f d_N$}{Vector space of non-symmetric test functions $\bigoplus_{N \geq 0} \f d_N(J)$}%
\nc{rdcx}{$\ca D = \ca D_N$}{Vector space of symmetric test functions $\bigoplus_{N \geq 0} \ca D_N(J)$}%
which are dense proper subsets of $\f h(J)$ and $\ca H(J)$, respectively.

\begin{rem} 
These dense subspaces are used to define (possibly unbounded) operators on $\f h$ and $\ca H$. 
In physics one is interested in Hamiltonians, certain (formally) self-adjoint operators with real unbounded spectrum, as the eigenvalues are interpreted as ``energy''.
Operators defined everywhere on an infinite-dimensional Hilbert space are necessarily bounded by virtue of the Hellinger-Toeplitz theorem (see e.g. \cite{ReedSimon}).
Hence, we must allow for operators that are only defined on dense subsets.
\end{rem}

Densely defined operators $T$ on $\f h(J)$ have a uniquely defined \emph{formal adjoint} $T^*$, which is an operator on $\f h(J)$ satisfying $\innerrnd{f}{T^*g} = \innerrnd{Tf}{g}$ for all $f,g\in \f d(J)$. Note that this formal adjoint may not be densely defined (it may even have trivial domain). Similarly we may consider formal adjoints of operators densely defined on $\ca H(J)$.\\

For all aforementioned Hilbert spaces $\ca X = \f h_N, \ca H_N, \f h, \ca H, \f d_N, \ca D_N, \f d, \ca D$ we will use the shorthand notation $\ca X = \ca X(J)$ if $J$ is clear from the context. 

\subsection{Periodicity and boundary conditions} \label{periodicity}

The $N$-particle sector of the Fock space associated to the system on a circle of circumference $L$ is obtained by imposing $L$-periodicity on the elements of the $N$-particle sector $\f h_N(\bR)$ and hence given by $\L^2 \left(\left(\bR/L \bZ \right)^N \right)$,
i.e. the Hilbert space of (symmetrized) $L$-periodic square-integrable functions$: \bR^N \to \bC$. 
Alternatively, we may define these $N$-particle sectors in terms of the $N$-particle sectors for the bounded interval $J=[-L/2,L/2]$ and use the Hilbert space
\[ \set{f \in \f h_N}{\forall \bm x \in J^N, \, \forall j, \, f(\bm x)|_{x_j=-L/2} = f(\bm x)|_{x_j=L/2}}; \]
note that this does not depend on our particular choice\footnote{Another common choice is $J=[0,L]$. The choice $J=[-L/2,L/2]$ has the benefit that, again, certain adjointness relations can be expressed more conveniently and that at least formally the system with $J=\bR$ is recovered in the limit $L \to \infty$.} of $J$.
The same discussion applies to the subspaces of symmetric functions.\\


Because of the physical interpretation of elements of a Hilbert space $\ca H$ as probability amplitudes in quantum mechanics (i.e. their squared absolute values are probability densities for quantum particles; this is also the reason why $\L^2$ functions are used), we are only interested in the subspace $\set{F \in \ca H}{F|_{\partial \Dom(F)} = 0}$, i.e. the set of functions which vanish at the boundary of their domain. This is a closed linear subspace, and hence a Hilbert space in its own right.
\begin{itemize}
\item For $J=\bR$, $F \in \ca H_N(J) = \L^2(\bR^N)$ implies $\lim_{x_j \to \infty} F(\bm x)=0$ and hence all elements of the Hilbert space are physically meaningful. 
\item However, in the Hilbert space $\ca H_N([-L/2,L/2])$ with $L \in \bR_{>0}$, this subspace is 
\[ \set{F \in \ca H_N([-L/2,L/2])}{F|_{x_j = \pm L/2} = 0 \n{ for all } j=1,\ldots,N},\]
a subspace of the Hilbert space of $L$-periodic functions. 
It turns out that in this case it is often more natural to study the subspace of $\ca H(J)$ of $L$-periodic functions.
\end{itemize}

\section{Quantum field theory} \label{seQFThy}

In the standard setup of the QISM \cite{Faddeev1995,KorepinBI} the Hamiltonian and the monodromy matrix for the QNLS model are introduced as expressions in terms of quantum field operators. We will briefly highlight some basic notions of non-relativistic quantum field theory (the formalism known as ``second quantization'').\\

\boxedenv{
\begin{notn}[Deleting and appending single variables] \label{notn2} Let $N$ be a nonnegative integer.
\begin{itemize}
\item Let $\bm x = (x_1,\ldots,x_N) \in J^N$ and $y \in J$. Then $(\bm x,y) = (x_1,\ldots,x_N,y) \in J^{N+1}$ and $(y,\bm x) = (y,x_1,\ldots,x_N) \in J^{N+1}$. 
\item Let $\bm x = (x_1,\ldots,x_{N+1}) \in J^{N+1}$ and $j=1,\ldots,N+1$. Then  
\[\bm x_{\hat \jmath} =  (x_1,\ldots,\hat x_j,\ldots,x_{N+1}) = (x_1,\ldots,x_{j-1},x_{j+1},\ldots,x_{N+1}) \in J^N.\]
\end{itemize}
\end{notn}
}

For each $y \in \bR$ two \emph{quantum fields} $\Phi(y), \Phi^*(y)$ are introduced as $\End(\ca D)$-valued distributions,  given by $\Phi(y)(\ca D_0)=0$ and 
\begin{equation} \label{quantumfields} \begin{aligned}
\left( \Phi(y) F \right) (\bm x) &= \sqrt{N+1}F(\bm x,y), && \n{for } F \in \ca D_{N+1}, \bm x \in J^N,  \\
\left( \Phi^*(y) F \right) (\bm x) &= \frac{1}{\sqrt{N+1}} \sum_{j=1}^{N+1} \delta(y-x_j) F(\bm x_{\hat \jmath}), && \n{for } F \in \ca D_N, \bm x \in J^{N+1}, 
\end{aligned} \end{equation}%
\nc{gvc}{$\Phi(y),\Phi^*(y)$}{Quantum fields \nomrefeqpage}%
for any positive integer $N$.
These definitions can be made rigorous using so-called \emph{smeared fields}, for which we refer the reader to \cite{ReedSimon2}. 
The key properties of the quantum fields is that they are formally adjoint and satisfy the \emph{canonical commutation relations}:
\[  [\Phi(x),\Phi^*(y)] = \delta(x-y), \quad [\Phi(x),\Phi(y)]=[\Phi^*(x),\Phi^*(y)]=0 \qquad (x,y \in J).\]

The quantum fields can be differentiated (in a distributional sense):
\begin{align*}
\left( \partial_y \Phi(y) F \right) (\bm x) &= \sqrt{N+1} \, \partial_y F(\bm x,y), && \n{for } F \in \ca D_{N+1}, \bm x \in \bR^{N}, \\
\left( \partial_y \Phi^*(y) F\right) (\bm x) &= \frac{1}{\sqrt{N+1}} \sum_{j=1}^{N+1} \delta'(y-x_j) F(\bm x_{\hat \jmath}), && \n{for } F \in \ca D_N, \bm x \in \bR^{N+1}.
\end{align*}
Note that $\delta'$, the distributional derivative of the Dirac delta, is characterized by $\innerrnd{\delta'}{F}=-\innerrnd{\delta}{F'} = -F'(0)$ for all $F \in \ca D_1$.

\subsection{The Hamiltonian}

We can (formally) express the QNLS Hamiltonian \rf{Hamiltonian} from the Introduction as follows
\begin{equation} \label{HamiltonianQuantumFields} H_\gamma =  \int_J \dd x \,\left( (\partial_x \Phi^*)(x)(\partial_x \Phi)(x) +\gamma \Phi^*(x)^2 \Phi(x)^2 \right). \end{equation}
In fact, $H_\gamma$ is an operator-valued distribution densely defined on $\End(\ca D)$.
From \rfeqn{HamiltonianQuantumFields} it follows that $H_\gamma$ is formally self-adjoint\footnote{We recall that $H_\gamma$ with a suitable choice of domain is essentially self-adjoint \cite{Dorlas}; in the case of bounded $J$ we recall the statement from Subsect. \ref{periodicity} that one needs to consider the subspace of $\ca H(J)$ of functions satisfying $L$-periodic boundary conditions}.
The derivation of \rfeqn{HamiltonianQuantumFields} uses that,
for all $F \in \ca D_N^\infty$ and $\bm x \in J^N$,
\begin{gather*}
\int_J \dd x \,\left( \partial_x \Phi^*(x) \right) \left( \partial_x \Phi(x) \right) = -\Delta, \\
\left( \Phi^*(y)^2 \Phi(y)^2 F \right) (\bm x) =  \sum_{j \ne k} \delta(y-x_j) \delta(y-x_k) F(\bm x_{\hat \jmath,\hat k},y,y)
\end{gather*}

\begin{rem}
The convention to write $\Phi^*(x)$ (and any of its derivatives) to the left of $\Phi(x)$ (and any of its derivatives) is called \emph{normal ordering}.
Physically this is required so that the system reproduces the correct ground-state energy.
Note that \rfeqn{HamiltonianQuantumFields} is what one would obtain by quantizing the Hamiltonian of the classical non-linear Schr\"odinger (CNLS) model using normal ordering:
\[ H_\gamma^\n{class}(\phi) = \int_{\bR} \dd x \left( |\partial_x \phi(x)|^2 +\gamma |\phi(x)|^4\right). \]
This explains the name ``quantum nonlinear Schr\"odinger model''. We also note that the CNLS model can be solved by the classical inverse scattering method.
\end{rem} 

\subsection{The monodromy matrix and the transfer matrix} \label{ssemonodromymatrix}

In the introductory remarks at the start of this chapter we touched upon an object called the monodromy matrix as an important tool to study the QNLS model. Here we will make this more precise. We will first define the monodromy matrix and then show that it is an important tool in solving the eigenvalue problem of the QNLS Hamiltonian. 
Define the (local) $\ca L$-matrix for the QNLS as
\begin{equation} \label{Lmatrix} \ca L_\lambda(x) = \begin{pmatrix} -\ii \lambda/2 &  \gamma \Phi^*(x) \\ \Phi(x) & \ii \lambda/2 \end{pmatrix}. \end{equation}%
\nc{rlcx}{$\ca L_{\lambda}(x)$}{QNLS (local) $\ca L$-matrix \nomrefeqpage}%
Using the time-ordered exponential \cite{Sklyanin1982,Sklyanin1989,KorepinBI} we can (formally) construct the (non-local) \emph{monodromy matrix} $\ca T_\lambda$, an operator on  $\End(\bC^2 \otimes \ca H)$ densely defined on $\End(\bC^2 \otimes \ca D)$:
\begin{equation} \label{Tmatrix} \ca T_\lambda = \quad : \! \exp_+ \int_J \dd x  \ca L_\lambda(x)\! : \quad = \quad \sum_{n \geq 0} \int_{J_+^n} \dd^n \bm x \, : \ca L_\lambda(x_n) \ldots \ca L_\lambda(x_1): 
\end{equation}%
\nc{rtcx}{$\ca T_{\lambda}$}{QNLS monodromy matrix \nomrefeqpage}%
where $: \ldots :$ indicates that normal ordering of the quantum fields $\Phi^*(x), \Phi(x)$ is applied.
That is, when expanding the product of matrices $\ca L_\lambda(x_n) \ldots \ca L_\lambda(x_1)$, in the resulting expressions any $\Phi^*(x)$ is moved to the left of any $\Phi(y)$.\\

The monodromy matrix is a 2x2-matrix with entries given by: 
\[ \ca T_\lambda = \begin{pmatrix} \ca T^{1 \, 1}_\lambda & \ca T^{1 \, 2}_\lambda \\ \ca T^{2 \, 1}_\lambda & \ca T^{2 \, 2}_\lambda \end{pmatrix} = \begin{pmatrix} A_\lambda & \gamma B_\lambda  \\ C_\lambda & D_\lambda  \end{pmatrix}. \]
We will see in \rfp{ABCDproperties} \ref{ABCDbounded} that the matrix entries of $\ca T_\lambda$ are bounded on the dense subspace $\ca H_\n{fin}$, which implies that $\ca T_\lambda$ can be viewed as an element of $\End(\bC^2 \otimes \ca H)$.\\

The \emph{transfer matrix} is obtained by taking the partial trace over $\bC^2$ of the monodromy matrix: 
\begin{equation} \label{transfermatrix} T_\lambda = \Tr_{\bC^2} \ca T_\lambda = A_\lambda+D_\lambda.\end{equation}%
\nc{rtca}{$T_{\lambda}$}{Transfer matrix \nomrefeqpage}%
We will see in Section \ref{YBalgebra} that $A_\lambda$ and $D_\lambda $ have the same dense domain $\ca D$ so that $T_\lambda$ is also densely defined. 
From \rfc{Tselfadjoint} we know that for $J$ bounded and $\lambda \in \bR$, $T_\lambda$ is self-adjoint. Furthermore, in \rft{commuteT} we will see that for all $\lambda, \mu \in \bC$,  $[T_\lambda,T_\mu]=0$, and in \rft{ABA} that $\Psi_{\bm \lambda} := B_{\lambda_N} \ldots B_{\lambda_1} \Vac$ is an eigenfunction of $T_\mu$, assuming that $\bm \lambda$ satisfies certain conditions. Dorlas \cite{Dorlas} has shown that the $\Psi_{\bm \lambda}$ are a complete set in $\ca H_N$.\\

The importance of the transfer matrix, and hence the monodromy matrix, to the study of the QNLS Hamiltonian $H_\gamma$ follows from the following argument (e.g. see \cite{KorepinBI}).
The Hamiltonian $H_\gamma$ is a linear combination of the coefficients obtained by asymptotically expanding $T_\mu$, i.e. by expanding $T_\mu$ in powers of $\mu$ in the limit $\mu \to \ii \infty$. 
Then expanding $[T_\lambda,T_\mu]=0$ with respect to $\mu$ implies that these expansion coefficients commute with $T_\lambda$. Hence, we obtain $[H_\gamma,T_\lambda] = 0$.
More precisely, 
\begin{align} 
\log\left( \e^{\ii \mu L/2} T_\mu \right) &\stackrel{\mu \to \ii\infty}{\sim} \frac{\ii \gamma}{\mu} H_\gamma^{[0]}+\frac{\ii \gamma}{\mu^2} \left( H_\gamma^{[1]}- \frac{\ii \gamma}{2}H_\gamma^{[0]} \right) + \nonumber \\
& \hspace{20mm} + \frac{\ii \gamma}{\mu^3} \left( H_\gamma^{[2]}-\ii \gamma H_\gamma^{[1]} - \frac{\gamma^2}{3}H_\gamma^{[0]} \right) + \mathcal O(\mu^{-4}), \label{Texpansion}
\end{align}
where $H_\gamma^{[2]} = H_\gamma$ and we have introduced the \emph{number} and \emph{momentum operators}
\[ H_\gamma^{[0]} = \int_J \dd x \Phi^*(x) \Phi(x), \quad H_\gamma^{[1]} = -\ii \int_J \dd x \Phi^*(x) \sd{x} \Phi(x). \]
These three operators are part of an infinite family of operators $H_\gamma^{[n]}$, for $n=0,1,\ldots$, which can be recovered recursively by calculating coefficients for higher powers of $\mu^{-1}$ in \rfeqn{Texpansion}. There exists a formalism \cite{Gutkin1985} for expressing all $H_\gamma^{[n]}$ in terms of the Bose fields $\Phi(x),\Phi^*(x)$.
Alternatively (again see \cite{KorepinBI}) the $H_\gamma^{[n]}$ can be rigorously defined by their restrictions to $\ca H_N$; for this purpose consider the $n$-th \emph{power sum polynomials} defined by
\begin{equation} \label{powersumpoly}
p_n(\lambda_1,\ldots,\lambda_N) = \sum_{j=1}^N \lambda_j^n.
\end{equation}%
\nc{rpl}{$p_{n}(\bm \lambda)$}{$n$-th power sum symmetric polynomial \nomrefeqpage}%
Then we have
\begin{equation} \label{higherintegralsofmotion} H_\gamma^{[n]}|_{\ca H_N} = p_n(-\ii \partial_1, \ldots, -\ii \partial_N) = \sum_{j=1}^N (-\ii \partial_j)^n, \end{equation}%
\nc{rhca}{$H_{\gamma}^{[n]}$}{Higher QNLS integrals of motion \nomrefeqpage}%
together with the \emph{higher-order derivative jump conditions at the hyperplanes}:
\[ \left( \partial_j-\partial_k \right)^{2l+1}|_{V_{j \, k}^+} -  \left( \partial_j-\partial_k \right)^{2l+1}|_{V_{j \, k}^-} = 2\gamma \left( \partial_j-\partial_k \right)^{2l}|_{V_{j \, k}}, \qquad \n{for } 1 \leq 2l+1 \leq n-1. \]
It follows from the above that the $H_\gamma^{[n]}$ all commute with $T_\lambda$ and hence with each other and with $H_\gamma$. Therefore they can be interpreted as integrals of motion.

\section{The $\ca R$-matrix and the Yang-Baxter equation} \label{secYBE}

The monodromy matrix satisfies the \emph{(quantum) Yang-Baxter equation} (QYBE), which involves another operator called the (QNLS) \emph{$\ca R$-matrix}.

\begin{defn} \label{Rmatrixdefn}
Let $\lambda \in \bC \setminus \{0\}$ and $\gamma  \in \bR$. Then
\begin{equation} \label{Rmatrix} \ca R_\lambda = 1 \otimes 1 -\frac{\ii \gamma}{\lambda} \ca P \in \End(\bC^2 \otimes \bC^2), \end{equation}%
\nc{rrc}{$\ca R_{\lambda}$}{QNLS $\ca R$-matrix \nomrefeqpage}%
where $\ca P \in \End(\bC^2 \otimes \bC^2)$ is the permutation operator: $\ca P(\bm v_1 \otimes \bm v_2)= \bm v_2 \otimes \bm v_1$ for $\bm v_1,\bm v_2 \in \bC^2$.%
\nc{rpcw}{$\ca P$}{Permutation operator of $\bC^2 \otimes \bC^2$}%
\end{defn}

\begin{rem}
This particular $\ca R$-matrix is relevant not only to the QNLS model, but also appears in other contexts such as the XXX model \cite{Bethe} and the Toda chain \cite{Toda}. These models are therefore algebraically related.
\end{rem}

Note that $\ca R_\lambda$ is invertible unless $\lambda = \pm \ii \gamma$; for these singular values of $\lambda$ we see that $\ca R_{\pm \ii \gamma} = 1 \otimes 1 \mp \ca P$ is proportional to a projection.
In general, $\ca R_\lambda \ca R_{-\lambda} = \frac{\lambda^2+\gamma^2}{\lambda^2}$. 
It can be easily checked that, for distinct nonzero $\lambda,\mu \in \bC$, $\ca R$ satisfies:
\begin{equation} \label{ybeR} \left(\ca R_{\lambda-\mu}\right)_{12} \left( \ca R_\lambda \right)_{13} \left( \ca R_\mu \right)_{23} = \left( \ca R_\mu \right)_{23} \left( \ca R_\lambda \right)_{13} \left(\ca R_{\lambda-\mu}\right)_{12} \in \End(\bC^2 \otimes \bC^2 \otimes \bC^2). \end{equation}
The subscript $_{i \, j}$ indicates which embedding $\End(\bC^2 \otimes \bC^2) \hookrightarrow \End(\bC^2 \otimes \bC^2 \otimes \bC^2)$ is used.

\begin{rem}
\rfeqn{ybeR} has many solutions, amongst which the $\ca R$-matrix defined in \rfd{Rmatrixdefn} stands out as the simplest nontrivial one. 
\end{rem}

\subsection{The Yang-Baxter equation}

We will now outline the method used in \cite{Gutkin1988} and \cite{KorepinBI} to derive the QYBE, an identity involving the $\ca R$-matrix and the monodromy matrix $\ca T$. Since $\ca T$ is defined in terms of the $\ca L$-matrix, it makes sense to establish first a more basic identity involving $\ca R$ and $\ca L$. 
For notational convenience we introduce the standard basis of $\f{sl}_2$: 
\[ \small \sigma_- = \begin{pmatrix} 0 & 0 \\ 1 & 0 \end{pmatrix}, \qquad
\sigma_+ = \begin{pmatrix} 0 & 1 \\ 0 & 0 \end{pmatrix}, \qquad \sigma_z = \begin{pmatrix} 1 & 0 \\ 0 & -1 \end{pmatrix}, \]%
\nc{gsl}{$\sigma_{\pm},\sigma_{z}$}{Basis elements of $\f{sl}_2$}%
satisfying the $\f{sl}_2$-relations $[\sigma_z,\sigma_\pm] = \pm 2 \sigma_\pm$, $[\sigma_+,\sigma_-]  = \sigma_z$. It can be checked that
\[ \ca P=\frac{1}{2}(1 \otimes 1+\sigma_z\otimes \sigma_z)+\sigma_+ \otimes \sigma_- + \sigma_- \otimes \sigma_+. \]

\begin{lem}
Let $\lambda, \mu \in \bC$ be distinct.
Then
\begin{align}
\lefteqn{ \hspace{-10mm} \ca R_{\lambda-\mu} \left( 1+ \ca L_\lambda(x) \otimes 1 + 1 \otimes \ca L_\mu(x) + \gamma \sigma_- \otimes \sigma_+ \right) =} \nonumber \\
\qquad &= \left( 1+ \ca L_\lambda(x) \otimes 1 + 1 \otimes \ca L_\mu(x) + \gamma \sigma_+ \otimes \sigma_- \right) \ca R_{\lambda-\mu} \in \End(\bC^2 \otimes \bC^2 \otimes \ca H).  \label{ybeRL}
\end{align}
\end{lem}

\begin{proof}
Expanding the identity by powers of $\gamma$, we see that the statement follows from
\[  \left[ \ca P, \ca L_\lambda(x) \otimes 1 + 1 \otimes \ca L_\mu(x) \right) = \ii (\lambda - \mu) \left(\sigma_+ \otimes \sigma_- - \sigma_- \otimes \sigma_+\right), \] 
which is a consequence of the elementary identities
\begin{align*} 
\ca P \left( \ca L_\lambda(x) \otimes 1 \right) - \left( 1 \otimes \ca L_\mu(x) \right) \ca P &= -\frac{\ii (\lambda-\mu)}{2} \ca P \left( \sigma_z \otimes 1 \right), \\
\ca P \left( 1 \otimes \ca L_\mu(x) \right) -  \left( \ca L_\lambda(x) \otimes 1 \right) \ca P &= \frac{\ii (\lambda-\mu)}{2} \left( \sigma_z \otimes 1 \right) \ca P,\\
\left[\left( \sigma_z \otimes 1 \right), \ca P\right] &= 2\left( \sigma_+ \otimes \sigma_- - \sigma_- \otimes \sigma_+ \right). \qedhere
\end{align*}
\end{proof}

\begin{thm} \emph{\cite{Gutkin1988, KorepinBI}} \label{ybeRTthm} 
Let $\lambda, \mu \in \bC$ be distinct.
Then $\ca T$ satisfies the \emph{quantum Yang-Baxter equation}\footnote{It is also known as the \emph{exchange relation} or simply the \emph{$\ca R \ca T \ca T$-relation}.}:
\begin{equation} \label{ybeRT} \ca R_{\lambda-\mu} \left(\ca T_\lambda \otimes 1 \right) \left( 1 \otimes \ca T_\mu\right) = \left(1 \otimes \ca T_\mu\right) \left(\ca T_\lambda \otimes 1 \right) \ca R_{\lambda-\mu} 
\in \End(\bC^2 \otimes \bC^2 \otimes \ca H). \end{equation}
\end{thm}

\begin{proof}
We explain the basic idea of the proof, referring to \cite{Gutkin1988, KorepinBI} for details.
The interval $J$ is split into subintervals; the monodromy matrix is redefined for each subinterval allowing the original monodromy matrix to be written as a product of monodromy matrices for subintervals. The key step is to use the ``local'' equation \rfeqn{ybeRL} and integrate it over each subinterval.
\end{proof}

We will now study the four matrix entries of $\ca T_\lambda$ in more detail; in particular, we will study their commutation relations which are encoded in the QYBE \rf{ybeRT}. 

\section{The Yang-Baxter algebra} \label{YBalgebra}

Let $\lambda \in \bC$ and $J=[-L/2,L/2]$ with $L \in \bR_{>0}$. 
The following expressions for the matrix entries of $\ca T_\lambda$, densely defined on each $\ca D_N$, are well-known \cite[Eqns.~(6.1.6)-(6.1.7)~and~Eqn.(6.1.13)]{Gutkin1988}.
\begin{align} 
\hspace{-2mm} A_\lambda &= \e^{-\ii \lambda L/2} \sum_{n = 0}^N \gamma^n \! \! \int_{J^{2n}_+} \! \! \dd^{2n}\bm y \e^{-\ii \lambda \sum_{j=1}^{2n} (-1)^j y_j} \! \! \left( \prod_{j=1 \atop j \n{ odd}}^{2n} \Phi^*(y_j) \right) \!  \!  \left( \prod_{j=1 \atop j \n{ even}}^{2n} \Phi(y_j) \right): \ca D_N \to \ca H_N \label{Aoperatorfields}, \displaybreak[2] \\%
\nc{rac}{$A_{\lambda}$}{Generator of symmetric Yang-Baxter algebra \nomrefeqpage} %
\hspace{-2mm} B_\lambda  &= \sum_{n = 0}^N \gamma^n \! \! \int_{J^{2n+1}_+} \! \! \dd^{2n+1}\bm y \e^{-\ii \lambda \sum_{j=1}^{2n+1} (-1)^j y_j} \! \! \left( \prod_{j=1 \atop j \n{ odd}}^{2n+1} \Phi^*(y_j) \right) \!  \!  \left( \prod_{j=1 \atop j \n{ even}}^{2n+1} \Phi(y_j) \right): \ca D_N \to \ca H_{N+1}, \displaybreak[2] \\%
\nc{rbc}{$B_{\lambda}$}{Generator of symmetric Yang-Baxter algebra; particle creation\\ operator \nomrefeqpage} %
\hspace{-2mm} C_\lambda &= \sum_{n = 0}^N \gamma^n \! \! \int_{J^{2n+1}_+} \! \! \dd^{2n+1}\bm y \e^{\ii \lambda \sum_{j=1}^{2n+1} (-1)^j y_j} \! \!  \left( \prod_{j=1 \atop j \n{ even}}^{2n+1} \Phi^*(y_j) \right) \! \!  \left( \prod_{j=1 \atop j \n{ odd}}^{2n+1} \Phi(y_j) \right): \ca D_{N+1} \to \ca H_N, \displaybreak[2] \\%
\nc{rcc}{$C_{\lambda}$}{Generator of symmetric Yang-Baxter algebra \nomrefeqpage} %
\hspace{-2mm} D_\lambda  &= \e^{\ii \lambda L/2} \sum_{n = 0}^N \gamma^n \! \! \int_{J^{2n}_+}\! \!  \dd^{2n}\bm y \e^{\ii\lambda \sum_{j=1}^{2n} (-1)^j y_j} \! \! \left( \prod_{j=1 \atop j \n{ even}}^{2n} \Phi^*(y_j) \right) \! \!  \left( \prod_{j=1 \atop j \n{ odd}}^{2n} \Phi(y_j) \right): \ca D_N \to \ca H_N. \label{Doperatorfields}
\end{align}%
\nc{rdc}{$D_{\lambda}$}{Generator of symmetric Yang-Baxter algebra \nomrefeqpage}%
To obtain explicit expressions for these, it is beneficial to introduce \emph{unit step functions}.
The one-dimensional unit step function $\theta: \bR \to \{0,1\}$ is defined by
\[ \theta(x) = \begin{cases} 1, & x>0, \\ 0, & x \leq 0. \end{cases} \]%
\nc{ghl}{$\theta(x), \theta(\bm x)$}{Unit step function}%
The multidimensional unit step function $\theta: \bR^N \to \{0,1\}$ is given by $\theta(\bm x) = \prod_{j=1}^N \theta(x_j)$.
Furthermore, $\theta(x_1 > \ldots > x_N)$ is shorthand for $\theta(x_1-x_2,x_2-x_3,\ldots,x_{N-1}-x_N)$.%
\nc{ghl}{$\theta(x_1 > \ldots > x_N)$}{Shorthand for $\theta(x_1-x_2,x_2-x_3,\ldots,x_{N-1}-x_N)$}%
\begin{exm}[$N=2$] \label{ABCDexamples}
The definitions in \rfeqnser{Aoperatorfields}{Doperatorfields} yield the following for the case $N=2$:
\begin{align*}
A_\lambda|_{\ca D_2} &= \e^{-\ii \lambda L/2} \left( 1 + \gamma \int_{J^2_+} \dd y_1 \dd y_2 \e^{\ii \lambda(y_1-y_2)} \Phi^*(y_1) \Phi(y_2) + \right. \\
& \hspace{20mm}\left. + \gamma^2  \int_{J^4_+} \dd y_1 \dd y_2 \dd y_3 \dd y_4 \e^{\ii \lambda(y_1-y_2+y_3-y_4)} \Phi^*(y_1) \Phi^*(y_3) \Phi(y_2)\Phi(y_4)  \right), \displaybreak[2] \\
B_\lambda|_{\ca D_2} &= \int_J \dd y \e^{\ii \lambda y} \Phi^*(y) + \gamma \int_{J^3_+} \dd y_1 \dd y_2 \dd y_3 \e^{\ii \lambda(y_1-y_2+y_3)} \Phi^*(y_1) \Phi^*(y_3) \Phi(y_2) + \\
& \hspace{20mm} + \gamma^2  \int_{J^5_+} \dd y_1 \dd y_2 \dd y_3 \dd y_4 \dd y_5 \e^{\ii \lambda(y_1-y_2+y_3-y_4+y_5)} \Phi^*(y_1) \Phi^*(y_3) \Phi^*(y_5) \Phi(y_2) \Phi(y_4) \displaybreak[2] \\
C_\lambda|_{\ca D_3} &= \int_J \dd y \e^{-\ii \lambda y} \Phi(y) + \gamma \int_{J^3_+} \dd y_1 \dd y_2 \dd y_3 \e^{\ii \lambda(-y_1+y_2-y_3)} \Phi^*(y_2) \Phi(y_1) \Phi(y_3) + \\
& \hspace{20mm} + \gamma^2  \int_{J^5_+} \dd y_1 \dd y_2 \dd y_3 \dd y_4 \dd y_5 \e^{\ii \lambda(-y_1+y_2-y_3+y_4-y_5)} \Phi^*(y_2) \Phi^*(y_4) \Phi(y_1) \Phi(y_3) \Phi(y_5) \displaybreak[2] \\
D_\lambda|_{\ca D_2} &= \e^{\ii \lambda L/2} \left( 1 + \gamma \int_{J^2_+} \dd y_1 \dd y_2 \e^{\ii \lambda(-y_1+y_2)} \Phi^*(y_2) \Phi(y_1) + \right. \\
& \hspace{20mm}\left. + \gamma^2  \int_{J^4_+} \dd y_1 \dd y_2 \dd y_3 \dd y_4 \e^{\ii \lambda(-y_1+y_2-y_3+y_4)} \Phi^*(y_2) \Phi^*(y_4) \Phi(y_1)\Phi(y_3)  \right).
\end{align*}
Focusing on the expression for $A_\lambda$, from \rfeqn{quantumfields} we have
\begin{align*} 
\left( \Phi^*(y_1) \Phi(y_2) F \right)(x_1,x_2) &= \delta(y_1 \! - \! x_1) F(x_2,y_2) + \delta(y_1 \! - \! x_2) F(x_1,y_2) \\
\left( \Phi^*(y_1) \Phi^*(y_3) \Phi(y_2)\Phi(y_4) F \right)(x_1,x_2) &= \delta(y_1 \! - \! x_1) \delta(y_3 \! - \! x_2) F(y_2,y_4)+\delta(y_1 \! - \! x_2)\delta(y_3 \! - \! x_1)F(y_2,y_4),
\end{align*}
for $F \in \ca D_2$, so that
\begin{align*}
\lefteqn{(A_\lambda F)(x_1,x_2) =} \\
&= \e^{-\ii \lambda L/2} \left( F(x_1,x_2) + \gamma \int_{J^2} \dd y_1 \dd y_2 \theta(y_1\! > \! y_2) \e^{\ii \lambda(y_1-y_2)} \delta(y_1 \! - \! x_1) F(x_2,y_2)  + \right. \\
& \hspace{20mm} + \gamma \int_{J^2} \dd y_1 \dd y_2 \theta(y_1\! > \! y_2) \e^{\ii \lambda(y_1-y_2)}\delta(y_1 \! - \! x_2) F(x_1,y_2) + \\
& \hspace{20mm}+ \gamma^2  \int_{J^4} \dd y_1 \dd y_2 \dd y_3 \dd y_4 \theta(y_1\! > \! y_2 \! > \! y_3 \! > \! y_4)  \e^{\ii \lambda(y_1-y_2+y_3-y_4)} \delta(y_1 \! - \! x_1) \delta(y_3 \! - \! x_2) F(y_2,y_4)  \\
& \hspace{20mm}\left. + \gamma^2  \int_{J^4} \dd y_1 \dd y_2 \dd y_3 \dd y_4 \theta(y_1\! > \! y_2 \! > \! y_3 \! > \! y_4)  \e^{\ii \lambda(y_1-y_2+y_3-y_4)} \delta(y_1 \! - \! x_2)\delta(y_3 \! - \! x_1)F(y_2,y_4)  \right) \displaybreak[2] \\
&= \e^{-\ii \lambda L/2} \left( F(x_1,x_2) + \gamma \int_{J}  \dd y \theta(x_1\! > \! y) \e^{\ii \lambda(x_1-y)}F(x_2,y)  + \gamma \int_{J}  \dd y \theta(x_2\! > \! y) \e^{\ii \lambda(x_2-y)}  F(x_1,y) + \right. \\
& \hspace{20mm}+ \gamma^2  \int_{J^2}  \dd y_1 \dd y_2 \theta(x_1\! > \! y_1 \! > \! x_2 \! > \! y_2)  \e^{\ii \lambda(x_1+x_2-y_1-y_2)} F(y_1,y_2)  \\
& \hspace{20mm}\left. + \gamma^2  \int_{J^2}  \dd y_1 \dd y_2 \theta(x_2\! > \! y_1 \! > \! x_1 \! > \! y_2)  \e^{\ii \lambda(x_1+x_2-y_1-y_2)} F(y_1,y_2)  \right) \displaybreak[2] \\
&= \e^{-\ii \lambda L/2} \left( F(x_1,x_2) + \gamma \int_{-L/2}^{x_1} \dd y \e^{\ii \lambda(x_1-y)}F(x_2,y)  + \gamma \int_{-L/2}^{x_2}  \dd y  \e^{\ii \lambda(x_2-y)}  F(x_1,y) + \right. \\
& \hspace{20mm}+ \gamma^2  \int_{x_2}^{x_1} \dd y_1 \int_{-L/2}^{x_2} \dd y_2  \e^{\ii \lambda(x_1+x_2-y_1-y_2)} F(y_1,y_2)  \\
& \hspace{20mm}\left. + \gamma^2  \int_{x_1}^{x_2} \dd y_1 \int_{-L/2}^{x_1} \dd y_2  \e^{\ii \lambda(x_1+x_2-y_1-y_2)} F(y_1,y_2)  \right). \displaybreak[2] 
\end{align*}
\end{exm}
We will provide explicit integral formulae such as the one for $A_\lambda$ in \rfex{ABCDexamples} shortly. 
First we highlight

\begin{prop}[Properties of $A_\lambda$, $B_\lambda $, $C_\lambda$, and $D_\lambda $] \label{ABCDproperties}
Let $\lambda \in \bC$.
\begin{enumerate}
\item \label{ABCDnumber} We have
\[ A_\lambda(\ca D_N), D_\lambda (\ca D_N) \subset \ca H_N, \qquad B_\lambda (\ca D_N) \subset \ca H_{N+1}, 
\qquad C_\lambda(\ca D_N) \subset \ca H_{N-1}, \qquad C_\lambda(\ca H_0) = 0. \]
\item \label{ABCDformaladjoint} 
$A_\lambda$ and $D_{\bar \lambda}$ are formally adjoint, as are $B_\lambda$ and $C_{\bar \lambda} $.
\item \label{ABCDvacuum} The actions of $A_\lambda$, $B_\lambda $, $C_\lambda$, $D_\lambda $ on $\Vac$ are as follows:
\[ A_\lambda\Vac = \e^{-\ii \lambda L/2} \Vac, \quad \left(B_\lambda \Vac\right)(x) = \e^{\ii \lambda x}, \quad
C_\lambda\Vac = 0, \quad D_\lambda \Vac = \e^{\ii \lambda L/2}\Vac. \]
\item \label{ABCDbounded} Crucially, the operators $A_\lambda$, $B_\lambda $, $C_\lambda$ and $D_\lambda $ are bounded on $\ca H_\n{fin}$ if $J$ is bounded.
More precisely, let $\lambda_0 \in \bR$ and $L \in \bR_{>0}$.
Then 
\begin{equation} \label{ABCDnorms1}
\begin{aligned}
\| A_\lambda|_{\ca H_N} \| & \leq \e^{(\lambda_0+|\gamma|N)L}, && \n{if } -\lambda_0 \leq\Im \lambda, \\
\| B_\lambda |_{\ca H_N} \|, \| C_\lambda|_{\ca H_N} \| &\leq N \sqrt{L} \e^{(\lambda_0/2+|\gamma|N)L}, && \n{if }
-\lambda_0 \leq \Im \lambda \leq \lambda_0, \\
\| D_\lambda |_{\ca H_N} \| & \leq \e^{(\lambda_0+|\gamma|N)L}, && \n{if }\Im \lambda \leq \lambda_0.
\end{aligned}
\end{equation}
In particular, for $\lambda \in \bR$ we have 
\begin{equation} \label{ABCDnorms2}
\| A_\lambda|_{\ca H_N} \|, \| D_\lambda |_{\ca H_N} \| \leq \e^{|\gamma|NL}, \| B_\lambda |_{\ca H_N} \|, \| C_\lambda|_{\ca H_N} \| \leq N \sqrt{L}\e^{|\gamma|NL}.
\end{equation}
\item \label{ADasBandCasAD} Furthermore, we can express $A_\lambda$ and $D_\lambda $ in terms of $B_\lambda $, and $C_\lambda$ in terms of $A_\lambda$ or $D_\lambda $, as follows:
\[ A_\lambda = \left[ \Phi(-L/2), B_\lambda  \right], \quad D_\lambda  = \left[ \Phi(L/2), B_\lambda  \right], \quad
C_\lambda = \frac{1}{\gamma} \left[  \Phi(L/2) , A_\lambda \right] = \frac{1}{\gamma} \left[ \Phi(-L/2), D_\lambda  \right].
\]
\end{enumerate}
\end{prop}

\begin{proof}
Properties \ref{ABCDnumber}-\ref{ABCDvacuum} follow immediately from \rfeqnser{Aoperatorfields}{Doperatorfields}. 
To establish property \ref{ABCDbounded}, we refer to \cite[Props. 6.2.1~and~6.2.2]{Gutkin1988} for the details of the proof of the norm bounds for $A_\lambda$ and $C_\lambda$.
Hence $A_\lambda$ (on any half-plane $\Im \lambda \geq -\lambda_0$) and $C_\lambda$ (on any ``strip'' $|\Im \lambda| \leq \lambda_0$) can be uniquely extended to bounded operators defined on each $\ca H_N$. These have bounded adjoints defined on each $\ca H_N$, of which $D_\lambda $ and $B_\lambda $ must be the restrictions to $\ca D$, because of property \ref{ABCDformaladjoint}. Since the adjoint of a bounded operator is bounded, we obtain the norm bounds for $D_\lambda $ (on the half-plane $\Im \lambda \leq \lambda_0$) and $B_\lambda $ as well. 
For $\lambda \in \bR$, the norm bounds for $A_\lambda$, $B_\lambda $, $C_\lambda$, and $D_\lambda $ follow immediately.
For property \ref{ADasBandCasAD}, we remark that since the operators $A_\lambda$, $B_\lambda $, $C_\lambda$ and $D_\lambda $ are bounded they can viewed as endomorphisms of $\ca H$. Then the desired result follows immediately from \rfeqnser{Aoperatorfields}{Doperatorfields}.
\end{proof}

In view of property \ref{ABCDnumber} we will refer to $B_\lambda $ as the \emph{QNLS (particle) creation operator}.
Combining properties \ref{ABCDformaladjoint} and \ref{ABCDbounded} we obtain that 
\begin{equation} \label{ABCDadjoint} A_\lambda^* = D_{\bar \lambda}, \qquad B_\lambda^* = C_{\bar \lambda} . \end{equation}

\rfeqn{ABCDnorms2} yields that $T_\lambda=A_\lambda+D_\lambda $ is bounded on each $\ca H_N$, and hence its unique extension to $\ca H_N$ is continuous. 
From property \ref{ABCDformaladjoint} it follows that $T_\lambda$ is formally self-adjoint; by the continuity of the $\L^2$ inner product its extension to $\ca H_N$ is also formally self-adjoint, and hence self-adjoint. We have obtained
\begin{cor} \label{Tselfadjoint}
For $J$ bounded and $\lambda \in \bR$, $T_\lambda$ is bounded and self-adjoint on $\ca H_\n{fin}$.
\end{cor}

\subsection{Integral formulae for the entries of the monodromy matrix}

To understand the matrix entries of $\ca T$ better, we study their action on $\ca H_N$.
Let $N$ and $n$ be nonnegative integers with $n \leq N$. Define the following sets of multi-indices
\begin{align} 
\f i^n_N &= \set{\bm i \in \{1, \ldots, N\}^n}{\n{for } l \ne m, \, i_l \ne i_m}, \label{seti}\\ %
\nc{rilwa}{$\f i^{n}_{N}$}{Set of $n$-tuples with distinct entries in $\{1,\ldots,N\}$\nomrefeqpage} %
\f I^n_N &= \set{\bm i \in \{1, \ldots, N\}^n}{i_1 < \ldots < i_n} \subset \f i^n_N, \label{setI} 
\end{align}%
\nc{ricw}{$\f I^{n}_{N}$}{Set of $n$-tuples with increasing entries in $\{1,\ldots,N\}$\nomrefeqpage}%
of $n$-tuples consisting of distinct and increasing elements from the index set $\{ 1, \ldots, N \}$, respectively. \\

\boxedenv{
\begin{notn}[Deleting and appending multiple variables] \label{notn3}
Let $X$ be a set; let $N$ and $n$ be nonnegative integers with $n \leq N$.
\begin{itemize}
\item Given $\bm x = (x_1,\ldots,x_N) \in X^N$ and $\bm i \in \f i^n_N$.
There is a unique element $\bm j = (j_1,\ldots,j_n) \in \f I^n_N \cap S_N(\bm i)$. We write 
\[ \bm x_{\hat{\bm \imath}} = \bm x_{\hat \imath_1 \ldots \hat \imath_n} = ((\ldots(\bm x_{\hat \jmath_n})\ldots)_{\hat \jmath_2})_{\hat \jmath_1}. \]
In other words, $\bm x_{\hat{\bm \imath}}$ is the element of $X^{N-n}$ obtained by deleting those $x_j$ with $j=i_m$ for some $m$.
\item Also, given $\bm x \in X^N$ and $\bm y = (y_1,\ldots,y_n) \in X^n$ for some positive integer $n$, then $(\bm x,\bm y) = (x_1,\ldots,x_N,y_1,\ldots,y_n) \in X^{N+n}$. 
\end{itemize}
\end{notn}}\\
\vspace{2mm}\\
\boxedenv{
\begin{notn}[Integral limits] \label{notn4}
Let $N$ and $n$ be nonnegative integers with $n \leq N$ and let $L \in \bR_{> 0}$. Whenever $\bm i \in \f i^n_N$ and $\bm x \in [-L/2,L/2]^N$, we will use the conventions that $x_{i_0}=L/2$ and $x_{i_{n+1}}=-L/2$; this facilitates the notation of certain integrals.
\end{notn} }\\

Let $J \subset \bR$ be a bounded or unbounded interval. For any set $X \subset J^N$, denote by $\chi_X$\nc{gxl}{$\chi_X$}{Characteristic function of the set $X$ \nomrefpage} the characteristic function of $X$.
Furthermore, we will use the same notation $\chi_X$ to denote the multiplication operator acting on $\ca F(J^N)$ or $\ca F(J^N_\n{reg})$ by sending $f$ to $\bm x \mapsto \chi_X(\bm x) f(\bm x)$ (for $\bm x \in J^N$ or $\bm x \in J^N_\n{reg}$, respectively).
To define (possibly unbounded) operators on $\ca H_N$ it is sufficient to specify the resulting function in the set $J^N_+$ and then symmetrize to $J^N_\n{reg}$ by applying $\sum_{w \in S_N} \chi_{w^{-1} J^N_+} w$; for each $\bm x$, only the term with $w \bm x \in J^N_+$ survives.
Moreover, one can extend the resulting function to the entire hypercube $J^N$ such that the extension is continuous at the hyperplanes $V_{j\,k}$ (the limits $x_k \to x_j$ in both domains $x_j \gtrless x_k$ are identical due to symmetry).

\begin{defn}[Elementary integral operators] \label{Eopsdefn}
Let $\lambda \in \bC$, $J=[-L/2,L/2]$ with $L \in \bR_{> 0}$ and $n=0,\ldots,N$. Define the \emph{elementary integral operators},
\begin{equation} \label{elemsymmintopdef}
\left.
\begin{aligned}
\hat E_{\lambda;\bm i} &= \hat E_{\lambda;i_1 \, \ldots \, i_{n+1}}:  && \ca H_N \to \ca H_{N+1}, && \n{for } \bm i \in \f I^{n+1}_{N+1}, \\
\bar E^\pm_{\lambda;\bm i} &= \bar E^\pm_{\lambda;i_1 \, \ldots \, i_n}: && \ca H_N \to \ca H_N, && \n{for } \bm i \in \f I^n_{N},   \\
\check E_{\lambda;\bm i} &= \check E_{\lambda;i_1 \, \ldots \, i_n}: && \ca H_{N+1} \to \ca H_N, && \n{for } \bm i \in \f I^n_{N}, 
\end{aligned}
\right\} \n{for }n=0,\ldots,N
\end{equation}%
\nc{rec}{$\hat E_{\lambda;\bm i}, \bar E^\pm_{\lambda;\bm i}, \check E_{\lambda;\bm i}$}{Elementary symmetric integral operators}%
by the following formulae.
For $F \in \ca D_N$, $\bm x \in J^{N+1}_+$, $\bm i \in \f I^{n+1}_{N+1}$ we have
\[ \left(\hat E_{\lambda;\bm i}F\right)(\bm x) =  \left( \prod_{m=1}^n \int_{x_{i_{m+1}}}^{x_{i_m}} \dd y_m \right) \e^{\ii \lambda \left( \sum_{m=1}^{n+1} x_{i_m} - \sum_{m=1}^n y_m \right)} F \left(\bm x_{\bm {\hat \imath}}, y_1, \ldots, y_n \right); \]
for $F \in \ca D_N$, $\bm x \in J^N_+$ and $\bm i \in \f I^n_{N}$ we have
\begin{align*}
\left(\bar E^+_{\lambda;\bm i}F\right)(\bm x) &= \e^{-\ii \lambda L/2} \left( \prod_{m=1}^n \int_{x_{i_{m+1}}}^{x_{i_m}} \dd y_m \right) \e^{\ii \lambda \sum_{m=1}^{n} \left( x_{i_m} - y_m \right)} F \left(\bm x_{\bm {\hat \imath}}, y_1, \ldots, y_n \right), \\
\left(\bar E^-_{\lambda;\bm i}F\right)(\bm x) &= \e^{\ii \lambda L/2} \left( \prod_{m=1}^n \int_{x_{i_m}}^{x_{i_{m-1}}} \dd y_m \right)  \e^{\ii \lambda \sum_{m=1}^{n} \left( x_{i_m} - y_m \right)} F \left(\bm x_{\bm {\hat \imath}}, y_1, \ldots, y_n \right);
\end{align*}
and for $F \in \ca D_{N+1}$, $\bm x \in J^N_+$ and $\bm i \in \f I^n_{N}$ we have
\[ \left(\check E_{\lambda;\bm i}F\right)(\bm x) = \left( \prod_{m=1}^n \int_{x_{i_m}}^{x_{i_{m-1}}} \dd y_m \right) \e^{\ii \lambda \left( \sum_{m=1}^n x_{i_m} - \sum_{m=1}^{n+1} y_m \right)} F \left(\bm x_{\bm {\hat \imath}}, y_1, \ldots, y_{n+1} \right). \]
\end{defn}

From \rfeqnser{Aoperatorfields}{Doperatorfields} the following expressions can be deduced (see e.g. \cite{Gutkin1988}):
\begin{align} 
A_\lambda & = \sum_{n=0}^N \gamma^n \sum_{\bm i \in \f I^n_{N}} \bar E^+_{\lambda;\bm i}: \hspace{-30mm} && \ca H_N \to \ca H_N, \label{Aintoperators} \displaybreak[2] \\
B_\lambda  & =  \frac{1}{N+1} \sum_{n=0}^N \gamma^n \sum_{\bm i\in \f I^{n+1}_{N+1}} \hat E_{\lambda;\bm i}: \hspace{-30mm} && \ca H_N \to \ca H_{N+1}, \label{Bintoperators} \displaybreak[2] \\
C_\lambda & = (N+1) \sum_{n=0}^N \gamma^n \sum_{\bm i \in \f I^n_{N}} \check E_{\lambda;\bm i}: \hspace{-30mm} && \ca H_{N+1} \to \ca H_N, \label{Cintoperators} \displaybreak[2] \\
D_\lambda  & = \sum_{n=0}^N \gamma^n \sum_{\bm i \in \f I^n_{N}} \bar E^-_{\lambda;\bm i}: \hspace{-30mm} && \ca H_N \to \ca H_N, \label{Dintoperators}
\end{align}
where $\lambda \in \bC$. 
We have changed the multiplicative factor of $B_\lambda$ from $\frac{1}{\sqrt{N+1}}$ to $\frac{1}{N+1}$, and similarly for $C_\lambda$ from $\sqrt{N+1}$ to $N+1$. This corresponds to conjugating the monodromy matrix by a matrix of the form $\omega^{\sigma_z} = \small \begin{pmatrix} \omega & 0 \\ 0 & \omega^{-1} \end{pmatrix}$. Because $\ca R_\lambda$ commutes with $\left(\omega^{\sigma_z} \otimes 1\right)\left( 1 \otimes \omega^{\sigma_z} \right)$ the YBE \rf{ybeRT} is preserved and hence this formalism describes the same physical system.

\begin{exm}[$N=2$, continuation of \rfex{ABCDexamples}] \label{ABCDexamples2}
We present the explicit formulae for the action of the matrix entries of $\ca T_\lambda$ on suitable functions $F$ restricted to the appropriate fundamental alcove.
For $F \in \ca D_2$ and $(x_1,x_2) \in J^2_+$ (i.e. $L/2>x_1>x_2>-L/2$) we have
\begin{align*}
(A_\lambda F)(x_1,x_2) &= \e^{- \! \ii \lambda L/2} \left( F(x_1,x_2) + \gamma \! \int_{-L/2}^{x_1} \dd y \e^{\ii \lambda(x_1\! - \! y)}F(x_2,y) + \right. \\
& \qquad\qquad + \gamma \! \int_{-L/2}^{x_2}  \dd y  \e^{\ii \lambda(x_2 \! -\! y)}  F(x_1,y) + \\
& \qquad\qquad + \left. \gamma^2  \int_{x_2}^{x_1} \dd y_1 \int_{-L/2}^{x_2} \dd y_2  \e^{\ii \lambda(x_1+x_2-y_1-y_2)} F(y_1,y_2)  \right), \displaybreak[2]  \\
(D_\lambda F)(x_1,x_2) &= \e^{\ii \lambda L/2} \left( F(x_1,x_2) + \gamma \! \int_{x_1}^{L/2} \dd y \e^{\ii \lambda(x_1\! - \! y)}F(x_2,y)  + \right. \\
& \qquad \qquad +\gamma \! \int_{x_2}^{L/2}  \dd y  \e^{\ii \lambda(x_2\! - \! y)}  F(x_1,y) + \\
& \qquad \qquad + \left. \gamma^2  \int_{x_1}^{L/2} \dd y_1 \int_{x_2}^{x_1} \dd y_2  \e^{\ii \lambda(x_1+x_2-y_1-y_2)} F(y_1,y_2)  \right).
\end{align*}
For $F \in \ca D_2$ and $(x_1,x_2,x_3) \in J^3_+$ we have
\begin{align*}
(B_\lambda F)(x_1,x_2,x_3) &=\frac{1}{3} \left( \e^{\ii \lambda x_1} F(x_2,x_3)+\e^{\ii \lambda x_2} F(x_1,x_3)+\e^{\ii \lambda x_3} F(x_1,x_2)  + \right.  \\
& \qquad + \gamma \int_{x_2}^{x_1} \dd y \e^{\ii \lambda(x_1+x_2-y)}F(x_3,y) 
+ \gamma \int_{x_3}^{x_1}  \dd y  \e^{\ii \lambda(x_1+x_3-y)}  F(x_2,y) +  \\
& \qquad + \gamma \int_{x_3}^{x_2}  \dd y  \e^{\ii \lambda(x_2+x_3-y)}  F(x_1,y) + \\
& \qquad + \left. \gamma^2  \int_{x_2}^{x_1} \dd y_1 \int_{x_3}^{x_2} \dd y_2  \e^{\ii \lambda(x_1+x_2+x_3-y_1-y_2)} F(y_1,y_2)  \right), 
\end{align*}
and finally for $F \in \ca D_3$ and $(x_1,x_2) \in J^2_+$ we have
\begin{align*}
(C_\lambda F)(x_1,x_2) &=3 \left( \int_{-L/2}^{L/2} \dd y \e^{-\ii \lambda y} F(x_1,x_2,y) + \right.  \\
& \qquad \qquad + \gamma \int_{x_1}^{L/2} \dd y_1 \int_{-L/2}^{x_1} \dd y_2 \e^{\ii (x_1-y_1-y_2)}F(x_2,y_1,y_2) + \\
& \qquad \qquad +  \gamma \int_{x_2}^{L/2} \dd y_1 \int_{-L/2}^{x_2} \dd y_2 \e^{\ii (x_2-y_1-y_2)}F(x_1,y_1,y_2)  + \\
& \qquad \qquad \left. + \gamma^2 \int_{x_1}^{L/2} \dd y_1 \int_{x_2}^{x_1} \dd y_2 \int_{-L/2}^{x_2} \dd y_3 \e^{\ii \lambda(x_1+x_2-y_1-y_2-y_3)} F(y_1,y_2,y_3) \right). 
\end{align*}
\end{exm}

Because of the above representation of the operators $A_\lambda$, $B_\lambda$, $C_\lambda$, $D_\lambda$ as integral operators we shall also refer to them as the \emph{QISM integral operators}.

\subsection{Commutation relations between the QISM integral operators}

Having established in \rfp{ABCDproperties} \ref{ABCDbounded} that $A_\lambda$, $B_\lambda $, $C_\lambda$, $D_\lambda $ are bounded in $\ca H_\n{fin}$, we consider the subalgebra of $\End(\ca H)$ generated by them and refer to it as the \emph{(symmetric) Yang-Baxter algebra} for the QNLS. Immediately from \rft{ybeRTthm} we have
\begin{cor} \label{Tcommrelcor}
For all choices of $j_1, k_1, j_2, k_2 \in \{ 1,2 \}$ and all $\lambda,\mu \in \bC$, $\lambda \ne \mu$ we have
\begin{equation} \label{Tcommrel}
\left[ \ca T^{j_1 \, k_1}_\lambda, \ca T^{j_2 \, k_2}_\mu \right] = -\frac{\ii \gamma}{\lambda-\mu} \left( \ca T^{j_2 \, k_1}_\lambda \ca T^{j_1 \, k_2}_\mu - \ca T^{j_2 \, k_1}_\mu \ca T^{j_1 \, k_2}_\lambda \right);
\end{equation} 
i.e. we have
\boxedalign{
[A_\lambda,A_\mu ] &= [B_\lambda ,B_\mu ] \, = \, [C_\lambda ,C_\mu ] \, = \, [D_\lambda ,D_\mu ] \, = \, 0 \nonumber \\
[A_\lambda,B_\mu ] &= - \frac{\ii \gamma}{\lambda-\mu} \left( B_\lambda A_\mu  - B_\mu A_\lambda \right), \label{ABcommrel} \\
[A_\lambda,C_\mu ] &= \frac{\ii \gamma}{\lambda-\mu} \left( C_\lambda A_\mu  - C_\mu A_\lambda \right), \nonumber \\
[B_\lambda ,A_\mu ] &= - \frac{\ii \gamma}{\lambda-\mu} \left( A_\lambda B_\mu  - A_\mu B_\lambda  \right), \label{BAcommrel} \\
[B_\lambda ,D_\mu ] &=  \frac{\ii \gamma}{\lambda-\mu} \left( D_\lambda B_\mu  - D_\mu B_\lambda \right), \label{BDcommrel} \\
[C_\lambda ,A_\mu ] &=  \frac{\ii \gamma}{\lambda-\mu} \left( A_\lambda C_\mu  - A_\mu C_\lambda \right), \label{CAcommrel} \\
[C_\lambda ,D_\mu ] &= - \frac{\ii \gamma}{\lambda-\mu} \left( D_\lambda C_\mu  - D_\mu C_\lambda \right), \nonumber \\
[D_\lambda ,B_\mu ] &= \frac{\ii \gamma}{\lambda-\mu} \left( B_\lambda D_\mu  - B_\mu D_\lambda  \right), \label{DBcommrel} \\
[D_\lambda ,C_\mu ] &= - \frac{\ii \gamma}{\lambda-\mu} \left( C_\lambda D_\mu  - C_\mu D_\lambda \right), \nonumber \\
[A_\lambda,D_\mu ] &= -\frac{\ii \gamma^2}{\lambda-\mu} \left( B_\lambda C_\mu -B_\mu C_\lambda  \right), \label{ADcommrel}  \\
[D_\lambda ,A_\mu ] &= -\frac{\ii \gamma^2}{\lambda-\mu} \left( C_\lambda B_\mu -C_\mu B_\lambda  \right), \nonumber \\
[B_\lambda ,C_\mu ] &= -\frac{\ii}{\lambda-\mu} \left( A_\lambda D_\mu -A_\mu D_\lambda \right), \label{BCcommrel}  \\
[C_\lambda ,B_\mu ] &= -\frac{\ii}{\lambda-\mu} \left( D_\lambda A_\mu -D_\mu A_\lambda\right). \label{CBcommrel} }
\end{cor}
We will at times refer to these identities by the letters of the operators appearing in the commutator on the left-hand side of the equals sign. For example, the $AB$-relation is the relation labelled \rf{ABcommrel} and is the relation used to move an $A$-operator past a $B$-operator from left to right.
We make the following observations.
\begin{itemize}
\item The $DD$- and $CC$-relations are the (formal) adjoints of the $AA$- and $BB$-relations, respectively. There is a subtlety in the derivation of these relations from \rfc{Tcommrelcor}. For example for the $AA$-relation one obtains $(\lambda-\mu+\ii \gamma)[A_\lambda,A_\mu ]=0$, so that an additional assumption $\lambda-\mu \ne -\ii \gamma$ seems necessary. However, by swapping $\lambda$ and $\mu$ this can be overcome. 
\item The $BD$-, $DB$-, $CD$- and $DC$-relations are the adjoints of the $AC$-, $CA$-, $AB$- and $BA$-relations, respectively.
\item The $AD$-, $BC$-, $CB$- and $DA$-relations are formally self-adjoint.
\end{itemize}
It is easily checked that the $BA$-relation is equivalent to the $AB$-relation; similarly, the $BD$-relation to the $DB$-relation. More generally, we have
\begin{lem}\label{ABBAequivalence}
Let $\lambda,\mu$ be distinct complex numbers, and $X$ and $Y$ parametrized elements of an associative algebra over $\bC$ satisfying
\[ [X_\lambda,Y_\mu] = h_{\lambda-\mu} \left( X_\lambda Y_\mu - X_\mu Y_\lambda \right) \]
for some $h: \bC\setminus\{0\} \to \bC\setminus\{0\}: z \mapsto h_z$ satisfying $h_{-z}=-h_z$.
Then
\[ [Y_\lambda,X_\mu] =  h_{\lambda-\mu} \left( Y_\lambda X_\mu - Y_\mu X_\lambda \right). \]
\end{lem}

\begin{proof}
Note that it follows immediately that $[X_\lambda,Y_\mu]=[X_\mu,Y_\lambda]$
We may write
\begin{align*} 
[Y_\lambda,X_\mu] &= -[X_\mu,Y_\lambda] &&= - h_{\mu-\lambda} \left( X_\mu Y_\lambda - X_\lambda Y_\mu \right) \\
&= h_{\lambda-\mu} \left( [X_\mu,Y_\lambda] + Y_\lambda X_\mu - [X_\lambda,Y_\mu] - Y_\mu X_\lambda \right) 
&&= h_{\lambda-\mu} \left(Y_\lambda X_\mu - Y_\mu X_\lambda \right). \qedhere
\end{align*}
\end{proof}

\subsection{Integrability of the QNLS model}

Referring back to the discussion at the end of Section \ref{seQFThy}, we now formulate the condition from which the integrability of the QNLS model follows.

\begin{thm}[Commuting transfer matrices] \label{commuteT}
Let $\lambda,\mu \in \bC$. Then $[T_\lambda,T_\mu]=0$.
\end{thm}

\begin{proof}
We may assume $\lambda \ne \mu$.
In view of $[A_\lambda,A_\mu ]=[D_\lambda ,D_\mu ]=0$, it is sufficient to prove that $[A_\lambda,D_\mu ]=-[D_\lambda ,A_\mu ]=[A_\mu ,D_\lambda ]$, i.e. that $[A_\lambda,D_\mu ]$ is invariant under $\lambda \leftrightarrow \mu$. This follows immediately from \rfeqn{ADcommrel}. 
\end{proof}

Let $\mu \in \bR$. From \rfc{Tselfadjoint} and \rft{commuteT} it follows that we have a commuting family of self-adjoint operators $T_\mu$, they have common eigenfunctions, necessarily independent of the spectral parameter $\mu$. 
We will now construct these eigenfunctions.

\section{The Bethe wavefunction and the algebraic Bethe ansatz}

We start this process by using the operators $B_\lambda $ to construct from an element of $\ca H_0$ an element of $\ca H_N$, i.e. a wavefunction representing an $N$-particle state. 

\subsection{The Bethe wavefunction}

Introduce the set of \emph{complex regular vectors}
\[ \bC^N_\n{reg} = \set{\bm \lambda \in \bC^N}{\lambda_j \ne \lambda_k \n{ if } j \ne k} \]%
\nc{rccw}{$\bC^N_\n{reg}$}{Set of complex regular vectors}%
and recall the pseudo-vacuum $\Vac = 1\in \ca H_0 \cong \bC$.
\begin{defn}[Bethe wavefunction] \label{Bethewavefnintops}
Let $\bm \lambda = (\lambda_1, \ldots, \lambda_N) \in \bC^N_\n{reg}$.
We define the \emph{Bethe wavefunction} as
\begin{equation} \label{Psi}
\Psi_{\bm \lambda} = \left( \prod_{j=1}^N B_{\lambda_j} \right) \Vac.
\end{equation}
\end{defn}

\begin{rem}
Note that the Bethe wavefunction is an $N$-particle wavefunction since we have applied the QNLS particle creation operator $N$ times to the pseudovacuum $\Vac$.
By virtue of the relation $[B_\lambda ,B_\mu ]=0$, the order of the product of the $B$-operators does not matter and hence for any $w \in S_N$, $\Psi_{\bm \lambda}= \Psi_{w\bm \lambda}$, so $\Psi_{\bm \lambda}$ is a wavefunction describing a system of bosons. 
If $J$ is bounded, we will see that further conditions need to be imposed on $\bm \lambda$ in order to make $\Psi_{\bm \lambda}$ an eigenfunction of $T_\mu$.
The wavefunctions $\Psi_{\bm \lambda}$ will turn out to be nontrivial linear combinations of plane waves with wavenumbers given by permutations of $\bm \lambda$; the interaction between the particles is encoded in the combinatorial nature of the coefficients of the Bethe wavefunction.
\end{rem}

Let $(\lambda_1,\ldots,\lambda_{N+1} )\in \bC^{N+1}_\n{reg}$.
Immediately from its definition we have the following key property of the Bethe wavefunction $\Psi_{\bm \lambda}$:
\begin{equation}  \label{Bethewavefnrecursion}
\Psi_{\lambda_1,\ldots,\lambda_{N+1}} = B_{\lambda_{N+1}}\Psi_{\lambda_1,\ldots,\lambda_N}.
\end{equation}

\subsection{The Bethe ansatz equations}

Let $L \in \bR_{>0}$ and $\gamma \in \bR$.
The \emph{Bethe ansatz equations} (BAEs) for $\bm \lambda \in \bC^N$ are
\begin{equation} \label{BAE}  \e^{\ii \lambda_j L} = \prod_{k \ne j} \frac{\lambda_j-\lambda_k+\ii \gamma}{\lambda_j-\lambda_k-\ii \gamma}, \qquad \n{for } j=1,\ldots,N. \end{equation}

Let $\bm \lambda \in \bC^N$ be a solution of the BAEs \rf{BAE}.
Then it is easily checked that $\bm \lambda$ satisfies the following conditions.
\begin{enumerate}
\item For any $w \in S_N$, $ w \bm \lambda$ is a solution. 
\item For any integer $m$, $\bm \lambda+ \frac{2 \pi}{L} m \bm 1$ is also a solution, where $\bm 1 = (1,\ldots, 1) \in \bC^N$.
\item The quantity $\sum_{j=1}^N \lambda_j$ is an integer multiple of $2 \pi/L$.
\end{enumerate}

\begin{rem}
The latter statement is physically relevant as it indicates that the total momentum is quantized.
\end{rem}

\begin{lem} \emph{\cite{KorepinBI}}
Let $\gamma > 0$. 
All solutions $\bm \lambda$ of the BAEs \rf{BAE} are in $\bR^N$. 
The set of solutions is injectively parametrized by $\bZ^N$ and hence there are infinitely many solutions. 
\end{lem}

\begin{proof}[Outline of proof]
We refer to \cite{KorepinBI} for the detailed proof. 
The reality of the solutions of the BAEs can be obtained by a straightforward estimate on the imaginary parts of the $\lambda_j$, where one uses that $\gamma$ is positive.
The existence of solutions, as well as the parametrization, follows by recasting \rfeqn{BAE} in logarithmic form; this gives
\[ L \lambda_j + \sum_{k=1}^N 2 \tan^{-1}\left(\frac{\lambda_j-\lambda_k}{\gamma}\right) = 2 \pi n_j, \qquad \n{for } j=1,\ldots,N, \]
where the $n_j$ are integers if $N$ is odd and half-integers if $N$ is even.
One then makes use of the fact that these equations also form the extremum condition for the so-called \emph{Yang-Yang action}
\[ S = \frac{L}{2} \sum_{j=1}^N \lambda_j^2 - 2 \pi \sum_{j=1}^N n_j \lambda_j + \frac{1}{2} \sum_{1 \leq j \ne k \leq N} \int_0^{\lambda_j-\lambda_k} \dd \mu \tan^{-1}\left(\frac{\mu}{\gamma}\right). \qedhere\]
\end{proof}

Let $\bm \lambda \in \bC^N_\n{reg}$.
Introduce the short-hand notation
\begin{equation} \label{BAEv} \tau^\pm_\mu(\bm \lambda) =   \prod_{j=1}^N \frac{\lambda_j-\mu \mp \ii \gamma}{\lambda_j-\mu}; \end{equation}%
\nc{gtl}{$\tau_{\mu}^{\pm}(\bm \lambda)$}{Shorthand for $\prod_{j=1}^N \frac{\lambda_j-\mu \mp \ii \gamma}{\lambda_j-\mu}$ \nomrefeqpage}%
note that $\tau_\mu^\pm \in \ca C^\omega((\bC \setminus \{\mu\})^N)$ is $S_N$-invariant.
\begin{lem}\label{BAEvariantslem}
Let $L \in \bR_{>0}$, $\gamma \in \bR$ and $\bm \lambda \in \bC^N_\n{reg}$.
The BAEs \rf{BAE} is equivalent to the conditions
\begin{align} 
\e^{\ii \lambda_j L} &= \frac{\tau^+_{\lambda_j}(\bm \lambda_{\hat \jmath})}{\tau^-_{\lambda_j}(\bm \lambda_{\hat \jmath})}, && \n{for } j=1,\ldots,N, \label{BAEvariant1}  \\
\tau^+_{\lambda_j}(\bm \lambda_{\hat \jmath}) &= \e^{\ii \lambda_j L} \tau^-_{\lambda_j}(\bm \lambda_{\hat \jmath}), && \n{for }j=1,\ldots,N. \label{BAEvariant} 
\end{align}
\end{lem}

\begin{proof}
For the case $\gamma=0$ the claims are immediate; suppose now that $\gamma \ne 0$.
The equivalence of the BAES with \rfeqn{BAEvariant1} follows immediately from the condition that $\bm \lambda \in \bC^N_\n{reg}$.
As for \rfeqn{BAEvariant}, it is certainly necessary for \rfeqn{BAEvariant1}; it is also sufficient since $\lambda_j-\lambda_k = \ii \gamma$ cannot occur for any pair $j,k$, as we prove now by supposing the opposite. 
Starting with such a pair, using \rfeqn{BAEvariant} one constructs a sequence $(j_l)_{l \geq 1}$ where for all $l$, $j_l \in \{ 1, \ldots, N \}$ and $\lambda_{j_l}-\lambda_{j_{l+1}} = \ii \gamma$. From the existence of $l<m$ with $j_l=j_m$ we obtain $0 = \lambda_{j_l}-\lambda_{j_m} = (m-l) \ii \gamma$, in contradiction with $\gamma \ne 0$.
\end{proof}

\subsection{The algebraic Bethe ansatz}

The central idea of the ABA is to impose the Bethe ansatz equations on $\bm \lambda$ to obtain an eigenfunction of the transfer matrix \cite{KorepinBI}. The following result is essential.
\begin{prop} \label{APsiDPsi}
Let $L \in \bR_{>0}$, $\gamma \in \bR$ and $(\bm \lambda,\mu) \in \bC^{N+1}_\n{reg}$.
Then
\begin{align}
A_\mu  \Psi_{\bm \lambda} &= \tau^+_\mu(\bm \lambda) \e^{-\ii \mu L/2} \Psi_{\bm \lambda} +\sum_{j=1}^N \tau^+_{\lambda_j}(\bm \lambda_{\hat \jmath}) \frac{\ii \gamma}{\lambda_j-\mu} \e^{-\ii \lambda_j L/2} \Psi_{\bm \lambda_{\hat \jmath},\mu}, \label{APsi} \\
D_\mu  \Psi_{\bm \lambda} &= \tau^-_\mu(\bm \lambda) \e^{\ii \mu L/2} \Psi_{\bm \lambda} - \sum_{j=1}^N \tau^-_{\lambda_j}(\bm \lambda_{\hat \jmath}) \frac{\ii \gamma}{\lambda_j-\mu} \e^{\ii \lambda_j L/2} \Psi_{\bm \lambda_{\hat \jmath},\mu}. \label{DPsi}
\end{align}
\end{prop}

\begin{proof}
By induction; for $N=0$ the statements are trivially true. 
We will complete the proof by showing that the statement for $A_\mu $ holds ``as is'' if the statement holds with $N \to N-1$.
Writing $\bm \lambda' = (\lambda_1,\ldots,\lambda_{N-1})$, we have that $\Psi_{\bm \lambda} = B_{\lambda_N} \Psi_{\bm \lambda'}$.
From \rfeqn{ABcommrel} we have
\[ A_\mu  \Psi_{\bm \lambda} = \frac{\lambda_N-\mu-\ii \gamma}{\lambda_N-\mu} B_{\lambda_N}A_\mu \Psi_{\bm \lambda'} + \frac{\ii \gamma}{\lambda_N-\mu} B_\mu A_{\lambda_N} \Psi_{\bm \lambda'}. \]
The induction hypothesis yields
\begin{align*}
A_\mu  \Psi_{\bm \lambda} &= \frac{\lambda_N-\mu-\ii \gamma}{\lambda_N-\mu} B_{\lambda_N} \left( \tau^+_\mu(\bm \lambda') \e^{-\ii \mu L/2} \Psi_{\bm \lambda'} + \sum_{j=1}^{N-1} \tau^+_{\lambda_j}(\bm \lambda'_{\hat \jmath}) \frac{\ii \gamma}{\lambda_j-\mu} \e^{-\ii \lambda_j L/2} \Psi_{\bm \lambda'_{\hat \jmath},\mu} \right) + \\
& \qquad + \frac{\ii \gamma}{\lambda_N-\mu} B_\mu  \left( \tau^+_{\lambda_N}(\bm \lambda') \e^{-\ii \lambda_N L/2} \Psi_{\bm \lambda'} + \sum_{j=1}^{N-1} \tau^+_{\lambda_j}(\bm \lambda_{\hat \jmath}) \frac{\ii \gamma}{\lambda_j-\lambda_N} \e^{-\ii \lambda_j L/2}\Psi_{\bm \lambda_{\hat \jmath}} \right) \displaybreak[2] \\
&= \tau^+_\mu(\bm \lambda) \e^{-\ii \mu L/2} \Psi_{\bm \lambda} + \tau^+_{\lambda_N}(\bm \lambda') \frac{\ii \gamma}{\lambda_N-\mu} \e^{-\ii \lambda_N L/2}\Psi_{\bm \lambda',\mu} + \\
& \qquad + \sum_{j=1}^{N-1} \tau^+_{\lambda_j}(\bm \lambda'_{\hat \jmath}) \left( \frac{\lambda_N-\mu-\ii \gamma}{\lambda_N-\mu} \frac{\ii \gamma}{\lambda_j-\mu} + \frac{\ii \gamma}{\lambda_N-\mu}\frac{\ii \gamma}{\lambda_j-\lambda_N} \right) \e^{-\ii \lambda_j L/2} \Psi_{\bm \lambda_{\hat \jmath},\mu}.
\end{align*}
Using that
\[ \frac{\lambda_N-\mu-\ii \gamma}{\lambda_N-\mu} \frac{\ii \gamma}{\lambda_j-\mu} + \frac{\ii \gamma}{\lambda_N-\mu}\frac{\ii \gamma}{\lambda_j-\lambda_N} = \frac{\lambda_N - \lambda_j - \ii \gamma}{\lambda_N - \lambda_j} \frac{\ii \gamma}{\lambda_j- \mu} \]
we obtain \rfeqn{APsi}.
The proof for \rfeqn{DPsi} goes analogously, using the commutation relation \rfeqn{DBcommrel} instead of \rfeqn{ABcommrel}.
\end{proof}

Now we can state and prove
\begin{thm}[Spectrum of the transfer matrix] \label{ABA}
Let $L \in \bR_{>0}$, $\gamma \in \bR$ and $(\bm \lambda,\mu) \in \bC^{N+1}_\n{reg}$.
Then $\Psi_{\bm \lambda}$ is an eigenfunction of the transfer matrix $T_\mu$ with eigenvalue 
\begin{equation} \label{BAEv2} \tau_\mu(\bm \lambda)=\e^{-\ii \mu L/2}\tau^+_\mu(\bm \lambda)+ \e^{\ii \mu L/2} \tau^-_\mu(\bm \lambda) \end{equation}%
\nc{gtl}{$\tau_{\mu}(\bm \lambda)$}{Eigenvalue of $T_\mu$ \nomrefeqpage}%
precisely if $\bm \lambda$ satisfies the BAEs \rf{BAE}.
\end{thm}

\begin{proof}
\rfp{APsiDPsi} allows us to write down a formula for $(A_\mu +D_\mu )\Psi_{\bm \lambda}$. Consequently the necessary and sufficient condition for $\Psi_{\bm \lambda}$ being an eigenfunction is that for all $j=1,\ldots,N$,
$\tau^+_{\lambda_j}(\bm \lambda_{\hat \jmath}) = \tau^-_{\lambda_j}(\bm \lambda_{\hat \jmath}) \e^{\ii \lambda_j L}$.
By virtue of \rfl{BAEvariantslem}, this is equivalent to the system of BAEs \rf{BAE}.
\end{proof} 

\begin{rem} \label{Pauli}
The restriction $\bm \lambda \in \bC^N_\n{reg}$ in \rft{ABA} is important: it has been shown (\cite[Section~VII.4]{KorepinBI} and \cite{IzerginKorepin}) that a Pauli principle holds for the QNLS model (and, more generally, for any one-dimensional interacting quantum system): one can show that eigenfunctions $\Psi_{\bm \lambda}$ of $T_\mu$ where some of the $\lambda_j$ coincide do not exist.
\end{rem}

\begin{exm}[$N=2$; the algebraic Bethe ansatz]
We recall from \rfex{Bethewavefn2} that
\[ \Psi_{\lambda_1,\lambda_2}|_{\bR^2_+} = \frac{1}{2} \left[ \frac{\lambda_1-\lambda_2-\ii \gamma}{\lambda_1-\lambda_2} \e^{\ii(\lambda_1 x_1 + \lambda_2 x_2)} + \frac{\lambda_1-\lambda_2+\ii \gamma}{\lambda_1-\lambda_2} \e^{\ii(\lambda_2 x_1 + \lambda_1 x_2)} \right] \]
and the reader should check that this equals $B_{\lambda_2} B_{\lambda_1} \Psi_\emptyset$.
The statement in \rfp{APsiDPsi} for the case $N=2$ reads
\begin{align*} 
A_\mu \Psi_{\lambda_1,\lambda_2} &= \frac{\lambda_1-\mu-\ii \gamma}{\lambda_1-\mu} \frac{\lambda_2-\mu-\ii \gamma}{\lambda_2-\mu} \e^{-\ii \mu L/2} \Psi_{\lambda_1,\lambda_2} + \\
& \qquad + \frac{\lambda_1-\lambda_2+\ii \gamma}{\lambda_1-\lambda_2} \frac{\ii \gamma}{\lambda_1-\mu} \e^{-\ii \lambda_1 L/2} \Psi_{\lambda_2,\mu} +  \frac{\lambda_1-\lambda_2-\ii \gamma}{\lambda_1-\lambda_2} \frac{\ii \gamma}{\lambda_2-\mu} \e^{-\ii \lambda_2 L/2}  \Psi_{\lambda_1,\mu}  \\
D_\mu \Psi_{\lambda_1,\lambda_2} &= \frac{\lambda_1-\mu+\ii \gamma}{\lambda_1-\mu}\frac{\lambda_2-\mu+\ii \gamma}{\lambda_2-\mu} \e^{\ii \mu L/2} \Psi_{\lambda_1,\lambda_2} + \\
& \qquad - \frac{\lambda_1-\lambda_2-\ii \gamma}{\lambda_1-\lambda_2}\frac{\ii \gamma}{\lambda_1-\mu} \e^{\ii \lambda_1 L/2} \Psi_{\lambda_2,\mu} - \frac{\lambda_1-\lambda_2+\ii \gamma}{\lambda_1-\lambda_2} \frac{\ii \gamma}{\lambda_2-\mu} \e^{\ii \lambda_2 L/2} \Psi_{\lambda_1,\mu}.
\end{align*}
In order for $\Psi_{\lambda_1,\lambda_2}$ to be an eigenfunction of $T_\lambda = A_\lambda+D_\lambda$ it is clear that it is necessary and sufficient that
\begin{align*} 
\frac{\lambda_1-\lambda_2+\ii \gamma}{\lambda_1-\lambda_2} \e^{-\ii \lambda_1 L/2} &= \frac{\lambda_1-\lambda_2-\ii \gamma}{\lambda_1-\lambda_2} \e^{\ii \lambda_1 L/2} \\
\frac{\lambda_1-\lambda_2-\ii \gamma}{\lambda_1-\lambda_2} \e^{-\ii \lambda_2 L/2} &= \frac{\lambda_1-\lambda_2+\ii \gamma}{\lambda_1-\lambda_2} \e^{\ii \lambda_2 L/2},
\end{align*}
i.e. that
\[ \frac{\lambda_1-\lambda_2+\ii \gamma}{\lambda_1-\lambda_2-\ii \gamma} = \e^{\ii \lambda_1 L} = \e^{-\ii \lambda_2 L}, \]
since it cannot happen that $\lambda_1-\lambda_2=\pm \ii \gamma$ (this would lead to $2=0$).
In this case, the eigenvalue of $T_\lambda$ equals
\[ \tau_\mu(\lambda_1.\lambda_2) = \frac{\lambda_1-\mu-\ii \gamma}{\lambda_1-\mu} \frac{\lambda_2-\mu-\ii \gamma}{\lambda_2-\mu} \e^{-\ii \mu L/2}+ \frac{\lambda_1-\mu+\ii \gamma}{\lambda_1-\mu}\frac{\lambda_2-\mu+\ii \gamma}{\lambda_2-\mu} \e^{\ii \mu L/2}. \]
\end{exm}

\subsection{The transfer matrix eigenvalue}

Here we consider some properties of the function $\tau_\mu(\bm \lambda)$. 
First of all, note that if $\bm \lambda \in \bR^N$ (i.e. in particular if $\bm \lambda$ solves the BAEs \rf{BAE}) and $\mu \in \bR$, then $\e^{-\ii \mu L/2}\tau^+_\mu(\bm \lambda)$ and $\e^{\ii \mu L/2} \tau^-_\mu(\bm \lambda)$ are conjugate complex numbers, so that $\tau_\mu(\bm \lambda) \in \bR$, as expected for a self-adjoint operator $T_\mu$. 
Moreover, we can recover the eigenvalues of the QNLS integrals of motion $\sum_{j=1}^N \lambda_j^n$ by expanding $\tau_\mu(\bm \lambda)$ in powers of $\mu^{-1}$, analogously to \rfeqn{Texpansion}.
Using
\[ \frac{\lambda_j-\mu-\ii \gamma}{\lambda_j-\mu} = 1 + \frac{\ii \gamma}{\mu} \sum_{k \geq 0} \left( \frac{\lambda_j}{\mu} \right)^k, \qquad \log(1+x) = \sum_{l \geq 1} (-1)^{l-1}x^l/l \]
we obtain
\begin{align*} 
\lefteqn{\log\left(\e^{\ii \mu L/2} \tau_\mu(\bm \lambda)\right) \stackrel{\mu \to \ii\infty}{\sim} \log\left(\tau^+_\mu(\bm \lambda)\right) = \sum_{j=1}^N \log\left( 1 + \frac{\ii \gamma}{\mu} \sum_{k \geq 0} \left( \frac{\lambda_j}{\mu} \right)^k \right) = } \displaybreak[2] \\
\qquad &\stackrel{\mu \to \ii\infty}{\sim} \sum_{j=1}^N \sum_{l \geq 1} \frac{(-1)^{l-1}}{l} \left( \frac{\ii \gamma}{\mu} \right)^l \left( \sum_{k \geq 0} \left( \frac{\lambda_j}{\mu} \right)^k \right)^l \displaybreak[2]  \\
&\stackrel{\mu \to \ii\infty}{\sim}
\frac{\ii \gamma}{\mu}  p_0+\frac{\ii \gamma}{\mu^2} \left( p_1 - \frac{\ii \gamma}{2}  p_0\right) + \frac{\ii \gamma}{\mu^3} \left( p_2  - \ii \gamma p_1 -\frac{\gamma^2}{3} p_0 \right) + \ca O(\mu^{-4}),
\end{align*}
where $p_n=p_n(\bm \lambda)$.\\

Next, we study the analyticity of $\mu \mapsto \tau_\mu(\bm\lambda)$.
\begin{prop}
Let $L \in \bR_{>0}$, $\gamma \in \bR$ and $\bm \lambda \in \bC^{N}_\n{reg}$.
Assume that $\bm \lambda$ satisfies the BAEs \rf{BAE}.
Then $\mu \mapsto \tau_\mu(\bm\lambda)$ is analytic.
\end{prop}

\begin{proof}
As for the meromorphic function $\mu \mapsto \tau^\pm_\mu(\bm \lambda)$, its singularities are simple poles at $\lambda_j$.
The residue at $\lambda_j$ is given by
\[ \Res_{\mu \to \lambda_j} \tau^\pm_\mu(\bm \lambda)= \lim_{\mu \to \lambda_j} (\mu-\lambda_j) \tau^\pm_\mu(\bm \lambda) = \pm \ii \gamma \lim_{\mu \to \lambda_j} \tau^\pm_\mu(\bm \lambda_{\hat \jmath}) = \pm \ii \gamma \tau^-_{\lambda_j}(\bm \lambda_{\hat \jmath}). \]

Any singularities of the meromorphic function $\mu \mapsto \tau_\mu(\bm \lambda)$ would be simple poles at $\lambda_j$.
We have
\[ \Res_{\mu \to \lambda_j} \tau_\mu(\bm \lambda) = \lim_{\mu \to \lambda_j} (\mu-\lambda_j)\tau_\mu(\bm \lambda) = \ii \gamma \e^{-\ii \lambda_j L/2} \left(  \tau^+_{\lambda_j}(\bm \lambda_{\hat \jmath}) -\tau^-_{\lambda_j}(\bm \lambda_{\hat \jmath})\e^{\ii \lambda_j L} \right) = 0 , \]
by virtue of the BAEs \rf{BAE}.
Riemann's theorem on removable singularities implies that this function is holomorphically extendable over $\lambda_j$; since this holds for each $j$, this function can be extended to a function that is holomorphic, and hence analytic, on $\bC^N$.
\end{proof}

We will still write $\mu \mapsto \tau_\mu(\bm \lambda)$ for this extended function.
In other words, the BAEs guarantee not only that $\Psi_{\bm \lambda}$ is an eigenfunction of $T_\mu$ but that the corresponding eigenvalue also depends analytically on $\mu \in \bC$.\\

There is a useful partial fraction expansion for $\mu \mapsto \tau_\mu(\bm \lambda)$, reminiscent to equation (2.10) in \cite{Macdonald1}.
\begin{lem} \label{BAevpartialfraction1}
Let $\gamma  \in\bR$ and $(\bm \lambda,\mu) \in \bC^{N+1}_\n{reg}$.
Then
\[ \tau^\pm_\mu(\bm \lambda) = 1\mp\sum_{j=1}^N \frac{\ii \gamma}{\lambda_j-\mu} \tau^\pm_{\lambda_j}(\bm \lambda_{\hat \jmath}). \]
\end{lem}

\begin{proof}
We only prove the case $\pm = -$; the case $\pm = +$ follows by replacing $\gamma$ by $-\gamma$.
Consider the following polynomial function in $\mu$:
\[ p(\mu) := \prod_{j=1}^N (\lambda_j-\mu+ \ii \gamma) - \left( \prod_{j=1}^N (\lambda_j-\mu) + \ii \gamma \sum_{k=1}^N \prod_{j \ne k} (\lambda_j -\lambda_k + \ii \gamma) \frac{\lambda_j-\mu}{\lambda_j-\lambda_k} \right). \]
Because the coefficient of $(\lambda_j-\mu)^N$ vanishes, $p(\mu)$ has degree at most $N-1$ in $\mu$.
On the other hand, we may evaluate, for $l=1,\ldots,N$, 
\[ p(\lambda_l) = \ii \gamma \prod_{j=1 \atop j \ne l}^N (\lambda_j-\lambda_l + \ii \gamma) - \ii \gamma \left( \sum_{k=1 \atop k \ne l}^N \prod_{j=1 \atop j \ne k}^N (\lambda_j-\lambda_k + \ii \gamma) \frac{\lambda_j-\lambda_l}{\lambda_j-\lambda_k} +  \prod_{j=1 \atop j \ne l}^N (\lambda_j-\lambda_l + \ii \gamma) \right) = 0. \]
Since $p(\mu)$ is a polynomial function of degree less than $N$ but with $N$ distinct zeros $\lambda_k$, we conclude that $p(\mu)$ is zero for all $\mu \in \bC$.
It follows that for all $\mu \in \bC \setminus \{\lambda_1,\ldots,\lambda_N\}$ 
\[ 0=p(\mu)/\prod_{j=1}^N(\lambda_j-\mu)=\tau^-_\mu(\bm \lambda) - \left( 1+ \ii \gamma \sum_{j=1}^N \frac{1}{\lambda_j-\mu} \tau^-_{\lambda_j}(\bm \lambda_{\hat \jmath}) \right). \qedhere \] 
\end{proof}

\begin{cor} \label{BAevpartialfraction}
Let $\gamma \in \bR$, $L \in \bR_{>0}$ and $(\bm \lambda,\mu) \in \bC^{N+1}_\n{reg}$.
Assume that $\bm \lambda$ satisfies the BAEs \rf{BAE}.
Then
\[ \tau_\mu(\bm \lambda) = \e^{-\ii \mu L/2}+ \e^{\ii \mu L/2}- \e^{-\ii \mu L/2}\ii \gamma \sum_{j=1}^N \tau^-_{\lambda_j}(\bm \lambda_{\hat \jmath})\frac{\e^{\ii \lambda_j L}-\e^{\ii \mu L}}{\lambda_j-\mu}. \]
\end{cor}

\begin{proof}
Using \rfl{BAevpartialfraction1} we obtain
\begin{align*} 
\tau_\mu(\bm \lambda) &= \tau^+_\mu(\bm \lambda)\e^{-\ii \mu L/2} +\tau^-_\mu(\bm \lambda) \e^{\ii \mu L/2} \\  
&= \e^{-\ii \mu L/2}+ \e^{\ii \mu L/2} - \ii \gamma \sum_{j=1}^N \frac{\tau^+_{\lambda_j}(\bm \lambda_{\hat \jmath}) \e^{-\ii \mu L/2}- \tau^-_{\lambda_j}(\bm \lambda_{\hat \jmath}) \e^{\ii \mu L/2} }{\lambda_j-\mu}. 
\end{align*}
The BAEs \rf{BAE} allow us to change the product $\tau^+_{\lambda_j}(\bm \lambda_{\hat \jmath})$ into $\tau^-_{\lambda_j}(\bm \lambda_{\hat \jmath}) \e^{\ii \lambda_j L}$, and consequently we obtain the desired expression.
\end{proof}

This form of the Bethe Ansatz eigenvalue allows for a nice application; we know that $\mu \mapsto \tau_\mu(\bm \lambda)$ is analytic at $\lambda_j$, and using \rft{BAevpartialfraction} we can calculate $\tau_{\lambda_j}(\bm \lambda)$.
We have, for $l=1,\ldots,N$, 
\begin{align*} 
\tau_{\lambda_j}(\bm \lambda) &= \lim_{\mu \to \lambda_j} \tau_\mu(\bm \lambda) 
\; = \; \e^{-\ii \lambda_j L/2} + \e^{\ii \lambda_j L/2}-\ii \gamma \e^{-\ii \lambda_j L/2} \sum_{k=1}^N \lim_{\mu \to \lambda_j} \tau^-_{\lambda_k}(\bm \lambda_{\hat k})\frac{\e^{\ii \lambda_k L}-\e^{\ii \mu L}}{\lambda_k-\mu}  \\
&= \e^{-\ii \lambda_j L/2} + \e^{\ii \lambda_j L/2}-\ii \gamma \e^{-\ii \lambda_j L/2}\sum_{k=1 \atop k \ne j}^N \tau^-_{\lambda_k}(\bm \lambda_{\hat k}) \frac{\e^{\ii \lambda_j L}-\e^{\ii \lambda_k L}}{\lambda_j-\lambda_k} +\gamma L\e^{-\ii \lambda_j L/2} \tau^-_{\lambda_j}(\bm \lambda_{\hat \jmath}), 
\end{align*}
using De l'H\^opital's rule.

\subsection{The action of $C_\mu$ on the Bethe wavefunction}

To complete this section, we compute $C_\mu  \Psi_{\bm \lambda}$. In physics one is interested in inner products of the form $\innerrnd{\Psi_{\bm \mu}}{\Psi_{\bm \lambda}}$ and for this purpose having expressions for $C_\mu \Psi_{\bm \lambda}$ as linear combinations of certain $\Psi_{\tilde{\bm \lambda}}$ is helpful \cite{Korepin}.

\begin{lem} \label{DADAPsi}
Let $\gamma \in \bR$, $L \in \bR_{>0}$ and $(\bm \lambda',\mu,\nu) \in \bC^{N+1}_\n{reg}$.
Then
\begin{align*}
\lefteqn{\left( D_{\mu}A_{\nu} - D_{\nu}A_{\mu} \right) \Psi_{\bm \lambda'} =} \displaybreak[2] \\
&=  \left( \tau^-_{\mu}(\bm \lambda')\tau^+_{\nu}(\bm \lambda')\e^{\ii (\mu-\nu) L/2} - \tau^-_{\nu}(\bm \lambda')\tau^+_{\mu}(\bm \lambda')\e^{-\ii (\mu-\nu) L/2}  \right) \Psi_{\bm \lambda'} + \\
& \quad - \sum_{j=1}^{N-1} \frac{\ii \gamma}{\lambda_j \! - \! \mu} \left( \tau^-_{\lambda_j}(\bm \lambda'_{\hat \jmath}) \tau^+_{\nu}(\bm \lambda')\e^{\ii (\lambda_j-\nu) L/2} +  \tau^-_{\nu}(\bm \lambda') \tau^+_{\lambda_j}(\bm \lambda'_{\hat \jmath}) \e^{-\ii (\lambda_j - \nu) L/2}  \right) \Psi_{\bm \lambda'_{\hat \jmath},\mu} + \\
& \quad + \sum_{j=1}^{N-1} \frac{\ii \gamma}{\lambda_j \! - \! \nu} \left( \tau^-_{\lambda_j}(\bm \lambda'_{\hat \jmath}) \tau^+_{\mu}(\bm \lambda') \e^{\ii (\lambda_j-\mu) L/2} + \tau^-_{\mu}(\bm \lambda') \tau^+_{\lambda_j}(\bm \lambda'_{\hat \jmath}) \e^{-\ii (\lambda_j-\mu) L/2} \right) \Psi_{\bm \lambda'_{\hat \jmath},\nu} + \\
& \quad - \sum_{j,k=1 \atop j < k}^{N-1} \left( \left( \frac{\lambda_j \! - \! \nu \! - \! \ii \gamma}{\lambda_j \! - \! \nu} \frac{\ii \gamma}{\lambda_j \! - \! \mu}\frac{\ii \gamma}{\lambda_k \! - \! \nu}  - \frac{\lambda_j \! - \! \mu \! - \! \ii \gamma}{\lambda_j \! - \! \mu} \frac{\ii \gamma}{\lambda_j \! - \! \nu}\frac{\ii \gamma}{\lambda_k \! - \! \mu} \right) \tau^-_{\lambda_j}(\bm \lambda'_{\hat \jmath}) \tau^+_{\lambda_k}(\bm \lambda'_{\hat k}) \e^{\ii (\lambda_j-\lambda_k) L/2} + \right. \\
& \quad \; \left. + \left( \frac{\lambda_k \! - \! \nu \! - \! \ii \gamma}{\lambda_k \! - \! \nu} \frac{\ii \gamma}{\lambda_k \! - \! \mu}\frac{\ii \gamma}{\lambda_j \! - \! \nu}  - \frac{\lambda_k \! - \! \mu \! - \! \ii \gamma}{\lambda_k \! - \! \mu} \frac{\ii \gamma}{\lambda_k \! - \! \nu}\frac{\ii \gamma}{\lambda_j \! - \! \mu} \right)  \tau^-_{\lambda_k}(\bm \lambda'_{\hat k}) \tau^+_{\lambda_j}(\bm \lambda'_{\hat \jmath})  \e^{-\ii (\lambda_j-\lambda_k) L/2} \right) \Psi_{\bm \lambda'_{\hat \jmath,\hat k},\mu,\nu}.
\end{align*}
\end{lem}

\begin{proof}
By virtue of \rfp{APsiDPsi} we have
\begin{align*}
\lefteqn{D_{\mu}A_{\nu} \Psi_{\bm \lambda'} =} \\
&=  \tau^-_{\mu}(\bm \lambda')\tau^+_{\nu}(\bm \lambda')\e^{\ii (\mu-\nu) L/2}  \Psi_{\bm \lambda'} + \\
& \quad - \sum_{j=1}^{N-1} \left( \tau^-_{\lambda_j}(\bm \lambda'_{\hat \jmath}) \tau^+_{\nu}(\bm \lambda')\frac{\ii \gamma}{\lambda_j \! - \! \mu} \e^{\ii (\lambda_j-\nu) L/2}  - \tau^-_{\nu}(\bm \lambda'_{\hat \jmath}) \tau^+_{\lambda_j}(\bm \lambda'_{\hat \jmath})\frac{\ii \gamma}{\lambda_j \! - \! \nu}\frac{\ii \gamma}{\mu \! - \! \nu} \e^{-\ii (\lambda_j-\nu) L/2}  \right) \Psi_{\bm \lambda'_{\hat \jmath},\mu} + \\
& \quad + \sum_{j=1}^{N-1} \tau^-_{\mu}(\bm \lambda'_{\hat \jmath}, \nu)\tau^+_{\lambda_j}(\bm \lambda'_{\hat \jmath}) \frac{\ii \gamma}{\lambda_j \! - \! \nu} \e^{-\ii (\lambda_j-\mu) L/2} \Psi_{\bm \lambda'_{\hat \jmath},\nu} + \\
& \quad - \sum_{j,k=1 \atop j \ne k}^{N-1} \tau^-_{\lambda_j}(\bm \lambda'_{\hat \jmath,\hat k},\nu) \tau^+_{\lambda_k}(\bm \lambda'_{\hat k}) \frac{\ii \gamma}{\lambda_j \! - \! \mu}\frac{\ii \gamma}{\lambda_k \! - \! \nu} \e^{\ii (\lambda_j-\lambda_k) L/2} \Psi_{\bm \lambda'_{\hat \jmath,\hat k},\mu,\nu}.
\end{align*}
Now subtract from this the corresponding expression for $D_{\nu}A_{\mu} \Psi_{\bm \lambda'}$. After collecting like terms, we obtain the lemma.
\end{proof}

\begin{prop} \label{CPsi}
Let $\gamma \in \bR$, $L \in \bR_{>0}$ and $(\bm \lambda,\mu) \in \bC^{N+1}_\n{reg}$.
Then $\gamma C_\mu \Psi_{\bm \lambda} =$
\begin{align*} 
&= -\sum_{j=1}^N \frac{\ii \gamma}{\lambda_j \! - \! \mu}\left(\tau^-_{\lambda_j}(\bm \lambda_{\hat \jmath})
\tau^+_\mu(\bm \lambda_{\hat \jmath}) \e^{\ii (\lambda_j-\mu) L/2}-\tau^-_\mu(\bm \lambda_{\hat \jmath})\tau^+_{\lambda_j}(\bm \lambda_{\hat \jmath})\e^{-\ii (\lambda_j-\mu) L/2}  \right) \Psi_{\bm \lambda_{\hat \jmath}} + \\
& \quad -\sum_{j,k=1 \atop j<k}^N \frac{\ii \gamma}{\lambda_j \! - \! \mu}\frac{\ii \gamma}{\lambda_k \! - \! \mu}\left( \tau^-_{\lambda_j}(\bm \lambda_{\hat \jmath})\tau^+_{\lambda_k}(\bm \lambda_{\hat \jmath,\hat k}) \e^{\ii (\lambda_j-\lambda_k) L/2} +
 \tau^-_{\lambda_k}(\bm \lambda_{\hat k}) \tau^+_{\lambda_j}(\bm \lambda_{\hat \jmath,\hat k})\e^{-\ii (\lambda_j-\lambda_k) L/2} \right) \Psi_{\bm \lambda_{\hat \jmath, \hat k},\mu}.
\end{align*}
If $\bm \lambda$ satisfies the BAEs \rf{BAE} this simplifies to
\begin{align*} 
\gamma C_\mu \Psi_{\bm \lambda} &=-\sum_{j=1}^N\frac{\ii \gamma \e^{\ii \lambda_j L/2}}{\lambda_j \! - \! \mu} \tau^-_{\lambda_j}(\bm \lambda_{\hat \jmath}) \left(\tau^+_\mu(\bm \lambda_{\hat \jmath})\e^{-\ii \mu L/2}-\tau^-_\mu(\bm \lambda_{\hat \jmath}) \e^{\ii \mu L/2}\right) \Psi_{\bm \lambda_{\hat \jmath}} +\\
& \quad - 2 \sum_{j,k=1 \atop j<k}^N \frac{\ii \gamma}{\lambda_j \! - \! \mu}\frac{\ii \gamma}{\lambda_k \! - \! \mu} \e^{\ii (\lambda_j+\lambda_k)L/2}\tau^-_{\lambda_j}(\bm \lambda_{\hat \jmath, \hat k}) \tau^-_{\lambda_k}(\bm \lambda_{\hat \jmath, \hat k})  \Psi_{\bm \lambda_{\hat \jmath, \hat k},\mu}. 
\end{align*}
\end{prop}

\begin{proof}
Note that the second statement follows immediately from the first statement,
which we prove by induction. Note that for $N=0$ we indeed recover $C_\mu \Vac = 0$.
For the induction step, write $\bm \lambda' = (\lambda_1,\ldots,\lambda_{N-1}) \in \bC^{N-1}_\n{reg}$.
Now by \rfeqn{CBcommrel} we obtain
\[ \gamma C_\mu \Psi_{\bm \lambda} = B_{\lambda_N}\gamma C_\mu \Psi_{\bm \lambda'}-\frac{\ii \gamma}{\lambda_N \! - \! \mu}\left(D_{\lambda_N}A_\mu -D_\mu A_{\lambda_N}\right) \Psi_{\bm \lambda'}. \]
The induction hypothesis yields
\begin{align*}
\lefteqn{B_{\lambda_N}\gamma C_\mu \Psi_{\bm \lambda'} =}\\
&= -\sum_{j=1}^{N-1} \frac{\ii \gamma}{\lambda_j \! - \! \mu}\left( \tau^-_{\lambda_j}(\bm \lambda'_{\hat \jmath})\tau^+_{\mu}(\bm \lambda'_{\hat \jmath}) \e^{\ii (\lambda_j-\mu) L/2} -\tau^-_{\mu}(\bm \lambda'_{\hat \jmath})\tau^+_{\lambda_j}(\bm \lambda'_{\hat \jmath}) \e^{-\ii (\lambda_j-\mu) L/2} \right) \Psi_{\bm \lambda_{\hat \jmath}} + \\
& \quad -\sum_{j,k=1 \atop j<k}^{N-1} \frac{\ii \gamma}{\lambda_j \! - \! \mu}\frac{\ii \gamma}{\lambda_k \! - \! \mu}\left( \tau^-_{\lambda_j}(\bm \lambda'_{\hat \jmath})\tau^+_{\lambda_k}(\bm \lambda'_{\hat \jmath, \hat k}) \e^{\ii (\lambda_j-\lambda_k) L/2} + \tau^-_{\lambda_k}(\bm \lambda'_{\hat k}) \tau^+_{\lambda_j}(\bm \lambda'_{\hat \jmath, \hat k}) \e^{-\ii (\lambda_j-\lambda_k) L/2} \right) \Psi_{\bm \lambda_{\hat \jmath, \hat k},\mu},
\end{align*}
whereas from \rfl{DADAPsi} we obtain the much-aligned expression
{\small
\begin{align*}
\lefteqn{- \frac{\ii \gamma}{\lambda_N \! - \! \mu} \left( D_{\lambda_N}A_\mu  - D_\mu A_{\lambda_N} \right) \Psi_{\bm \lambda'} =} \nonumber \displaybreak[2] \\
&= - \frac{\ii \gamma}{\lambda_N \! - \! \mu} \left( \tau^-_{\lambda_N}(\bm \lambda')\tau^+_\mu(\bm \lambda')\e^{\ii (\lambda_N-\mu) L/2} - \tau^-_\mu(\bm \lambda')\tau^+_{\lambda_N}(\bm \lambda')\e^{-\ii(\lambda_N-\mu) L/2}  \right) \Psi_{\bm \lambda'} +  \\
& \quad + \frac{\ii \gamma}{\lambda_N \! - \! \mu}\sum_{j=1}^{N-1} \frac{\ii \gamma}{\lambda_j \! - \! \lambda_N}\left( \tau^-_{\lambda_j}(\bm \lambda'_{\hat \jmath}) \tau^+_\mu(\bm \lambda') 
\e^{\ii (\lambda_j-\mu) L/2} + \tau^-_\mu(\bm \lambda') \tau^+_{\lambda_j}(\bm \lambda'_{\hat \jmath}) \e^{-\ii (\lambda_j-\mu) L/2}  \right) \Psi_{\bm \lambda_{\hat \jmath}} +  \\
& \quad - \frac{\ii \gamma}{\lambda_N \! - \! \mu} \sum_{j=1}^{N-1} \frac{\ii \gamma}{\lambda_j \! - \! \mu} \left( \tau^-_{\lambda_j}(\bm \lambda'_{\hat \jmath}) \tau^+_{\lambda_N}(\bm \lambda') \e^{\ii (\lambda_j-\lambda_N) L/2} + \tau^-_{\lambda_N}(\bm \lambda') \tau^+_{\lambda_j}(\bm \lambda'_{\hat \jmath}) \e^{-\ii (\lambda_j-\lambda_N) L/2}  \right) \Psi_{\bm \lambda'_{\hat \jmath},\mu} +  \\
& \quad + \frac{\ii \gamma}{\lambda_N \! - \! \mu}\sum_{j,k=1 \atop j < k}^N \left( \left(
\frac{\lambda_j \! - \! \mu \! - \! \ii \gamma}{\lambda_j \! - \! \mu} \frac{\ii \gamma}{\lambda_j \! - \! \lambda_N}\frac{\ii \gamma}{\lambda_k \! - \! \mu}  - \frac{\lambda_j \! - \! \lambda_N \! - \! \ii \gamma}{\lambda_j \! - \! \lambda_N} \frac{\ii \gamma}{\lambda_j \! - \! \mu}\frac{\ii \gamma}{\lambda_k \! - \! \lambda_N} \right) \tau^-_{\lambda_j}(\bm \lambda'_{\hat \jmath}) \tau^+_{\lambda_k}(\bm \lambda'_{\hat k}) \e^{\ii (\lambda_j-\lambda_k) L/2} + \right. \\
& \qquad \left. + \left( \frac{\lambda_k \! - \! \mu \! - \! \ii \gamma}{\lambda_k \! - \! \mu} \frac{\ii \gamma}{\lambda_k \! - \! \lambda_N}\frac{\ii \gamma}{\lambda_j \! - \! \mu}  - \frac{\lambda_k \! - \! \lambda_N \! - \! \ii \gamma}{\lambda_k \! - \! \lambda_N} \frac{\ii \gamma}{\lambda_k \! - \! \mu}\frac{\ii \gamma}{\lambda_j \! - \! \lambda_N} \right) \tau^-_{\lambda_k}(\bm \lambda'_{\hat k}) \tau^+_{\lambda_j}(\bm \lambda'_{\hat \jmath})  \e^{-\ii (\lambda_j-\lambda_k) L/2} \right) \Psi_{\bm \lambda_{\hat \jmath,\hat k},\mu}.
\end{align*}}
It follows that the coefficient of $\Psi_{\bm \lambda_{\hat \jmath}}$ ($1 \leq j \leq N-1$) in $\gamma  C_\mu \Psi_{\bm \lambda}$ equals
\begin{gather*}
\left( \frac{\ii \gamma}{\lambda_N \! - \! \mu} \frac{\ii \gamma}{\lambda_j \! - \! \lambda_N} \tau^-_{\lambda_j}(\bm \lambda'_{\hat \jmath}) \tau^+_\mu(\bm \lambda') -\frac{\ii \gamma}{\lambda_j \! - \! \mu} \tau^+_{\mu}(\bm \lambda'_{\hat \jmath})\tau^-_{\lambda_j}(\bm \lambda'_{\hat \jmath}) \right) \e^{\ii (\lambda_j-\mu) L/2}+ \qquad \qquad \\
\qquad \qquad \qquad + \left( \frac{\ii \gamma}{\lambda_j \! - \! \mu} \tau^-_{\mu}(\bm \lambda'_{\hat \jmath})\tau^+_{\lambda_j}(\bm \lambda'_{\hat \jmath}) + \frac{\ii \gamma}{\lambda_N \! - \! \mu} \frac{\ii \gamma}{\lambda_j \! - \! \lambda_N} \tau^-_\mu(\bm \lambda') \tau^+_{\lambda_j}(\bm \lambda'_{\hat \jmath}) \right) \e^{-\ii (\lambda_j-\mu) L/2} \\
\qquad = \frac{-\ii \gamma}{\lambda_j \! - \! \mu}\left( \tau^-_{\lambda_j}(\bm \lambda_{\hat \jmath}) \tau^+_\mu(\bm \lambda_{\hat \jmath}) \e^{\ii (\lambda_j-\mu) L/2} - \tau^-_\mu(\bm \lambda_{\hat \jmath}) \tau^+_{\lambda_j}(\bm \lambda_{\hat \jmath}) \e^{-\ii (\lambda_j-\mu) L/2} \right);
\end{gather*}
the corresponding terms together with the term proportional to $\Psi_{\bm \lambda'}$ combine to the desired sum over $j=1,\ldots,N$.
Furthermore, the coefficient of $\Psi_{\bm \lambda_{\hat \jmath, \hat k}, \mu}$ ($1 \leq j<k \leq N $) equals
\begin{align*}
\lefteqn{ \left( -\frac{\ii \gamma}{\lambda_j \! - \! \mu}\frac{\ii \gamma}{\lambda_k \! - \! \mu} + \frac{\ii \gamma}{\lambda_N \! - \! \mu}\left(
\frac{\lambda_j \! - \! \mu \! - \! \ii \gamma}{\lambda_j \! - \! \mu} \frac{\ii \gamma}{\lambda_j \! - \! \lambda_N}\frac{\ii \gamma}{\lambda_k \! - \! \mu} - \frac{\lambda_j \! - \! \lambda_N \! - \! \ii \gamma}{\lambda_j \! - \! \lambda_N} \frac{\ii \gamma}{\lambda_j \! - \! \mu}\frac{\ii \gamma}{\lambda_k \! - \! \lambda_N} \right) \right) \cdot } \\
& \hspace{60mm} \cdot \tau^-_{\lambda_j}(\bm \lambda'_{\hat \jmath}) \tau^+_{\lambda_k}(\bm \lambda'_{\hat k}) \e^{\ii (\lambda_j-\lambda_k) L/2} + \\
& + \left( -\frac{\ii \gamma}{\lambda_j \! - \! \mu}\frac{\ii \gamma}{\lambda_k \! - \! \mu} + \frac{\ii \gamma}{\lambda_N \! - \! \mu}\left( \frac{\lambda_k \! - \! \mu \! - \! \ii \gamma}{\lambda_k \! - \! \mu} \frac{\ii \gamma}{\lambda_k \! - \! \lambda_N}\frac{\ii \gamma}{\lambda_j \! - \! \mu}  - \frac{\lambda_k \! - \! \lambda_N \! - \! \ii \gamma}{\lambda_k \! - \! \lambda_N} \frac{\ii \gamma}{\lambda_k \! - \! \mu}\frac{\ii \gamma}{\lambda_j \! - \! \lambda_N} \right) \right) \cdot \\
& \hspace{60mm} \cdot \tau^-_{\lambda_k}(\bm \lambda'_{\hat k}) \tau^+_{\lambda_j}(\bm \lambda'_{\hat \jmath})  \e^{-\ii (\lambda_j-\lambda_k) L/2} \\
&= \frac{-\ii \gamma}{\lambda_j \! - \! \mu}\frac{\ii \gamma}{\lambda_k \! - \! \mu} \left( \tau^-_{\lambda_j}(\bm \lambda_{\hat \jmath}) \tau^+_{\lambda_k}(\bm \lambda_{\hat \jmath,\hat k}) \e^{\ii (\lambda_j-\lambda_k) L/2} +\tau^-_{\lambda_k}(\bm \lambda_{\hat k}) \tau^+_{\lambda_j}(\bm \lambda_{\hat \jmath,\hat k}) \e^{-\ii (\lambda_j-\lambda_k) L/2}  \right);
\end{align*}
the corresponding terms together with the term proportional to $\Psi_{\bm \lambda'_{\hat \jmath},\mu}=\Psi_{\bm \lambda_{\hat \jmath,\hat N},\mu}$ give the desired expression for the sum over $j$ and $k$.
\end{proof}

\section{The limiting case $J = \bR$} \label{infiniteJ}

Na\" ively one could expect that as $L \to \infty$ in the ABA for the bounded interval, we obtain that $\Psi_{\bm \lambda}$ describes the system of $N$ bosonic particles moving along $\bR$. However, as made clear in \cite[Section 8]{Gutkin1988}, this limit is very subtle. In particular, the QNLS creation operators $\lim_{L \to \infty} B_\lambda $ when seen as operators on $\ca H(\bR)$ have a trivial domain for $\lambda \in \bR$. Although the operator $B_\lambda $ does not explicitly depend on $L$, cf. \rfeqn{Bintoperators}, simply letting $B_\lambda $ act on an element of $\ca H_N(\bR)$ does not produce an element of $\ca H_{N+1}(\bR)$. \\

In particular, the Bethe wavefunctions $\Psi_{\bm \lambda}$ are not square-integrable. This is not surprising; the solutions for the non-interacting case, i.e. symmetrized plane waves $\ca S_N \e^{\ii \bm \lambda}$ are not square-integrable either, because the non-symmetric plane waves $\e^{\ii \bm \lambda}$ are not.
However, there is a weaker sense of completeness of the plane waves in $\f h_N = \L^2(\bR^N)$, and hence for the symmetrized plane waves in $\ca H_N(\bR)$, afforded by the Fourier transform on $\L^2(\bR^N)$, which is a unitary operator as per the Plancherel theorem (see, e.g. \cite{ReedSimon}).
For all $f \in \f h_N$
\[ f(\bm x) = \int_{\bR^N} \dd^N \bm \lambda \tilde f(\bm \lambda) \e^{\ii \inner{\bm \lambda}{\bm x}} \]
for some $\tilde f \in \f h_N$, which can be seen as a linear combination of (possibly uncountably infinitely many) plane waves $\e^{\ii \bm \lambda}$.
It is not difficult to see that an analogous result holds for all $F \in \ca H_N$:
\[ F(\bm x) = \int_{\bR^N} \dd^N \bm \lambda \tilde F(\bm \lambda) \left( \ca S_N \e^{\ii \bm \lambda}\right)(\bm x), \]
which is what is sometimes referred to in the literature as the ``completeness (or closure) of the free eigenstates in $\L^2(\bR^N)$'', despite the fact that the symmetrized plane waves are not square-integrable (\cite{Gaudin1971-1}, \cite[Section 2]{Gutkin1988}).
The result by Gaudin \cite{Gaudin1971-1,Gaudin1971-2} shows that the $\Psi_{\bm \lambda}$ play a similar role in $\ca H_N$; for all $F \in \ca H_N$ we have
\[ F(\bm x) = \int_{\bR^N} \dd^N \bm \lambda \tilde F_\gamma(\bm \lambda) \Psi_{\bm \lambda}(\bm x), \]
for some ``deformed Fourier coefficients'' $\tilde F_\gamma \in \ca H_N$.

\subsection{Action of the QISM integral operators on $\Psi_{\bm \lambda}$ defined on the line}

The operation $B_\lambda$ defined by the formulae in \rfeqn{Bintoperators} is well-defined as an operator$: \ca C(\bR^N)^{S_N} \to \ca C(\bR^{N+1})^{S_{N+1}}$ (the limits of integration are all bounded in \rfeqn{Bintoperators}).
The operators $A_\lambda,D_\lambda$ explicitly depend on $L$, cf. \rfeqn{Aintoperators} and \rfeqn{Dintoperators}.
\begin{prop} \emph{\cite[Section 8]{Gutkin1988}} \label{Linfty}
Let $L \in \bR_{>0}$, $\gamma \in \bR$, $\bm \lambda \in \bR^N_\n{reg}$ and $\mu \in \bC \setminus \{ \lambda_1,\ldots,\lambda_N\}$.
Then 
\[ B_\mu \Psi_{\bm \lambda} = \Psi_{\bm \lambda,\mu} \in \ca C(\bR^N) \]
and, for $\bm x \in \bR^N$,
\begin{align*}
\lim_{L \to \infty} \e^{\ii \mu L/2} \left(A_\mu \Psi_{\bm \lambda} \right)(\bm x) &= \tau^+_\mu(\bm \lambda) \Psi_{\bm \lambda}(\bm x), && \n{if } \Im \mu>0, \\
\lim_{L \to \infty} \e^{-\ii \mu L/2} \left(D_\mu \Psi_{\bm \lambda} \right)(\bm x) &= \tau^-_\mu(\bm \lambda) \Psi_{\bm \lambda}(\bm x), && \n{if } \Im \mu<0.
\end{align*}
Furthermore
\begin{align*}
\lim_{L \to \infty} \e^{\ii \mu L/2} \left(T_\mu \Psi_{\bm \lambda} \right)(\bm x) &= \tau^+_\mu(\bm \lambda) \Psi_{\bm \lambda}(\bm x), && \n{if } \Im \mu>0, \\
\lim_{L \to \infty} \e^{-\ii \mu L/2} \left(T_\mu \Psi_{\bm \lambda} \right)(\bm x) &= \tau^-_\mu(\bm \lambda) \Psi_{\bm \lambda}(\bm x), && \n{if } \Im \mu<0.
\end{align*}
Using these expressions and the aforementioned deformed Fourier formalism one can define $\lim_{L \to \infty} \e^{\ii \mu L/2} A_\mu$, $\lim_{L \to \infty} \e^{-\ii \mu L/2} D_\mu$ and $\lim_{L \to \infty} \e^{\pm \ii \mu L/2} T_\mu$ as operators on $\ca H_N(\bR)$ for suitable non-real values of $\mu$; this defines bounded operators on $\ca H_N(\bR)$.
\end{prop}

\begin{proof}
The expression for $B_\mu  \Psi_{\bm \lambda}$ in the limit $L \to \infty$ is proven in \cite[Thm.~8.2.3]{Gutkin1988}. 
We refer to \cite[Eqns.~(8.3.1)~and~(8.3.2)]{Gutkin1988} for the first expressions for $\e^{\ii \mu L/2} A_\mu  \Psi_{\bm \lambda}$ and $\e^{-\ii \mu L/2} D_\mu  \Psi_{\bm \lambda}$ in the limit $L \to \infty$ and \cite[Eqn.~(8.3.3)]{Gutkin1988} for the boundedness of $A$ and $D$.
The statement about the transfer matrices follows by noting that for $\Im \mu>0$, $\e^{\ii \mu L/2}D_\mu $ is exponentially damped as $L \to \infty$, and likewise for the case $\Im \mu<0$.
\end{proof}

\section{The quantum determinant and the Yangian}

Using the commutation relations from \rfeqn{Tcommrel} we may construct an element of $\ca H$ which commutes with $A$, $B$, $C$ and $D$, i.e. it is in the centre of the Yang-Baxter algebra.

\begin{defn}
Let $\gamma \in \bR$ and $\lambda \in \bC$. Write $\lambda_\pm = \lambda \mp \ii \gamma/2$.%
\nc{gllpm}{$\lambda_{\pm}$}{Shorthand for $\lambda \mp \ii \gamma/2$} 
The \emph{quantum determinant} of the monodromy matrix of the QNLS model is the operator
\begin{equation} 
\label{qdet}
\qdet \ca T_\lambda = A_{\lambda_+} D_{\lambda_-} - \gamma B_{\lambda_+} C_{\lambda_-}. 
\end{equation}
\nc{rql}{$\qdet \ca T_\lambda$}{Quantum determinant of the monodromy matrix \nomrefeqpage}
\end{defn}
\vspace{-8mm}

Note that $\qdet \ca T_\lambda \in \End(\ca H_N)$ and it satisfies
\begin{equation} \label{qdetselfadjoint} \qdet \ca T_\lambda^* = \qdet \ca T_{\bar \lambda}.
\end{equation}

\begin{prop} \label{qdetMcomm}
$[\ca T^{j \, k}_\lambda,\qdet \ca T_\mu]=0$ for all $j,k = 1,2$ and all $\lambda,\mu \in \bC$ with $\lambda \ne \mu_\pm$. That is, $\qdet \ca T_\mu$ is in the centre of the Yang-Baxter algebra.
\end{prop}

\begin{proof}
In light of \rfeqn{ABCDadjoint} and \rfeqn{qdetselfadjoint} it suffices to prove $[A_\lambda,\qdet \ca T_\mu]=[B_\lambda ,\qdet \ca T_\mu]=0$.
Note that 
\[ \left[A_\lambda,\qdet \ca T_\mu\right] = \left[ A_\lambda,A_{\mu_+}D_{\mu_-} \right]-\gamma \left[ A_\lambda,B_{\mu_+}C_{\mu_-}  \right]. \]
Dealing with the first commutator, we have
\[ \left[ A_\lambda,A_{\mu_+}D_{\mu_-}  \right] = A_{\mu_+} \left[ A_\lambda,D_{\mu_-} \right] \\
= -\frac{\ii \gamma^2}{\lambda \! - \! \mu_-}A_{\mu_+}\left( B_\lambda C_{\mu_-}  - B_{\mu_-}C_\lambda \right),\]
by virtue of \rfeqn{ADcommrel}.
It follows that for $\left[A_\lambda,\qdet \ca T_\mu\right]=0$ it suffices to prove
\begin{equation} \label{eqn261} \left[ A_\lambda,B_{\mu_+}C_{\mu_-}  \right] = -\frac{\ii \gamma}{\lambda \! - \! \mu_-}A_{\mu_+}\left( B_\lambda C_{\mu_-}  - B_{\mu_-}C_\lambda \right). \end{equation}
\rfeqn{CAcommrel} yields
\begin{align}
C_{\mu_-} A_\lambda 
&= \frac{\lambda - \mu_+}{\lambda - \mu_-} A_\lambda C_{\mu_-}  - \frac{\ii \gamma}{\lambda - \mu_-} A_{\mu_-}C_\lambda . \label{eqn262}
\end{align}
On the other hand, \rfeqn{BAcommrel} gives
\begin{align}
B_{\mu_+}A_\lambda 
&= \frac{\lambda - \mu_-}{\lambda - \mu_+} A_\lambda B_{\mu_+} 
+ \frac{\ii \gamma}{\lambda - \mu_+} A_{\mu_+} B_\lambda . \label{eqn263}
\end{align}
In particular it follows that
\begin{equation} B_{\mu_+}A_{\mu_-} = A_{\mu_+}B_{\mu_-}. \label{eqn264} \end{equation}
Combining \rfeqnser{eqn262}{eqn264} we find that
\begin{align*}
B_{\mu_+}C_{\mu_-} A_\lambda &= \frac{\lambda - \mu_+}{\lambda - \mu_-} B_{\mu_+}A_\lambda C_{\mu_-}  - \frac{\ii \gamma}{\lambda - \mu_-} B_{\mu_+}A_{\mu_-}C_\lambda \\
&= \frac{\lambda - \mu_+}{\lambda - \mu_-} \left( \frac{\lambda - \mu_-}{\lambda - \mu_+} A_\lambda B_{\mu_+} + \frac{\ii \gamma}{\lambda - \mu_+} A_{\mu_+} B_\lambda  \right) C_{\mu_-} - \frac{\ii \gamma}{\lambda - \mu_-} A_{\mu_+}B_{\mu_-}C_\lambda \\
&=A_\lambda B_{\mu_+}C_{\mu_-}  + \frac{\ii \gamma}{\lambda - \mu_-} A_{\mu_+} \left( B_\lambda  C_{\mu_-}  - B_{\mu_-}C_\lambda  \right),
\end{align*}
which is equivalent to \rfeqn{eqn261}. \\

To show that $B_\lambda $ and $\qdet \ca T_\mu$ commute, we note that 
\[ \left[B_\lambda ,\qdet \ca T_\mu\right] =\left[ B_\lambda , A_{\mu_+}D_{\mu_-} \right]-\gamma B_{\mu_+} \left[ B_\lambda ,C_{\mu_-}  \right]. \]
We have from \rfeqn{DBcommrel}
\begin{equation} \label{eqn265} D_{\mu_-}B_\lambda  = \frac{\lambda-\mu_+}{\lambda-\mu_-} B_\lambda D_{\mu_-} - \frac{\ii \gamma}{\lambda-\mu_-} B_{\mu_-}D_\lambda , \end{equation}
and from \rfeqn{BAcommrel}
\begin{equation} \label{eqn266} B_\lambda A_{\mu_+} 
= \frac{\lambda - \mu_+ + \ii \gamma}{\lambda - \mu_+} A_{\mu_+} B_\lambda  
- \frac{\ii \gamma}{\lambda - \mu_+} A_\lambda B_{\mu_+}. \end{equation}
From \rfeqnser{eqn265}{eqn266} we obtain
\begin{align*} \left[ B_\lambda ,A_{\mu_+}D_{\mu_-} \right] 
&= \left( \frac{\lambda - \mu_+ + \ii \gamma}{\lambda - \mu_+} A_{\mu_+} B_\lambda  
- \frac{\ii \gamma}{\lambda - \mu_+} A_\lambda B_{\mu_+}\right) D_{\mu_-} + \\
& \qquad - A_{\mu_+} \left( \frac{\lambda-\mu_+}{\lambda-\mu_-} B_\lambda D_{\mu_-} - 
\frac{\ii \gamma}{\lambda-\mu_-} B_{\mu_-}D_\lambda  \right) \\
&= -\frac{\ii \gamma}{\lambda-\mu_-}\frac{\ii \gamma}{\lambda-\mu_+}A_{\mu_+}B_\lambda D_{\mu_-} 
+ \frac{\ii \gamma}{\lambda-\mu_-} A_{\mu_+} B_{\mu_-}D_\lambda  
- \frac{\ii \gamma}{\lambda-\mu_+} A_\lambda B_{\mu_+} D_{\mu_-}. 
\end{align*}
On the other hand, we have by virtue of \rfeqn{BCcommrel} and \rfeqnser{eqn263}{eqn264},
\begin{align*} 
\gamma B_{\mu_+} \left[ B_\lambda ,C_{\mu_-}  \right] &= -\frac{\ii \gamma}{\lambda-\mu_-} B_{\mu_+} \left( A_\lambda D_{\mu_-}-A_{\mu_-}D_\lambda \right) \displaybreak[2] \\
&= -\frac{\ii \gamma}{\lambda-\mu_-} \left(  \left( \frac{\lambda - \mu_-}{\lambda - \mu_+} A_\lambda B_{\mu_+} - \frac{\ii \gamma}{\lambda - \mu_+} A_{\mu_+} B_\lambda  \right) D_{\mu_-} - A_{\mu_+} B_{\mu_-}D_\lambda \right)  \displaybreak[2] \\
&= \frac{\ii \gamma}{\lambda-\mu_-} A_{\mu_+} B_{\mu_-}D_\lambda  
 - \frac{\ii \gamma}{\lambda- \mu_+} A_\lambda B_{\mu_+} D_{\mu_-}  - \frac{\ii \gamma}{\lambda-\mu_-} \frac{\ii \gamma}{\lambda - \mu_+} A_{\mu_+} B_\lambda  D_{\mu_-},
\end{align*}
and indeed we see that $\left[B_\lambda , A_{\mu_+}D_{\mu_-} \right] = \gamma \left[ B_\lambda , B_{\mu_+}C_{\mu_-}  \right]$.
\end{proof}

\begin{rem}
This result means that $\qdet \ca T_\mu$ plays the role of the Casimir element in the Yang-Baxter algebra.
\end{rem}

\begin{cor} \label{qdetPsi}
Let $\bm \lambda \in \bC^N_\n{reg}$.
Also, let $\mu \in \bC$ such that both $\mu_\pm$ are unequal to any of the $\lambda_j$.
Then
\begin{equation} \label{qdetBethevector} \qdet \ca T_\mu \Psi_{\bm \lambda} = \e^{-\gamma L/2} \Psi_{\bm \lambda}. \end{equation}
\end{cor}

\begin{proof}
Using \rfl{qdetMcomm}, we have
\[ \qdet \ca T_\mu \Psi_{\bm \lambda} = \qdet \ca T_\mu \left( \prod_{j=1}^N B_{\lambda_j} \right) \Vac =  \left( \prod_{j=1}^N B_{\lambda_j} \right) \qdet \ca T_\mu \Vac, \]
and
\[ \qdet \ca T_\mu \Vac = A_{\mu_+}D_{\mu_-} \Vac = \e^{-\ii(\mu-\ii \gamma/2)L/2} \e^{\ii(\mu+\ii \gamma/2)L/2} \Vac  = \e^{-\gamma L/2}\Vac. \qedhere\]
\end{proof}

\rfl{qdetPsi} tells us that the Bethe wavefunction $\Psi_{\bm \lambda}$ is an eigenfunction of the quantum determinant, with the corresponding eigenvalue independent of $\bm \lambda$. We note that this is without imposing the BAEs \rfeqn{BAE} on $\bm \lambda$. Since the $\Psi_{\bm \lambda}$ are a complete set in $\ca H$, it follows that $\qdet \ca T_\mu$ acts as multiplication by the constant $\e^{-\gamma L/2}$ throughout $\ca H$.\\

The quantum determinant provides a connection with the Yangian \cite{ChariPressley, Molev, Drinfeld1985, Drinfeld1986}.
The \emph{Yangian} of $\f{gl}_2$ is a Hopf algebra, more precisely a deformation of the current algebra of $\f{gl}_2$ \cite[Chapter 12]{ChariPressley}. More precisely, it is defined to be the algebra $Y(\f{gl}_2)$ generated by the elements $\ca T^{j \, k}_{(l)}$ where $j,k \in \{1,2\}$ and $l \in \bZ_{\geq 0}$ subject to
\[ \left[\ca T^{j_1 \, k_1}_{(l_1+1)},\ca T^{j_2 \, k_2}_{(l_2)}\right] - \left[\ca T^{j_1 \, k_1}_{(l_1)},\ca T^{j_2 \, k_2}_{(l_2+1)}\right] = - \left( \ca T^{j_2 \, k_1}_{(l_1)} \ca T^{j_1 \, k_2}_{(l_2)} - \ca T^{j_2 \, k_1}_{(l_2)} \ca T^{j_1 \, k_2}_{(l_1)} \right), \]
which is what one would obtain from \rfc{Tcommrelcor} by writing $\ca T^{j \, k}_\lambda = \sum_{l \geq -1}  \ca T^{j \, k}_{(l)} \lambda^{-l+1}$ and defining $\ca T^{j \, k}_{(-1)} = \delta_{j \, k}$, %
\nc{gdl}{$\delta_{x \, y}$}{Kronecker delta}%
where we have introduced the \emph{Kronecker delta}
\[ \delta: X \times X \to \{0,1\}: (x,y) \mapsto \delta_{x \, y} := \begin{cases} 1,& x=y \\ 0,& \n{otherwise}, \end{cases} \]
for any set $X$.
Owing to the exchange relation \rf{ybeRT} $Y(\f{gl}_2)$ becomes a (quasitriangular) Hopf algebra with comultiplication $\Delta$, counit $\epsilon$ and antipode $s$ given by
\[ (\Delta \otimes \n{id}) \ca T_\lambda = \left(\ca T_\lambda\right)_{1 \, 2} \left(\ca T_\lambda\right)_{1 \, 3}, \qquad 
\left( \epsilon \otimes \n{id} \right) \ca T_\lambda = 1, \qquad \left(s \otimes \n{id}\right) \ca T_\lambda = \ca T_\lambda^{-1}. \]
The quantum determinant $\qdet \ca T_\lambda$ generates the centre of $Y(\f{gl}_2)$.

\section{Alternative formulae for the QISM integral operators} \label{altformulae}

We return to the explicit integral formulae for the operators $A_\mu,B_\mu,C_\mu,D_\mu$ as presented in \rfeqnser{Aintoperators}{Dintoperators}. There is a practical disadvantage to these formulae, which we will explain now.
Functions $f \in \ca H_N$ are generally defined on the fundamental alcove $J^N_+$ (or its closure) and then extended to $J^N$ by ordering the entries of the argument $\bm x$ in decreasing order, i.e. by applying $\sum_{w \in S_N} w \chi_{J^N_+}$ (only that term with $w^{-1} \bm x \in J^N_+$ remains). 
This means that in order to calculate $(A_\lambda f)(\bm x)$, say, for each $\bm i \in \f I^n_{N}$ we need to arrange $(x_1, \ldots, \widehat{x_{i_1}}, \ldots, \widehat{x_{i_n}}, \ldots, x_N, y_1, \ldots, y_n) \in J^N$ in a decreasing order so that we can use the formula for $F(\bm x)$ for $\bm x \in J^N_+$.
Because $y_m$ runs from $x_{i_m}$ down to $x_{i_{m+1}}$, it can assume any position among the intermediate coordinates $x_{i_m-1} > x_{i_m-2} > \ldots > x_{i_{m+1}+1}$.
In other words, in general we have no control over the order of the arguments of $F$, which leads to computational issues; this we will now address.

\begin{exm} \label{examplealtform}
Consider the action of $A_\lambda$ on a function $F \in \ca H_4$.
Specifically, concentrate on the terms in the summand with $n=2$.
The set $\f I^2_4$ contains the tuples $(1,2)$, $(1,3)$, $(1,4)$, $(2,3)$, $(2,4)$, $(3,4)$.
Taking $\bm i = (3,4)$ gives
\[ \left( \bar E^+_{\lambda;3 \, 4} F \right)(x_1,x_2,x_3,x_4) = \int_{x_4}^{x_3} \dd y_1 \int_{-L/2}^{x_4} \dd y_2 \e^{\ii \lambda (x_3+x_4-y_1-y_2)}  F \left( x_1, x_2, y_1, y_2 \right). \]
Note that $x_1 > x_2 > y_1 > y_2$ so the arguments of $F$ for this term can be unconditionally rearranged to be in the alcove $J^4_+$.
However, taking $\bm i = (1,4)$ gives
\[ \left( \bar E^+_{\lambda;1 \, 4} F \right)(x_1,x_2,x_3,x_4)= \int_{x_4}^{x_1} \dd y_1 \int_{-L/2}^{x_4} \dd y_2 \e^{\ii \lambda (x_1+x_4-y_1-y_2)}  F \left( x_2, x_3, y_1, y_2 \right). \]
We do have $x_2>x_3>y_2$ but for $y_1$ we have three possibilities: $y_1 > x_2$, $x_2>y_1>x_3$, and $x_3>y_1>x_4$, so we need to do some work before we can rearrange the arguments of $F$ for this term.
The idea is simply to split up the integral over $[x_4,x_1]$ in three integrals: $\int_{x_4}^{x_1} = \int_{x_2}^{x_1}+\int_{x_3}^{x_2}+\int_{x_4}^{x_3}$.
Hence
\begin{align*} 
\left( \bar E^+_{\lambda;1 \, 4} F \right)(x_1,x_2,x_3,x_4) &= \int_{x_2}^{x_1} \dd y_1 \int_{-L/2}^{x_4} \dd y_2 \e^{\ii \lambda (x_1+x_4-y_1-y_2)} F \left( y_1,x_2, x_3, y_2 \right) + \\
& \qquad + \int_{x_3}^{x_2} \dd y_1 \int_{-L/2}^{x_4} \dd y_2 \e^{\ii \lambda (x_1+x_4-y_1-y_2)} F \left( x_2, y_1,x_3, y_2 \right)+ \\
& \qquad +\int_{x_4}^{x_3} \dd y_1 \int_{-L/2}^{x_4} \dd y_2 \e^{\ii \lambda (x_1+x_4-y_1-y_2)} F \left( x_2, x_3, y_1, y_2 \right).  
\end{align*}
Now note that in each of the three sets of arguments of $f$, we have the correct ordering.
\end{exm}

We will now formalize this in generality. 
Given nonnegative integers $n \leq N$ and $\bm i \in \f I^n_{N}$, introduce the following sets:
\begin{align*}
\f I^+_{\bm i} &:= \set{\bm j=(j_1,\ldots,j_n) \in \f I^n_{N}}{i_m \leq j_m < i_{m+1}, \n{ for } m=1,\ldots,n}, \\
\tilde{\f I}^+_{\bm i} &:= \set{\bm j=(j_1,\ldots,j_{n-1}) \in \f I^{n-1}_{N-1}}{i_m \leq j_m < i_{m+1}, \n{ for } m=1,\ldots,n-1}, \\
\tilde{\f I}^-_{\bm i} &:= \set{\bm j=(j_1,\ldots,j_{n+1}) \in \f I^{n+1}_{N+1}}{i_{m-1} < j_m \leq i_m, \n{ for }m=1,\ldots,n+1}, \\
\f I^-_{\bm i} &:= \set{\bm j=(j_1,\ldots,j_n) \in \f I^n_{N}}{i_{m-1} < j_m \leq i_m, \n{ for }m=1,\ldots,n}. 
\end{align*}
We recall that if $\bm j \in \f I^n_{N}$ we use the notation $x_{j_0}=L/2$ and $x_{j_{n+1}}=-L/2$.

\begin{lem} \label{elementaryintopsaltformulae}
Let $\lambda \in \bC$ and $n=0, \ldots, N$. 
For $F \in \ca H_N$, $\bm x \in J^{N+1}_+, \bm i \in \f I^{n+1}_{N+1}$ we have
\begin{align*}
\left( \hat E_{\lambda;\bm i} F \right)(\bm x) &= \sum_{\bm j \in \tilde{\f I}^+_{\bm i}} \left( \prod_{m=1}^n \int_{x_{j_m}+1}^{x_{j_m}} \dd y_m \right) \e^{\ii \lambda \left( \sum_{m=1}^{n+1} x_{i_m} - \sum_{m=1}^n y_m \right)}  \cdot \\
& \qquad \cdot F\left( x_1, \ldots, \left( \widehat{x_{i_m}}, \ldots, x_{j_m}, y_m, x_{j_m+1} \right)_{m=1}^n, \ldots, \widehat{x_{i_{n+1}}}, \ldots, x_N \right);
\end{align*}
for $F \in \ca H_N$, $\bm x \in J^N_+, \bm i \in \f I^n_{N}$ we have
\begin{align*}
\left( \bar E^+_{\lambda;\bm i} F \right)(\bm x) &= \sum_{\bm j \in \f I^+_{\bm i}} \left( \prod_{m=1}^n \int_{x_{j_m}+1}^{x_{j_m}} \dd y_m \right) \e^{\ii \lambda \sum_{m=1}^n (x_{i_m}-y_m)}  \cdot \\
& \qquad \cdot F\left( x_1, \ldots, \left( \widehat{x_{i_m}}, \ldots, x_{j_m}, y_m, x_{j_m+1} \right)_{m=1}^n, \ldots, x_N \right), \displaybreak[2] \\
\left( \bar E^-_{\lambda;\bm i} F \right)(\bm x) &= \sum_{\bm j \in \f I^-_{\bm i}} \left( \prod_{m=1}^n \int_{x_{j_m}}^{x_{j_m}-1} \dd y_m \right) \e^{\ii \lambda \sum_{m=1}^n (x_{i_m}-y_m)}  \cdot \\
& \qquad \cdot F\left( x_1, \ldots, \left( x_{j_m-1}, y_m, x_{j_m}, \ldots, \widehat{x_{i_m}} \right)_{m=1}^n, \ldots, x_N \right);
\end{align*}
and for $F \in \ca H_{N+1}$, $\bm x \in J^N_+, \bm i \in \f I^n_{N}$ we have
\begin{align*}
\left( \check E_{\lambda;\bm i} F \right)(\bm x) &= \sum_{\bm j \in \tilde{\f I}^-_{\bm i}} \left( \prod_{m=1}^n \int_{x_{j_m}+1}^{x_{j_m}} \dd y_m \right) \e^{\ii \lambda \left( \sum_{m=1}^n x_{i_m} - \sum_{m=1}^{n+1} y_m \right)}  \cdot \\
& \quad \cdot F\left( x_1, \ldots, \left( x_{j_m-1}, y_m, x_{j_m}, \ldots, \widehat{x_{i_m}} \right)_{m=1}^n, \ldots, x_{j_{n+1}-1}, y_{n+1}, x_{j_{n+1}}, \ldots, x_N \right).
\end{align*}
\end{lem}

\begin{proof}
These formulae are obtained by splitting each integration interval into a union of intervals between adjacent $x_j$, e.g. for $\bar E^+_{\lambda;\bm i}F$
\[ \int_{x_{i_{m+1}}}^{x_{i_m}} \hspace{-2mm} \dd y_m  = \sum_{j_m=i_m}^{i_{m+1}-1} \int_{x_{j_m+1}}^{x_{j_m}} \hspace{-2mm} \dd y_m . \]
Taking all summations over $j_m$ together as a summation over $\bm j= (j_1, \ldots, j_n)$, we obtain $\bm j \in \f I^+_{\bm i}$.
Now for each $m=1,\ldots,n$ we have $x_{i_m} > \ldots >x_{j_m} > y_m > x_{j_m+1}$ so that we can write the argument of $F$ as indicated. The other formulae follow analogously.
\end{proof}

\rfl{elementaryintopsaltformulae} essentially solves the problem alluded to at the start of this section. Wishing to treat all arguments of $F$ in these formulae in the same manner, we introduce unit step functions.

\begin{exm}[Continuation of \rfex{examplealtform}]
We obtain
\begin{align*} 
\left( \bar E^+_{\lambda;1 \, 4} F \right)(x_1,x_2,x_3,x_4) &= \int_J \dd y_1 \int_J \dd y_2  \e^{\ii \lambda (x_1+x_4-y_1-y_2)} \cdot \\
& \qquad \cdot \left( F \left( y_1,x_2, x_3, y_2 \right) \theta(x_1>y_1>x_2)\theta(x_4 > y_2 > -L/2) + \right. \\
& \qquad \qquad +  F \left( x_2, y_1,x_3, y_2 \right) \theta(x_2>y_1>x_3) \theta(x_4>y_2> -L/2) + \\
& \qquad \qquad \left. + F \left( x_2, x_3, y_1, y_2 \right) \theta(x_3>y_1>x_4) \theta(x_4>y_2 > -L/2) \right). 
\end{align*}
We can simplify matters further by introducing Dirac deltas for those arguments of $F$ over which we do not integrate; writing $\bm y = (y_1,y_2,y_3,y_4)$ this gives the terms
\begin{align*} `
\left( \bar E^+_{\lambda;1 \, 4} F \right)(x_1,x_2,x_3,x_4) &= \int_{J^4_+} \dd \bm y \e^{-\ii \lambda (y_1+y_2+y_3+y_4)} F(y_1,y_2,y_3,y_4) \cdot \\
& \qquad \cdot \left( \theta(x_1>y_1>x_2) \delta(y_2-x_2) \delta(y_3-x_3) \theta(x_4>y_2>-L/2)  +   \right. \\
& \qquad \qquad +  \delta(y_1-x_2) \theta(x_2>y_2>x_3) \delta(y_3-x_3) \theta(x_4>y_4>-L/2)+ \\
& \qquad \qquad \left. + \delta(y_1-x_2) \delta(y_2-x_3) \theta(x_3>y_3>x_4) \theta(x_4>y_4>-L/2)\right). 
\end{align*}
\end{exm}

To generalize this, given $n = 0,\ldots,N$ and $\bm i \in \f I^n_{N}$, consider $\bm i^\n{c} = (i^\n{c}_1, \ldots, i^\n{c}_{N-n}) = (1,\ldots,N)_{\hat{\bm \imath}} \in \f I^{N-n}_{N}$. It is the unique element of $\f I^{N-n}_{N}$ which has no entries in common with $\bm i$.
For example, for $N=5$, $n=2$, and $\bm i = (1,4)$, we have $\bm i^\n{c} = (2,3,5)$.
For $\bm i \in \f I^n_{N}$, note that if $\bm j \in \f I^+_{\bm i}$ then $j^\n{c}_m \in \{ i^\n{c}_m-1, i^\n{c}_m\}$, and if $\bm j \in \f I^-_{\bm i}$ then $j^\n{c}_m \in \{i^\n{c}_m,i^\n{c}_m+1 \}$.
Also, for $\bm i \in \f I^{n+1}_{N+1}$, $\bm j \in \tilde{\f I}^+_{\bm i} \subset \f I^n_{N}$ so that both $i^\n{c}_m$ and $j^\n{c}_m$ are in $\f I^{N-n}_{N+1}$ and indeed $j^\n{c}_m \in  \{ i^\n{c}_m-1, i^\n{c}_m\}$. We can now state and prove
\begin{lem}
Let $\lambda \in \bC$ and $n=0, \ldots, N$. 
For $F \in \ca H_N$, $\bm x \in J^{N+1}_+, \bm i \in \f I^{n+1}_{N+1}$ we have
\begin{align*}
\left( \hat E_{\lambda;\bm i} F \right)(\bm x) &= \int_{J^N_+} \dd^N \bm y \e^{\ii \lambda \left( \sum_{k=1}^{N+1} x_k - \sum_{k=1}^N y_k \right)} F(\bm y) \cdot \\
& \qquad \cdot \sum_{\bm j \in \tilde{\f I}^+_{\bm i}}  \left( \prod_{m=1}^n \theta(x_{j_m}>y_{j_m}>x_{j_m+1}) \right) \left( \prod_{m=1}^{N-n} \delta\left( y_{j^\n{c}_m} - x_{i^\n{c}_m} \right) \right) ; 
\end{align*}
for $F \in \ca H_N$, $\bm x \in J^N_+, \bm i \in \f I^n_{N}$ we have
\begin{align*}
\left( \bar E^+_{\lambda;\bm i} F \right)(\bm x) &= \int_{J^N_+} \dd^N \bm y \, \e^{\ii \lambda \sum_{k=1}^N (x_k - y_k)} F(\bm y) \cdot  \\
& \qquad \cdot \sum_{\bm j \in \f I^+_{\bm i}} \left( \prod_{m=1}^n \theta(x_{j_m} > y_{j_m} > x_{j_m+1} ) \right) \left(\prod_{m=1}^{N-n}  \delta(y_{j^\n{c}_m} \! - \! x_{i^\n{c}_m}) \right), \displaybreak[2] \\
\left( \bar E^-_{\lambda;\bm i} F \right)(\bm x) &= \int_{J^N_+} \dd^N \bm y \, \e^{\ii \lambda \sum_{k=1}^N (x_k-y_k)} F(\bm y) \cdot \\
& \qquad \cdot \sum_{\bm j \in \f I^-_{\bm i}} \left( \prod_{m=1}^n \theta(x_{j_m-1} > y_{j_m} > x_{j_m} ) \right) \left(\prod_{m=1}^{N-n}  \delta(y_{j^\n{c}_m} \! - \! x_{i^\n{c}_m}) \right);
\end{align*}
and for $F \in \ca H_{N+1}$, $\bm x \in J^N_+, \bm i \in \f I^n_{N}$ we have
\begin{align*}
\left( \check E_{\lambda;\bm i} F \right)(\bm x) &= \int_{J^N_+}\dd^N \bm y \e^{\ii \lambda (\sum_{k=1}^N x_k - \sum_{k=1}^{N+1} y_k)} F(\bm y) \cdot \\
& \qquad \cdot \sum_{\bm j \in \tilde{\f I}^-_{\bm i}} \left( \prod_{m=1}^{n+1} \theta(x_{i_{m-1}} \! > \!  y_m \! > \!  x_{i_m}) \right) \left( \prod_{m=1}^{N-n-1} \delta(y_{n+m+1}-x_{i^\n{c}_m}) \right).\end{align*}
\end{lem}

\begin{proof}
Again we focus on the expression for $\bar E^+_{\lambda;\bm i}F$. 
Starting from \rfl{elementaryintopsaltformulae} we relabel the integration variables $y_m$ to $y_{j_m}$; we like to think of the integrations as being over $\bR$ and hence introduce a step function $\theta(x_{j_m}>y_{j_m}>x_{j_m+1})$.
Also, we introduce integrations over variables $y_{j^\n{c}_m}$ for $m=1,\ldots,N-n$ with Dirac deltas $\delta(y_{j^\n{c}_m}-x_{i^\n{c}_m})$. We multiply the integrand by $\e^{\ii \lambda(x_{i^\n{c}_m}-y_{j^\n{c}_m})} = 1$, producing an overall factor $\e^{\ii \lambda \sum_{k=1}^N (x_k-y_k)}$. \\

We claim that the arguments of $F$ in decreasing order are now given by $y_1,\ldots,y_N$; it suffices to show that under the restrictions $x_{j_m}>y_{j_m} > x_{j_m+1}$, and $y_{j^\n{c}_m}=x_{i^\n{c}_m}$, we have $y_l > y_m$ if $l < m$, for all $l,m=1,\ldots,N$.
Taking the restrictions given by the step functions and Dirac deltas into account, it is clear that the $(y_{j_1},\ldots,y_{j_n})$ satisfy $y_{j_1} > \ldots > y_{j_n}$ and $(y_{j^\n{c}_1},\ldots,y_{j^\n{c}_{N-n}})$ satisfy $y_{j^\n{c}_1} > \ldots > y_{j^\n{c}_{N-n}}$.
Also, if $j_l < j^\n{c}_m$, then $j_l+1 \leq j^\n{c}_m$ and we have $y_{j_l} > x_{j_l+1} \geq x_{j^\n{c}_m} \geq x_{i^\n{c}_m} = y_{j^\n{c}_m}$, because $j^\n{c}_m \leq i^\n{c}_m$.
Finally, if $j^\n{c}_m < j_l$, then $j^\n{c}_m+1 \leq j_l$ and hence $y_{j^\n{c}_m}=x_{i^\n{c}_m} \geq x_{j^\n{c}_m+1} \geq x_{j_l} > y_{j_l}$, because $i^\n{c}_m \leq j^\n{c}_m+1$.
This produces the desired formula for $\bar E^+_{\lambda;\bm i}F$. The other formulae can be dealt with in the same way.
\end{proof}

By summing over $\bm i$ and including the appropriate factors, we obtain
\begin{thm} \label{ABCDaltformulae}
Let $\lambda \in \bC$.
For $F \in \ca H_N$ and $\bm x \in J^N_+$ we have
\begin{align*}
(A_\lambda F)(\bm x) &= \int_{J^N_+} \dd^N \bm y \, \e^{\ii \lambda (\sum_{k=1}^{N+1} x_k - \sum_{k=1}^N y_k)} F(\bm y) \cdot \\
& \qquad \cdot  \sum_{n=0}^N \sum_{\bm i \in \f I^n_N} \sum_{\bm j \in \f I^+_{\bm i}} \left( \prod_{m=1}^n \gamma \theta(x_{j_m} > y_{j_m} > x_{j_m+1} ) \right) \left(\prod_{m=1}^{N-n}  \delta(y_{j^\n{c}_m} \! - \! x_{i^\n{c}_m}) \right), \\
(D_\lambda F)(\bm x) &= \int_{J^N_+} \dd^N \bm y \, \e^{\ii \lambda ( \sum_{k=0}^N x_k- \sum_{k=1}^N y_k)} F(\bm y) \cdot \\
& \qquad \cdot \sum_{n=0}^N \sum_{\bm i \in \f I^n_N} \sum_{\bm j \in \f I^-_{\bm i}} \left( \prod_{m=1}^n \gamma \theta(x_{j_m-1} > y_{j_m} > x_{j_m} ) \right) \left(\prod_{m=1}^{N-n}  \delta(y_{j^\n{c}_m} \! - \! x_{i^\n{c}_m}) \right); 
\end{align*}
for $F \in \ca H_N$ and $\bm x \in J^{N+1}_+$ we have
\begin{align*}
(B_\lambda F)(\bm x) &= \frac{1}{N+1} \int_{J^N_+} \dd^N \bm y \, \e^{\ii \lambda( \sum_{k=1}^{N+1} x_k-\sum_{k=1}^N y_k)} F(\bm y) \cdot \\
& \qquad \cdot \sum_{n=0}^N \sum_{\bm i \in \f I^{n+1}_{N+1}}  \sum_{\bm j \in \tilde{\f I}^+_{\bm i}}  \left( \prod_{m=1}^n \gamma \theta(x_{j_m} > y_{j_m} > x_{j_m+1} ) \right) \left(\prod_{m=1}^{N-n}  \delta(y_{j^\n{c}_m} \! - \! x_{i^\n{c}_m}) \right); 
\end{align*}
and for $F \in \ca H_{N+1}$ and $\bm x \in J^N_+$ we have
\begin{align*}
\gamma \left(C_\lambda  F \right) (\bm x) &=  (N+1) \int_{J^{N+1}_+} \dd^N \bm y \, \e^{\ii \lambda (\sum_{k=0}^{N+1} x_k - \sum_{k=1}^{N+1} y_k)} F(\bm y) \cdot \\
& \qquad \cdot \sum_{n=0}^{N-1} \sum_{\bm i \in \f I^n_{N - 1}} \sum_{\bm j \in \tilde{\f I}^-_{\bm i}}  \left( \prod_{m=1}^{n+1} \gamma \theta(x_{j_{m-1}} \! > \!  y_{j_m} \! > \!  x_{j_m}) \right) \left( \prod_{m=1}^{N-n-1} \delta(y_{j^\n{c}_m}-x_{i^\n{c}_m}) \right).
\end{align*}
Furthermore, the domains of integration $J^N_+$, $J^{N+1}_+$ may be replaced by $\bR^N$, $\bR^{N+1}$, respectively, because of the unit step functions in the integrands.
\end{thm}

These alternative expressions for the generators of the Yang-Baxter algebra appear to be new.

\newpage

\chapter[The degenerate affine Hecke algebra]{The degenerate affine Hecke algebra (dAHA)} \label{chdAHA}

In this chapter we will review another established method for solving the QNLS eigenvalue problem \rfeqnser{QNLS1}{QNLS2}, which involves a deformation of the group algebra of the symmetric group $S_N$, called the \emph{degenerate affine Hecke algebra}, or also the \emph{graded Hecke algebra of type A$_{N-1}$}. Its main advantage compared to the QISM is that it can be naturally generalized to different reflection groups, both finite and affine. Most of these, the classical Weyl groups, allow for meaningful interpretations in physical systems; they are the symmetry groups of one-dimensional systems of quantum particles with certain boundary conditions. The systematic study of these systems was begun by Gaudin \cite{Gaudin1971-3}, and Gutkin and Sutherland \cite{Gutkin1982,GutkinSutherland}.\\

Affine Hecke algebras were introduced and studied initially by Lusztig and Kazhdan \cite{KazhdanLusztig,Lusztig}, and Drinfel\cprime d \cite{Drinfeld1986}. 
Their relevance to the QNLS problem was highlighted by Heckman and Opdam \cite{HeckmanOpdam1996,HeckmanOpdam1997} who used an infinitesimal version of them, the \emph{graded Hecke algebra}.
A generalization of the affine Hecke algebra, the double affine Hecke algebra, has been used by Cherednik \cite{Cherednik} to prove Macdonald's constant term conjecture for Macdonald polynomials.\\

In this thesis we will restrict ourselves to the case of the Weyl group $S_N$. 
We will review some theory of the symmetric group and its group algebra; subsequently the degenerate affine Hecke algebra (dAHA)\footnote{The abbreviation DAHA is usually reserved for the aforementioned double affine Hecke algebra.} is introduced, which is a deformation of the group algebra of the symmetric group with the coupling constant $\gamma$ functioning as a deformation parameter. We will highlight three of its representations that play a role in finding the QNLS wavefunction:
\begin{itemize}
\item The regular representation in momentum space, generated by deformed transpositions $\tilde s_{j, \gamma}$ and multiplication operators $\lambda_j$, acting on analytic functions in $\bm \lambda$.
\item The integral representation in position space, generated by deformed transpositions $s_{j,\gamma}$ (the term proportional to $\gamma$ in these are integral operators) and partial differential operators $-\ii \partial_j$, acting on smooth functions on $\bR^N$. This can be viewed as the Fourier transform of the first representation.
\item The Dunkl-type representation in position space, generated by transpositions and deformed partial differential operators $-\ii \partial_{j,\gamma}$, acting on continuous functions on $\bR^N$ whose restrictions to the set of regular vectors is smooth.
\end{itemize} 

The crucial \emph{propagation operator} \cite{Gutkin1982,EmsizOS,Hikami} is constructed, which intertwines the integral and Dunkl-type representations and maps plane waves, i.e. functions of the form $\bm x \mapsto \e^{\ii \inner{\bm \lambda}{\bm x}}$ for some $\bm \lambda \in \bC^N$, to \emph{non-symmetric} functions $\psi_{\bm \lambda}$ that solve the QNLS eigenvalue problem \rfeqnser{QNLS1}{QNLS2}. For the QNLS model these non-symmetric functions were first considered by Komori and Hikami \cite{KomoriHikami}.
We will call these solutions \emph{pre-wavefunctions}, since upon symmetrizing them one recovers the QNLS wavefunction $\Psi_{\bm \lambda}$.
They will return to our attention in Chapter \ref{ch5}. 

\section{The symmetric group} \label{secsymmgroup}

In this section we review some facts involving the symmetric group. 
There is a natural way of embedding $S_N$ in $S_{N+1}$ which is useful for recursive constructions.

\begin{lem} \label{symmgrouprecursion}
We have
\[ S_{N+1} = S_N \cdot \set{s_{m \, N+1}}{m=1,\ldots,N+1} = \set{s_{m \, N+1}}{m=1,\ldots,N+1} \cdot S_N. \]
\end{lem}

\begin{proof}
Since each $w \in S_N$ permutes the set $\set{s_{m \, N+1}}{m=1,\ldots,N+1}$, the second equality follows and it suffices to prove the first, which can be done by noting that the following mapping is a one-to-one correspondence:
\[ p: S_{N+1} \to S_N \times \{ 1, \ldots, N+1 \} : w \mapsto (ws_{w^{-1}(N+1) \, N+1} , w^{-1}(N+1)). \]
Note that for $w \in S_N$, $ws_{w^{-1}(N+1) \, N+1}(N+1) = w(w^{-1}(N+1)) = N+1$, so that $ws_{w^{-1}(N+1) \, N+1}$ can be thought of as an element of $S_N$, and indeed $p$ maps into $ S_N \times \{ 1, \ldots, N+1 \} $. Since $S_{N+1}$ and $S_N \times \{ 1, \ldots, N+1 \} $ are finite sets of the same cardinality it suffices to show $p$ is injective. This follows from the fact that $p$ has a left-inverse$: S_N \times \{ 1, \ldots, N+1 \} \to S_{N+1}$ given by
$(w' , j) \mapsto w' s_{j \, N+1}$. 
\end{proof}




\subsection{The group algebra of the symmetric group}

The group algebra of the symmetric group is the unital associative algebra of all linear combinations of elements of $S_N$,
\[ \bC S_N = \set{\sum_{w \in S_N} c_w w}{c_w \in \bC}, \]
where the multiplication in $S_N$ is extended linearly.
The \emph{symmetrizer} is given by
\begin{equation} \label{symmetrizer} \ca S^{(N)} = \frac{1}{N!} \sum_{w \in S_N} w \in \bC S_N. \end{equation}%
\nc{rscx}{$\ca S^{(N)}$}{Symmetrizer \nomrefeqpage}%
It satisfies $w \ca S^{(N)} = \ca S^{(N)} w = \ca S^{(N)}$ for all $w \in S_N$ and hence is a projection: $\left(\ca S^{(N)}\right)^2= \ca S^{(N)}$.
The symmetrizer can be constructed recursively:
\begin{lem} \label{symmetrizerrecursion}
Let $N\geq1$ be an integer.
We have
\begin{align*} 
\ca S^{(N+1)} &= \ca S^{(N)} \frac{1}{N+1} \sum_{m=1}^{N+1} s_{m \, N+1} = \ca S^{(N)} \frac{1}{N+1}\sum_{m=1}^{N+1} s_N \ldots s_m \\
&= \frac{1}{N+1}\sum_{m=1}^{N+1} s_{m \, N+1} \ca S^{(N)} = \frac{1}{N+1} \sum_{m=1}^{N+1} s_m \ldots s_N  \ca S^{(N)}. 
\end{align*}
\end{lem}

\begin{proof}
The expressions $\ca S^{(N)} \frac{1}{N+1} \sum_{m=1}^{N+1} s_{m \, N+1}$ and $\frac{1}{N+1} \sum_{m=1}^{N+1} s_{m \, N+1} \ca S^{(N)}$ follow directly from \rfl{symmgrouprecursion}. 
To obtain the expression $\ca S^{(N)} \sum_{m=1}^{N+1} s_N \ldots s_m$ note that 
\[ s_{m \, N} = s_m \ldots s_{N-1} s_N s_{N-1} \ldots s_m.\] 
Since $s_m \ldots s_{N-1} \in S_{N}$, it can be absorbed into $\ca S^{(N)}$ and hence we find
\[  \ca S^{(N)} \sum_{m=1}^{N+1} s_{m \, N+1} = \ca S^{(N)} \sum_{m=1}^{N+1} s_N \ldots s_m. \]
In the same fashion we obtain 
\[ \sum_{m=1}^{N+1} s_{m \, N+1} \ca S^{(N)} = \sum_{m=1}^{N+1} s_m \ldots s_N \ca S^{(N)}. \qedhere \]
\end{proof}

The group algebra $\bC S_N$ acts on the polynomial algebra $\bC[X_1,\ldots,X_N]$ (the normal symmetric group action can be extended linearly).
The indeterminates $X_j$ themselves also act on $\bC[X_1,\ldots,X_N]$ by multiplication. 
The combined algebra, written $\f H^N$, %
\nc{rhcw}{$\f H^N$}{Combined algebra isomorphic to $\bC S_N \otimes \bC[X_1,\ldots,X_N]$ as \\ vector space}%
is isomorphic to $\bC S_N \otimes \bC[X_1,\ldots,X_N]$ as a vector space, and is generated by $s_1,\ldots,s_{N-1},X_1,\ldots,X_N$ with relations:
\begin{align*}
s_j s_{j+1} s_j &= s_{j+1} s_j s_{j+1}, && \n{for } j =1,\ldots,N-2,  \\
\left[ s_j, s_k \right] &= 0, && \n{for }j,k = 1, \ldots, N-1: |j-k|>1,  \\
s_j^2 &= 1, && \n{for }j=1,\ldots,N-1, \\
s_j X_k - X_{s_j(k)} s_j &= 0, && \n{for } j=1,\ldots,N-1,k=1,\ldots,N, \\
\left[ X_j, X_k \right] &=0, && \n{for } j,k=1,\ldots,N. 
\end{align*}

\subsection{The length function} \label{subseclength}

For $w \in S_N$, consider
\begin{equation} \label{Sigmadef} \Sigma(w) = \set{(j,k) \in \{1,\ldots,N\}}{j<k, \, w(j)>w(k)}, \end{equation}%
\nc{gsc}{$\Sigma(w)$}{Set of ordered pairs in $\{1,\ldots,N\}$ whose order is inverted by \\ $w \in S_N$ \nomrefeqpage}%
i.e. the set of ordered pairs whose order is inverted by $w$.
The \emph{length} of $w$ is defined to be 
\begin{equation} l(w) =  |\Sigma(w)|. \end{equation}%
\nc{rll}{$l(w)$}{Length of $w \in \Sigma$ \nomrefeqpage}%
We note that such a characterization fits in the context of the definition of the length of an element of a general Weyl group in terms of the (positive) root system. Hence the results in \cite[\S 2.2]{Macdonald2} can be used. Here we review some of these results applied to the case of $S_N$. \\

Evidently we have $l(w) = 0$ if and only if $w = 1$. Since $\Sigma(s_j) = \{ (j,j+1) \}$ it follows that $l(s_j)=1$ for $j=1,\ldots,N-1$. Also, $l(w^{-1}) = l(w)$ follows from the observation $\Sigma(w^{-1}) = w \set{(j,k) \in \{ 1, \ldots, N \}^2}{(k,j) \in \Sigma(w)}$.\\

For all $w_1,w_2 \in S_N$, $l(w_1w_2) \leq l(w_1)+l(w_2)$ and the following conditions are equivalent:
\begin{align}
l(w_1w_2) &= l(w_1)+l(w_2); \label{lengthcond1} \\
\Sigma(w_1 w_2) &= w_2^{-1} \Sigma(w_1) \cup \Sigma(w_2); \label{lengthcond2} \\
(j,k) \in w_2^{-1} \Sigma(w_1) &\implies j<k. \label{lengthcond3}
\end{align}

Let $j=1,\ldots,N-1$ and $w \in S_N$.
Using the equivalence of \rfeqn{lengthcond1} and \rfeqn{lengthcond3}, once with $w_1=w$, $w_2 = s_j$, and once with $w_1=ws_j$, $w_2=s_j$, we obtain
\begin{equation}
l(ws_j) = l(w)+\sgn(w(j+1)-w(j)). \label{length2}
\end{equation}

A simple induction argument may be used to obtain that for all $w \in S_N$
\[ l(w) = \min_{w = s_{i_1} \ldots s_{i_l}} l, \]
i.e. the length of a permutation $w$ is simply the minimum number of simple transpositions in a decomposition for $w$. Such a decomposition is called \emph{reduced}.\\

The set $\Sigma(w)$ will be used several times during this chapter; the following lemma is useful for inductions on the length of $w$.
\begin{lem} \label{lengthlem1}
Let $w \in S_N$ with a reduced expression $w = s_{i_1} \ldots s_{i_l}$. Then
\[ \set{w s_{j \, k} }{(j,k) \in \Sigma(w)} = \set{s_{i_1} \ldots \hat s_{i_m} \ldots s_{i_l} }{m=1,\ldots,l}, \]
where the hat placed over $s_{i_m}$ indicates that this particular transposition is removed from the product.
In particular, the length of each $w s_{j \, k}$, where $(j,k) \in \Sigma(w)$, is strictly less than the length of $w$.
\end{lem}

\begin{proof}
By induction on $l$. The statement for $l=0$ is vacuously true. To see that the statement for $l+1$ follows from the statement for $l$, given a reduced composition $w=s_{i_1} \ldots s_{i_{l+1}}$, write $w' = ws_{i_{l+1}} = s_{i_1}\ldots s_{i_l}$ and note that the induction hypothesis implies
\[ \set{w' s_{\bar \jmath \, \bar k}}{(j,k) \in s_{i_l}\Sigma(w')} =  \set{s_{i_1} \ldots \hat s_{i_m} \ldots s_{i_l} }{m=1,\ldots,l}, \]
where for $j=1,\ldots,N$, $\bar \jmath = s_{i_{l+1}}(j)$.
Right-multiplying by $s_{i_l}$ we obtain
\[ \set{w s_{j \, k}}{(j,k) \in s_{i_l} \Sigma(w')} = \set{s_{i_1} \ldots \hat s_{i_m} \ldots s_{i_{l+1}}}{m=1,\ldots,l}. \]
Now make use of the equivalence of \rfeqnser{lengthcond2}{lengthcond3}.
\end{proof}

\subsection{Duality of $S_N$-actions}

We will be considering functions of two $N$-tuples $\bm \lambda \in \bC^N$, $\bm x \in \bR^N$ (or subsets thereof), and we wish to study the two distinct $S_N$-actions on such functions. An example of such a function is the plane wave: $\bC^N \times \bR^N \to \bC: (\bm \lambda,\bm x)\mapsto \e^{\ii \inner{\bm \lambda}{\bm x}}$.
We will denote these two actions as follows.\\

\boxedenv{\begin{notn}[Distinction of $S_N$-actions] \label{notn5}
Let $w \in S_N$ and $f: \bC^N \times \bR^N \to \bC$. Then we may define two functions $w f, \tilde w f: \bC^N \times \bR^N \to \bC$ as follows:
\begin{align*}
(w f)(\bm \lambda,\bm x) &= f(\bm \lambda,w^{-1} \bm x) \\
(\tilde w f)(\bm \lambda,\bm x) &= f(w^{-1}  \bm \lambda,\bm x).
\end{align*}
For functions $f: \bC^N \to \bC$ whose argument is denoted $\bm \lambda$ we will sometimes denote the action of $w \in S_N$ on such a function by $\tilde w$, as well.\\

In the context of root systems, this notation is reminiscent to the duality of the action of the Weyl group on the Euclidean space, spanned by the co-roots, and the action on its dual, spanned by the roots. 
\end{notn}}

Let $\bm \lambda \in \bC^N$ and note that the plane wave $\e^{\ii \bm \lambda}$ satisfies
\[ (w \e^{\ii \bm \lambda})(\bm x) = \e^{\ii \bm \lambda}(w^{-1} \bm x) = \e^{\ii \inner{\bm \lambda}{w^{-1} \bm x}} = \e^{\ii \inner{w \bm \lambda}{\bm x} } = \e^{\ii w \bm \lambda}(\bm x) =\left({\tilde w}^{-1} \e^{\ii \bm \lambda}\right)(\bm x), \]
for all $\bm x \in \bR^N$, i.e. $w \e^{\ii \bm \lambda}= {\tilde w}^{-1} \e^{\ii \bm \lambda}$.

\section{The degenerate affine Hecke algebra} \label{secdAHA}
 
We will now introduce a deformation of the symmetric group algebra which is the central object in this chapter. It was introduced independently by Lusztig \cite{Lusztig} and Drinfel\cprime d \cite{Drinfeld1986}.

\begin{defn} \label{dAHAdefn}
Let $\gamma \in \bR$.
The \emph{degenerate affine Hecke algebra} (dAHA), denoted $\f H^{N}_\gamma$, %
\nc{rhcw}{$\f H^N_\gamma$}{Degenerate affine Hecke algebra}%
is the algebra with generators $s_1,\ldots,s_{N-1},X_1,\ldots,X_N$ and relations
\begin{align}
s_j s_{j+1} s_j &= s_{j+1} s_j s_{j+1}, && \n{for } j =1,\ldots,N-2, \label{dAHArel1} \tag{dAHA 1}\\
\left[ s_j, s_k \right] &= 0, && \n{for } j,k = 1, \ldots, N-1: |j-k|>1, \label{dAHArel2} \tag{dAHA 2} \\
s_j^2 &= 1, && \n{for } j=1,\ldots,N-1,  \label{dAHArel3} \tag{dAHA 3} \\
s_j X_k -X_{s_j(k)} s_j & = -\ii \gamma \left( \delta_{j \, k} - \delta_{j+1 \, k} \right), && \n{for } j=1,\ldots,N-1,k=1,\ldots,N, \label{dAHArel4} \tag{dAHA 4} \\
\left[ X_j, X_k \right] & =0, && \n{for } j,k=1,\ldots,N. \label{dAHArel5} \tag{dAHA 5}
\end{align}
\end{defn}

$\f H^{N}_\gamma$ can be viewed as a deformation of $\f H^{N}= \f H^{N}_0$, controlled by $\gamma$. We will identify $\bC[\bm X]$ and $\bC S_N$ as subalgebras of $\f H^{N}_\gamma$. 
There are some well-known properties of $\f H^{N}_\gamma$ \cite{Opdam, Lusztig, Cherednik2} that can be directly obtained from \rfeqnser{dAHArel1}{dAHArel5}.
For all $j=1,\ldots,N$, $w \in S_N$ we have
\begin{equation} \label{dAHAeqn1} w X_j = X_{w(j)} w- \ii \gamma w \left( \sum_{k: (j,k) \in \Sigma(w)} s_{j \, k} - \sum_{k: (k,j) \in \Sigma(w)} s_{j \, k}  \right) \in \f H^{N}_\gamma. \end{equation}
This can be proven by induction on $l(w)$.
Also, for all $j=1,\ldots,N-1$ and $p \in \bC[\bm X]$, 
\begin{equation} \label{dAHAeqn2} s_j p(X_1,\ldots,X_N) - p(X_{s_j(1)},\ldots,X_{s_j(N)})s_j = -\ii \gamma (\Delta_j p)(X_1,\ldots,X_N) \in \f H^{N}_\gamma,\end{equation}
where we have introduced the \emph{divided difference operator}\footnote{Alternatively, it is known as the \emph{Bern\v{s}te\v{\i}n-Gel\cprime\hspace{-2mm} fand-Gel\cprime\hspace{-2mm} fand operator} or \emph{Lusztig-Demazure operator}.} $\Delta_j \in \End(\bC[\bm X])$
\begin{equation} \label{divdiffop} (\Delta_j p)(X_1,\ldots,X_N) = \frac{p(X_1,\ldots,X_N)-p(X_{s_j(1)},\ldots,X_{s_j(N)})}{X_j-X_{j+1}}. \end{equation}%
\nc{gdc}{$\Delta_j$}{Divided difference operator acting on $\bC[\bm X]$ \nomrefeqpage}%
Since the polynomial $p(X_1,\ldots,X_N)-p(X_{s_j(1)},\ldots,X_{s_j(N)})$ is alternating in $X_j,X_{j+1}$ it is divisible by $X_j-X_{j+1}$, and hence indeed $\Delta_j p \in \bC[\bm X]$. \rfeqn{dAHAeqn2} can be demonstrated by induction on the degree of $p$. 
Finally, combining \rfeqnser{dAHAeqn1}{dAHAeqn2} we obtain that the centre of $\f H^{N}_\gamma$ is the subalgebra of symmetric polynomials:
\begin{equation} \label{dAHAcentre}
Z(\f H^{N}_\gamma) = \bC[\bm X]^{S_N}.
\end{equation}

\section{The regular representation in momentum space}

This representation is also known as the Bern\v ste\v \i n-Gel\cprime\hspace{-2mm} fand-Gel\cprime\hspace{-2mm} fand representation, as well as the Demazure representation \cite{BernsteinGG, Demazure}.
For this representation we will consider the vector space of polynomial functions $\ca P(\bC^N) \cong \bC[\lambda_1,\ldots,\lambda_N]$. In particular, we will identify a copy of $\f H^{N}_\gamma$ as a subalgebra of $\End(\ca P(\bC^N))$. \\

The divided difference operator $\Delta_j$ introduced in \rfeqn{dAHAeqn2} can be ``dualized'' as follows. 
\begin{defn} \emph{\cite{Emsiz,Gutkin1987}} \label{divdiffoperator}
Let $1 \leq j \ne k \leq N$ and $\gamma \in \bR$.
Then $\tilde \Delta_{j \, k} = \frac{1-\tilde s_{j \, k}}{\lambda_j-\lambda_k} \in \End(\ca P(\bC^N))$ is defined by
\begin{equation} \label{divdiffop2}
\left( \tilde \Delta_{j \, k} p \right)(\lambda_1,\ldots,\lambda_N) = \frac{p(\lambda_1,\ldots,\lambda_N)-p(\lambda_{s_{j \, k} 1},\ldots,\lambda_{s_{j\, k} N})}{\lambda_j-\lambda_k}, 
\end{equation}%
\nc{gdc}{$\tilde \Delta_j, \tilde \Delta_{j \, k}$}{Divided difference operator appearing in regular representation of \\ dAHA \nomrefeqpage}%
where $p \in \ca P(\bC^N)$ and $\bm \lambda \in \bC^N$.
Note that $p-\tilde s_{j \, k} p$ is a polynomial alternating in $\lambda_j,\lambda_k$, so that it is divisible by $\lambda_j-\lambda_k$; in other words $\tilde \Delta_{j \, k}$ is a bona fide operators on $\End(\ca P(\bC^N))$.
For $j=1,\ldots,N-1$ we write $\tilde \Delta_j = \tilde \Delta_{j \, j+1}$ and introduce the deformed simple transposition
\begin{equation} \label{deformedtransposition} \tilde s_{j,\gamma} = \tilde s_j - \ii \gamma \tilde \Delta_j \in \End(\ca P(\bC^N)). \end{equation}%
\nc{rslz}{$\tilde s_{j,\gamma}$}{Deformed simple transposition appearing in regular representation of dAHA in momentum space  \nomrefeqpage}%
\vspace{-7mm}
\end{defn}

\begin{lem} \label{divdiffprops}
We list some useful properties of the divided difference operators that follow immediately from the definition. Let $1 \leq j \ne k \leq N$. 
\begin{enumerate}
\item \label{divdiffequivariance} For  $1 \leq l \ne m \leq N$, we have $\tilde s_{j \, k} \tilde \Delta_{l \, m} = \tilde \Delta_{s_{j \, k}(l) \, s_{j \, k}(m)} \tilde s_{j \, k}$.
\item \label{divdiffnilpotency} $\tilde s_{j \, k} \tilde \Delta_{j \, k} =  \tilde \Delta_{k \, j} \tilde s_{j \, k} = -\tilde \Delta_{j \, k} \tilde s_{j \, k} = \tilde \Delta_{j \, k}$ and hence $\tilde \Delta_{j \, k}^2 = 0$. 
\item \label{divdiffformula} For $a,b \in \bZ_{\geq 0}$ we have:
\[ \tilde \Delta_{j \, k} \lambda_j^a \lambda_k^b = \begin{cases} -\displaystyle \sum_{a \leq A, B \leq b-1 \atop A+B=a+b-1} \lambda_j^A \lambda_k^B, & \n{if } a \leq b \\  \displaystyle \sum_{b \leq A,B \leq a-1 \atop A+B=a+b-1} \lambda_j^A \lambda_k^B, & \n{if } a \geq b. \end{cases} \]
\end{enumerate}
\end{lem}

\begin{proof}
Property \ref{divdiffequivariance} follows from $\tilde \Delta_{l \, m}$ being a linear combination of 1 and $\tilde s_{l \, m}$ whose coefficients only depend on $\lambda_l$ and $\lambda_m$. Next, property \ref{divdiffnilpotency} is an immediate consequence of property \ref{divdiffequivariance}, and property \ref{divdiffformula} is a straightforward calculation.
\end{proof}

A representation of an associative algebra can be defined by fixing the images of its generators. Because the underlying vector space structure can be preserved by linearly extending these assignments one only needs to check that the relations used in the definition of the associative algebra are also preserved.

\begin{prop} \emph{\cite{Emsiz,Gutkin1987}} \label{dAHAregrep}
The following assignments define a representation $\rho^\n{reg}_\gamma$ 
of $\f H^{N}_\gamma$ on $\ca P(\bC^N)$ $\cong$ $\bC[\lambda_1,\ldots,\lambda_N]$.
\begin{equation} \label{dAHAregrepeqns}
\rho^\n{reg}_\gamma( s_j ):=\tilde s_{j,\gamma}, \qquad \rho^\n{reg}_\gamma(X_j) := \lambda_j. \end{equation}%
\nc{grl}{$\rho^\n{reg}_\gamma$}{Regular representation in momentum space of dAHA \nomrefeqpage}%
\end{prop}
\vspace{-10mm}
\begin{proof}
We will repeatedly refer to Appendix \ref{dAHAregrepprops}.
We only need to show that with definitions \rfeqn{dAHAregrepeqns}, the axioms \rfeqnser{dAHArel1}{dAHArel5} hold, with $(s_j, X_k) \to (\tilde s_{j,\gamma}, \lambda_k)$.
\begin{description}
\item[\rf{dAHArel1}] This involves the most work; we have 
\begin{align*}
\lefteqn{\tilde s_{j, \gamma} \tilde s_{j\!+\!1, \gamma} \tilde s_{j, \gamma} - \tilde s_{j\!+\!1, \gamma} \tilde s_{j, \gamma} \tilde s_{j\!+\!1, \gamma} =} \displaybreak[2] \\
&= \tilde s_j  \tilde s_{j\!+\!1}  \tilde s_j - \tilde s_{j\!+\!1} \tilde s_j  \tilde s_{j\!+\!1}  + \\
& \qquad - \ii \gamma \left( \tilde s_j  \tilde s_{j\!+\!1} \tilde \Delta_j +  \tilde s_j \tilde \Delta_{j\!+\!1}  \tilde s_j + \tilde \Delta_j  \tilde s_{j\!+\!1}  \tilde s_j - \tilde s_{j\!+\!1}  \tilde s_j \tilde \Delta_{j\!+\!1} -  \tilde s_{j\!+\!1} \tilde \Delta_j  \tilde s_{j\!+\!1} - \tilde \Delta_{j\!+\!1}  \tilde s_j  \tilde s_{j\!+\!1} \right) +\\
& \qquad -\gamma^2 \left(  \tilde s_j \tilde \Delta_{j\!+\!1} \tilde \Delta_j + \tilde \Delta_j  \tilde s_{j\!+\!1} \tilde \Delta_j + \tilde \Delta_j \tilde \Delta_{j\!+\!1}  \tilde s_j - \tilde s_{j\!+\!1} \tilde \Delta_j \tilde \Delta_{j\!+\!1} + \tilde \Delta_{j\!+\!1} \tilde s_j \tilde \Delta_{j\!+\!1} + \tilde \Delta_{j\!+\!1} \tilde \Delta_j  \tilde s_{j\!+\!1} \right) +\\
& \qquad + \ii \gamma^3 \left( \tilde \Delta_j \tilde \Delta_{j\!+\!1} \tilde \Delta_j - \tilde \Delta_{j\!+\!1} \tilde \Delta_j \tilde \Delta_{j\!+\!1} \right) \; = \; 0.
\end{align*}
This follows from applying, for the term proportional to $\gamma$, \rfeqn{regrep6} with $(j,k,l) \to (j,j+1,j+2)$ and \rfeqn{regrep6a} twice, once with $(j,k,l) \to (j,j+1,j+2)$ and once with $(j,k,l) \to (j+2,j+1,j)$; for the term proportional to $\gamma^2$, \rfeqn{regrep9} (again, once with $(j,k,l) \to (j,j+1,j+2)$ and once with $(j,k,l) \to (j+2,j+1,j)$); for the term proportional to $\gamma^3$, \rfeqn{regrep10} with $(j,k,l) \to (j,j+1,j+2)$.
\item[\rf{dAHArel2}] We write, for $|j-k|>1$,
\[ \left[\tilde s_{j,\gamma},\tilde s_{k,\gamma}\right]=  [ \tilde s_j, \tilde s_k] - \ii \gamma \left( [ \tilde s_j, \tilde \Delta_k ] - [ \tilde s_k,\tilde \Delta_j]\right) -\gamma^2 [\tilde \Delta_j,\tilde \Delta_k]. \]
These commutators vanish because of \rfl{regrep5lem}.
\item[\rf{dAHArel3}] We simply have, by virtue of \rfl{divdiffprops}, Property \ref{divdiffnilpotency},
\[ \tilde s_{j,\gamma}^2 = ( \tilde s_j - \ii \gamma \tilde \Delta_j)( \tilde s_j - \ii \gamma \tilde \Delta_j) = \tilde s_j^2 - \ii \gamma ( \tilde s_j \tilde \Delta_j + \tilde \Delta_j \tilde s_j) -\gamma^2 \tilde \Delta_j^2 = 1. \]
\item[\rf{dAHArel4}] This follows from \rfl{regrep4lem} with $k=j+1$:
\begin{align*} 
\tilde s_{j,\gamma} \lambda_k - \lambda_{s_j(k)} \tilde s_{j,\gamma} 
&=(\tilde s_j - \ii \gamma \tilde \Delta_j) \lambda_k - \lambda_{s_j(k)} (\tilde s_j - \ii \gamma \tilde \Delta_j)\\ &= - \ii \gamma \left( \tilde \Delta_j \lambda_k - \lambda_{s_j(k)} \tilde \Delta_j \right) \; = \; - \ii \gamma \left( \delta_{j \, k} - \delta_{j + 1 \, k} \right); 
\end{align*}
\item[\rf{dAHArel5}] This is trivial. \qedhere
\end{description}
\end{proof}

Because the polynomial functions form a dense subspace of the analytic functions, we immediately get a representation of $\f H^{N}_\gamma$ on this larger vector space.

\begin{cor} \label{dAHAregrepanalytic}
The assignments given by \rfeqn{dAHAregrepeqns} define a representation of $\f H^{N}_\gamma$ on $\ca C^\omega(\bC^N)$.
\end{cor}

For any $w \in S_N$ with decomposition $w=s_{i_1} \ldots s_{i_l}$ for some $i_1,\ldots,i_l \in \{1,\ldots,N-1\}$, we will write 
\begin{equation}
\tilde w_\gamma = \tilde s_{i_1,\gamma} \ldots \tilde s_{i_l,\gamma};
\end{equation}%
\nc{rwlz}{$\tilde w_\gamma$}{Deformed permutation acting in momentum space \nomrefeqpage}%
because of \rfp{dAHAregrep} this does not depend on the choice of the decomposition and hence is a well-defined map: $S_N \to \ca C^\omega(\bC^N)$. 
We extend this notation linearly to any element $t$ of the group algebra $\bC S_N$, writing $\tilde t_\gamma$ for $\rho^\n{reg}_\gamma(t)$.
In particular, we may consider 
\begin{equation}
\tilde{\ca S}^{(N)}_\gamma = \rho^\n{reg}_\gamma(\ca S^{(N)}) = \frac{1}{N!} \sum_{w \in S_N} \tilde w_\gamma.
\end{equation}%
\nc{rscx}{$\tilde{\ca S}^{(N)}_\gamma$}{Image of symmetrizer under regular representation of dAHA in momentum space \nomrefeqpage}%
If $w=s_{j \, k}$, then for $\tilde w_\gamma=(\tilde s_{j\,k})_\gamma$ we may also write $\tilde s_{j\,k,\gamma}$.\\

\boxedenv{
\begin{notn}[Comparing operators in function spaces with different particle numbers] \label{notn6}
When it is important to highlight in which $\ca F(J^N)$ we consider the action of an operator, we will indicate this by adding a superscript $(N)$ to the operator in question, such as $s_1^{(2)}$ for the transposition acting on $\ca F(J^2)$ by $(s_1^{(2)} f)(x_1,x_2)=f(x_2,x_1)$ and $s_1^{(3)}$ for the transposition acting on $\ca F(J^3)$ by $(s_1^{(3)} f)(x_1,x_2,x_3)=f(x_2,x_1,x_3)$.
\end{notn} 
}

We now introduce $G_\gamma = G_\gamma^{(N)} \in \ca C^\omega(\bC^N_\n{reg})$ by
\begin{equation} G_\gamma(\bm \lambda) = \prod_{j,k=1 \atop j < k}^N \frac{\lambda_j-\lambda_k-\ii \gamma}{\lambda_j-\lambda_k}. \end{equation}%
\nc{rgc}{$G_\gamma(\bm \lambda)$}{Coefficient in Bethe wavefunction \nomrefeqpage}%
Recalling the notation $\tau^\pm_\mu(\bm \lambda)$ from \rfeqn{BAEv},
we have the obvious property that 
\begin{equation} \label{GLambda} 
G_\gamma^{(N+1)}(\bm \lambda,\mu) = \tau^+_\mu(\bm \lambda) G_\gamma^{(N)}(\bm \lambda), \qquad
G_\gamma^{(N+1)}(\mu,\bm \lambda) = \tau^-_\mu(\bm \lambda) G_\gamma^{(N)}(\bm \lambda), 
\end{equation}
for $(\bm \lambda,\mu) \in \bC^{N+1}_\n{reg}$.
Denote by the same symbol $G_\gamma(\bm \lambda)=G^{(N)}_\gamma(\bm \lambda)$ the corresponding multiplication operator in $\End(\ca C^\omega(\bC^N_\n{reg}))$. Then we have
\begin{prop} \label{regrepsymmetrizers}
The following identity holds in $\End(\ca C^\omega(\bC^N))$:
\begin{equation} \label{regrepsymmetrizerseqn} \tilde{\ca S}^{(N)}_\gamma  =  \tilde{\ca S}^{(N)} G_\gamma^{(N)}(\bm \lambda). \end{equation}
\end{prop}

\begin{proof}
By induction on $N$; the case $N=1$ is trivial. 
To complete the proof, for $\bm \lambda \in \bC^N$ and $\bm \lambda' = (\lambda_1,\ldots,\lambda_{N-1}) \in \bC^{N-1}$, and note that 
\begin{align*} 
\tilde{\ca S}^{(N)} G^{(N)}_\gamma(\bm \lambda) &= \frac{1}{N} \sum_{m=1}^{N} \tilde s_m \ldots \tilde s_{N-1} \tilde{\ca S}^{(N-1)} \tau^+_{\lambda_N}(\bm \lambda') G^{(N-1)}_\gamma(\bm \lambda') \displaybreak[2] \\
&= \frac{1}{N} \sum_{m=1}^{N} \tilde s_m \ldots \tilde s_{N-1} \tau^+_{\lambda_N}(\bm \lambda') \tilde{\ca S}^{(N-1)} G^{(N-1)}_\gamma(\bm \lambda') 
\end{align*}
where we have used \rfl{symmetrizerrecursion} and \rfeqn{GLambda}, as well as the fact that $\tau^+_{\mu}( \bm \lambda')$ is symmetric in the $\lambda_j$. 
Now by virtue of \rfl{regreplem2}, we obtain that
\[ \tilde{\ca S}^{(N)} G^{(N)}_\gamma(\bm \lambda) = \frac{1}{N} \sum_{m=1}^{N} \tilde s_{m,\gamma} \ldots \tilde s_{N,\gamma} \tilde{\ca S}^{(N-1)} G^{(N-1)}_\gamma(\bm \lambda'), \]
and now the induction hypothesis, together with \rfl{symmetrizerrecursion} once more, yields
\[ \tilde{\ca S}^{(N)} G^{(N)}_\gamma(\bm \lambda) \! = \! \frac{1}{N}  \sum_{m=1}^{N} \tilde s_{m,\gamma} \ldots \tilde s_{N-1,\gamma} \tilde{\ca S}^{(N-1)}_\gamma  = \tilde{\ca S}^{(N)}_\gamma. \! \qedhere \]
\end{proof}

\section{The integral representation}

We now turn to a second representation of the dAHA, introduced as a tool to study the QNLS problem in \cite{Gutkin1982,GutkinSutherland}.
Consider the space $\ca C(\bR^N)$ of continuous complex-valued functions on $\bR^N$.
For $1 \leq j \ne k \leq N$, we introduce the integral operator $I_{j \, k} \in \End\left(\ca C(\bR^N)\right)$ defined by
\begin{equation} (I_{j \, k} f)(\bm x) = \int_0^{x_j-x_k} \dd y f \! \left( \bm x - y (\bm e_j-\bm e_k) \right) \end{equation}%
\nc{rica}{$I_j, I_{j \, k}$}{Integral operator appearing in integral \\ representation of dAHA \nomrefeqpage}%
for $f\in \ca C(\bR^N)$ and $\bm x \in \bR^N$.
Note that, for $f \in \ca C(\bR^N)$ and $\bm x \in V_{j \, k} = \set{\bm x \in \bR^N}{x_j=x_k}$, we have $(I_{j \, k}f)(\bm x) = 0$. 
Also, $I_{j \, k}$ restricts to an operator on $\ca C^\infty(\bR^N)$ and indeed to $\ca C^\omega(\bR^N)$.
For $j=1,\ldots,N-1$, we write $I_j = I_{j \, j+1}$ and introduce 
\begin{equation} s_{j,\gamma} =  s_j  + \gamma I_j \in \End\left(\ca C^\infty(\bR^N)\right), \end{equation}%
\nc{rslw}{$s_{j,\gamma}$}{Deformed simple transposition appearing in integral representation of \\ dAHA \nomrefeqpage}%
We remark that also $s_{j,\gamma}$ restricts to an operator on $\ca C^\omega(\bR^N)$.

\begin{prop} \emph{\cite{HeckmanOpdam1997}} \label{dAHAintrep}
The following assignments define a representation $\rho^\n{int}_\gamma$ of $\f H^{N}_\gamma$ on $\ca C^\infty(\bR^N)$:
\begin{equation} \rho^\n{int}_\gamma( s_j) := s_{j,\gamma}, \qquad \rho^\n{int}_\gamma(X_j) := -\ii \partial_j. \end{equation}%
\nc{grl}{$\rho^\n{int}_\gamma$}{Integral representation of dAHA \nomrefeqpage}%
In other words, $\ca C^\infty(\bR^N)$ is an $\f H^N_\gamma$-module; furthermore, $\ca C^\omega(\bR^N)$ is a submodule.
\end{prop}

The proof given by Heckman and Opdam in \cite{HeckmanOpdam1997} refers to \cite{Gutkin1982} for the Coxeter relation $s_{j,\gamma}s_{j+1,\gamma}s_{j,\gamma}$ = $s_{j+1,\gamma}s_{j,\gamma}s_{j+1,\gamma}$, which is the trickiest relation to prove. 
Here, we will present a different proof, relying on the representation $\rho^\n{reg}_\gamma$.

\begin{proof}[Proof of \rfp{dAHAintrep}]
Consider the Fourier expansion of an arbitrary $f \in \ca C^\infty(\bR^N)$,
\[ f = \int_{\bR^N} \dd \bm \lambda \tilde f(\bm \lambda) \e^{\ii \bm \lambda},  \]
for certain Fourier coefficients $\tilde f(\bm \lambda) \in \bC$.
We can use this and \rfl{planewaveregintrep} to turn the axioms (\ref{dAHArel1}-\ref{dAHArel5}) that we need to prove into the axioms of the regular representation, which we already know to hold by virtue of \rfp{dAHAregrep}.\\

For example, to prove that $[s_{j,\gamma} ,s_{k,\gamma}]= 0$ for $|j-k|>1$ we can write
\[ [s_{j,\gamma} ,s_{k,\gamma}]f  = \int_{\bR^N} \dd \bm \lambda \tilde f(\bm \lambda) [s_{j,\gamma} ,s_{k,\gamma}] \e^{\ii \bm \lambda} . \]
\rfl{planewaveregintrep} yields $(s_{j,\gamma} s_{k,\gamma} - s_{k,\gamma} s_{j,\gamma}) \e^{\ii \bm \lambda}  =  (s_{j,\gamma} \tilde s_{k,\gamma} - s_{k,\gamma} \tilde s_{j,\gamma}) \e^{\ii \bm \lambda}$;
because $s_{j,\gamma}$ and $\tilde s_{k,\gamma}$ act on different spaces they commute so that 
$(s_{j,\gamma} s_{k,\gamma} - s_{k,\gamma} s_{j,\gamma}) \e^{\ii \bm \lambda} = ( \tilde s_{k,\gamma} s_{j,\gamma} - \tilde s_{j,\gamma} s_{k,\gamma} ) \e^{\ii \bm \lambda}$
and we can use \rfl{planewaveregintrep} once more.
This gives
\[ [s_{j,\gamma} ,s_{k,\gamma}]f  = - \int_{\bR^N} \dd \bm \lambda \tilde f(\bm \lambda)  \left[ \tilde s_{j,\gamma}, \tilde s_{k,\gamma} \right] \e^{\ii \bm \lambda}, \]
which is zero by virtue of \rfp{dAHAregrep}, so that indeed $[s_{j,\gamma},s_{k,\gamma}]=0$ for $|j-k|>1$.
All axioms can be dealt with in this way.
\end{proof}

Similar to the notation for the regular representation $\rho^\n{reg}_\gamma$, for any $w \in S_N$ with $w=s_{i_1} \ldots s_{i_l}$ for some $i_1,\ldots,i_l \in \{1,\ldots,N-1\}$, we will write 
\begin{equation} w_\gamma = s_{i_1,\gamma} \ldots s_{i_l,\gamma}; \end{equation}%
\nc{rwl}{$w_\gamma$}{Deformed permutation acting in position space \nomrefeqpage}%
because of \rfp{dAHAintrep} this does not depend on the choice of the $s_{i_m}$.
Similarly to the notation for the regular representation, we extend this linearly to the whole group algebra $\bC S_N$; in particular we have
\begin{equation}
\ca S^{(N)}_\gamma = \rho^\n{int}_\gamma(\ca S^{(N)}) = \frac{1}{N!} \sum_{w \in S_N} w_\gamma 
\in \End(\ca C^\infty(\bR^N)) .
\end{equation}%
\nc{rscx}{$\ca S^{(N)}_\gamma$}{Image of symmetrizer under integral representation of \\ dAHA \nomrefeqpage}%
Again, if $w=s_{j \, k}$, then for $ w_\gamma=(s_{j\,k})_\gamma$ we may also write $s_{j\,k,\gamma}$.

\section{The Dunkl-type representation}

Recall the notation $\Sigma(w)$ introduced in \rfeqn{Sigmadef}.
For $j=1,\ldots,N$ we introduce $\Lambda_j\in \End(\ca C^\infty(\bR^N_\n{reg}))$ defined by specifying its action on each alcove:
\begin{equation} \Lambda_j|_{w^{-1} \bR^N_+} = \sum_{k: (k,j) \in \Sigma(w)} s_{j \, k} - \sum_{k: (j,k) \in \Sigma(w)} s_{j \, k}, \qquad \n{for } w \in S_N. 
\end{equation}%
\nc{glc}{$\Lambda_j= \Lambda^{(N)}_j$}{Auxiliary function used in definition of Dunkl-type operator \\ $\partial_{j,\gamma}$ \nomrefeqpage}%
\vspace{-5mm}

\begin{defn}\label{Dunkldefn} \emph{\cite{KomoriHikami,Opdam,MurakamiWadati}}
Let $j=1,\ldots,N$ and $\gamma \in \bR$.
The \emph{Dunkl-type operator} is given by
\begin{equation} \label{Dunklreprestr} 
\partial_{j,\gamma}|_{w^{-1}\bR^N_+} = \partial_j - \gamma \Lambda_j \in \End(\ca C^\infty(\bR^N_\n{reg})),%
\nc{rdlvb}{$\partial_{j,\gamma}$}{Dunkl-type operator \nomrefeqpage}%
\end{equation}
i.e. for $f \in \ca C^\infty(\bR^N_\n{reg})$, $\bm x \in w^{-1} \bR^N_\n{reg}$ and $w \in S_N$ we have
\[ (\partial_{j,\gamma} f)(\bm x) = (\partial_j f)(\bm x) - \gamma 
\left( \sum_{k: (k,j) \in \Sigma(w)} f(s_{j \, k} \bm x)- \sum_{k: (j,k) \in \Sigma(w)} f(s_{j \, k} \bm x) \right). \]
In particular, we have $\partial_{j,\gamma}|_{\bR^N_+} = \partial_j$.
\end{defn}

Alternatively, we may provide a single formula for $\Lambda_j$, and hence $\partial_{j, \gamma}$, on the entire $\bR^N_\n{reg}$ as follows.
Given $\bm i = (i_1, \ldots i_n) \in \{1,\ldots,N\}^n$, introduce the notation 
\begin{equation} \theta_{\bm i} = \theta_{i_1 \, \ldots \, i_n} \end{equation}%
\nc{ghl}{$\theta_{\bm i} = \theta_{i_1 \, \ldots \, i_n}$}{Multiplication operator corresponding to step \\ function \nomrefeqpage}%
for the multiplication operator on $\ca F(\bR^N_\n{reg})$ determined by
\[ (\theta_{\bm i} f)(\bm x) = \theta(x_{i_1}>x_{i_2}> \ldots >x_{i_n}) f(\bm x), \]
for $f \in \ca F(\bR^N_\n{reg})$ and $\bm x \in \bR^N_\n{reg}$.
We remark that $\theta_{\bm i}$ restricts to an endomorphism of $\ca C^\infty(\bR^N_\n{reg})$.
Also note that if $\bm i \not \in \f i^n_N$ (i.e. if some of the $i_m$ are the same) then $\theta_{\bm i}=0$.
It then follows that 
\begin{equation} \label{Lambdaalt} \Lambda_j = \sum_{k<j} \theta_{j \, k}  s_{j \, k}-\sum_{k>j}  \theta_{k \, j} s_{j \, k} \in \End(\ca C^\infty(\bR^N_\n{reg})). \end{equation} 
Hence, for $f \in \ca C^\infty(\bR^N_\n{reg})$ and $\bm x \in \bR^N_\n{reg}$ we have
\[ \partial_{j,\gamma}f(\bm x) = \partial_j f(\bm x) - \gamma \sum_{k<j} \theta(x_j-x_k) f(s_{j \, k}\bm x) +\gamma \sum_{k>j} \theta(x_k-x_j) f(s_{j \, k}\bm x). \]

\begin{prop} \emph{\cite{Opdam,MurakamiWadati}} \label{Dunklrep}
The following assignments define a representation $\rho^\n{Dunkl}_\gamma$ of $\f H^{N}_\gamma$ on $\ca C^\infty(\bR^N_\n{reg})$.
\begin{equation} \rho^\n{Dunkl}_\gamma( s_j) := s_j, \quad \rho^\n{Dunkl}_\gamma(X_j) := -\ii \partial_{j,\gamma}. \label{Dunklrepeqns} \end{equation}%
\nc{grl}{$\rho^\n{Dunkl}_\gamma$}{Dunkl-type representation \nomrefeqpage}%
\end{prop}
\vspace{-12mm}

\begin{proof}
This follows immediately from \rfc{Dunklrep2} and \rfl{Dunklrep5}.
\end{proof}

\begin{lem} \label{Dunklreppoly}
Let $F \in \bC[\lambda_1,\ldots,\lambda_N]^{S_N}$.
Then 
\begin{align}
[w,F(\partial_{1,\gamma},\ldots,\partial_{N,\gamma})] &=0 && \in \End(\ca C^\infty(\bR^N_\n{reg})), \qquad w \in S_N, \label{eqn350}\\
F(\partial_{1,\gamma},\ldots,\partial_{N,\gamma}) &= F(\partial_{1},\ldots,\partial_{N}) && \in \End(\ca C^\infty(\bR^N_\n{reg})). \label{eqn351}
\end{align}
\end{lem}

\begin{proof}
\rfeqn{eqn350} follows immediately from \rfeqn{dAHAcentre} applied to the image of $\f H^{N}_\gamma$ under $\rho^\n{Dunkl}_\gamma$. To obtain \rfeqn{eqn351},
let $f \in \ca C^\infty(\bR^N_\n{reg})$, $\bm x \in \bR^N_+$ and $w \in S_N$.
It suffices to prove that $\left( F(\partial_{1,\gamma},\ldots,\partial_{N,\gamma})f \right)(w^{-1} \bm x) = \left( F(\partial_{1},\ldots,\partial_{N})f \right) (w^{-1} \bm x)$. 
From $\partial_{j,\gamma}|_{\bR^N_+} = \partial_j|_{\bR^N_+}$ it follows that
\begin{align*}
\left( F(\partial_{1,\gamma},\ldots,\partial_{N,\gamma})f \right)(w^{-1} \bm x) &= \left( w F(\partial_{1,\gamma},\ldots,\partial_{N,\gamma}) f \right)(\bm x) &&=  \left( F(\partial_{1,\gamma},\ldots,\partial_{N,\gamma}) w f\right) (\bm x) \\
&= \left( F(\partial_{1},\ldots,\partial_{N}) w f\right) (\bm x)&&= \left( w F(\partial_{1},\ldots,\partial_{N}) f \right)(\bm x) \\
&= \left( F(\partial_{1},\ldots,\partial_{N})f\right)(w^{-1} \bm x). && \hspace{56mm} \qedhere
\end{align*}
\end{proof}
\begin{exm}[Dunkl-type operators for $N=2$]
From \rfd{Dunkldefn} it follows that
\[ \partial_{1,\gamma} = \partial_1 + \gamma \theta_{2 \, 1} s_{1 \, 2}, \qquad 
\partial_{2,\gamma} = \partial_2 - \gamma \theta_{2 \, 1} s_{1 \, 2}, \]
i.e.
\begin{align*} 
(\partial_{1,\gamma}f)(x_1,x_2) &= (\partial_1 f)(x_1,x_2) + \gamma \theta(x_2-x_1) f(x_2,x_1), \\
(\partial_{2,\gamma}f)(x_1,x_2) &= (\partial_2 f)(x_1,x_2) - \gamma  \theta(x_2-x_1) f(x_2,x_1) 
\end{align*}
for $f \in \ca C^\infty(\bR^2_\n{reg})$ and $(x_1,x_2) \in \bR^2_\n{reg}$.
The reader should check that these satisfy the dAHA axioms, viz.
\[ s_1 \partial_{1, \gamma}-\partial_{2, \gamma} s_1 = \gamma, \qquad [ \partial_{1, \gamma}, \partial_{2, \gamma}]=0. \]
\end{exm}

\subsection{Common eigenfunctions of the Dunkl-type operators}

In order to connect the Dunkl-type operators $\partial_{j,\gamma}$ to the study of the QNLS eigenvalue problem, it is important to allow the study of \rfeqnser{QNLS1}{QNLS2} for non-symmetric functions. Consider the following subspaces of $\ca C(\bR^N)$, which were introduced in \cite{Gutkin1982,EmsizOS}.
\begin{align} 
\ca{CB}^1(\bR^N) &= \left\{ f \in \ca C(\bR^N) \, : \, \forall w \, f|_{w \bR^N_+} \n{ has a } \ca C^1 \n{ extension to some open } \right. \nonumber \\
& \hspace{70mm} \left. \n{neighbourhood of }\overline{w \bR^N_+} \right\}; \\%
\nc{rccyb}{$\ca{CB}^1(\bR^N)$}{Set of continuous functions with $\ca C^1$ restriction to regular \\ vectors \nomrefeqpage}%
\ca{CB}^\infty(\bR^N) &= \set{f \in \ca C(\bR^N)}{f|_{\bR^N_\n{reg}} \in \ca C^\infty(\bR^N_\n{reg})}; \\%
\nc{rccyc}{$\ca{CB}^\infty(\bR^N)$}{Set of continuous functions with smooth restriction to regular \\ vectors \nomrefeqpage}%
\ca C^1_\gamma(\bR^N) &= \set{f \in \ca{CB}^1(\bR^N)}{\left(\partial_j-\partial_k\right)f|_{V_{j \, k}^+}-\left(\partial_j-\partial_k\right)f|_{V_{j \, k}^-} =  2 \gamma f|_{V_{j \, k}} \n{ for } 1 \leq j<k \leq N} .%
\end{align}%
\nc{rcczb}{$\ca C^1_\gamma(\bR^N)$}{Set of functions in $\ca{CB}^1(\bR^N)$ satisfying the derivative jump \\
conditions \nomrefeqpage}%
\vspace{-7mm}

Note that $\ca{CB}^\infty(\bR^N) \subset \ca{CB}^1(\bR^N)$.
Furthermore, it has been observed \cite[Prop. 2.2]{EmsizOS} that, due to the hypoellipticity of the Laplacian, 
\[ f \in \ca C^1_\gamma(\bR^N) \n{ and } \Delta f|_{\bR^N_\n{reg}} = -E f|_{\bR^N_\n{reg}} \n{ as distributions } \implies f \in \ca{CB}^\infty(\bR^N). \]

\begin{lem} \label{context}
Let $j=1,\ldots,N$ and $f \in \ca{CB}^\infty(\bR^N)$; suppose that $\partial_{j,\gamma} (f|_{\bR^N_\n{reg}}) \in \ca{C}^\infty(\bR^N_\n{reg})$ is a constant multiple of $f|_{\bR^N_\n{reg}}$.
Then $\partial_{j,\gamma} (f|_{\bR^N_\n{reg}})$ can be continuously extended to $\bR^N$. Hence, $\partial_{j,\gamma} f$ may be viewed as an element of $\ca{CB}^\infty(\bR^N)$. 
\end{lem}

\begin{proof}
There exists $m \in \bC$ such that $\partial_{j,\gamma} (f|_{\bR^N_\n{reg}})=mf|_{\bR^N_\n{reg}}$.
For any $1 \leq k < l \leq N$ we simply define $\partial_{j,\gamma} f|_{V_{k \, l}}$ to be $mf|_{V_{k \, l}}$; since $f$ is continuous this ensures that this extension is continuous.
\end{proof}

For $\bm \lambda = (\lambda_1,\ldots,\lambda_N) \in \bC^N$, we consider the following eigenvalue problem for $f \in \ca{CB}^\infty(\bR^N)$:
\begin{equation} \label{Dunklrepsystem} 
\partial_{j,\gamma}f = \ii \lambda_j f, \qquad \n{for } j=1,\ldots,N.
\end{equation}

\begin{lem}[Uniqueness of solutions] \label{Dunklrepsystemunique}
Let $\bm \lambda \in \bC^N$ and $\gamma \in \bR$. Suppose that $f \in \ca{CB}^\infty(\bR^N)$ satisfies the system \rf{Dunklrepsystem}. 
Then $f$ is uniquely defined up to an overall scalar factor, i.e. the subspace of $\ca{CB}^\infty(\bR^N)$ consisting of solutions of \rf{Dunklrepsystem} is 1-dimensional.
\end{lem}

\begin{proof}
Suppose that $f, g \in \ca{CB}^\infty(\bR^N)$ both satisfy \rfeqn{Dunklrepsystem}. 
We may assume that both $f$ and $g$ are nonzero, and after multiplying one of them by a nonzero complex number, that $f(\bm 0)=g(\bm 0)$. 
Note that $h=f-g\in \ca{CB}^\infty(\bR^N)$ satisfies the same system \rfeqn{Dunklrepsystem}, and $h(\bm 0)=0$. It is sufficient to prove that $h = 0$.\\

\emph{Claim:} Given $w \in S_N$, if $h|_{(w')^{-1} \bR^N_+}=0$ for all $w' \in S_N$ with $l(w')<l(w)$, then $h|_{w^{-1}\bR^N_+}=0$.\\

It is clear that from the claim the lemma follows; in particular it follows that $h|_{\bR^N_+}=0$ and by induction on $l(w)$ we obtain $h|_{\bR^N_\n{reg}}=0$; finally by continuity we have $h=0$.
To prove the claim, by virtue of \rfeqn{Dunklreprestr} we have for $\bm x \in w^{-1} \bR^N_+$,
\begin{equation} \label{Dunklrepeqn5} 
\partial_j h(\bm x) = \ii \lambda_j h(\bm x) + \gamma \left( \sum_{k:(k,j) \in \Sigma(w)} h(s_{j \, k}\bm x) - \sum_{k:(j,k) \in \Sigma(w)} h(s_{j \, k}\bm x) \right),
\quad \n{for } j=1,\ldots,N, \end{equation}
where $s_{j \, k} \bm x \in (w s_{j \, k})^{-1} \bR^N_+$ with $l(w s_{j \, k})<l(w)$ as follows from \rfl{lengthlem1}.
Hence \rfeqn{Dunklrepeqn5} reduces to $\partial_j h|_{w^{-1} \bR^N} =  \ii \lambda_j h$, $ j=1,\ldots,N$, i.e. $h|_{w^{-1} \bR^N} = c_w \e^{\ii \bm \lambda}$ for some $c_w \in \bC$. 
Continuity at $\bm x = \bm 0$ yields that $c_w = 0$, i.e. $h|_{w^{-1} \bR^N}=0$.
\end{proof}

The relevance of system \rf{Dunklrepsystem} to the QNLS Hamiltonian is expressed in
\begin{prop}[The Dunkl-type operators and the QNLS eigenvalue problem] \label{DunklrepQNLS}
Suppose that $f \in \ca{CB}^\infty(\bR^N)$ satisfies the system \rf{Dunklrepsystem} for some $\bm \lambda \in \bC^N$.
Then $f \in \ca C^1_\gamma(\bR^N)$ and $-\Delta f|_{\bR^N_\n{reg}} =  \sum_{j=1}^N \lambda_j^2 f|_{\bR^N_\n{reg}}$,
i.e. $f$ solves the QNLS eigenvalue problem \rfeqnser{QNLS1}{QNLS2} with $E = \sum_{j=1}^N \lambda_j^2$ except for $S_N$-invariance.
\end{prop}

\begin{proof}
That $f$ is an eigenfunction of $-\Delta$ on the regular vectors with eigenvalue $p_2(\bm \lambda)$ follows from $\Delta = p_2(\partial_{1,\gamma},\ldots,\partial_{N,\gamma})$, which in itself is a consequence of \rfl{Dunklreppoly}, \rfeqn{eqn351} applied to $F = p_2$. 
As for the claim that $f$ satisfies the derivative jump conditions, note that from \rfeqn{Lambdaalt} it follows that
\[  \partial_j = \partial_{j,\gamma} +\gamma \sum_{l < j} \theta_{j \, l} s_{j \, l} -\gamma \sum_{l > j} \theta_{l \, j} s_{j \, l} \]
and therefore
\begin{align*} 
\partial_j - \partial_k  &= \partial_{j,\gamma} -  \partial_{k,\gamma} + \gamma \sum_{l<j} \theta_{j \, l} s_{j \, l}
- \gamma \sum_{j<l<k} \theta_{l \, j} s_{j \, l} -\gamma  \theta_{k \, j}  s_{j \, k} 
- \gamma \sum_{l>k} \theta_{j \, l} s_{j \, l} + \\
& \qquad - \gamma \sum_{l < j}  \theta_{k \, l} s_{k \, l} -\gamma \theta_{k \, j} s_{j \, k} - \gamma \sum_{j<l<k} \theta_{k \, l} s_{k \, l} + \gamma \sum_{l > k} \theta_{l \, k} s_{k \, l}  \displaybreak[2] \\
&= \partial_{j,\gamma} -  \partial_{k,\gamma} -2\gamma \theta_{k \, j} s_{j \, k}  + \gamma \sum_{l < j} \left( \theta_{j \, l}s_{j \, l}  - \theta_{k \, l} s_{k \, l} \right) \\
& \qquad + \gamma \sum_{j < l < k} \left( - \theta_{l \, j} s_{j \, l}  - \theta_{k \, l} s_{k \, l} \right) 
+ \gamma \sum_{k < l} \left(  - \theta_{l \, j} s_{j \, l} + \theta_{l \, k} s_{k \, l} \right).
\end{align*}
Applying this to $f \in \ca{CB}^\infty(\bR^N)$ satisfying the system \rf{Dunklrepsystem} we have
\begin{align*} 
(\partial_j - \partial_k) f|_{V_{j \,k}^+} &= \ii (\lambda_j-\lambda_k) f|_{V_{j \, k}^+} + \gamma \sum_{l < j}  \left( \theta_{j \, l} s_{j \, l}  - \theta_{k \, l} s_{k \, l} \right) f|_{V_{j \, k}^+} +\\
& \qquad + \gamma \sum_{j < l < k} \left( - \theta_{l \, j} s_{j \, l}  - \theta_{k \, l} s_{k \, l} \right) f|_{V_{j \, k}^+} + \gamma \sum_{k < l} \left(  - \theta_{l \, j} s_{j \, l} + \theta_{l \, k} s_{k \, l} \right) f|_{V_{j \, k}^+}\\
&= \ii (\lambda_j-\lambda_k) f|_{V_{j \, k}} - \gamma \sum_{j < l < k} s_{j \, l} f|_{V_{j \, k}},
\end{align*}
and
\begin{align*} 
(\partial_j - \partial_k) f|_{V_{j \, k}^-} &= \ii (\lambda_j-\lambda_k) f|_{V_{j \, k}^-} -2\gamma s_{j \, k} f|_{V_{j \, k}^-} + \gamma \sum_{l < j} \left( \theta_{j \, l} s_{j \,l}  - \theta_{k \, l} s_{k \, l} \right) f|_{V_{j \, k}^-} +\\
& \qquad + \gamma \sum_{j < l < k} \left( - \theta_{l \, j} s_{j \, l}  - \theta_{k \, l} s_{k \, l} \right) f|_{V_{j \, k}^-} + \gamma \sum_{k < l} \left(  - \theta_{l \, j} s_{j \, l} + \theta_{l \, k} s_{k \, l} \right) f|_{V_{j \, k}^-}\\
&= \ii (\lambda_j-\lambda_k) f|_{V_{j \, k}} - 2 \gamma f|_{V_{j \, k}} - \gamma \sum_{j < l < k} s_{j \, l} f|_{V_{j \, k}}.
\end{align*}
We conclude that $\left( \partial_j - \partial_k \right) f|_{V_{j \, k}^+} - \left( \partial_j - \partial_k \right) f|_{V_{j \, k}^-} = 2\gamma f|_{V_{j \, k}}$.
\end{proof}

\section{The propagation operator and the pre-wavefunction}

There exists a special element of $\Hom(\ca C^\infty(\bR^N),\ca {CB}^\infty(\bR^N))$ that allows us to construct solutions of the system \rf{Dunklrepsystem}.
  
\begin{defn} \emph{\cite{Hikami}}
Let $\gamma \in \bR$.
The \emph{propagation operator} or \emph{intertwiner} is the following element of $\End(\ca C(\bR^N_\n{reg}))$:
\begin{equation} P_\gamma = \sum_{w \in S_N} w^{-1} \chi_{\bR^N_+} w_\gamma = \sum_{w \in S_N} \chi_{w^{-1} \bR^N_+} w^{-1}  w_\gamma, \end{equation}%
\nc{rpca}{$P_\gamma$}{Propagation operator \nomrefeqpage}%
where $\chi_{w^{-1} \bR^N_+}$ is the multiplication operator associated to the characteristic function of the set $w^{-1} \bR^N_+$.
In other words, $P_\gamma$ is the element of $\End(\ca C(\bR^N_\n{reg}))$ determined by
\begin{equation} \label{propoprestr}
P_\gamma|_{w^{-1} \bR^N_+} = w^{-1} w_\gamma, \qquad \n{for }  w \in S_N.
\end{equation}
\end{defn}
Note that $P_0$ is the identity operator on $\ca C(\bR^N_\n{reg})$.
The propagation operator was introduced by Gutkin \cite{Gutkin1982}. 
Some of its properties in the case of the Weyl group $S_N$ were elucidated by Hikami \cite{Hikami}. 
A vector-valued analogue was considered by Emsiz \cite{Emsiz}. 

\begin{exm}[The propagation operator for $N=2$] \label{propopexample}
For $N=2$ we have
\[ P_\gamma = \chi_{\bR^2_+} + \chi_{s_1 \bR^2_+} s_1  s_{1,\gamma} = 1 + \chi_{s_1 \bR^2_+} (s_1  s_{1,\gamma} -1) = 1 - \gamma \chi_{s_1 \bR^2_+} I_1,\]
i.e.
\[ (P_\gamma f)(x_1,x_2) = f(x_1,x_2) + \gamma \theta(x_2-x_1) \int_0^{x_2-x_1} \dd y f(x_1+y,x_2-y). \]
The reader is invited to check that the following identities hold formally, and consider the proper domain for each identity (i.e. on which function space it acts):
\[ s_1 P_\gamma = P_\gamma s_{1,\gamma}, \qquad \partial_{1,\gamma} P_\gamma = P_\gamma \partial_1, \qquad  \partial_{2,\gamma} P_\gamma = P_\gamma \partial_2. \]
\end{exm}

The statements of the next two lemmas are (at least implicitly) already present in \cite{Gutkin1982,Emsiz,Hikami}.

\begin{lem} \label{propopcont}
Let $\gamma \in \bR$. For $f \in \ca C(\bR^N)$, $P_\gamma f|_{\bR^N_\n{reg}}$ can be continuously extended to $\bR^N$.
As a consequence, $P_\gamma$ restricts to an element of $\End(\ca C(\bR^N))$.
\end{lem}

\begin{proof}
Let $w \in S_N$.
The neighbouring alcoves of $w^{-1} \bR^N_+$ are $(s_j w)^{-1} \bR^N_+$, where $j=1,\ldots,N-1$, and the shared boundary of $w^{-1} \bR^N_+$ and $(s_j w)^{-1} \bR^N_+$ is a subset of the hyperplane $V_{w^{-1}(j) \, w^{-1}(j+1)}$. It is sufficient to prove that, for $j=1,\ldots,N-1$, we have
\[ \lim_{x_{w^{-1}(j+1)} \to x_{w^{-1}(j)} \atop \bm x \in w^{-1} \bR^N_+} (P_\gamma f)(\bm x) = 
\lim_{x_{w^{-1}(j+1)} \to x_{w^{-1}(j)} \atop \bm x \in (s_j w)^{-1} \bR^N_+} (P_\gamma f)(\bm x). \]
Using \rfeqn{propoprestr} this is equivalent to
\begin{equation} \label{eqn311} \lim_{x_{w^{-1}(j+1)} \to x_{w^{-1}(j)}} \left(w^{-1} w_\gamma f \right)(\bm x) = 
\lim_{x_{w^{-1}(j+1)} \to x_{w^{-1}(j)}} \left(w^{-1} s_j s_{j,\gamma} w_\gamma f \right)(\bm x). \end{equation}
For the right-hand side of \rfeqn{eqn311} we note that
\[ w^{-1} s_j s_{j,\gamma} w_\gamma = w^{-1} (1 - \gamma I_{j \, j+1} ) w_\gamma = (1-\gamma I_{w^{-1}(j) \, w^{-1}(j+1)}) w^{-1} w_\gamma. \]
Now using that $\lim_{x_j \to x_k} I_{j \, k} = 0$ establishes \rfeqn{eqn311}.
\end{proof}

\begin{lem} \label{propopcont2}
Let $\gamma \in \bR$. For $f \in \ca C^\infty(\bR^N)$, $P_\gamma f \in \ca C^\infty(\bR^N_\n{reg})$. 
As a consequence, $P_\gamma$ restricts to an element of $\Hom(\ca C^\infty(\bR^N),\ca C^\infty(\bR^N_\n{reg}))$.
\end{lem}
\begin{proof}
Let $w \in S_N$.
Then $P_\gamma f|_{w^{-1} \bR^N_+} = w^{-1} w_\gamma f|_{w^{-1} \bR^N_+}$, which is a linear combination of products of reflection operators $s_{j \, k}$ and integral operators $I_{j \, k}$, both of which send smooth functions to smooth functions.
\end{proof}

By combining \rfls{propopcont}{propopcont2} we obtain
\begin{cor}
Let $\gamma \in \bR$. $P_\gamma$ restricts to an element of $\Hom(\ca C^\infty(\bR^N),\ca{CB}^\infty(\bR^N))$.
\end{cor}

The crucial property of $P_\gamma$ is that it \emph{intertwines} the integral and Dunkl-type representations of the dAHA. 
\begin{thm}[Intertwining property, \cite{Hikami}] \label{intertwine}
Let $\gamma \in \bR$. We have
\begin{align} 
w P_\gamma &= P_\gamma w_\gamma && \in \Hom(\ca C^\infty(\bR^N),\ca{CB}^\infty(\bR^N)), && w \in S_N, \label{intertwine1}\\
\partial_{j,\gamma} (P_\gamma|_{\bR^N_\n{reg}}) &= (P_\gamma \partial_j)|_{\bR^N_\n{reg}} \hspace{-12mm} && \in \Hom(\ca C^\infty(\bR^N),\ca C^\infty(\bR^N_\n{reg})), && j=1,\ldots,N. \label{intertwine2}
\end{align}
\end{thm}

\begin{proof}
For \rfeqn{intertwine1} we simply have
\[ w P_\gamma = \sum_{v \in S_N} w v^{-1} \chi_{\bR^N_+} v_\gamma = \sum_{v \in S_N} v^{-1} \chi_{\bR^N_+} (vw)_\gamma = P_\gamma w_\gamma.\]
To prove \rfeqn{intertwine2} it is sufficient to show that $\partial_{j,\gamma} P_\gamma = P_\gamma \partial_j$ on each alcove $w^{-1} \bR^N_+$ ($w \in S_N$).
Indeed, on $w^{-1} \bR^N_+$ we have
\begin{align*} \partial_{j,\gamma} P_\gamma - P_\gamma \partial_j 
&= \left[ \partial_j, P_\gamma \right] +\gamma \left( \sum_{k: (j,k) \in \Sigma(w)} s_{j \, k} - \sum_{k: (k,j) \in \Sigma(w)} s_{j \, k} \right) P_\gamma \\
&= \left[ \partial_j, P_\gamma \right] +\gamma P_\gamma \left( \sum_{k: (j,k) \in \Sigma(w)} s_{j \, k, \gamma} - \sum_{k: (k,j) \in \Sigma(w)} s_{j \, k, \gamma} \right).
\end{align*}
by virtue of \rfeqn{Dunklreprestr} and \rfeqn{intertwine1}.
Next, \rfeqn{propoprestr} yields
\begin{equation} \label{eqn364}   \partial_{j,\gamma} P_\gamma - P_\gamma \partial_j = w^{-1} \left( \partial_{w(j)} w_\gamma - w_\gamma \partial_j + \gamma w_\gamma \left( \sum_{k: (j,k) \in \Sigma(w)} s_{j \, k,\gamma} - \sum_{k: (k,j) \in \Sigma(w)} s_{j \, k,\gamma} \right) \right). \hspace{-1mm} \end{equation}
Since $(w,X_j) \mapsto (w_\gamma,-\ii \partial_j)$ defines a representation of the dAHA, we can use \rfeqn{dAHAeqn1}, which for this representation reads
\[ w_\gamma \partial_j = \partial_{w(j)} w_\gamma + \gamma w_\gamma \left( \sum_{k: (j,k) \in \Sigma(w)} s_{j \, k,\gamma} - \sum_{k: (k,j) \in \Sigma(w)} s_{j \, k,\gamma}  \right), \]
so that the right-hand side in \rfeqn{eqn364} vanishes.
\end{proof}
We note that the above proof is different from the one given in \cite{Hikami}, although the key ingredient \rfeqn{dAHAeqn1} is the same. 

\subsection{The action of the propagation operator on analytic functions}

We will review further established properties of the propagation operator which are relevant to the study of the QNLS model involving the following subspaces of $\ca C(\bR^N)$:
\begin{align} 
\ca{CB}^\omega(\bR^N) &= \set{f \in \ca C(\bR^N)}{\forall w \, f|_{w \bR^N_+} \n{ has a real-analytic extension to } \bR^N}, %
\nc{rccyd}{$\ca{CB}^\omega(\bR^N)$}{Set of continuous functions with real-analytic restriction to regular \\ vectors \nomrefeqpage} \\%
\ca C^\omega_\gamma(\bR^N) &= \left\{ f \in \ca{CB}^\omega(\bR^N) : \left(\partial_j-\partial_k\right)^r f|_{V_{j \, k}^+}-\left(\partial_j-\partial_k\right)^r f|_{V_{j \, k}^-} = \right. \nonumber \\
& \qquad \left.  (1-(-1)^r) \gamma \left(\partial_j-\partial_k\right)^{r-1} f|_{V_{j \, k}}  \n{ for } 1 \leq j<k \leq N \n{ and } r \in \bZ_{>0} \right\}.
\end{align}%
\nc{rcczd}{$\ca C^\omega_\gamma(\bR^N)$}{Set of functions in $\ca{CB}^\omega(\bR^N)$ satisfying the higher-order boundary jump \\ conditions \nomrefeqpage}%

\vspace{-7mm}

Note that $\ca C^\omega_\gamma(\bR^N) \subset \ca C^1_\gamma(\bR^N)$. Firstly, by considering power series expansions, one obtains
\begin{lem}[Action of $P_\gamma$ on analytic functions] \label{propopanalytic}
Let $\gamma \in \bR$.
$P_\gamma$ restricts to an injective element of $\Hom(\ca C^\omega(\bR^N),\ca{CB}^\omega(\bR^N))$.
\end{lem}

The following further statements are due to \cite{EmsizOS} to which we refer for the detailed proofs.

\begin{prop}[Invertibility of the propagation operator {\cite[Thm. 5.3(ii)]{EmsizOS}}] \label{propopinvtbl}
Let $\gamma \in \bR$. 
Then $P_\gamma$ defines a bijection between $\ca C^\omega(\bR^N)$ and $\ca C^\omega_\gamma(\bR^N)$.
\end{prop}

\begin{proof}
Given $f \in \ca C^\omega_\gamma(\bR^N)$ one considers the unique analytic function $g$ that coincides with $f$ on the fundamental alcove. Then it can be shown by continuity, the derivative jump conditions, and an induction argument that $f=P_\gamma g$ everywhere. 
\end{proof}

Since $\ca C^\omega(\bR^N)$ is a $\f H_\gamma$-submodule of $\ca C^\infty(\bR^N)$ (in terms of the $\rho_\n{int}$-action) from \rfp{propopinvtbl} we obtain
\begin{cor} \label{intertwineanalytic}
Let $\gamma \in \bR$. The operators $\partial_{j,\gamma}-\partial_{k,\gamma}$ ($1 \leq j < k \leq N$) preserve the space $\ca C^\omega_\gamma(\bR^N)$.
Hence, $\ca C^\omega_\gamma(\bR^N)$ is a $\f H^N_\gamma$-module.
Furthermore, in $\Hom(\ca C^\omega(\bR^N),\ca C^\omega_\gamma(\bR^N))$ we have
\begin{align} 
w P_\gamma &= P_\gamma w_\gamma, && w \in S_N, \label{intertwineanalytic1}\\
\hspace{20mm} \left( \partial_{j,\gamma}-\partial_{k,\gamma} \right) P_\gamma&=P_\gamma \left( \partial_j-\partial_k \right), && 1 \leq j < k \leq N, \hspace{20mm} \label{intertwineanalytic2}
\end{align}
\end{cor}

\subsection{The pre-wavefunction}

Since for all $\bm \lambda \in \bC^N$, $\e^{\ii \bm \lambda} \in \ca C^\omega(\bR^N)$, from \rfp{propopinvtbl} we infer that $P_\gamma \e^{\ii \bm \lambda} \in \ca C^\omega_\gamma(\bR^N)$ which shows the relevance of this function to the QNLS model.
However, we will be able to arrive at this statement in a different way. 
For now, we will merely use the analyticity of $\e^{\ii \bm \lambda}$ and \rfl{propopanalytic} to conclude that $P_\gamma \e^{\ii \bm \lambda} \in \ca{CB}^\omega(\bR^N)$. 
First, since this function will be the central object of study for the rest of this thesis, we have
\begin{defn} \label{prewavefndefn}
Let $\gamma \in \bR$ and $\bm \lambda \in \bC^N$.
The \emph{pre-wavefunction} is the function
\begin{equation} \psi_{\bm \lambda} := P_\gamma \e^{\ii \bm \lambda} \in \ca {CB}^\omega(\bR^N). \end{equation}%
\nc{gyl}{$\psi_{\bm \lambda}$}{Pre-wavefunction \nomrefeqpage}%
\end{defn}
\vspace{-15mm}

\begin{rem}
The pre-wavefunction will turn out to provide an intermediate step in the construction of the symmetric wavefunction, but in Chapter \ref{ch5} it will play a more central role.
The reader should think of $\psi_{\bm \lambda}$ as a deformation of the plane wave $\e^{\ii \bm \lambda}$; indeed, for $\gamma=0$ the propagation operator is the identity operator.
The next lemma highlights this further.
\end{rem}

The action of the regular representation in momentum space and the Dunkl-type representation on the pre-wavefunction are intimately related.

\begin{lem} \label{Dunklrepprewavefn} 
Let $\gamma \in \bR$ and $\bm \lambda \in \bC^N$.
Then 
\begin{align} 
\partial_{j,\gamma} \psi_{\bm \lambda} &= \ii \lambda_j \psi_{\bm \lambda}, && \n{for } j=1,\ldots,N, \label{prewavefn1} \\
w \psi_{\bm \lambda} &= \tilde w_\gamma^{-1} \psi_{\bm \lambda}, && \n{for } w \in S_N. \label{prewavefn2}
\end{align}
\end{lem}

\begin{proof}
This follows rather straightforwardly by virtue of the intertwining property of the propagation operator.
More precisely, from \rfeqn{intertwine2} we have
\[  \partial_{j,\gamma} \left( \psi_{\bm \lambda}|_{\bR^N_\n{reg}} \right) 
=  \partial_{j,\gamma} \left( P_\gamma \e^{\ii \bm \lambda}|_{\bR^N_\n{reg}} \right) 
= \left( P_\gamma  \partial_j \e^{\ii \bm \lambda} \right) |_{\bR^N_\n{reg}} = \ii \lambda_j  \left( P_\gamma\e^{\ii \bm \lambda} \right)|_{\bR^N_\n{reg}} = \ii \lambda_j \psi_{\bm \lambda}|_{\bR^N_\n{reg}}, \]
and following \rfl{context} we obtain \rf{prewavefn1}.
Also, from  \rfeqn{intertwine1} we have
\[ w \psi_{\bm \lambda} = w P_\gamma \e^{\ii \bm \lambda} = P_\gamma w_\gamma \e^{\ii \bm \lambda} = P_\gamma \tilde w_\gamma^{-1} \e^{\ii \bm \lambda} = \tilde w_\gamma^{-1} P_\gamma \e^{\ii \bm \lambda} = \tilde w_\gamma^{-1} \psi_{\bm \lambda}, \]
where we have used \rfeqn{planewaveregintrep5}.
We note that if $\bm \lambda \in \bC^N \setminus \bC^N_\n{reg}$, \rfeqn{prewavefn2} may have to be interpreted in accordance with the proof of \rfl{planewaveregintrep}, \rfeqn{planewaveregintrep1} for the irregular case ($x_j = x_k$), i.e. as the result of an appropriate limit in momentum space.
\end{proof}

Let $\bm \lambda \in \bC^N$.
We recall that the plane wave $\e^{\ii \bm \lambda} \in \ca C^\infty(\bR^N)$ is the unique solution (up to a constant factor) of the system of partial differential equatons $\partial_j f = \ii \lambda_j f$ for $j=1,\ldots,N$.
From \rfeqn{prewavefn1} it follows that $\psi_{\bm \lambda}$ for $\bm \lambda \in \bC^N$ solves the ``deformed'' system \rf{Dunklrepsystem} in $\ca{CB}^\infty(\bR^N)$.
Hence, from \rfl{Dunklrepsystemunique} we obtain the following result which will be important in Chapter \ref{ch5}:
\begin{cor} \label{prewavefnDunkl}
Let $\gamma \in \bR$ and $\bm \lambda \in \bC^N$.
Then the solution set of the system \rf{Dunklrepsystem} in $\ca{CB}^\infty(\bR^N)$ is one-dimensional, and it is spanned by $\psi_{\bm \lambda}$.
\end{cor}

Moreover, by virtue of \rfp{DunklrepQNLS} we obtain that $\psi_{\bm \lambda}$ satisfies the QNLS eigenvalue problem:
\begin{cor} \label{prewavefnQNLS}
Let $\gamma \in \bR$ and $\bm \lambda \in \bC^N$.
Then $\psi_{\bm \lambda}$ satisfies the derivative jump conditions and $-\Delta \psi_{\bm \lambda}|_{\bR^N_\n{reg}} =  \sum_{j=1}^N \lambda_j^2 \psi_{\bm \lambda}|_{\bR^N_\n{reg}}$.
\end{cor}

\begin{rem}
We have remarked on the fact that $\psi_{\bm \lambda}$ satisfies the derivative jump conditions simply by virtue of being the image of a real-analytic function under $P_\gamma$, cf. \rfp{propopinvtbl}. However, we wish to highlight the point that the pre-wavefunction satisfies these conditions by virtue of being a common eigenfunction of the Dunkl-type operators.
\end{rem}

\begin{lem} \label{prewavefnregrep}
Let $\gamma \in \bR$ and $\bm \lambda \in \bC^N$. We have 
\[ \psi_{\bm \lambda} = \sum_{w \in S_N} \chi_{w^{-1} \bR^N_+} \tilde w_\gamma^{-1} \tilde w  \e^{\ii \bm \lambda}. \] 
\end{lem}

\begin{proof}
We have, for each $w \in S_N$,
\[ w^{-1} w_\gamma \e^{\ii \bm \lambda} = w^{-1} \tilde w_\gamma^{-1} \e^{\ii \bm \lambda} = \tilde w_\gamma^{-1} w^{-1} \e^{\ii \bm \lambda} = \tilde w_\gamma^{-1} \tilde w \e^{\ii \bm \lambda},\]
by virtue of \rfc{planewaveregintrepcor}. 
Now multiply by $\chi_{w^{-1} \bR^N_+}$ and sum over all $w \in S_N$. 
\end{proof}
 
\begin{exm}[The pre-wavefunction for $N=2$] \label{prewavefnexample}
From \rfl{prewavefnregrep} we obtain that
\begin{align*} 
\psi_{\lambda_1,\lambda_2} &= \left( \chi_{ \bR^2_+} + \chi_{s_1 \bR^2_+} \tilde s_{1,\gamma} \tilde s_1 \right) \e^{\ii (\lambda_1,\lambda_2)} 
\; = \; \left(1 + \ii \gamma \chi_{s_1 \bR^2_+} \tilde \Delta_1\right) \e^{\ii (\lambda_1,\lambda_2)}\\
&= \e^{\ii(\lambda_1,\lambda_2)}+\frac{\ii \gamma}{\lambda_1-\lambda_2} \theta_{2\,1} \left(  \e^{\ii(\lambda_1,\lambda_2)} -  \e^{\ii(\lambda_2,\lambda_1)}\right), 
\end{align*}
i.e.
\[ \psi_{\lambda_1,\lambda_2}(x_1,x_2) = \e^{\ii(\lambda_1 x_1 + \lambda_2 x_2)}+\frac{\ii \gamma}{\lambda_1-\lambda_2} \theta(x_2-x_1) \left(  \e^{\ii(\lambda_1 x_1 + \lambda_2 x_2)} -  \e^{\ii(\lambda_2 x_1 + \lambda_1 x_2)}\right) . \]
It follows from \rfc{intertwine} that we have
\[ \partial_{1,\gamma} \psi_{\lambda_1,\lambda_2} = \ii \lambda_1 \psi_{\lambda_1,\lambda_2}, \qquad
\partial_{2,\gamma} \psi_{\lambda_1,\lambda_2} = \ii \lambda_2 \psi_{\lambda_1,\lambda_2}, \]
and hence
\begin{gather*}
-(\partial_1^2+\partial_2^2)\psi_{\lambda_1,\lambda_2}(x_1,x_2) = (\lambda_1^2+\lambda_2^2) \psi_{\lambda_1,\lambda_2}(x_1,x_2), \qquad \n{if } x_1 \ne x_2 \\
\lim_{x_2 \to x_1 \atop x_1 > x_2} (\partial_1-\partial_2) \psi_{\lambda_1,\lambda_2}(x_1,x_2)-
\lim_{x_2 \to x_1 \atop x_1 < x_2} (\partial_1-\partial_2) \psi_{\lambda_1,\lambda_2}(x_1,x_2) =
2 \gamma \psi_{\lambda_1,\lambda_2}(x_1,x_1).
\end{gather*}
\end{exm}

\section{The Bethe wavefunction}

Note that for any $\bm \lambda \in \bC^N_\n{reg}$ the pre-wavefunction $\psi_{\bm \lambda} \in \ca{CB}^\infty(\bR^N)$ solves the QNLS eigenvalue problem - except that it is not $S_N$-invariant (either in momentum or position space).
This can be rectified by symmetrizing $\psi_{\bm \lambda}$ in position space.

\begin{defn} \label{wavefndefn}
Let $\gamma \in \bR$ and $\bm \lambda \in \bC^N$.
The \emph{Bethe wavefunction} is given by
\[ \Psi_{\bm \lambda} := \ca S^{(N)} \psi_{\bm \lambda} = \frac{1}{N!} \sum_{w \in S_N} w \psi_{\bm \lambda} \in \ca{CB}^\infty(\bR^N)^{S_N}. \]
\end{defn}
%

\begin{rem}
This $\Psi_{\bm \lambda}$ is the same function as $\Psi_{\bm \lambda}$ defined in Chapter \ref{chQISM} (\rfd{Bethewavefnintops}) using the quantum inverse scattering method. 
We will discuss this equality further in Chapter \ref{chInterplay} and we will give a (new) proof of it in Chapter \ref{ch5} (\rfc{Brecursion}). 
\end{rem}

\begin{thm} \label{wavefnsymmpolys} \cite{Hikami,EmsizOS,Gutkin1987}
Let $\gamma \in \bR$ and $\bm \lambda \in \bC^N$.
$\Psi_{\bm \lambda}$ is an eigenfunction of $F( \partial_{1,\gamma},\ldots, \partial_{N,\gamma})$ with eigenvalue $F(\ii \bm \lambda)$, for any $F \in \bC[\bm \lambda]^{S_N}$.
Furthermore, $\Psi_{\bm \lambda}$ solves the QNLS eigenvalue problem: it satisfies \rfeqn{QNLS1} with $E=\| \bm \lambda \|^2$ and the derivative jump conditions \rfeqn{QNLS2}.
\end{thm}

\begin{proof}
Using \rfl{Dunklreppoly} we have
\begin{align*} 
F( \partial_{1,\gamma},\ldots, \partial_{N,\gamma}) \Psi_{\bm \lambda} 
&= F( \partial_{1,\gamma},\ldots, \partial_{N,\gamma}) \ca S^{(N)} \psi_{\bm \lambda} &&= \ca S^{(N)} F( \partial_{1,\gamma},\ldots, \partial_{N,\gamma}) \psi_{\bm \lambda} \\
&=  \ca S^{(N)} F(\ii \bm \lambda) \psi_{\bm \lambda} &&= F(\ii \bm \lambda) \Psi_{\bm \lambda}, 
\end{align*}
proving the first statement. 
Hence, by virtue of \rfc{prewavefnQNLS}, $\Psi_{\bm \lambda}$ is an eigenfunction of $-\Delta$.
As for the derivative jump conditions \rfeqn{QNLS2}, it can be easily checked that if $f \in \ca{CB}^\infty(\bR^N)$ satisfies it, so does $wf$, for any $w \in S_N$. Taking $f = \psi_{\bm \lambda}$, using \rfp{DunklrepQNLS}, summing over all $w$ and dividing out a factor $N!$ we obtain that $\Psi_{\bm \lambda}$ also satisfies \rfeqn{QNLS2}.
\end{proof}

\begin{rem}
By virtue of \rft{wavefnsymmpolys}, we see that the QNLS integrals of motion arise as symmetric expressions $F(\partial_{1,\gamma},\ldots,\partial_{N,\gamma})$ in the Dunkl-type operators $\partial_{j,\gamma}$, which mutually commute and act on Bethe wavefunctions as multiplication by $F(\ii \bm \lambda)$. 
Recall the power sum polynomials $p_n$ defined by $p_n(\bm \lambda) = \sum_{j=1}^N \lambda_j^n$.
It is well-known \cite{Macdonald1} that symmetric polynomials in $\bm \lambda$ are themselves polynomial expressions in the $p_n(\bm \lambda)$, where $n=1,\ldots,N$. There are other sets that generate $\bC[\bm \lambda]^{S_N}$ in this way, but the $p_n$ allow for a useful physical interpretation.
More precisely, the $p_n(-\ii \partial_{1,\gamma},\ldots,-\ii \partial_{N,\gamma})$ reproduce the integrals of motion discussed in Subsect. \ref{ssemonodromymatrix}. 
In particular, 
\begin{align*}
p_0(-\ii \partial_{1,\gamma},\ldots,-\ii \partial_{N,\gamma}) &= N, \\
p_1(-\ii \partial_{1,\gamma},\ldots,-\ii \partial_{N,\gamma}) &=  -\ii (\partial_{1}+\ldots+\partial_{N}), \\
p_2(-\ii \partial_{1,\gamma},\ldots,-\ii \partial_{N,\gamma}) &= -\Delta. 
\end{align*}
\end{rem}

The Bethe wavefunction $\Psi_{\bm \lambda}$ can also be obtained from $\psi_{\bm \lambda}$ through a symmetrization in momentum space; in particular this demonstrates that $\Psi_{\bm \lambda}$ is $S_N$-invariant not only in the particle coordinates, but also in the particle momenta.

\begin{prop} \label{wavefnregrep}
Let $\gamma \in \bR$ and $\bm \lambda \in \bC^N$.
Then 
\begin{equation} \label{wavefnexpr1} \Psi_{\bm \lambda} = \tilde{\ca S}^{(N)} G^{(N)}_\gamma(\bm \lambda) \psi_{\bm \lambda}.
\end{equation}
\end{prop}

\begin{proof}
We have
\begin{align*} 
\tilde{\ca S}^{(N)} G^{(N)}_\gamma(\bm \lambda) \psi_{\bm \lambda} 
& = \tilde{\ca S}^{(N)}_\gamma P^{(N)}_\gamma \e^{\ii \bm \lambda} 
&& = P^{(N)}_\gamma \tilde{\ca S}^{(N)}_\gamma \e^{\ii \bm \lambda} 
&&= P^{(N)}_\gamma \ca S^{(N)}_\gamma \e^{\ii \bm \lambda} \\
&= \ca S^{(N)} P^{(N)}_\gamma \e^{\ii \bm \lambda} 
&& = \ca S^{(N)} \psi_{\bm \lambda} 
&& = \Psi_{\bm \lambda}, 
\end{align*}
where we have used \rfp{regrepsymmetrizers}, \rfeqn{planewaveregintrep5} and \rfeqn{intertwine1}.
\end{proof}

This approach leads us to the following well-known statement expressing the Bethe wavefunction in terms of plane waves. This particular proof of it does not appear to be in the literature.

\begin{prop}\cite{Hikami,EmsizOS,Gutkin1987} \label{wavefnexprs}
Let $\gamma \in \bR$ and $\bm \lambda \in \bC^N$. Then 
\begin{equation} \label{eqn370} \Psi_{\bm \lambda} = \sum_{w \in S_N} \chi_{w^{-1} \bR^N_+} \tilde{\ca S}^{(N)} G^{(N)}_\gamma(w \bm \lambda) \e^{\ii \bm \lambda}. \end{equation}
In particular,
\begin{equation} \label{eqn371} \Psi_{\bm \lambda}|_{\bR^N_+} = \tilde{\ca S}^{(N)} G^{(N)}_\gamma(\bm \lambda) \e^{\ii \bm \lambda}. \end{equation}
\end{prop}

\begin{proof}
First, \rfeqn{eqn371} is established by restricting \rfeqn{wavefnexpr1} to the fundamental alcove and using that $\psi_{\bm \lambda}|_{\bR^N_+} = \e^{\ii \bm \lambda}$.
\rfeqn{eqn370} follows from the principle that a symmetric continuous function is completely determined by its values in the fundamental alcove. More precisely,
\begin{align*} 
\Psi_{\bm \lambda}(\bm x) &= \sum_{w \in S_N} \chi_{\bR^N_+}(w \bm x)\Psi_{\bm \lambda}(w \bm x) 
&&= \sum_{w \in S_N} \chi_{\bR^N_+}(w \bm x) \tilde{\ca S}^{(N)} G^{(N)}_\gamma(\bm \lambda) \e^{\ii \inner{\bm \lambda}{w \bm x}} \\
&= \sum_{w \in S_N} \chi_{w^{-1}\bR^N_+}(\bm x) \tilde{\ca S}^{(N)} G^{(N)}_\gamma(\bm \lambda) \e^{\ii \inner{w^{-1} \bm \lambda}{\bm x}} 
&& = \sum_{w \in S_N} \chi_{w^{-1}\bR^N_+}(\bm x) \tilde{\ca S}^{(N)} G^{(N)}_\gamma(w \bm \lambda) \e^{\ii \bm \lambda}(\bm x),
\end{align*}
since $\tilde{\ca S}^{(N)} \tilde w = \tilde{\ca S}^{(N)}$ for all $w \in \ca S^{(N)}$.
\end{proof}

\subsection{Periodicity} \label{subsectperiodicity}

In order to solve the QNLS problem on a interval of length $L \in \bR_{>0}$, we need to impose $L$-periodicity on $\Psi_{\bm \lambda}$ and $\partial_j \Psi_{\bm \lambda}$ in the $j$-th argument for $j=1,\ldots,N$. 
We will recover the Bethe ansatz equations as conditions on the $\bm \lambda$ that ensure periodicity. Denote $t_- = s_{N-1} \ldots s_2 s_1\in S_N$, so that $t_-(j) = j-1$ (mod $N$).

\begin{lem} \label{Gprop}
Let $\gamma \in \bR$, $L \in \bR_{>0}$ and $\bm \lambda \in \bC^N$.
Assume $\bm \lambda$ satisfies the BAEs \rf{BAE}.
Then
\[ \e^{\ii \lambda_N L}  G^{(N)}_\gamma(\tilde t_- \bm \lambda) = G^{(N)}_\gamma(\bm \lambda). \]
\end{lem}
\begin{proof}
This follows from \rfeqn{GLambda} and \rfeqn{BAEvariant} with $j=N$, writing $\bm \lambda' = (\lambda_1,\ldots,\lambda_{N-1})$:
\begin{align*}
\e^{\ii \lambda_N L}  G^{(N)}_\gamma(\tilde t_- \bm \lambda) &=  \e^{\ii \lambda_N L}  G^{(N)}_\gamma(\lambda_N,\bm \lambda') && = \e^{\ii \lambda_N L} G^{(N-1)}_\gamma(\bm \lambda') \tau^-_{\lambda_N}(\bm \lambda') \\
&=G^{(N-1)}_\gamma(\bm \lambda') \tau^+_{\lambda_N}(\bm \lambda') && = G^{(N)}_\gamma(\bm \lambda). \qedhere \end{align*}
\end{proof}

\begin{prop} \label{Psiperiodic} \cite{EmsizOS,Gutkin1987}
Let $\gamma \in \bR$, $L \in \bR_{>0}$ and $\bm \lambda \in \bC^N$.
Assume $\bm \lambda$ satisfies the BAEs \rf{BAE}.
Then $\Psi_{\bm \lambda}$ and $\partial_j \Psi_{\bm \lambda}$, for $j=1,\ldots,N$, are $L$-periodic in each argument.
\end{prop}

\begin{proof}
Because of the $S_N$-invariance of $\Psi_{\bm \lambda}$ it suffices to prove that
\begin{align*} 
\Psi_{\bm \lambda}(x_1,\ldots,x_{N-1},-L/2) &= \Psi_{\bm \lambda}(L/2,x_1,\ldots,x_{N-1}), \\
\partial_N \Psi_{\bm \lambda}(x_1,\ldots,x_N)|_{x_N = -L/2} &= \partial_N \Psi_{\bm \lambda}(x_N,x_1,\ldots,x_{N-1})|_{x_N = L/2}, 
\end{align*}
where $(x_1,\ldots,x_{N-1}) \in \overline{J^N_+}$, with $J=[-L/2,L/2]$.
Because of continuity of $\Psi_{\bm \lambda}$ we may relax this to $(x_1,\ldots,x_{N-1}) \in J^N_+$, i.e. $L/2 > x_1 > \ldots > x_{N-1} > -L/2$.
Note that $\bm \lambda \in \bR^N$ because of the BAEs.
We have
\begin{align*}
\Psi_{\bm \lambda}(L/2,x_1,\ldots,x_{N-1})  &= \tilde{\ca S}^{(N)} G^{(N)}_\gamma(\bm \lambda) \e^{\ii \lambda_1 L/2} \e^{\ii (\lambda_2 x_1 + \ldots + \lambda_N x_{N-1})} \displaybreak[2] \\
&= \tilde{\ca S}^{(N)} \tilde t_- G^{(N)}_\gamma(\bm \lambda) \e^{\ii \lambda_1 L/2} \e^{\ii (\lambda_2 x_1 + \ldots + \lambda_N x_{N-1})} \displaybreak[2] \\
&= \tilde{\ca S}^{(N)} G^{(N)}_\gamma(\tilde t_- \bm \lambda) \e^{\ii \lambda_{t_- (1)} L/2} \e^{\ii (\lambda_{t_- (2)} x_1 + \ldots + \lambda_{t_- (N)} x_{N-1})} \displaybreak[2] \\
&= \tilde{\ca S}^{(N)} G^{(N)}_\gamma(\tilde t_- \bm \lambda) \e^{\ii \lambda_N L/2} \e^{\ii (\lambda_1 x_1 + \ldots + \lambda_{N-1} x_{N-1})}.
\end{align*}
On the other hand
\[ \Psi_{\bm \lambda}(x_1,\ldots,x_{N-1},-L/2) = \tilde{\ca S}^{(N)} G^{(N)}_\gamma(\bm \lambda) \e^{-\ii \lambda_N L/2} \e^{\ii (\lambda_1 x_1 + \ldots + \lambda_{N-1} x_{N-1})}, \]
so that it is sufficient to prove
\[ \tilde{\ca S}^{(N)} G^{(N)}_\gamma(\tilde t_- \bm \lambda) \e^{\ii \lambda_N L/2} \e^{\ii (\lambda_1 x_1 + \ldots + \lambda_{N-1} x_{N-1})} = \tilde{\ca S}^{(N)} G^{(N)}_\gamma(\bm \lambda) \e^{-\ii \lambda_N L/2} \e^{\ii (\lambda_1 x_1 + \ldots + \lambda_{N-1} x_{N-1})}. \]
Similarly, for the condition on the derivative, we obtain that it is sufficient to prove
\begin{gather*} 
\tilde{\ca S}^{(N)} G^{(N)}_\gamma(\tilde t_- \bm \lambda) \lambda_N \e^{\ii \lambda_N L/2}  \e^{\ii (\lambda_1 x_1 + \ldots + \lambda_{N-1} x_{N-1})} = \qquad \\
\qquad \tilde{\ca S}^{(N)} G^{(N)}_\gamma(\bm \lambda) \lambda_N \e^{-\ii \lambda_N L/2} \e^{\ii (\lambda_1 x_1 + \ldots + \lambda_{N-1} x_{N-1})}. \end{gather*}
Applying \rfl{Gprop} completes the proof for both $\Psi_{\bm \lambda}$ and its derivative.
\end{proof}

We draw the reader's attention to the fact that if $\gamma \ne 0$ the pre-wavefunction $\psi_{\bm \lambda}$ ($\bm \lambda \in \bC^N$) cannot be made periodic by imposing a condition on $\bm \lambda$. We will illustrate this for the case $N=2$.

\begin{exm}[$N=2$]\label{prewavefnperiodicity}
From \rfex{prewavefnexample} we will show that imposing $L$-periodicity on $\psi_{\lambda_1,\lambda_2}$ in either argument leads to a contradiction. We have
\[ \psi_{\lambda_1,\lambda_2}(x_1,x_2) = \e^{\ii(\lambda_1 x_1 + \lambda_2 x_2)}+\frac{\ii \gamma}{\lambda_1-\lambda_2} \theta(x_2-x_1) \left(  \e^{\ii(\lambda_1 x_1 + \lambda_2 x_2)} -  \e^{\ii(\lambda_2 x_1 + \lambda_1 x_2)}\right) . \]
$L$-periodicity in the first argument, viz. $\psi_{\lambda_1,\lambda_2}(-L/2,x) = \psi_{\lambda_1,\lambda_2}(L/2,x)$ for $-L/2 < x< L/2$, translates as
\[ \frac{\lambda_1-\lambda_2+\ii \gamma}{\lambda_1-\lambda_2} \e^{-\ii \lambda_1 L/2} \e^{\ii \lambda_2 x} - \frac{\ii \gamma}{\lambda_1-\lambda_2} \e^{-\ii \lambda_2 L/2} \e^{\ii \lambda_1 x} = \e^{\ii \lambda_1 L/2} \e^{\ii \lambda_2 x} , \]
i.e. for all $x \in (-L/2,L/2)$
\[  \e^{\ii (\lambda_1-\lambda_2) x} = \frac{\lambda_1-\lambda_2+\ii \gamma}{\ii \gamma} \e^{-\ii (\lambda_1-\lambda_2) L/2} - \frac{\lambda_1-\lambda_2}{\ii \gamma} \e^{\ii (\lambda_1+\lambda_2) L/2}. \]
It follows that $\lambda_1=\lambda_2$, which leads to a contradiction as follows. 
By De l'H\^opital's rule we have 
\[ \psi_{\lambda,\lambda}(x_1,x_2) := \lim_{\lambda_2, \lambda_1 \to \lambda} \psi_{\lambda_1,\lambda_2}(x_1,x_2) = \e^{\ii \lambda(x_1+x_2)}\left(1+\gamma \theta(x_2-x_1) (x_2-x_1)\right). \]
Hence $\psi_{\lambda,\lambda}(-L/2,x) = \psi_{\lambda,\lambda}(L/2,x)$ for all $x \in (-L/2,L/2)$ implies that for all such $x$, $1+\gamma(x+\frac{L}{2})=\e^{\ii \lambda L}$. Substituting $x = \pm L/2$ leads to $\gamma L =0$, contradictory to assumptions. $L$-periodicity in the second argument can be ruled out by applying \rfl{Dunklrepprewavefn}. 
\end{exm}
\newpage

\chapter[Interplay between the QISM and the dAHA]{Interplay between the quantum inverse scattering method and the degenerate affine Hecke algebra} \label{chInterplay}

The purpose of this chapter is to highlight connections between the two discussed methods for solving the QNLS eigenvalue problem, some of which may be known to experts in the field, but which are not discussed in the literature. This interplay can be seen as something reminiscent of Schur-Weyl duality; the Yangian of $\f{gl}_2$, the algebraic object underlying the QYBE, is a deformation of the current algebra of $\f{gl}_2$ and its representation theory should be related to that of the degenerate affine Hecke algebra, which is a deformation of the group algebra of $S_N$.

\section{Equality of the wavefunctions and dimension of the solution spaces}

In chapters 2 and 3 we have reviewed the construction of the Bethe wavefunctions $\Psi_{\bm \lambda}$ using two different methods. To distinguish them, for now we refer to them as $\Psi^\n{QISM}_{\bm \lambda} = B_{\lambda_N} \ldots B_{\lambda_1} \Vac$ and $\Psi^\n{dAHA}_{\bm \lambda} = \ca S^{(N)} \psi_{\bm \lambda}$. 
In order to prove that $\Psi^\n{QISM}_{\bm \lambda}=\Psi^\n{dAHA}_{\bm \lambda}$, it would be helpful if the solution space of the set of equations they solve were one-dimensional.
Unfortunately, the solution space of the system \rfeqnser{QNLS1}{QNLS2} is not one-dimensional;
We illustrate this by the case $N=1$. Let $\lambda \in \bC$. 
The most general eigenfunction of $-\partial_x^2$ with eigenvalue $\lambda^2$ is given by 
\[ \Psi(x) = c_0 \left( \cos(\lambda x)+\frac{\ii \mu}{\lambda} \sin(\lambda x) \right)= \frac{c_0}{2} \left( \frac{\lambda+\mu}{\lambda} \e^{\ii \lambda x} + \frac{\lambda-\mu}{\lambda} \e^{-\ii \lambda x} \right), \]
with $c_0,\mu \in \bC$.
For $L$-periodicity we need the Bethe ansatz condition $\e^{\ii \lambda L}=1$, but no additional condition on $\mu$. Thus we have a 2-dimensional space parametrized by the constants $c_0$ and $\mu$.
If we choose $\mu = \pm \lambda$, $\Psi$ is also an eigenfunction of $-\ii \partial_x$ with eigenvalue $\mu$ and we recover $\Psi_\mu \propto \e^{\ii \mu}$.\\

However in \cite{EmsizOS} a useful result is obtained for the generalization of the eigenvalue problem to affine Weyl groups. In the case of the symmetric group, and using the present notation, we have
\begin{thm}\cite[Thm.~2.6]{EmsizOS} \label{QNLS1D}
Write $J=[-L/2,L/2]$ for some $L \in \bR_{>0}$.
Let $\gamma \in \bR$ and $\bm \lambda \in \bC^N$.
Then the vector space of functions $\Psi \in \ca{CB}^\infty(\bR^N)^{S_N}$ satisfying the derivative jump conditions \rfeqn{QNLS2} and 
\begin{equation}
\label{QNLS3}
P(-\ii \partial_{1,\gamma},\ldots,-\ii \partial_{N,\gamma})\Psi|_{J^N_\n{reg}} = P(\bm \lambda)\Psi|_{J^N_\n{reg}}, \qquad \n{for } P \in \bC[\bm X]^{S_N} 
\end{equation}
has dimension at most 1; the dimension equals 1 if and only if $\bm \lambda$ is a solution of the Bethe ansatz equations \rfeqn{BAE} and $\bm \lambda \in \bR^N_\n{reg}$.
\end{thm}

It is important to impose \rfeqn{QNLS3}, which is a stronger condition than \rfeqn{QNLS1}, 
and can be seen as defining a particular self-adjoint extension of $-\Delta$.
Fortunately, both the QISM and the dAHA approach take this into account; both yield that $\Psi_{\bm \lambda}$ satisfies \rf{QNLS3}, so we obtain from \rft{QNLS1D} that 
$\Psi^\n{QISM}_{\bm \lambda}$ and $\Psi^\n{dAHA}_{\bm \lambda}$ are proportional. \\

Having discussed the dimensionality of the solution space of the QNLS eigenvalue problem, we can use this to demonstrate that $\Psi^\n{QISM}_{\bm \lambda} = \Psi^\n{dAHA}_{\bm \lambda}$. 

\begin{thm} 
Let $\bm \lambda \in \bC^N_\n{reg}$.
Then $\Psi^\n{QISM}_{\bm \lambda} =\Psi^\n{dAHA}_{\bm \lambda}$.
\end{thm}

\begin{proof}
In view of \rft{QNLS1D} it is sufficient to establish that $\Psi^\n{QISM}_{\bm \lambda}(\bm 0) =\Psi^\n{dAHA}_{\bm \lambda}(\bm 0)$; we will in fact show that evaluating both expressions for the Bethe wavefunction at $\bm x = \bm 0$ yields 1.
For the QISM wavefunction, we use induction. The case $N=0$ is obvious.
To go from $N$ to $N+1$, let $\bm i \in \f I^{n+1}_{N+1}$ and note that $\lim_{\bm x \to \bm 0} \hat E_{\bm i} = 0$ unless $n=0$. Hence for $\bm \lambda \in \bC^{N+1}_\n{reg}$ write $\bm \lambda' = (\lambda_1,\ldots,\lambda_N)$ and note that
\begin{align*} 
\Psi^\n{QISM}_{\bm \lambda}(\underbrace{0,\ldots,0}_{N+1}) &= \lim_{x_1,\ldots,x_{N+1} \to 0} \left(B_{\lambda_{N+1}}\Psi^\n{QISM}_{\bm \lambda'}\right)(\bm x) \displaybreak[2] \\
&= \lim_{x_1,\ldots,x_{N+1} \to 0}\frac{1}{N+1} \sum_{j=1}^N \e^{\ii \lambda_{N+1} x_j}\Psi^\n{QISM}_{\bm \lambda'}(\bm x_{\hat \jmath}) \; = \; 1,
\end{align*}
by virtue of the induction hypothesis, and the fact that $B_\lambda F$ is continuous at the hyperplanes $V_{j\, k}$ for all $F \in \ca H$.
To prove the statement for $\Psi^\n{dAHA}_{\bm \lambda}$, simply let \rfeqn{regrepsymmetrizerseqn} act on $1 \in \ca C^\omega(\bC^N_{\n reg})$ and use that $w_\gamma(1)=1$ for all $w \in S_N$.
\end{proof}
We will give a new proof in Chapter \ref{ch5} that $\Psi^\n{QISM}_{\bm \lambda} = \Psi^\n{dAHA}_{\bm \lambda}$.

\begin{rem}
Having fixed $\Psi_{\bm \lambda}(\bm 0) = 1$, we have chosen a normalization of the wavefunction. This does not make $\Psi_{\bm \lambda}$ a probability amplitude, i.e. $\| \Psi_{\bm \lambda} \| \ne 1$. By obtaining formulae for the norm $\| \Psi_{\bm \lambda} \|$ one finds the corresponding probability amplitude $\frac{1}{\| \Psi_{\bm \lambda} \|}\Psi_{\bm \lambda}$.
\end{rem}

\section{Recursive constructions using the dAHA}\label{dAHArecursions}

As mentioned earlier, an advantage of the QISM is the appearance of a recursive formula \rfeqn{Psi} immediately from the definition of $\Psi_{\bm \lambda}$. In the dAHA method these recursive relations can be obtained from the embedding of $S_N$ into $S_{N+1}$ in \rfl{symmgrouprecursion}. We will do this for the pre-wavefunction $\psi_{\bm \lambda}$ as well as the Bethe wavefunction $\Psi_{\bm \lambda}$. The regular representation of the dAHA in momentum space plays a key role here. \\

For $\lambda \in \bC$, denote by $\hat e^-_\lambda$ the element of $\End(\ca F(\bR^N),\ca F(\bR^{N+1}))$ determined by
\[ \left( \hat e^-_\lambda f \right)(x_1,\ldots,x_{N+1}) = \e^{\ii \lambda x_{N+1}} f(x_1,\ldots,x_N), \]
for $f \in \ca F(\bR^N)$ and $(x_1,\ldots,x_{N+1}) \in \bR^{N+1}$.

\begin{prop}[Recursion for the pre-wavefunction] \label{prewavfnrecursionregrep}
Let $\bm \lambda$ = $(\lambda_1,\ldots,\lambda_{N+1})$ $\in$ $\bC^{N+1}$ and $\bm x= (x_1,\ldots,x_{N+1}) \in \bR^{N+1}_\n{reg}$. 
Writing $\bm \lambda' = (\lambda_1,\ldots,\lambda_N) \in \bC^N$, we have
\[ \psi_{\bm \lambda} = \sum_{m=1}^{N+1} \tilde s_{m \, N+1, \gamma} s_{m \, N+1} \left( \prod_{k=1}^N \theta_{k \, N+1} \right) \hat e^-(\lambda_{N+1})  \psi_{\bm \lambda'}. \]
\end{prop}

\begin{proof}
Applying \rfl{symmgrouprecursion} and \rfl{prewavefnregrep} we obtain
\begin{align*}
\psi_{\bm \lambda} &=\sum_{w \in S_{N+1}}  \chi_{w^{-1} \bR^{N+1}_+} \tilde w_\gamma^{-1} \tilde w \e^{\ii \bm \lambda} 
\; = \; \sum_{m=1}^{N+1} \sum_{w' \in S_N} \! \chi_{s_{m \, N+1} (w')^{-1} \bR^{N+1}_+} \tilde s_{m \, N+1, \gamma} (\tilde w'_\gamma)^{-1} \tilde w' \tilde s_{m \, N+1} \e^{\ii \bm \lambda} \displaybreak[2] \\
&= \sum_{m=1}^{N+1} \tilde s_{m \, N+1, \gamma} \sum_{w' \in S_N} \! \chi_{s_{m \, N+1} (w')^{-1} \bR^{N+1}_+}  (\tilde w'_\gamma)^{-1} \tilde w' s_{m \, N+1} \e^{\ii \bm \lambda} \displaybreak[2] \\
&= \sum_{m=1}^{N+1} \tilde s_{m \, N+1, \gamma} s_{m \, N+1}  \sum_{w' \in S_N} \! \chi_{(w')^{-1} \bR^{N+1}_+}  (\tilde w'_\gamma)^{-1} \tilde w' \e^{\ii \bm \lambda}.
\end{align*}
Write $\bm x' = (x_1,\ldots,x_N) \in \bR^N_\n{reg}$.
Now, $\chi_{(w')^{-1} \bR_+^{N+1}}(\bm x)=1$ precisely if $x_{w'(1)}>\ldots>x_{w'(N+1)}$, but $w'(N+1)=N+1$, so the condition $\chi_{(w')^{-1} \bR_+^{N+1}}(\bm x)=1$ is equivalent to $\chi_{(w')^{-1} \bR_+^N}(\bm x')=1$ and $x_{w'(k)}>x_{N+1}$ for all $k=1,\ldots,N$, i.e. $x_k>x_{N+1}$ for all $k=1,\ldots,N$.
This yields
\[ \chi_{(w')^{-1} \bR_+^{N+1}}(\bm x) = \left( \prod_{k=1}^N \theta(x_k-x_{N+1}) \right) \chi_{(w')^{-1} \bR_+^N}(\bm x') \]
and hence
\[ \psi_{\bm \lambda}(\bm x) = \sum_{m=1}^{N+1} \tilde s_{m \, N+1, \gamma} s_{m \, N+1} \left( \prod_{k=1}^N \theta(x_k-x_{N+1}) \right) \sum_{{w'} \in S_N} \chi_{{w'} \bR^N_+}(\bm x') (\tilde w'_\gamma)^{-1} \tilde w' \e^{\ii \inner{\bm \lambda}{\bm x}}. \]
Finally, using $\e^{\ii \inner{\bm \lambda}{\bm x}} = \e^{\ii \lambda_{N+1} x_{N+1}} \e^{\ii \inner{\bm \lambda'}{\bm x'}}$, we have
\begin{align*}
\psi_{\bm \lambda}(\bm x) &= \sum_{m=1}^{N+1} \tilde s_{m \, N+1, \gamma} s_{m \, N+1} \left( \prod_{k=1}^N \theta(x_k-x_{N+1}) \right) \e^{\ii  \lambda_{N+1} x_{N+1}} \cdot \\
& \hspace{50mm} \sum_{{w'} \in S_N}  \chi_{(w')^{-1} \bR^N_+}(\bm x') (\tilde w'_\gamma)^{-1} \tilde w' \e^{\ii \inner{\bm \lambda'}{\bm x}},
\end{align*}
where we recognize the summation over $w'$ as $\psi_{\bm \lambda'}(\bm x')$, as per \rfl{prewavefnregrep}.
\end{proof}

\begin{prop}[Recursion for the Bethe wavefunction] \label{Psiinduction}
Let $\bm \lambda$ = $(\lambda_1,\ldots,\lambda_{N+1})$ $\in$ $\bC^{N+1}$ and $\bm x = (x_1,\ldots,x_{N+1}) \in \bR^{N+1}_+$. Write $\bm \lambda'=(\lambda_1,\ldots,\lambda_N)$.
We have
\[ \Psi_{\bm \lambda} = \frac{1}{N+1} \sum_{m=1}^{N+1} \tilde s_{m \, N+1} \tau^+_{\lambda_N}(\bm \lambda') \hat e^-( \lambda_{N+1}) \Psi_{\bm \lambda'} = \frac{1}{N+1} \sum_{m=1}^{N+1} \tilde s_{m \, N+1, \gamma} \hat e^-( \lambda_{N+1}) \Psi_{\bm \lambda'}.\]
\end{prop}

\begin{proof}
Using \rfp{wavefnregrep}, one of the identities from \rfl{symmetrizerrecursion} and \rfp{dAHAregrep}, we obtain 
\begin{align*} \Psi_{\bm \lambda} &= \tilde{\ca S}^{(N+1)}_\gamma \e^{\ii \bm \lambda} \; = \; \frac{1}{N+1} \sum_{m=1}^{N+1} \left( \tilde s_{m \, N+1} \right)_\gamma \tilde{\ca S}^{(N)}_\gamma \hat e^-( \lambda_{N+1}) \e^{\ii \bm \lambda'} \displaybreak[2] \\
&= \frac{1}{N+1} \sum_{m=1}^{N+1} \tilde s_{m \, N+1, \gamma} \hat e^-( \lambda_{N+1}) \tilde{\ca S}^{(N)}_\gamma \e^{\ii \bm \lambda'} 
\; = \; \frac{1}{N+1} \sum_{m=1}^{N+1} \tilde s_{m \, N+1, \gamma} \hat e^-( \lambda_{N+1}) \Psi_{\bm \lambda'}. \end{align*}
establishing the first expression for $\Psi_{\bm \lambda}$. The second expression is equivalent to the first by virtue of \rfl{regreplem2} and the identity $s_m \ldots s_N \ca S^{(N)} = s_{m \, N+1} \ca S^{(N)}$ in $S_{N+1}$.
\end{proof}

\section{The Yang-Baxter algebra and the regular representation of the dAHA} 
\sectionmark{The Yang-Baxter algebra and the regular repr. of the dAHA}
\label{YBalgebradAHAsec}

In this section we note that the regular representation of the dAHA occurs in a natural way in the Yang-Baxter commutation relations as presented in \rfc{Tcommrelcor}.
This does not appear to have been documented in the literature on the subject. However, it is another indicator of the close relationship between the dAHA and the QISM.

\begin{lem} \label{YBalgebradAHA}
Let $\lambda_j, \lambda_{j+1} \in \bC$.
Then
\[ [\ca T^{k_1 \, l_1}_{\lambda_j},\ca T^{k_2 \, l_2}_{\lambda_{j+1}}] = \frac{-\ii \gamma}{\lambda_j-\lambda_{j+1}} \left( \ca T^{k_2 \, l_1}_{\lambda_j}\ca T^{k_1 \, l_2}_{\lambda_{j+1}} - \ca T^{k_2 \, l_1}_{\lambda_{j+1}} \ca T^{k_1 \, l_2}_{\lambda_j}\right) = -\ii \gamma \left( \tilde \Delta_j \ca T^{k_2 \, l_1}_{\lambda_j} \ca T^{k_1 \, l_2}_{\lambda_{j+1}} \right). \]
In particular, \rfeqn{ABcommrel}, \rfeqn{DBcommrel} and \rfeqn{CBcommrel} can be written as
\begin{align} 
A_{\lambda_{j+1}} B_{\lambda_j} &= \left( \tilde s_{j,-\gamma} B_{\lambda_{j+1}} A_{\lambda_j} \right)
\label{ABcommrel2}, \\
D_{\lambda_{j+1}} B_{\lambda_j} &= \left( \tilde s_{j,\gamma} B_{\lambda_{j+1}} D_{\lambda_j} \right) \label{DBcommrel2}, \\
C_{\lambda_{j+1}}  B_{\lambda_j} &= \left( \ii \tilde \Delta_j D_{\lambda_{j+1}}A_{\lambda_j} \right) \label{CBcommrel2}.
\end{align}
\end{lem}

\begin{proof}
This follows immediately from \rfeqn{Tcommrel}. We have used $\tilde s_{j,\gamma} = \tilde s_j - \ii \gamma \tilde \Delta_j$ and $\tilde s_j \tilde s_{j,\gamma} = \tilde s_{j,-\gamma} \tilde s_j$ in the derivation of \rfeqnser{ABcommrel2}{DBcommrel2}.
\end{proof}

\boxedenv{
\begin{notn}[Delimiting the action of divided difference operators and deformed permutations] \label{notn7}
The parentheses in $\left( \tilde \Delta_j \ca T^{k_2 \, l_1}_{\lambda_j} \ca T^{k_1 \, l_2}_{\lambda_{j+1}} \right)$ delimit the action of the divided difference operator $\tilde \Delta_j$: it only acts on the $\lambda_j$ and $\lambda_{j+1}$ appearing inside the parentheses. The same holds for expressions such as $\left( \tilde s_{j,\gamma} B_{\lambda_{j+1}} D_{\lambda_j} \right)$ involving the deformed permutations.
We emphasize that typically the integral operators $\ca T^{k \, l}_\mu$ act on functions that do not depend on $\mu$, in which case the parentheses are unnecessary.
\end{notn}}

From \rfeqnser{ABcommrel2}{DBcommrel2} we can derive various expressions for the action of $A_\mu $ and $D_\mu $ on the wavefunction $\Psi_{\bm \lambda}$:

\begin{prop}
Let $\gamma \in \bR$, $L \in \bR_{>0}$ and $(\bm \lambda,\lambda_{N+1}) \in \bC^{N+1}$.
Denote by $w_0$ the longest element in $S_N$ (sending $j$ to $N+1-j$, for $j=1,\ldots,N$). %
\nc{rwl}{$w_0$}{Longest element in $S_N$}%
Then
\begin{align*} 
A_{\lambda_{N+1}} \Psi_{\bm \lambda} &= \tilde s_{N,-\gamma} \ldots \tilde s_{1,-\gamma} \tilde s_1 \ldots \tilde s_N \e^{-\ii \lambda_{N+1} L/2} \Psi_{\bm \lambda} \\
&= \left(1-\ii \gamma \tilde \Delta_{N \, N+1} \right) \ldots \left(1-\ii \gamma \tilde \Delta_{1 \, N+1} \right) \e^{-\ii \lambda_{N+1} L/2} \Psi_{\bm \lambda} \\
&= \tilde w_0 \left( 1-\ii \gamma \tilde \Delta_{1 \, N+1} \right) \ldots \left( 1-\ii \gamma \tilde \Delta_{N \, N+1} \right) \e^{-\ii \lambda_{N+1} L/2} \Psi_{\bm \lambda}, \displaybreak[2] \\
D_{\lambda_{N+1}} \Psi_{\bm \lambda} &= \tilde s_{N,\gamma} \ldots \tilde s_{1,\gamma} \tilde s_1 \ldots \tilde s_N  \e^{\ii \lambda_{N+1} L/2} \Psi_{\bm \lambda} \\
&= \left(1+ \ii \gamma \tilde \Delta_{N \, N+1} \right) \ldots \left(1+ \ii \gamma \tilde \Delta_{1 \, N+1} \right) \e^{\ii \lambda_{N+1} L/2} \Psi_{\bm \lambda} \\
&= \tilde w_0 \left( 1+\ii \gamma \tilde \Delta_{1 \, N+1} \right) \ldots \left( 1+\ii \gamma \tilde \Delta_{N \, N+1} \right) \e^{\ii \lambda_{N+1} L/2} \Psi_{\bm \lambda}.
\end{align*}
\end{prop}

\begin{proof}
We will derive the expressions for $D_{\lambda_{N+1}} \Psi_{\bm \lambda}$ from \rfeqn{DBcommrel2}.
We have
\begin{align*} 
D_{\lambda_{N+1}} \Psi_{\bm \lambda} &= D_{\lambda_{N+1}}B_{\lambda_N} \ldots B_{\lambda_1} \Vac \\
&= \tilde s_{N,\gamma} B_{\lambda_{N+1}} D_{\lambda_N} B_{\lambda_{N-1}} \ldots B_{\lambda_1} \Vac \; = \; \ldots \; = \tilde s_{N,\gamma} \ldots \tilde s_{1,\gamma} B_{\lambda_{N+1}} \ldots B_{\lambda_2} D_{\lambda_1} \Vac \\
&= \tilde s_{N,\gamma} \ldots \tilde s_{1,\gamma} \tilde s_1 \ldots \tilde s_N B_{\lambda_N} \ldots B_{\lambda_1} D_{\lambda_{N+1}} \Vac.
\end{align*}
We have moved the $D$ operator to the right as far as possible; now we use $D_\lambda \Vac = \e^{\ii \lambda L/2}\Vac$ and we obtain
\[ D_{\lambda_{N+1}} \Psi_{\bm \lambda} 
= \tilde s_{N,\gamma} \ldots \tilde s_{1,\gamma} \tilde s_1 \ldots \tilde s_N \e^{\ii \lambda_{N+1} L/2} B_{\lambda_N}  \ldots B_{\lambda_1}  \Vac , \]
leading to the first expression for $D_{\lambda_{N+1}} \Psi_{\bm \lambda}$.
\rfl{regrep20} provides us with the second expression.
From $\tilde \Delta_{j \, N+1} = \tilde w_0 \tilde \Delta_{N+1-j \, N+1} \tilde w_0$ and $w_0 \in S_N$ we obtain the third expression.
The expressions for $A_{\lambda_{N+1}} \Psi_{\bm \lambda}$ are found in the same way, starting from \rfeqn{ABcommrel2}.
\end{proof}

We can now easily recover more properties of $A_{\lambda_{N+1}}\Psi_{\bm \lambda}$ and $D_{\lambda_{N+1}} \Psi_{\bm \lambda}$.
\begin{itemize}
\item Since $\Psi_{\bm \lambda}$ is $S_N$-invariant in position space, so are $A_{\lambda_{N+1}}\Psi_{\bm \lambda}$ and $D_{\lambda_{N+1}} \Psi_{\bm \lambda}$ (any permutation $w$ commutes with the operators $\tilde s_{j,\gamma}$ acting in momentum space and with the factor $\e^{\ii \lambda_{N+1}L}$). Because $\set{\Psi_{\bm \lambda}}{\bm \lambda \n{ satisfies the BAEs \rf{BAE}}}$ form a complete set in $\ca H_N$, we obtain that $A_\mu $ and $D_\mu $ are operators on $\ca H_N$.
\item $A_{\lambda_{N+1}} \Psi_{\bm \lambda}$ and $D_{\lambda_{N+1}} \Psi_{\bm \lambda}$ can be obtained from each other by $\gamma \to -\gamma, L \to -L$.
\end{itemize}

In fact, using this formalism we can provide an alternative proof for \rfp{APsiDPsi}.
\begin{prop}
Let $\gamma \ne 0$, $L \in \bR_{>0}$ and $(\bm \lambda,\mu) \in \bC^{N+1}$.
Then
\begin{align*}
A_\mu  \Psi_{\bm \lambda} &= \tau^+_\mu(\bm \lambda) \e^{-\ii \mu L/2} \Psi_{\bm \lambda} +\sum_{j=1}^N \tau^+_{\lambda_j}(\bm \lambda_{\hat \jmath}) \frac{\ii \gamma}{\lambda_j-\mu} \e^{-\ii \lambda_j L/2} \Psi_{\bm \lambda_{\hat \jmath}, \mu}, \\
D_\mu  \Psi_{\bm \lambda} &= \tau^-_\mu(\bm \lambda) \e^{\ii \mu L/2} \Psi_{\bm \lambda} - \sum_{j=1}^N \tau^-_{\lambda_j}(\bm \lambda_{\hat \jmath}) \frac{\ii \gamma}{\lambda_j-\mu} \e^{\ii \lambda_j L/2} \Psi_{\bm \lambda_{\hat \jmath},\mu}. 
\end{align*}
\end{prop}

\begin{proof}
Write $\lambda_{N+1} = \mu$. It is sufficient to prove 
\begin{align}
\lefteqn{\tilde s_{N,\gamma} \ldots \tilde s_{1,\gamma} \tilde s_1 \ldots \tilde s_N \e^{\ii \lambda_{N+1} L/2} \tilde{\ca S}^{(N)} =} \nonumber \displaybreak[2] \\
&= \left( \e^{\ii \lambda_{N+1} L/2} \tau^-_{\lambda_{N+1}}(\bm \lambda)
- \sum_{j=1}^N \e^{\ii \lambda_j L/2} \tau^-_{\lambda_j}(\bm \lambda_{\hat \jmath}) \frac{\ii \gamma}{\lambda_j-\lambda_{N+1}} \tilde s_j \ldots \tilde s_N  \right) \tilde{\ca S}^{(N)}; \label{eqn456}
\end{align}
this would concern the formula for $D_\mu \Psi_{\bm \lambda}$; to obtain the equivalent statement for $A_\mu  \Psi_{\bm \lambda}$ apply $\gamma \to -\gamma$, $L \to -L$.
\rfeqn{eqn456} can be proven by induction; the case $N=0$ is trivial. To establish the induction step, we first observe that 
\begin{align*} 
\lefteqn{\tilde s_{N,\gamma} \frac{\ii \gamma}{\lambda_j-\lambda_N} \tilde s_j \ldots \tilde s_N \tilde{\ca S}^{(N)} =}\\
&= \left( \frac{\lambda_N-\lambda_{N+1}+\ii \gamma}{\lambda_N-\lambda_{N+1}} \tilde s_N - \frac{\ii \gamma}{\lambda_N-\lambda_{N+1}}  \right)\frac{\ii \gamma}{\lambda_j-\lambda_N}  
\tilde s_j \ldots \tilde s_N \tilde{\ca S}^{(N)} \\
&=  \frac{\lambda_N-\lambda_{N+1}+\ii \gamma}{\lambda_N-\lambda_{N+1}}  \frac{\ii \gamma}{\lambda_j-\lambda_{N+1}} \tilde s_j \ldots \tilde s_{N-2} \tilde s_N \tilde s_{N-1} \tilde s_N \tilde{\ca S}^{(N)} - \frac{\ii \gamma}{\lambda_N-\lambda_{N+1}} \frac{\ii \gamma}{\lambda_j-\lambda_N} \tilde s_j \ldots \tilde s_N \tilde{\ca S}^{(N)} \\
&=   \frac{\lambda_N-\lambda_{N+1}+\ii \gamma}{\lambda_N-\lambda_{N+1}}  \frac{\ii \gamma}{\lambda_j-\lambda_{N+1}} \tilde s_j \ldots \tilde s_N \tilde s_{N-1} \tilde{\ca S}^{(N)} - \frac{\ii \gamma}{\lambda_N-\lambda_{N+1}} \frac{\ii \gamma}{\lambda_j-\lambda_N}  \tilde s_j \ldots \tilde s_N \tilde{\ca S}^{(N)}\\
&= \left(  \frac{\lambda_N-\lambda_{N+1}+\ii \gamma}{\lambda_N-\lambda_{N+1}}  \frac{\ii \gamma}{\lambda_j-\lambda_{N+1}} - \frac{\ii \gamma}{\lambda_N-\lambda_{N+1}} \frac{\ii \gamma}{\lambda_j-\lambda_N} \right) \tilde s_j \ldots \tilde s_N \tilde{\ca S}^{(N)}\\
&= \frac{\lambda_j-\lambda_N-\ii \gamma}{\lambda_j-\lambda_N} \frac{\ii \gamma}{\lambda_j-\lambda_{N+1}} \tilde s_j \ldots \tilde s_N \tilde{\ca S}^{(N)}.
\end{align*}
Now write $\bm \lambda' = (\lambda_1,\ldots,\lambda_{N-1}) \in \bC^{N-1}$ and note that $\tilde{\ca S}^{(N-1)} \tilde{\ca S}^{(N)} = \tilde{\ca S}^{(N)}$ so that, by virtue of the induction hypothesis, we have 
\begin{align*}
\lefteqn{\tilde s_{N,\gamma} \ldots \tilde s_{1,\gamma}\tilde s_1 \ldots \tilde s_N \e^{\ii \lambda_{N+1} L} \tilde{\ca S}^{(N)} \; = \; \tilde s_{N,\gamma} \ldots \tilde s_{1,\gamma}\tilde s_1 \ldots \tilde s_{N-1} \e^{\ii \lambda_N L}\tilde s_N  \tilde{\ca S}^{(N)} \; =}\\
&= \tilde s_{N,\gamma}  \left( \e^{\ii \lambda_{N} L/2} \tau^-_{\lambda_N}(\bm \lambda')
- \sum_{j=1}^{N-1} \e^{\ii \lambda_j L/2} \tau^-_{\lambda_j}(\bm \lambda'_{\hat \jmath}) \frac{\ii \gamma}{\lambda_j-\lambda_{N}} \tilde s_j \ldots \tilde s_{N-1} \right) \tilde{\ca S}^{(N-1)} \tilde s_N \tilde{\ca S}^{(N)} \\
&= \left( \tilde s_{N,\gamma}  \e^{\ii \lambda_{N} L/2} \tau^-_{\lambda_N}(\bm \lambda')\tilde s_N
-  \sum_{j=1}^{N-1} \e^{\ii \lambda_j L/2} \tau^-_{\lambda_j}(\bm \lambda'_{\hat \jmath})\tilde s_{N,\gamma} \frac{\ii \gamma}{\lambda_j-\lambda_{N}} \tilde s_j \ldots \tilde s_N \right) \tilde{\ca S}^{(N)} \displaybreak[2] \\
&= \left(  \e^{\ii \lambda_{N+1} L/2} \tau^-_{\lambda_{N+1}}(\bm \lambda) - \e^{\ii \lambda_{N} L/2} \tau^-_{\lambda_N}(\bm \lambda')\frac{\ii \gamma}{\lambda_N-\lambda_{N+1}} \tilde s_N + \right. \\
& \left. \qquad -  \sum_{j=1}^{N-1} \e^{\ii \lambda_j L/2} \tau^-_{\lambda_j}(\bm \lambda'_{\hat \jmath})\frac{\lambda_j-\lambda_N-\ii \gamma}{\lambda_j-\lambda_N} \frac{\ii \gamma}{\lambda_j-\lambda_{N+1}} \tilde s_j \ldots \tilde s_N \right) \tilde{\ca S}^{(N)} \displaybreak[2] \\
&= \left( \e^{\ii \lambda_{N+1} L/2} \tau^-_{\lambda_{N+1}}(\bm \lambda) -  \sum_{j=1}^N \e^{\ii \lambda_j L/2} \tau^-_{\lambda_j}(\bm \lambda_{\hat \jmath}) \frac{\ii \gamma}{\lambda_j-\lambda_{N+1}}  \tilde s_j \ldots \tilde s_N \right) \tilde{\ca S}^{(N)},
\end{align*}
which establishes the induction step. 
\end{proof}

\begin{rem}
As the statements made in this section are a direct consequence of the QYBE \rf{ybeRT}, it follows that the dAHA plays a role, through its regular representation, in the theory of any physical system with an $\ca R$-matrix given by \rfd{Rmatrixdefn}.
\end{rem}

\section{The $Q$-operator}

For the QNLS model we can follow a method from statistical mechanics, initially formulated by Baxter for the eight-vertex model \cite{Baxter1972,Baxter1989}, which has been deployed for other solvable lattice models and spin chains (see, e.g. \cite{KricheverLWZ, BazhanovLZ,KuznetsovS1998,Sklyanin2000,PasquierGaudin,Korff2006}). In this method one constructs an operator $Q_\lambda$, known as ``Baxter's $Q$-operator'', which is in involution with the transfer matrix $T_\mu$ and typically used to find the spectrum of $T_\mu$ as a substitute for the Bethe ansatz. 
Its application to the QNLS model is particularly interesting since it affords a connection to the Dunkl-type operators from Chapter \ref{chdAHA}.

\begin{thm}
Let $\gamma \in \bR$ and $L \in \bR_{>0}$.
There exists a family of operators $\set{Q_\mu}{\mu \in \bC}$, densely defined on $\ca H_N = \ca H_N([-L/2,L/2])$, satisfying the following conditions.
\begin{enumerate}
\item For all $\lambda,\mu \in \bC$, 
\begin{equation} [T_\lambda,Q_\mu]=[Q_\lambda,Q_\mu]=0. \label{TQcomm} \end{equation}
\item The \emph{$TQ$-equation} holds for all $\mu \in \bC$:
\begin{equation} \label{TQeqn}
T_\mu Q_\mu  =  \e^{-\ii \mu L/2}Q_{\mu+\ii \gamma} + \e^{\ii \mu L/2} Q_{\mu-\ii \gamma}.
\end{equation}
\item \label{QonPsi} If $\bm \lambda \in \bC^N_\n{reg}$ satisfies the BAEs \rf{BAE} then $Q_{\lambda_j}\Psi_{\bm \lambda} = 0$ for all $j=1,\ldots,N$.
\end{enumerate}
\end{thm}

\begin{proof}
First we assume that $(\bm \lambda,\mu) \in \bC^{N+1}_\n{reg}$ and $\bm \lambda$ satisfies the BAEs \rf{BAE}.
Then from \rft{ABA} we know that the Bethe wavefunction $\Psi_{\bm \lambda}$ is an eigenfunction of the transfer matrix $T_\mu$:
\[ T_\mu \Psi_{\bm \lambda} = \left( \e^{-\ii \mu L/2} \tau^+_\mu(\bm \lambda) + \e^{\ii \mu L/2} \tau^-_\mu(\bm \lambda) \right) \Psi_{\bm \lambda}. \]
Multiplying by $ \prod_{j=1}^N (\lambda_j-\mu) $ we arrive at
\begin{equation} \label{TQ1} 
T_\mu \left( \prod_{j=1}^N (\lambda_j-\mu) \right) \Psi_{\bm \lambda} = \left( \e^{-\ii \mu L/2} \prod_{j=1}^N (\lambda_j-\mu-\ii \gamma) + \e^{\ii \mu L/2} \prod_{j=1}^N (\lambda_j-\mu+\ii \gamma) \right) \Psi_{\bm \lambda}. 
\end{equation}
Using the completeness of $\set{\Psi_{\bm \lambda}}{\bm \lambda \n{ satisfies the BAEs}}$ in $\ca H_N$, we may define the operator $Q_\mu$ on a dense subset of $\ca H_N$ by specifying its action on Bethe wavefunctions as follows,
\begin{equation} \label{Qdef} Q_\mu \Psi_{\bm \lambda} := \left( \prod_{j=1}^N (\lambda_j-\mu) \right) \Psi_{\bm \lambda}, \end{equation}%
\nc{rqc}{$Q_\lambda$}{QNLS $Q$-operator \nomrefeqpage}%
and extending this linearly. Property \ref{QonPsi} follows immediately.
\rfeqn{TQ1} now reads
\[ T_\mu Q_\mu \Psi_{\bm \lambda} = \left( \e^{-\ii \mu L/2} Q_{\mu+\ii \gamma} + \e^{\ii \mu L/2} Q_{\mu-\ii \gamma} \right) \Psi_{\bm \lambda}. \]
\rfeqnser{TQcomm}{TQeqn} follow by applying them on $\Psi_{\bm \lambda}$, where $\bm \lambda$ satisfies the BAEs \rf{BAE}, again by using completeness.
\end{proof}

We recall the Dunkl-type operators $\partial_{j,\gamma}$, $j=1,\ldots,N$, which act as $\partial_j$ in the fundamental alcove ${\bR}^N_+$.
With this in mind, we now present the main result of this section.
\begin{thm} \label{QopDunklops}
Let $\bm \lambda \in \bC^N_\n{reg}$ be a solution of the BAEs \rf{BAE}. Also, let $\mu \in \bC$.
Then
\[ Q_\mu \Psi_{\bm \lambda} = \left( \prod_{j=1}^N (-\ii \partial_{j,\gamma}-\mu) \right) \Psi_{\bm \lambda} \]
and we have the identity
\[ Q_\mu = \prod_{j=1}^N (-\ii \partial_{j,\gamma}-\mu) = (-1)^N \sum_{n=0}^N \ii^n e_n(-\ii \partial_{1,\gamma},\ldots,-\ii \partial_{N,\gamma}) \mu^{N-n} \]
on a dense subspace of $\ca H_N$, where the \emph{elementary symmetric polynomial} $e_n$ is given by
\begin{equation} e_n(\bm \lambda) = \sum_{\bm i \in \f I^N_n} \prod_{m=1}^n \lambda_{i_m} = \sum_{1 \leq i_1 < \ldots < i_n \leq N} \lambda_{i_1} \ldots \lambda_{i_n} \end{equation}%
\nc{rel}{$e_n(\bm \lambda)$}{$n$-th elementary symmetric polynomial \nomrefeqpage}%
\end{thm}

\begin{proof}
Since $\prod_{j=1}^N (-\ii \partial_{j,\gamma}-\mu)$ is a symmetric polynomial in $\partial_{1,\gamma},\ldots,\partial_{N,\gamma}$, we may use the first part of \rft{wavefnsymmpolys}. 
Together with \rfeqn{Qdef} this implies that
\[ \left( \prod_{j=1}^N (-\ii \partial_{j,\gamma}-\mu) \right) \Psi_{\bm \lambda} = Q_\mu \Psi_{\bm \lambda}. \]
As $\set{\Psi_{\bm \lambda}}{\bm \lambda \n{ satisfies the BAEs \rf{BAE}}}$ is complete in $\ca H_N$ and both $Q_\mu$ and  symmetric expressions in the $\partial_{j,\gamma}$ are defined on a dense subspace, the second result follows.
\end{proof}

\begin{rem}
Referring back to \rfeqn{higherintegralsofmotion} and \rfeqn{eqn351}, we have seen the power sum symmetric polynomials in the Dunkl-type operators appearing as expansion coefficients with respect to $\mu$ in the transfer matrix $T_\mu$. The expansion of Baxter's $Q$-operator yields the elementary symmetric polynomials in the Dunkl-type operators. We remark that both sets of polynomials generate the algebra of symmetric polynomials in $N$ indeterminates. From a physical perspective, $T$ and $Q$ are generating functions for alternative sets of integrals of motion. $T$ is the natural object in the QISM and $Q$ is the natural object in the dAHA method. The $TQ$-equation \rf{TQeqn} connects these two approaches.
\end{rem}

\clearpage

\chapter{The non-symmetric Yang-Baxter algebra} \label{ch5}

In this section we set out the main body of original work in this thesis. 
We will describe operators $a_\mu$, $b^\pm_\mu$, $c^\pm_\mu $, $d_{\mu}$ that are extensions of the operators $A_\mu $, $B_\mu $, $C_\mu $, $D_\mu $ $\in \End(\ca H([-L/2,L/2]))$ from Chapter \ref{chQISM} to a dense subspace of $\f h = \f h([-L/2,L/2])$. These operators generate a subalgebra of $\End(\f h)$ which we will call the \emph{non-symmetric Yang-Baxter algebra}. 
For example, we will discuss operators $b^\pm_{\mu} \in \Hom(\f h_N,\f h_{N+1})$ that generate the pre-wavefunctions recursively, i.e.
\begin{equation} \label{prewavefnrecursion} 
\psi_{\lambda_1,\ldots,\lambda_{N+1}} = b^-_{\lambda_{N+1}} \psi_{\lambda_1,\ldots,\lambda_N} = b^+_{\lambda_1} \psi_{\lambda_2,\ldots,\lambda_{N+1}}, 
\end{equation}
in correspondence with \rfeqn{Bethewavefnrecursion}:
\[ \Psi_{\lambda_1,\ldots,\lambda_{N+1}} = B_{\lambda_{N+1}} \Psi_{\lambda_1,\ldots,\lambda_N}; \]
note that since $\psi_{\lambda_1,\ldots,\lambda_{N+1}}$ is not $S_{N+1}$-invariant (in position or momentum space), it is natural that there should be two ways of recursively expressing it as in \rfeqn{prewavefnrecursion}.
We will prove this recursive property based on \rfc{prewavefnDunkl} and identities $\partial_{j,\gamma} b^-_{\lambda_{N+1}} \ldots b^-_{\lambda_1}$ $=$ $\lambda_j b^-_{\lambda_{N+1}} \ldots b^-_{\lambda_1}$ and $\partial_{j,\gamma} b^+_{\lambda_1} \ldots b^+_{\lambda_{N+1}}$ $=$ $\lambda_j b^+_{\lambda_1} \ldots b^+_{\lambda_{N+1}}$ for all $j=1,\ldots,N+1$.
Using the definition of $\Psi_{\bm \lambda}$ in terms of the $\psi_{\bm \lambda}$, which are in turn defined in terms of representations of the dAHA, and the above recursion \rfeqn{prewavefnrecursion}, we are able to give a new proof (\rfc{Brecursion}) for the QISM recursion $\Psi_{\bm \lambda} = \prod_{j=1}^N B_{\lambda_j} \Vac$.\\

Also, in Section \ref{abcdcommrelssec} we will obtain commutation relations for the non-symmetric Yang-Baxter algebra on a subspace of $\f h$ by expressing the operators $a_\mu$, $d_{\mu}$, and $c^\pm_\mu $ in terms of the $b^\pm_\mu$ and using \rfeqn{prewavefnrecursion}.

\section{The non-symmetric integral operators}

\subsection{Definitions and basic properties}

For $N \in \bZ_{\geq 0}$, $n = 0, \ldots, N$ we recall
\[ \f i^n_N = \set{(i_1,\ldots,i_n) \in \{1,\ldots,N\}^n}{i_l \ne i_m \n{ for } l \ne m}. \]
We also recall that given $\bm i \in \f i^n_{N}$ and $\bm x \in J^N$, unless otherwise specified, $x_{i_0}$ denotes $L/2$, and $x_{i_{n+1}}$ denotes $-L/2$.\\

Let $N$ be a nonnegative integer and $f \in \ca F(J^N)$.
We can act on such functions by ``deletion'', ``insertion'', and ``replacement'' of variables, as follows: 
\begin{description}
\item[Deletion.] Let $\bm x \in J^{N+1}$ and $j=1,\ldots,N+1$. Then 
\begin{equation} (\hat \phi_j f)(\bm x) = f(\bm x_{\hat \jmath}) = f(x_1,\ldots,\hat x_j,\ldots,x_{N+1}). \end{equation}%
\nc{gvl}{$\hat \phi_j$}{Deletion operator \nomrefeqpage}%
Note that $\hat \phi_j \in \Hom(\ca F(J^N),\ca F(J^{N+1}))$.
\item[Insertion.] Let $y \in J$. We consider $\check \phi^\pm(y) \in \Hom(\ca F(J^{N+1}),\ca F(J^N))$ given by
\begin{equation} (\check \phi^+(y) f)(\bm x) = f(y,x_1,\ldots,x_N), \qquad (\check \phi^-(y) f)(\bm x) = f(x_1,\ldots,x_N,y). \end{equation}%
\nc{gvl}{$\check \phi^\pm(y)$}{Insertion operator \nomrefeqpage}%
Also note that 
\begin{equation} \label{eqn42} [\check \phi^+(y),\check \phi^-(z)]=0. \end{equation}
The operators $\check \phi^\pm(y)$ are non-symmetric versions of the quantum field $\Phi(y)$ (up to a scalar factor).
\item[Replacement.] For $j=1,\ldots,N$ and $y \in J$, we have
\[ (\phi_{j \to y}f)(\bm x) = f(x_1, \ldots, x_{j-1}, y, x_{j+1},\ldots, x_N), \]
and, for $\bm i \in \f i^n_{N}$, $\bm y \in J^n$,
\begin{equation} \phi_{\bm i \to \bm y}= \prod_{m=1}^n \phi_{i_m \to y_m}. \end{equation}%
\nc{gvl}{$\phi_{j \to y}, \phi_{\bm i \to \bm y}$}{Replacement operator \nomrefeqpage}%
We note that $\phi_{\bm i \to \bm y} \in \End(\ca F(J^N))$; 
furthermore, for all $w \in S_N$ we have
\begin{equation} \label{eqn43} w \phi_{\bm i \to \bm y} = \phi_{w \bm i \to \bm y} w. \end{equation}
\end{description}

Given $\bm i = (i_1,\ldots,i_n) \in \f i^n_N$, write $\bm i_+ = (i_1+1,\ldots,i_n+1) \in \f i^n_{N+1}$.
\nc{rilb}{$\bm i_+$}{Shorthand for $(i_1+1,\ldots,i_n+1)$}
Recall the notation $\theta_{\bm i}$ for the multiplication operator in $\End(\ca F(\bR^N_\n{reg}))$ defined by $(\theta_{\bm i} f)(\bm x) = \theta(x_{i_1}>x_{i_2}>\ldots>x_{i_n})f(\bm x)$, for $f \in \ca F(\bR^N_\n{reg})$, $\bm x \in \bR^N_\n{reg}$. 

\begin{defn} \label{eopsdefn}
Let $L \in \bR_{>0}$, $\lambda \in \bC$ and $n = 0,\ldots, N$.
Given $\bm i \in \f i^n_N$, the associated \emph{elementary (non-symmetric) integral operators},
\begin{align*}
\hat e^\pm_{\lambda;\bm i} &= \hat e^\pm_{\lambda;i_1 \, \ldots \, i_n}: \f h_N \to \f h_{N+1} \displaybreak[2] \\
\bar e^\pm_{\lambda;\bm i} &= \bar e^\pm_{\lambda;i_1 \, \ldots \, i_n}: \f h_N \to \f h_N, \displaybreak[2] \\
\check e^\pm_{\lambda;\bm i} &= \check e^\pm_{\lambda;i_1 \, \ldots \, i_n}: \f h_{N+1} \to \f h_N, 
\end{align*}%
\nc{rel}{$\hat e^\pm_{\lambda;\bm i},\bar e^\pm_{\lambda;\bm i},\check e^\pm_{\lambda;\bm i}$}{Elementary non-symmetric integral operators}%
are densely defined on $\f d([-L/2,L/2])$ by means of:
\begin{align*}
\hat e^+_{\lambda;\bm i}|_{\f d_N} &=  \e^{\ii \lambda x_1} \theta_{\bm i_+ \, 1} \left( \prod_{m=1}^n  \int_{x_{i_{m+1}+1}}^{x_{i_m+1}} \dd y_m \e^{\ii \lambda (x_{i_m+1}-y_m)} \right) \phi_{\bm i^+ \to \bm y} \hat \phi_1, && \n{where } x_{i_{n+1}+1} = x_1, \\
\hat e^-_{\lambda; \bm i}|_{\f d_N} &= \e^{\ii \lambda x_{N+1}} \theta_{N+1 \, \bm i} \left( \prod_{m=1}^n \int_{x_{i_m}}^{x_{i_{m-1}}} \dd y_m \e^{\ii \lambda (x_{i_m} -y_m)} \right) \phi_{\bm i \to \bm y} \hat \phi_{N+1}, && \n{where } x_{i_0} = x_{N+1}, \\
\bar e^+_{\lambda;\bm i}|_{\f d_N} &=\e^{-\ii \lambda L/2} \theta_{\bm i} \left( \prod_{m=1}^n \int_{x_{i_{m+1}}}^{x_{i_m}} \dd y_m \e^{\ii \lambda (x_{i_m}-y_m)} \right) \phi_{\bm i \to \bm y}, \\
\bar e^-_{\lambda; \bm i}|_{\f d_N} &= \e^{\ii \lambda L/2} \theta_{\bm i} \left( \prod_{m=1}^n \int_{x_{i_m}}^{x_{i_{m-1}}} \dd y_m \e^{\ii \lambda (x_{i_m}-y_m)} \right) \phi_{\bm i \to \bm y}, \\
\check e^+_{\lambda;\bm i}|_{\f d_{N+1}} &= \theta_{\bm i} \left( \prod_{m=0}^n \int_{x_{i_{m+1}}}^{x_{i_m}} \dd y_m \e^{\ii \lambda(x_{i_m}-y_m)} \right) \phi_{\bm i \to \bm y} \check \phi^-(y_0) \\
\check e^-_{\lambda; \bm i}|_{\f d_{N+1}} &= \theta_{\bm i} \left( \prod_{m=1}^{n+1} \int_{x_{i_m}}^{x_{i_{m-1}}} \dd y_m \e^{\ii \lambda(x_{i_m}-y_m)} \right) \phi_{\bm i \to \bm y} \check \phi^+(y_{n+1}).
\end{align*}
For the latter four operators, the standard conventions $x_{i_{n+1}} = -L/2$, $x_{i_0} = L/2$ apply.
\end{defn}
Note that for $n=0$ the above definitions $\hat e^\pm_\lambda:=\hat e^\pm_{\lambda; \emptyset}$\nc{rel}{$\hat e^\pm_\lambda$}{Shorthand for $\hat e^\pm_{\lambda;\emptyset}$} simplify to
\begin{align*}
(\hat e^+_\lambda f)(x_1,\ldots,x_{N+1}) &= \e^{\ii \lambda x_1} f(x_2,\ldots,x_{N+1})\\
(\hat e^-_\lambda f)(x_1,\ldots,x_{N+1}) &= \e^{\ii \lambda x_{N+1}} f(x_1,\ldots,x_N).
\end{align*} 
We have already encountered the operator $\hat e^-_\lambda$ in Section \ref{dAHArecursions}.
Note that 
\[ \hat e^+_\mu \e^{\ii \bm \lambda} = \e^{\ii (\mu,\bm \lambda)}, \qquad \hat e^-_\mu \e^{\ii \bm \lambda} = \e^{\ii (\bm \lambda, \mu)}, \]
which can be seen as the $\gamma = 0$ case for the intended recursions \rfeqn{prewavefnrecursion}.

\begin{exm}
Let $\lambda \in \bC$. The statement
\[ \hat e^+_{\lambda;2}|_{\f d_2} = \e^{\ii \lambda x_1} \theta_{3 \, 1} \int_{x_1}^{x_3} \dd y \e^{\ii \lambda (x_3-y)} \phi_{1 \to y} \hat \phi_3 \]
means that for all $f \in \f d_2$, $(x_1,x_2,x_3) \in J^3$ the function $\hat e^+_{\lambda;2}f$ is defined by
\[ \left(\hat e^+_{\lambda;2}f\right)(x_1,x_2,x_3) = \e^{\ii \lambda x_1} \theta(x_3-x_1) \int_{x_1}^{x_3} \dd y \e^{\ii \lambda (x_3-y)} f(y,x_2). \]
\end{exm}

Useful properties of these maps are listed and proven in Appendix \ref{nonsymmintopscalcs}.

\begin{defn}[Non-symmetric integral operators] 
Let $L \in \bR_{>0}$, $\gamma \in \bR$ and $\lambda \in \bC$. We define the operators $a_\lambda ,b^\pm_\lambda ,c^\pm_\lambda ,d_{\lambda}$ on $\f h$ by specifying their action on $\f d$:
\begin{align}
a_\lambda  &=  \sum_{n \geq 0} \gamma^n a_{\lambda;n}, && \n{where} \hspace{-20mm} & 
a_{\lambda;n}|_{\f d_N}  &= \sum_{\bm i \in \f i^n_{N}} \bar e_{\lambda;\bm i}^+, \displaybreak[2] \\%
\nc{ral}{$a_\lambda,a_{\lambda;n}$}{Generator of non-symmetric Yang-Baxter algebra \nomrefeqpage}%
b^\pm_\lambda  &=  \sum_{n \geq 0} \gamma^n b^\pm_{\lambda;n}, && \n{where} \hspace{-20mm} & 
b^\pm_{\lambda;n}|_{\f d_N} &= \sum_{\bm i \in \f i^{n}_{N}} \hat e^\pm_{\lambda;\bm i}, \displaybreak[2] \\%
\nc{rbl}{$b^\pm_\lambda,b^\pm_{\lambda;n}$}{Generator of non-symmetric Yang-Baxter algebra; \\ non-symmetric particle creation operators \nomrefeqpage}%
c^\pm_\lambda  &=\sum_{n \geq 0} \gamma^n c^\pm_{\lambda;n}, && \n{where} \hspace{-20mm} & 
c^\pm_{\lambda;n}|_{\f d_{N+1}} &= \sum_{\bm i \in \f i^n_N } \check e^\pm_{\lambda;\bm i} \\%
\nc{rcl}{$c^\pm_\lambda,c^\pm_{\lambda;n}$}{Generator of non-symmetric Yang-Baxter algebra \nomrefeqpage}%
d_{\lambda} &= \sum_{n \geq 0} \gamma^n d_{\lambda;n}, && \n{where} \hspace{-20mm} & d_{\lambda;n}|_{\f d_N} &= \sum_{\bm i \in \f i^n_{N}} \bar e^-_{\lambda;\bm i}.
\end{align}%
\nc{rdla}{$d_\lambda,d_{\lambda;n}$}{Generator of non-symmetric Yang-Baxter algebra \nomrefeqpage}%
\end{defn}
\vspace{-7mm}

\begin{prop}[Properties of the non-symmetric integral operators] \label{abcdproperties}
Let $L \in \bR_{>0}$, $\gamma \in \bR$ and $\lambda \in \bC$. 
The operators $a_\lambda ,b^\pm_\lambda,c^\pm_\lambda,d_{\lambda}$ have the following properties.
\begin{enumerate}
\item \label{abcdnumber} We have
\[ a_\lambda (\f d_N), d_{\lambda}(\f d_N) \subset \f h_N, \qquad b^\pm_\lambda (\f d_N) \subset \f h_{N+1},
\qquad c^\pm_\lambda (\f d_N) \subset \f h_{N-1}, \qquad c^\pm_\lambda (\f h_0) = 0. \]
\item \label{abcdformaladjoint} $b^\pm_\lambda $ is the formal adjoint of $c^\mp_{\bar \lambda}$, and $a_\lambda $ is the formal adjoint of $d_{\bar \lambda}$.
\item \label{abcdvacuum} The actions of $a_\lambda $, $b^\pm_\lambda $, $c^\pm_\lambda $, $d_{\lambda}$ on $\Vac$ are as follows:
\[ a_\lambda \Vac = \e^{-\ii \lambda L/2} \Vac, \qquad \left(b^\pm_\lambda \Vac\right)(x) = \e^{\ii \lambda x}, \qquad
c^\pm_\lambda \Vac = 0, \qquad d_{\lambda}\Vac = \e^{\ii \lambda L/2}\Vac. \]
\item \label{abcdbounded} The operators $a_\lambda $, $b^\pm_\lambda $, $c^\pm_\lambda $, $d_{\lambda}$ are bounded on the dense subset $\f h_\n{fin}$. Hence, they can be seen as endomorphisms of $\f h$ and can be composed with other elements of $\End(\f h)$.
\item \label{abcdpermutation}
For $w \in S_N$, we have 
\[ [w, a_\lambda ]|_{\f h_N} =  [w, b^-_{\lambda}]|_{\f h_N} =  [w, c^+_{\lambda}]|_{\f h_{N+1}} = [w, d_{\lambda}]|_{\f h_N} = 0. \]
Furthermore, 
\[ [\ca S^{(N)}, a_\lambda ]|_{\f h_N} =  [\ca S^{(N)}, b^-_{\lambda}]|_{\f h_N} =  [\ca S^{(N)}, c^+_{\lambda}]|_{\f h_{N+1}} = [\ca S^{(N)}, d_{\lambda}]|_{\f h_N} = 0. \]
\end{enumerate}
\end{prop}

\begin{proof}
We prove each property separately.
\begin{enumerate}
\item This follows from the definition of the elementary integral operators.
\item Using \rfl{eopsadjointness} we see that
\[ \innerrnd{b^+_\lambda f}{g}_{N} 
= \sum_{n=0}^N \gamma^n \sum_{\bm i \in \f i^n_{N}} \innerrnd{\hat e^+_{\lambda;\bm i} f}{g}_{N} 
= \sum_{n=0}^N \gamma^n \sum_{\bm i \in \f i^n_{N}} \innerrnd{f}{\check e^-_{\bar \lambda;\bm i} g}_{N-1} \hspace{-4mm} 
=  \innerrnd{f}{c^-_{\bar \lambda} g}_{N-1},\]
where $f \in \f h_{N-1}$ and $g \in \f h_N$. The other adjointness relations are demonstrated in the same fashion.
\item This follows immediately by using the definition of these operators in terms of the elementary integral operators.
\item This follows from the triangle inequality for the operator norm and the fact that the $a_\lambda $, $b^\pm_\lambda $, $c^\pm_\lambda $, $d_{\lambda}$ are all finite linear combinations of the elementary integral operators.
\item This is obtained by summing the statements in \rfl{eopspermutation} over the appropriate multi-indices $\bm i$. \qedhere
\end{enumerate}
\end{proof}

In light of \rfp{abcdproperties}, Property \ref{abcdnumber}, the operators $b^\pm_\lambda $ will be referred to as the \emph{non-symmetric creation operators}.

\begin{exm}[$N=2$]
We will present here expressions for the action of the non-symmetric integral operators $a_\lambda,b^\pm_\lambda,c^\pm_\lambda,d_\lambda$ on suitable functions $f$.
The reader should compare this with \rfex{ABCDexamples2}.
For $f \in \f d_2$ and $(x_1,x_2) \in J^2$ we have
\begin{align*}
(a_\lambda f)(x_1,x_2) &= \e^{- \! \ii \lambda L/2} \left( f(x_1,x_2) + \gamma  \int_{-L/2}^{x_1} \dd y \e^{\ii \lambda(x_1\! - \! y)}f(y,x_2)  +\gamma  \int_{-L/2}^{x_2}  \dd y  \e^{\ii \lambda(x_2 \! -\! y)}  f(x_1,y) +\right. \\
& \qquad + \gamma^2  \theta(x_1>x_2)  \int_{x_2}^{x_1} \dd y_1 \int_{-L/2}^{x_2} \dd y_2  \e^{\ii \lambda(x_1+x_2-y_1-y_2)} f(y_1,y_2)  +\\
& \qquad + \left. \gamma^2 \theta(x_2>x_1) \int_{x_1}^{x_2} \dd y_1 \int_{-L/2}^{x_1} \dd y_2  \e^{\ii \lambda(x_1+x_2-y_1-y_2)} f(y_2,y_1)  \right), \displaybreak[2]  \\
(d_\lambda f)(x_1,x_2) &= \e^{\ii \lambda L/2} \left( f(x_1,x_2) + \gamma  \int_{x_1}^{L/2} \dd y \e^{\ii \lambda(x_1\! - \! y)}f(y,x_2)  +\gamma  \int_{x_2}^{L/2}  \dd y  \e^{\ii \lambda(x_2\! - \! y)}  f(x_1,y) + \right. \\
& \qquad + \gamma^2 \theta(x_1>x_2) \int_{x_1}^{L/2} \dd y_1 \int_{x_2}^{x_1} \dd y_2  \e^{\ii \lambda(x_1+x_2-y_1-y_2)} f(y_1,y_2) + \\
& \qquad + \left. \gamma^2 \theta(x_2>x_1) \int_{x_2}^{L/2} \dd y_1 \int_{x_1}^{x_2} \dd y_2  \e^{\ii \lambda(x_1+x_2-y_1-y_2)} f(y_2,y_1)  \right).
\end{align*}
For $f \in \f d_2$ and $(x_1,x_2,x_3) \in J^3$ we have
\begin{align*}
(b^-_\lambda f)(x_1,x_2,x_3) &= \e^{\ii \lambda x_3} \left( f(x_1,x_2)+  \gamma \theta(x_3 > x_1) \int_{x_1}^{x_3} \dd y \e^{\ii \lambda(x_1-y)} f(y,x_2) + \right. \\
& \qquad + \gamma \theta(x_3 > x_2) \int_{x_2}^{x_3} \dd y \e^{\ii \lambda(x_2-y)} f(x_1,y) + \\
& \qquad + \gamma^2 \theta(x_3 > x_1 > x_2) \int_{x_1}^{x_3} \dd y_1 \int_{x_2}^{x_1} \dd y_2 \e^{\ii \lambda (x_1+x_2-y_1-y_2)} f(y_1,y_2) + \\
& \qquad  \left. + \gamma^2 \theta(x_3 > x_2 > x_1) \int_{x_2}^{x_3} \dd y_1 \int_{x_1}^{x_2} \dd y_2 \e^{\ii \lambda (x_1+x_2-y_1-y_2)} f(y_2,y_1) \right), \displaybreak[2]\\ 
(b^+_\lambda f)(x_1,x_2,x_3)&= \e^{\ii \lambda x_1} \left( f(x_2,x_3)+ \gamma \theta(x_2 > x_1) \int_{x_1}^{x_2} \dd y \e^{\ii \lambda (x_2-y)} f(y,x_3) + \right. \\
& \qquad + \gamma \theta(x_3 > x_1) \int_{x_1}^{x_3} \dd y \e^{\ii \lambda (x_3-y)} f(x_2,y) + \\
& \qquad + \gamma^2 \theta(x_2> x_3 > x_1) \int_{x_3}^{x_2} \dd y_1 \int_{x_1}^{x_3} \dd y_2 \e^{\ii \lambda (x_2+x_3-y_1-y_2)} f(y_1,y_2) + \\
& \qquad \left. + \gamma^2  \theta(x_3> x_2>x_1) \int_{x_2}^{x_3} \dd y_1 \int_{x_1}^{x_2} \dd y_2 \e^{\ii \lambda (x_2+x_3-y_1-y_2)} f(y_2,y_1) \right).
\end{align*}
Finally, for $f \in \f d_3$ and $(x_1,x_2) \in J^2$ we have
\begin{align*}
(c^-_\lambda f)(x_1,x_2) &=  \int_{-L/2}^{L/2} \dd y \e^{-\ii \lambda y} f(y,x_1,x_2)+  \gamma  \int_{x_1}^{L/2} \dd y_1 \int_{-L/2}^{x_1} \dd y_2 \e^{\ii \lambda(x_1-y_1-y_2)} f(y_2,y_1,x_2) + \\
& \qquad +  \gamma \int_{x_2}^{L/2} \dd y_1 \int_{-L/2}^{x_2} \dd y_2  \e^{\ii \lambda(x_2-y_1-y_2)} f(y_2,x_1,y_1) + \\
& \qquad + \gamma^2 \theta(x_1 > x_2)\int_{x_1}^{L/2} \dd y_1 \int_{x_2}^{x_1} \dd y_2  \int_{-L/2}^{x_2} \dd y_3  \e^{\ii \lambda (x_1+x_2-y_1-y_2-y_3)} f(y_3,y_1,y_2) + \\
& \qquad + \gamma^2 \theta(x_2 > x_1)\int_{x_2}^{L/2} \dd y_1 \int_{x_1}^{x_2} \dd y_2  \int_{-L/2}^{x_1} \dd y_3  \e^{\ii \lambda (x_1+x_2-y_1-y_2-y_3)} f(y_3,y_2,y_1), \displaybreak[2]\\ 
(c^+_\lambda f)(x_1,x_2)&=  \int_{-L/2}^{L/2} \dd y \e^{-\ii \lambda y} f(x_1,x_2,y)+  \gamma  \int_{x_1}^{L/2} \dd y_0 \int_{-L/2}^{x_1} \dd y_1 \e^{\ii \lambda(x_1-y_0-y_1)} f(y_1,x_2,y_0) + \\
& \qquad +  \gamma \int_{x_2}^{L/2} \dd y_0 \int_{-L/2}^{x_2} \dd y_1 \e^{\ii \lambda(x_2-y_0-y_1)} f(x_1,y_1,y_0) + \\
& \qquad + \gamma^2 \theta(x_1 > x_2)\int_{x_1}^{L/2} \dd y_0 \int_{x_2}^{x_1} \dd y_1  \int_{-L/2}^{x_2} \dd y_2 \e^{\ii \lambda (x_1+x_2-y_0-y_1-y_2)} f(y_1,y_2,y_0) + \\
& \qquad + \gamma^2 \theta(x_2 > x_1)\int_{x_2}^{L/2} \dd y_0 \int_{x_1}^{x_2} \dd y_1  \int_{-L/2}^{x_1} \dd y_2  \e^{\ii \lambda (x_1+x_2-y_0-y_1-y_2)} f(y_2,y_1,y_0).
\end{align*}
\end{exm}

\subsection{Connections with the QNLS symmetric integral operators}

We now highlight why the operators $a_\lambda ,b^\pm_\lambda,c^\pm_\lambda,d_{\lambda}$ are relevant to the study of the QNLS model. First of all, we have
\begin{lem} \label{abcdrestr}
Let $J = [-L/2,L/2]$ with $L \in \bR_{>0}$, $\gamma \in \bR$, $\lambda \in \bC$ and $F \in \ca H_N$. Then
\begin{align*} 
a_\lambda F|_{J^N_+} &= A_\lambda F|_{J^N_+} & 
\ca S^{(N+1)}b^\pm_\lambda  F|_{J^{N+1}_+} &= B_\lambda F|_{J^{N+1}_+}\\
c^\pm_\lambda  F|_{J^{N-1}_+} &= C_\lambda F|_{J^{N-1}_+} & 
d_{\lambda}F|_{J^N_+} &= D_\lambda F|_{J^N_+}. 
\end{align*}
\end{lem}

\begin{proof}
By summing over all suitable $\bm i$ in \rfl{eopsrestr}; we remark that in the definition of $b^\pm$ there is an extra condition on $\bm i$ (involving the labels 1 and $N+1$, respectively), so that we need to symmetrize to get a sum of all $\bm i \in \f I^{n+1}_{N+1}$.
\end{proof}

\begin{prop}[Symmetric integral operators as restrictions of non-symmetric integral operators] \label{abcdrestrABCD}
Let $L \in \bR_{>0}$, $\gamma \in \bR$ and $\lambda \in \bC$. 
Then
\[ a_\lambda |_{\ca H_N}= A_\lambda, \quad \ca S^{(N+1)} b^\pm_\lambda |_{\ca H_N} = B_\lambda , \quad
c^\pm_\lambda |_{\ca H_N}= C_\lambda , \quad d_{\lambda}|_{\ca H_N} = D_\lambda . \]
\end{prop}

\begin{proof}
We make the following observations.
\begin{itemize}
\item From \rfp{abcdproperties} \ref{abcdpermutation} we may conclude that $a_\lambda $ and $d_{\lambda}$ map $S_N$-invariant functions to $S_N$-invariant functions. 
\item The two operators $\ca S^{(N+1)} b^\pm_\lambda $ evidently map to $S_{N+1}$-invariant functions and by virtue of \rfl{abcdrestr} their actions on an element of $\ca H_N$ coincide on the alcove and hence everywhere. 
\item $c^+_{\lambda}$ maps $S_{N+1}$-invariant functions to $S_N$-invariant functions due to \rfp{abcdproperties} \ref{abcdpermutation}; more precisely,
\[ c^+_{\lambda} \ca S^{(N+1)} f = c^+_{\lambda} \ca S^{(N)} \ca S^{(N+1)} f = \ca S^{(N)} c^+_{\lambda} \ca S^{(N+1)} f, \]
for an arbitrary $f \in \f h_{N+1}$.
Also, by summing over all suitable $\bm i$ in \rfl{eopssymmetric} we obtain that $c^+_{\lambda}|_{\ca H_N} = c^-_{\lambda}|_{\ca H_N}$, so that $c^-_{\lambda}$ has the same property. 
\end{itemize}
Since $S_N$-invariant functions are determined by their behaviour on the fundamental alcove $J^N_+$, application of \rfl{abcdrestr} completes the proof.
\end{proof}

An important consequence of \rfp{abcdrestrABCD} is the following
\begin{cor} \label{bsymmetrizer}
Let $L \in \bR_{>0}$, $\gamma \in \bR$ and $\lambda \in \bC$. 
\[ B_\lambda  \ca S^{(N)} = \ca S^{(N+1)} b^-_{\lambda} \in \Hom(\f h_N,\ca H_{N+1}). \]
\end{cor}

\begin{proof}
From \rfp{abcdrestrABCD} we obtain that $B_\lambda  \ca S^{(N)} = \ca S^{(N+1)} b^-_{\lambda} \ca S^{(N)}$. Applying \rfp{abcdproperties} \ref{abcdpermutation} gives the desired result.
\end{proof}

\rfc{bsymmetrizer} can be neatly conveyed as a commuting diagram (\rff{commutingdiagram}).
\begin{figure}[h] 
\[ \xymatrix@R=1.5cm@C=2cm{\f h_N\ar[d]^{\ca S^{(N)}} \ar[r]^{b^-_{\lambda}} & \f h_{N+1} \ar[d]^{\ca S^{(N+1)}} \\ \ca H_N \ar[r]^{B_\lambda } & \ca H_{N+1} } \]
\caption{Diagrammatic presentation of \rfc{bsymmetrizer}} \label{commutingdiagram} 
\end{figure}

\subsection{Relations among the non-symmetric integral operators}

The following proposition will be used to reduce statements involving the operators $a_\lambda $, $c^\pm_\lambda $ and $d_{\lambda}$ to statements about $b^\pm_\lambda $.
\begin{prop} \label{abcdproducts}
Let $L \in \bR_{>0}$ and $\gamma \in \bR$. 
\begin{enumerate}
\item \label{brel1}
For $\lambda \in \bC$, in $\End(\f h_N)$ we have 
\[ b^+_\lambda \check \phi^+(-L/2) = \check \phi^+(-L/2) s_1 b^+_\lambda, \qquad
b^-_{\lambda} \check \phi^-(L/2) = \check \phi^-(L/2) s_N b^-_{\lambda}, \]
i.e. for $f \in \f d_N$ and $(x_1,\ldots,x_N) \in J^N$ we have
\begin{align*}
\left( b^+_\lambda \check \phi^+(-L/2) f \right)(x_1,\ldots,x_N) &= \left( b^+_\lambda f \right)(x_1,-L/2,x_2,\ldots,x_N), \\
\left( b^-_\lambda \check \phi^-(L/2) f \right)(x_1,\ldots,x_N) &= \left( b^-_\lambda f \right)(x_1,\ldots,x_{N-1},L/2,x_N).
\end{align*}
\item \label{adintermsofb}
For $\lambda \in \bC$, in $\End(\f h_N)$ we have 
\[ a_\lambda  = \check \phi^+(-L/2) b^+_\lambda , \qquad d_{\lambda} = \check \phi^-(L/2) b^-_{\lambda}, \]
i.e. for $f \in \f d_N$ and $(x_1,\ldots,x_N) \in J^N$ we have
\begin{align*}
\left(a_\lambda f\right)(x_1,\ldots,x_N) &= \left(b^+_\lambda f\right)(-L/2,x_1,\ldots,x_N), \\
\left(d_{\lambda}f\right)(x_1,\ldots,x_N) &= \left(b^-_{\lambda}f\right)(x_1,\ldots,x_N,L/2).
\end{align*}
\item \label{adintermsofb2}
Let $\lambda, \mu \in \bC$. Then in $\Hom(\f h_N,\f h_{N+1})$
\[ b^+_\lambda a_\mu = \check \phi^+(-L/2) s_1 b^+_\lambda  b^+_\mu , 
\qquad b^-_{\lambda}d_{\mu} = \check \phi^-(L/2) s_{N+1} b^-_{\lambda} b^-_{\mu}, \]
i.e. for $f \in \f d_N$ and $(x_1,\ldots,x_{N+1}) \in J^{N+1}$ we have
\begin{align*}
\left(b^+_\lambda a_\mu f\right)(x_1,\ldots,x_{N+1}) &= \left(b^+_\lambda b^+_\mu f\right)(x_1,-L/2,x_2,\ldots,x_{N+1}), \\
\left(b^-_{\lambda}d_{\mu}f\right)(x_1,\ldots,x_{N+1}) &= \left(b^-_{\lambda}b^-_{\mu}f\right)(x_1,\ldots,x_N,L/2,x_{N+1}).
\end{align*}
\item \label{cintermsofb} 
Let $\lambda \in \bC$. In $\Hom(\f h_{N+1},\f h_N)$ we have
\begin{align*} 
\gamma c^+_{\lambda} &= [\check \phi^-(L/2), a_\lambda ] = \check \phi^+(-L/2) [\check \phi^-(L/2), b^+_\lambda ], \\
\gamma c^-_{\lambda} &= [\check \phi^+(-L/2), d_{\lambda}] = \check \phi^-(L/2) [\check \phi^+(-L/2),  b^-_{\lambda}]. 
\end{align*}
\end{enumerate}
\end{prop}

\begin{proof}
The properties are proven separately.
\begin{enumerate}
\item This follows from the definitions of $\hat e^\pm_\lambda$. For the product $\check \phi^+(-L/2) s_1 b^+_\lambda $ note that those terms in the underlying summation over $\bm i \in \f i^n_{N}$ with $i_1=1$ vanish, because of $x_1 \in J$: either $\theta(\ldots>-L/2>x_1)=0$ will occur or there will be an integral with upper and lower limit  equal to $x_1=-L/2$.
A similar argument can be made for the product $\check \phi^-(L/2) s_N b^-_{\lambda}$.
\item This follows straightforwardly from the definitions of $\bar e^\pm$.
\item This is established by combining Properties \ref{brel1}-\ref{adintermsofb}.
\item In $(b^+_\lambda f)(L/2,\bm x,-L/2)$ ($\bm x \in \bR^N$) split the summation over $\bm i \in \f i^n_{N}$ into those terms with $i_1=N$, corresponding to the terms appearing in $(c^+_{\lambda}f)(\bm x)$, and those with $i_1 \ne N$, corresponding to the terms appearing in $(a_\lambda f)(L/2,\bm x)$. Then use $\check \phi^+(-L/2) \check \phi^-(L/2) b^+_\lambda $ = $\check \phi^-(L/2) a_\lambda $ by virtue of Property \ref{adintermsofb}.
A similar argument for $b^-_{\lambda}$ is used. \qedhere
\end{enumerate}
\end{proof}

\begin{rem}
Properties \ref{adintermsofb}-\ref{cintermsofb} in \rfp{abcdproducts} may in fact be used to define $a_\lambda, d_\lambda, c^\pm_\lambda$. 
These identities are generalizations of \rfp{ABCDproperties} \ref{ADasBandCasAD}, the corresponding identities for the symmetric operators $A_\lambda,B_\lambda ,C_\lambda ,D_\lambda  \in \End(\ca H)$.

\end{rem}

\section{Recursive construction of the pre-wavefunction}

In order to prove the important recursive property \rfeqn{prewavefnrecursion}, we will present commutation relations between $b^\pm_{\lambda}$ and $\partial_{j,\gamma}$ which are backed up by lemmas from Appendix \ref{nonsymmintopscommrels}. Crucially, this will also rely on \rfc{prewavefnDunkl}.\\

First we need to address the difference in domains between the Dunkl-type operators (acting on smooth functions on the regular vectors $\bR^N_\n{reg}$) and the creation operators $b^\pm_\lambda$ (acting on square-integrable functions on the bounded interval $[-L/2,L/2]$).
Note that the formulae defining the $b^\pm_\lambda$ can be seen as defining linear maps, denoted by the same symbols, from $\ca C(\bR^N)$ to $\ca C(\bR^{N+1})$, for any nonnegative integer $N$; the reader should compare this to the interpretation of $B_\lambda$ as an element of $\Hom(\ca C(\bR^N)^{S_N},\ca C(\bR^{N+1})^{S_{N+1}})$ in Section \ref{infiniteJ}.
Furthermore, both $b^\pm_\lambda$ restrict to elements of $\Hom(\ca{CB}^\infty(\bR^N),\ca{CB}^\infty(\bR^{N+1}))$, as they clearly preserve smoothness on the (open) alcoves of $\bR^N$.
By restricting the functions upon which the $b^\pm_\lambda$ act to $\bR^N_\n{reg}$, finally we may view $b^\pm_\lambda$ as elements of $\Hom(\ca C^\infty(\bR^N_\n{reg}),\ca C^\infty(\bR^{N+1}_\n{reg}))$.
Therefore, compositions of the form $\partial_{j,\gamma}^{(N+1)} b^\pm_{\lambda}$ for $j=1,\ldots,N+1$ and $ b^\pm_{\lambda} \partial_{j,\gamma}^{(N)}$ for $j=1,\ldots,N$ make sense as linear operators: $\ca C^\infty(\bR^N_\n{reg}) \to \ca C^\infty(\bR^{N+1}_\n{reg})$. 

\begin{prop} \label{bDunkl}
Let $\gamma \in \bR$ and $\lambda \in \bC$. 
Then in $\Hom(\ca C^\infty(\bR^N_\n{reg}),\ca C^\infty(\bR^{N+1}_\n{reg}))$ we have
\begin{align} 
\partial^{(N+1)}_{N+1,\gamma} b^-_{\lambda} &= \ii \lambda b^-_{\lambda}, & 
\partial^{(N+1)}_{1,\gamma} b^+_\lambda &= \ii \lambda b^+_\lambda, \label{eqn500} \\
\partial_{j,\gamma}^{(N+1)} b^-_{\lambda} &= b^-_{\lambda} \partial_{j,\gamma}^{(N)}, & 
\partial_{j+1,\gamma}^{(N+1)} b^+_\lambda &= b^+_\lambda \partial_{j,\gamma}^{(N)}, \qquad j=1,\ldots,N. \label{eqn501}
\end{align}
\end{prop}

\begin{proof}
We write $\partial^{(N)}_j$ for the partial derivative $\partial_j$ acting on functions defined on (an open set in) $\bR^N$.
For \rfeqn{eqn500}, by collecting powers of $\gamma$ we see that it suffices to prove
\begin{gather}
(\partial_{N+1}^{(N\! +\! 1)}-\ii \lambda)b^-_{\lambda;0} = 0, \qquad 
(\partial_{1}^{(N\! +\! 1)}-\ii \lambda)b^+_{\lambda;0} = 0, \label{eqn301} \\
\begin{aligned}
(\partial_{N+1}^{(N\! +\! 1)}-\ii \lambda)b^-_{\lambda;n+1} &= \Lambda^{(N\! +\! 1)}_{N+1} b^-_{\lambda;n}, \\
(\partial_{1}^{(N\! +\! 1)}-\ii \lambda)b^+_{\lambda;n+1} &= \Lambda^{(N\! +\! 1)}_{1} b^+_{\lambda;n}, \end{aligned} \qquad \n{for } n=0,\ldots,N-1,   \label{eqn302}  \\
\Lambda^{(N\! +\! 1)}_{N+1} b^-_{\lambda;N} = 0, \qquad 
\Lambda^{(N\! +\! 1)}_{1} b^+_{\lambda;N} = 0.  \label{eqn303} 
\end{gather}
\rfeqn{eqn301} follows from $(b^-_{\lambda;0}f)(x_1,\ldots,x_{N+1}) = \e^{\ii \lambda x_{N+1}} \, f(x_1,\ldots,x_N)$ and $(b^+_{\lambda;0}f)(x_1,\ldots,x_{N+1}) $ $= \e^{\ii \lambda x_1} f(x_2,\ldots,x_{N+1})$.
\rfeqn{eqn302} is established in \rfl{bnlem1}, and \rfeqn{eqn303} in \rfl{bnlem2}.\\

Again by collecting powers of $\gamma$, for \rfeqn{eqn501} we see that it is sufficient to prove
\begin{gather}
\partial_j^{(N\! +\! 1)} b^-_{\lambda;0} = b^-_{\lambda;0}  \partial_j^{(N)}, \qquad
\partial_{j+1}^{(N\! +\! 1)} b^+_{\lambda;0} = b^+_{\lambda;0} \partial_j^{(N)}, \label{eqn440} \\
\begin{aligned}
\partial_j^{(N\! +\! 1)} b^-_{\lambda;n+1} \! - \! \Lambda^{(N\! +\! 1)}_{j} b^-_{\lambda;n} &= b^-_{\lambda;n+1} \partial_j^{(N)} \! - \!  b^-_{\lambda;n} \Lambda^{(N)}_{j},\\
\partial_{j+1}^{(N\! +\! 1)} b^+_{\lambda;n+1} \! - \! \Lambda^{(N\! +\! 1)}_{j+1} b^+_{\lambda;n} &= b^+_{\lambda;n+1} \partial_j^{(N)} \! - \! b^+_{\lambda;n} \Lambda^{(N)}_{j}, \end{aligned} \quad \n{for } n=0,\ldots,N-1,   \label{eqn441}  \\
\Lambda^{(N\! +\! 1)}_{j} b^-_{\lambda;N} = b^-_{\lambda;N} \Lambda^{(N)}_{j}, \qquad 
\Lambda^{(N\! +\! 1)}_{j+1} b^+_{\lambda;N} = b^+_{\lambda;N} \Lambda^{(N)}_{j}.   \label{eqn442}
\end{gather}
We note that \rfeqn{eqn440} is trivial since $(b^-_{\lambda;0} f)(x_1,\ldots,x_{N+1}) = \e^{\ii \lambda x_{N+1}} f(x_1,\ldots,x_N)$ and $(b^+_{\lambda;0} f)(x_1,\ldots,x_{N+1}) = \e^{\ii \lambda x_1} f(x_2,\ldots,x_{N+1})$.
\rfeqn{eqn441} follows from \rfls{bnlem3}{bnlem4} and \rfeqn{eqn442} is established in \rfl{bnlem5}.
\end{proof}

We can now prove one of the main results of this thesis.
\begin{thm}[Recursive construction for the pre-wavefunction] \label{brecursion}
Let $\gamma \in \bR$ and $(\lambda_1,\ldots,\lambda_N) \in \bC^N$.
We have 
\begin{equation} \label{psirecursion1} \psi_{\lambda_1,\ldots,\lambda_N} = b^-_{\lambda_N}  \ldots b^-_{\lambda_1} \Vac 
= b^+_{\lambda_1}  \ldots b^+_{\lambda_N} \Vac \in \ca{CB}^\infty(\bR^N).\end{equation}
Hence, we have the recursion
\begin{equation} \label{psirecursion2} 
\psi_{\lambda_1,\ldots,\lambda_N} = b^-_{\lambda_N}  \psi_{\lambda_1,\ldots,\lambda_{N-1}} 
= b^+_{\lambda_1}  \psi_{\lambda_2,\ldots,\lambda_N} \in \ca{CB}^\infty(\bR^N). 
\end{equation}
Since the functions $\psi_{\bm \lambda}$ and the operators $b^\pm_\mu$ do not depend on $L$, we may interpret \rfeqnser{psirecursion1}{psirecursion2} as identities in $\f h_N([-L/2,L/2])$, where we have denoted the restriction of $\psi_{\bm \lambda}$ to the bounded interval $[-L/2,L/2]$ by the same symbol. 
\end{thm}

\begin{proof}
First of all, we note that $b^-_{\lambda_N}  \ldots b^-_{\lambda_1} \Vac$ and $b^+_{\lambda_1}  \ldots b^+_{\lambda_N}  \Vac$ are both elements of $\ca{CB}^\infty(\bR^N)$ since $\Vac \in \ca{CB}^\infty(\bR^0) \cong \bC$.
In light of \rfc{prewavefnDunkl}, showing that $b^-_{\lambda_N}  \ldots b^-_{\lambda_1} \Vac, b^+_{\lambda_1}  \ldots b^+_{\lambda_N}  \Vac \in \ca{CB}^\infty(\bR^N)$ are solutions of the system \rf{Dunklrepsystem} would imply that they are at least proportional to $\psi_{\lambda_1,\ldots,\lambda_N}$.
Let $j=1,\ldots,N$.
Then 
\[ \partial^{(N)}_{j,\gamma} b^-_{\lambda_N}  \ldots b^-_{\lambda_1} = b^-_{\lambda_N}  \ldots b^-_{\lambda_{j+1}} \partial^{(j)}_{j,\gamma} b^-_{\lambda_j} \ldots b^-_{\lambda_1} \]
by repeated application of \rfeqn{eqn501} in \rfp{bDunkl}.
We invoke \rfeqn{eqn500} and obtain 
\[ \partial^{(N)}_{j,\gamma} b^-_{\lambda_N}  \ldots b^-_{\lambda_1} = 
b^-_{\lambda_N}  \ldots b^-_{\lambda_{j+1}} \ii \lambda_j b^-_{\lambda_j} \ldots b^-_{\lambda_1} = 
\ii \lambda_j b^-_{\lambda_N}  \ldots b^-_{\lambda_{j+1}} b^-_{\lambda_j} \ldots b^-_{\lambda_1}, \]
an identity in $\Hom(\ca C^\infty(\bR^0_\n{reg}),\ca C^\infty(\bR^N_\n{reg}))$, with $\ca C^\infty(\bR^0_\n{reg}) \cong \bC$.
We obtain that $b^-_{\lambda_N}  \ldots b^-_{\lambda_1} \Vac \in \ca{CB}^\infty(\bR^N)$ is an eigenfunction of $\partial_{j,\gamma}^{(N)}$ with eigenvalue $\ii \lambda_j$, for all $j=1,\ldots,N$, as required. Similarly, we see that $b^+_{\lambda_1}  \ldots b^+_{\lambda_N}  \Vac$ is an eigenfunction of $\partial_{j,\gamma}$ with eigenvalue $\lambda_j$:
\begin{align*}
\partial_{j,\gamma}^{(N)} b^+_{\lambda_1}  \ldots b^+_{\lambda_N}  &= b^+_{\lambda_1}  \partial^{(N-1)}_{j-1,\gamma}  b^+_{\lambda_2}  \ldots b^+_{\lambda_N} \; = \; \ldots \; = \; b^+_{\lambda_1}  \ldots b^+_{\lambda_{j-1}}  \partial^{(j)}_{1,\gamma} b^+_{\lambda_j}  \ldots b^+_{\lambda_N}  \\
&= b^+_{\lambda_1}  \ldots b^+_{\lambda_{j-1}}  \ii \lambda_j b^+_{\lambda_j}  \ldots b^+_{\lambda_N}  \; = \; \ii \lambda_j b^+_{\lambda_1}  \ldots b^+_{\lambda_{j-1}}  b^+_{\lambda_j}  \ldots b^+_{\lambda_N} .
\end{align*}

It follows that $b^-_{\lambda_N}  \ldots b^-_{\lambda_1} \Vac$ and $b^+_{\lambda_1}  \ldots b^+_{\lambda_N}  \Vac$ are multiples of $\psi_{\lambda_1,\ldots,\lambda_N}$.
To see that they are in fact equal, it suffices to show that the functions coincide on the fundamental alcove ${\bR}^N_+ = \set{ \bm x \in \bR^N }{x_1 > \ldots > x_N}$.
\rfeqn{propoprestr} yields $\psi_{\lambda_1,\ldots,\lambda_N}|_{{\bR}^N_+} = \e^{\ii (\lambda_1,\ldots,\lambda_N)}$.
To see that $b^-_{\lambda_N}  \ldots b^-_{\lambda_1} \Vac|_{{\bR}^N_+} = \e^{\ii (\lambda_1,\ldots,\lambda_N)}$ we need a simple inductive argument; in the fundamental alcove for all $j=1,\ldots,N-1$, we have $x_N<x_j$, so that for all $f_{N-1} \in \ca{CB}^\infty(\bR^{N-1})$ and $(x_1,\ldots,x_N) \in {\bR}^N_+$ we have 
\[(b^-_{\lambda_N}  f_{N-1})(x_1,\ldots,x_N) = \e^{\ii \lambda_N x_N} f_{N-1}(x_1,\ldots,x_{N-1}). \]
Now $(x_1,\ldots,x_{N-1}) \in {\bR}^{N-1}_+$ and if $f_{N-1} = b^-_{\lambda_{N-2}} f_{N-2}$ for some $f_{N-2} \in \ca{CB}^\infty(\bR^{N-2})$ then again we see that $f_{N-1}(x_1,\ldots,x_{N-1}) = \e^{\ii \lambda_{N-1} x_{N-1}} f_{N-2}(x_1,\ldots,x_{N-2})$ so that 
\[ (b^-_{\lambda_N}  b^-_{\lambda_{N-1}} f_{N-2})(x_1,\ldots,x_N) = \e^{\ii (\lambda_{N-1}x_{N-1}+\lambda_N x_N)} f_{N-2}(x_1,\ldots,x_{N-3}).\] 
Continuing thus, we obtain that $(b^-_{\lambda_N}  \ldots b^-_{\lambda_1} \Vac)(x_1,\ldots,x_N)$ = 
$\e^{\ii \sum_j \lambda_j x_j} \Vac$ for \\
$(x_1,\ldots,x_N) \in {\bR}^N_+$; hence $b^-_{\lambda_N}  \ldots b^-_{\lambda_1} \Vac|_{{\bR}^N_+} = \e^{\ii (\lambda_1,\ldots,\lambda_N)}$ as required. A similar argument shows that $b^+_{\lambda_1}  \ldots b^+_{\lambda_N}  \Vac|_{{\bR}^N_+} = \e^{\ii (\lambda_1,\ldots,\lambda_N)}$.
\end{proof}

\begin{rem} 
\rfeqn{psirecursion2} can be seen as a Fourier transform of the statement in \rfp{prewavfnrecursionregrep}, and as a generalization of \rfeqn{Bethewavefnrecursion} to non-symmetric wavefunctions.
\end{rem}

We have demonstrated how the pre-wavefunction can be generated by the successive application of operators $b^-_{\lambda}$.
We are now able to give an alternative proof that the $\Psi_{\bm \lambda}$ defined using the QISM (i.e. by repeated application of the operators $B_{\lambda_j}$) is in fact equal to the $\Psi_{\bm \lambda}$ defined using the dAHA (i.e. by symmetrizing the pre-wavefunction):
\begin{cor} \label{Brecursion}
Let $L \in \bR_{>0}$, $\gamma \in \bR$ and $\bm \lambda = (\lambda_1,\ldots,\lambda_N) \in \bC^N$.
Then 
\[ \Psi_{\bm \lambda} =  B_{\lambda_N} \Psi_{\lambda_1,\ldots,\lambda_{N-1}}. \]
Hence, $\Psi_{\bm \lambda} = B_{\lambda_N} \ldots B_{\lambda_1} \Vac$. 
These may be viewed as identities in both $\ca{CB}^\infty(\bR^N)^{S_N}$ and $\ca H_N([-L/2,L/2])$.
\end{cor}

\begin{proof}
Straightforwardly we have, by virtue of \rfeqn{psirecursion2} and \rfc{bsymmetrizer}, 
\[ \Psi_{\lambda_1,\ldots,\lambda_N} = \ca S^{(N)} \psi_{\lambda_1,\ldots,\lambda_N} = 
\ca S^{(N)} b^-_{\lambda_N}  \psi_{\lambda_1,\ldots,\lambda_{N-1}} = 
B_{\lambda_N} \ca S^{(N-1)} \psi_{\lambda_1,\ldots,\lambda_{N-1}} = 
B_{\lambda_N} \Psi_{\lambda_1,\ldots,\lambda_{N-1}}. \qedhere \]
\end{proof}

\begin{prop}[Relation with the non-symmetric propagation operator] \label{bintertwiner}
Let $L \in \bR_{>0}$, $\gamma \in \bR$ and $\mu \in \bC$. We have
\[ P^{(N+1)}_\gamma \hat e^-_\mu = b^-_{\mu} P^{(N)}_\gamma \in \Hom(\ca{CB}^\infty(\bR^N),\ca{CB}^\infty(\bR^{N+1})). \]
\end{prop}

\begin{proof}
Note that the $\e^{\ii \bm \lambda}$ ($\bm \lambda \in \bC^N$) form a complete set in $\f h_N(\bR^N)$ (in the sense discussed in Section \ref{infiniteJ}).
Hence the proposition follows from the observations
\[ P^{(N+1)}_\gamma \hat e^-_\mu \e^{\ii \bm \lambda} = P^{(N+1)}_\gamma \e^{\ii(\bm \lambda,\mu)} = \psi_{\bm \lambda,\mu} \]
and
\[ b^-_{\mu} P^{(N)}_\gamma  \e^{\ii \bm \lambda} = b^-_{\mu} \psi_{\bm \lambda} = \psi_{\bm \lambda,\mu}. \qedhere \]
\end{proof}

The commuting diagram in \rff{commutingdiagram2} conveys \rfp{bintertwiner} succinctly.
Furthermore, we can summarize the contents of \rft{brecursion} and \rfp{bintertwiner} in \rff{recconstr}.

\begin{figure}[h] 
\[ \xymatrix@R=1.5cm@C=2cm{\f h_N\ar[d]^{P^{(N)}_\gamma} \ar[r]^{\hat e^-_\lambda} & \f h_{N+1} \ar[d]^{P^{(N+1)}_\gamma} \\ \f h_N \ar[r]^{b^-_{\lambda}} & \f h_{N+1} } \]
\caption{Diagrammatic presentation of \rfp{bintertwiner} which conveys the notion that first applying the propagation operator $P^{(N)}_\gamma$ and then the $\gamma$-dependent particle creation operator $b^-_\lambda$ gives the same result as first applying the free particle ($\gamma=0$) creation operator $e^-_\lambda$ and then $P^{(N+1)}_\gamma$.}
\label{commutingdiagram2} 
\end{figure}

\begin{figure}[h] 
\[ \xymatrix@R=1.2cm@C=1.5cm{
1 \ar[d]^{P^{(0)}_\gamma = 1} \ar[r]^{\hat e^-_{\lambda_1}} & \e^{\ii \lambda_1} \ar[d]^{P^{(1)}_\gamma = 1} \ar[r]^{\hat e^-_{\lambda_2}} & \e^{\ii (\lambda_1,\lambda_2)} \ar[d]^{P^{(2)}_\gamma} \ar[r]^{\hat e^-_{\lambda_3}}& \ldots \ar[r]^{\hspace{-4mm} \hat e^-_{\lambda_{N-1}}} & \e^{\ii (\lambda_1,\ldots,\lambda_{N-1})} \ar[d]^{P^{(N-1)}_\gamma} \ar[r]^{\hat e^-_{\lambda_{N}}} & \e^{\ii (\lambda_1,\ldots,\lambda_{N})} \ar[d]^{P^{(N)}_\gamma}\\ 
1 \ar[d]^{\ca S^{(0)} = 1} \ar[r]^{b^-_{\lambda_1}} & \psi_{\lambda_1} \ar[d]^{\ca S^{(1)} = 1} \ar[r]^{b^-_{\lambda_2}} & \psi_{\lambda_1,\lambda_2} \ar[d]^{\ca S^{(2)}} \ar[r]^{b^-_{\lambda_3}}& \ldots \ar[r]^{\hspace{-4mm} b^-_{\lambda_{N-1}}} & \psi_{\lambda_1,\ldots,\lambda_{N-1}} \ar[d]^{\ca S^{(N-1)}} \ar[r]^{b^-_{\lambda_N}} & \psi_{\lambda_1,\ldots,\lambda_{N}}  \ar[d]^{\ca S^{(N)}} \\  
1 \ar[r]^{B_{\lambda_1}} & \Psi_{\lambda_1} \ar[r]^{B_{\lambda_2}} & \Psi_{\lambda_1,\lambda_2} \ar[r]^{B_{\lambda_3}}& \ldots \ar[r]^{\hspace{-4mm} B_{\lambda_{N-1}}} & \Psi_{\lambda_1,\ldots,\lambda_{N-1}} \ar[r]^{B_{\lambda_N}} & \Psi_{\lambda_1,\ldots,\lambda_{N}} } \]
\caption{Scheme for recursive constructions of the wavefunction $\Psi_{\bm \lambda}$. In the three rows we have the non-symmetric plane waves, the non-symmetric pre-wavefunctions (highlighting the new formulae from \rft{brecursion}), and the symmetric Bethe wavefunctions. 
Note that the three operators $\hat e^-_\lambda$, $b^-_{\lambda}$, $B_\lambda $ coincide when acting on $\f h_0$ or $\f h_1$. 
The plane waves can be obtained from the pre-wavefunctions by setting $\gamma =0$. This also works on the level of the creation operators; when setting $\gamma =0$ in $b^-_{\lambda}$ one obtains $\hat e^-_\lambda$.
} \label{recconstr} 
\end{figure}

\section{Commutation relations} \label{abcdcommrelssec}

Let $J=[-L/2,L/2]$ be bounded.
We will now prove commutation relations of the non-symmetric integral operators $a_\lambda,b^\pm_\lambda,c^\pm_\lambda,d_\lambda$, which we will refer to as the \emph{non-symmetric Yang-Baxter relations}. 
These operators are defined on a dense subset of $\f h_N=\f h_N(J)$; however, we will only prove these relations here on the subset
\begin{equation} \f z_N = \overline{\left\langle \psi_{\bm \lambda}: \bm \lambda \in \bC^N \right\rangle}, \end{equation}%
\nc{rzlw}{$\f z_N$}{Completion of span of all $\psi_{\bm \lambda}$ in $\f h_N$. \nomrefeqpage}%
i.e. the completion of the subspace of $\f h_N$ spanned by the pre-wavefunctions $\psi_{\bm \lambda}$.

\begin{prop} \label{bbcommrel1}
Let $L \in \bR_{>0}$, $\gamma \in \bR$ and $\mu,\nu \in \bC$.
Then 
\[ [b^-_{\mu},b^+_{\nu} ] = 0 \in \Hom(\f z_N,\f z_{N+2}). \]
\end{prop}

\begin{proof}
This follows from the fact that $b^-_{\mu}b^+_{\nu}  \psi_{\bm \lambda}$ and $b^+_{\nu}  b^-_{\mu}\psi_{\bm \lambda}$ both equal $\psi_{\nu,\bm \lambda,\mu}$, for $\bm \lambda \in \bC^N$, as per \rft{brecursion}.
\end{proof}

The next theorem generalizes the connection of the Yang-Baxter algebra with the regular representation of the dAHA in momentum space studied in Section \ref{YBalgebradAHAsec} to the non-symmetric case.
\begin{thm} \label{bbcommrel2}
Let $L \in \bR_{>0}$, $\gamma \in \bR$ and $\lambda_1,\lambda_2\lambda_{N+1}, \lambda_{N+2} \in \bC$. Then we have
\begin{align*} 
s_{N+1} b^-_{\lambda_{N+2}} b^-_{\lambda_{N+1}} &= \left( \tilde s_{N+1,\gamma} b^-_{\lambda_{N+2}} b^-_{\lambda_{N+1}} \right) && \in \Hom(\f z_N,\f z_{N+2}) \\
s_1 b^+_{\lambda_1}  b^+_{\lambda_2}  &= \left( \tilde s_{1,\gamma} b^+_{\lambda_1}  b^+_{\lambda_2}  \right) && \in \Hom(\f z_N,\f z_{N+2}).
\end{align*}
Moreover, we have
\begin{align*} 
b^-_{\lambda_{N+2}}d_{\lambda_{N+1}} &= \left( \tilde s_{N+1,\gamma} d_{\lambda_{N+2}}b^-_{\lambda_{N+1}}\right) &&\in \Hom(\f z_N,\f z_{N+1}), \\
b^+_{\lambda_1}  a_{\lambda_2} &= \left( \tilde s_{1,\gamma} a_{\lambda_1}b^+_{\lambda_2}  \right) && \in \Hom(\f z_N,\f z_{N+1}),
\end{align*}
or, alternatively,
\begin{align}  \label{dbcommrel} d_{\lambda_{N+2}}b^-_{\lambda_{N+1}} &= \left( \tilde s_{N+1,\gamma} b^-_{\lambda_{N+2}}d_{\lambda_{N+1}}\right) && \in \Hom(\f z_N,\f z_{N+1})\\  
\label{abcommrel} a_{\lambda_1}b^+_{\lambda_2}  &= \left( \tilde s_{1,\gamma} b^+_{\lambda_1} a_{\lambda_2}\right) && \in \Hom(\f z_N,\f z_{N+1}).\end{align}
\end{thm}

\begin{proof}
From \rfl{Dunklrepprewavefn}, \rfeqn{prewavefn2} we have, for arbitrary $(\lambda_1,\ldots,\lambda_{N+2}) \in \bC^{N+2}$
\[ s_{N+1} \psi_{\lambda_1,\ldots,\lambda_{N+2}} = \tilde s_{N+1,\gamma} \psi_{\lambda_1,\ldots,\lambda_{N+2}}. \]
Using \rfeqn{psirecursion2} this is equivalent to
\[ s_{N+1} b^-_{\lambda_{N+2}} b^-_{\lambda_{N+1}} \psi_{\lambda_1,\ldots,\lambda_N} = \tilde s_{N+1,\gamma} b^-_{\lambda_{N+2}} b^-_{\lambda_{N+1}} \psi_{\lambda_1,\ldots,\lambda_N}, \]
which yields the desired identity for compositions of $b^-$. 
A similar argument applies for compositions of $b^+$, using  
$s_1 \psi_{\lambda_1,\ldots,\lambda_{N+2}} = \tilde s_{1,\gamma} \psi_{\lambda_1,\ldots,\lambda_{N+2}}$,
and $\psi_{\lambda_1,\ldots,\lambda_{N+2}} = b^+_{\lambda_1} b^+_{\lambda_2} \psi_{\lambda_3,\ldots,\lambda_{N+2}}$.\\

The statements connecting the operators $b^-$ and $d$, and $b^+$ and $a$, respectively, are obtained by sending the appropriate $x_j$ to its minimum or maximum value (i.e., $\pm L/2$) in \rft{bbcommrel2}, relabelling the remaining $x_j$ and using \rfl{abcdproducts} \ref{adintermsofb}-\ref{adintermsofb2}. The alternative formulations \rfeqnser{dbcommrel}{abcommrel} follow from \rfl{ABBAequivalence}.
\end{proof}

``Unpacking'' the notation involving the operators $\tilde s_{j,\gamma}$ in \rft{bbcommrel2} we obtain
\begin{align*}
s_{N+1} b^-_{\lambda} b^-_{\mu} - b^-_{\mu} b^-_{\lambda} &= \frac{\ii \gamma}{\lambda-\mu} \left[ b^-_{\lambda},b^-_{\mu} \right] && \in \Hom(\f z_N,\f z_{N+2}), \\
s_1 b^+_\lambda  b^+_\mu  - b^+_\mu  b^+_\lambda  &= \frac{-\ii \gamma}{\lambda-\mu} \left[ b^+_\lambda ,b^+_\mu  \right] && \in \Hom(\f z_N,\f z_{N+2}), \displaybreak[2] \\
\left[ b^-_{\lambda}, d_{\mu} \right] &= \frac{\ii \gamma}{\lambda-\mu} \left( d_{\lambda}b^-_{\mu}-d_{\mu}b^-_{\lambda} \right) && \in \Hom(\f z_N,\f z_{N+1}), \displaybreak[2] \\
\left[ b^+_\lambda , a_\mu  \right] &= \frac{-\ii \gamma}{\lambda-\mu} \left( a_\lambda b^+_\mu -a_\mu b^+_\lambda  \right) && \in \Hom(\f z_N,\f z_{N+1}).
\end{align*}

\begin{lem} \label{apsidpsilem}
Let $L \in \bR_{>0}$, $\gamma \in \bR$ and $(\bm \lambda,\lambda_{N+1}) \in \bC^{N+1}$. Then
\begin{align}
a_{\lambda_{N+1}}\psi_{\bm \lambda} &= \tilde s_{N,-\gamma} \ldots \tilde s_{1,-\gamma} \tilde s_1 \ldots \tilde s_N \e^{-\ii \lambda_{N+1}L/2} \psi_{\bm \lambda} \nonumber \\
&= \left(1-\ii \gamma \tilde \Delta_{N \, N+1} \right) \ldots \left(1-\ii \gamma \tilde \Delta_{1 \, N+1} \right) \e^{-\ii \lambda_{N+1}L/2} \psi_{\bm \lambda} \label{apsi}\\
d_{\lambda_{N+1}}\psi_{\bm \lambda} &= \tilde s_{N,\gamma} \ldots \tilde s_{1,\gamma} \tilde s_1 \ldots \tilde s_N \e^{\ii \lambda_{N+1}L/2} \psi_{\bm \lambda} \nonumber \\
&= \left(1+\ii \gamma \tilde \Delta_{N \, N+1} \right) \ldots \left(1+\ii \gamma \tilde \Delta_{1 \, N+1} \right) \e^{\ii \lambda_{N+1}L/2} \psi_{\bm \lambda} \label{dpsi}.
\end{align}
\end{lem}

\begin{proof}
Using \rfeqn{dbcommrel} we can move the operator $d$ to the right, as follows:
\begin{align*}
d_{\lambda_{N+1}}\psi_{\bm \lambda} &= d_{\lambda_{N+1}} b^-_{\lambda_N}  \ldots b^-_{\lambda_1} \Vac \; = \; \tilde s_{N+1,\gamma} b^-_{\lambda_{N+1}} d_{\lambda_N} b^-_{\lambda_{N-1}} \ldots b^-_{\lambda_1} \Vac \\
&= \tilde s_{N+1,\gamma} b^-_{\lambda_{N+1}} \tilde s_{N,\gamma} b^-_{\lambda_N}  d_{\lambda_{N-1}} b^-_{\lambda_{N-2}}\ldots b^-_{\lambda_1} \Vac \; = \; \ldots \; = \;  \tilde s_{N,\gamma} \ldots \tilde s_{1,\gamma} b^-_{\lambda_{N+1}} \ldots b^-_{\lambda_2} d_{\lambda_1} \Vac.
\end{align*}
From $d_{\mu} \Vac = \e^{\ii \mu L/2}\Vac$ we obtain
\[ d_{\lambda_{N+1}}\psi_{\bm \lambda} =  \tilde s_{N,\gamma} \ldots \tilde s_{1,\gamma}  \e^{\ii \lambda_1 L/2}\psi_{\lambda_2, \ldots, \lambda_{N+1}} =  \tilde s_{N,\gamma} \ldots \tilde s_{1,\gamma} \tilde s_{1,\gamma} \ldots \tilde s_{N,\gamma} \e^{\ii \lambda_{N+1} L/2}\psi_{\bm \lambda},\]
as required. Using \rfl{regrep20} we also find the other expression for $d_{\lambda_{N+1}}\psi_{\bm \lambda}$. The expressions for $a_{\lambda_{N+1}} \psi_{\bm \lambda}$ are obtained analogously.
\end{proof}

\begin{cor} \label{ddcommrel}
Let $L \in \bR_{>0}$, $\gamma \in \bR$ and $\lambda, \mu \in \bC$. Then $[d_{\lambda},d_{\mu}] = 0 \in \End(\f z_N)$.
\end{cor}

\begin{proof}
It suffices to show that $d_{\lambda_{N+2}}d_{\lambda_{N+1}}\psi_{\bm \lambda}$ is invariant under $\lambda_{N+1} \leftrightarrow \lambda_{N+2}$, for $(\bm \lambda, \lambda_{N+1},\lambda_{N+2}) $ $\in \bC^{N+2}$.
Applying \rfeqn{dpsi} we note that 
\begin{align*} 
d_{\lambda_{N+2}}d_{\lambda_{N+1}}\psi_{\bm \lambda} 
&= \left(1+\ii \gamma \tilde \Delta_{N \, N+1} \right) \ldots \left(1+\ii \gamma \tilde \Delta_{1 \, N+1} \right) \e^{\ii \lambda_{N+1}L/2} d_{\lambda_{N+2}} \psi_{\bm \lambda}\\
&= \left(1+\ii \gamma \tilde \Delta_{N \, N+1} \right) \ldots \left(1+\ii \gamma \tilde \Delta_{1 \, N+1} \right)  \cdot \\
& \qquad \cdot \left(1+\ii \gamma \tilde \Delta_{N \, N+2} \right) \ldots \left(1+\ii \gamma \tilde \Delta_{1 \, N+2} \right) \e^{\ii (\lambda_{N+1}+\lambda_{N+2})L/2}  \psi_{\bm \lambda} \\
&= \left(1+\ii \gamma \tilde \Delta_{N \, N+1} \right)\left(1+\ii \gamma \tilde \Delta_{N \, N+2} \right)  \ldots \cdot \\
& \qquad \cdot \ldots \left(1+\ii \gamma \tilde \Delta_{1 \, N+1} \right)\left(1+\ii \gamma \tilde \Delta_{1 \, N+2} \right) \e^{\ii (\lambda_{N+1}+\lambda_{N+2})L/2}  \psi_{\bm \lambda}.
\end{align*}
Writing 
\[ y_{N+1,N+2}(j):=\left(1+\ii \gamma \tilde \Delta_{j \, N+1} \right)\left(1+\ii \gamma \tilde \Delta_{j \, N+2} \right) =1+\ii \gamma(\tilde \Delta_{j \, N+1}+\tilde \Delta_{j \, N+2})-\gamma^2 \tilde \Delta_{j \, N+1}\tilde \Delta_{j \, N+2} \] 
for $j=1,\ldots,N$, it follows that 
\[ d_{\lambda_{N+2}}d_{\lambda_{N+1}}\psi_{\bm \lambda}  = y_{N+1,N+2}(N) \ldots y_{N+1,N+2}(1) \e^{\ii (\lambda_{N+1}+\lambda_{N+2})L/2}  \psi_{\bm \lambda}. \]
Note that the expression $\e^{\ii (\lambda_{N+1}+\lambda_{N+2})L/2}  \psi_{\bm \lambda}$ is invariant under $\lambda_{N+1} \leftrightarrow \lambda_{N+2}$, i.e. under the action of the symmetrizer $\frac{1}{2}(1+\tilde s_{N+1})$, so that it is sufficient to show that
\[ y_{N+1,N+2}(N) \ldots y_{N+1,N+2}(1) (1+\tilde s_{N+1}) = y_{N+2,N+1}(N) \ldots y_{N+2,N+1}(1)  (1+\tilde s_{N+1}). \]
This in turn follows from repeatedly applying
\[ y_{N+1,N+2}(j) (1+\tilde s_{N+1}) = y_{N+2,N+1}(j) (1+\tilde s_{N+1}), \]
which is a consequence of
\[ \tilde \Delta_{j \, N+1}\tilde \Delta_{j \, N+2}  (1+\tilde s_{N+1}) = \tilde \Delta_{j \, N+2}\tilde \Delta_{j \, N+1}  (1+\tilde s_{N+1}) \]
a restatement of \rfl{regrep22}.
\end{proof}

Write $\f z_\n{fin} = \f h_\n{fin} \cap \bigcup_{N \geq 0} \f z_N$.
Using \rfp{abcdproperties} \ref{abcdformaladjoint} we may take adjoints of the statements in \rfp{bbcommrel1}, \rft{bbcommrel2} and \rfc{ddcommrel} and obtain
\begin{thm}[Non-symmetric Yang-Baxter algebra] \label{abcdcommrels}
Let $L \in \bR_{>0}$, $\gamma \in \bR$ and $\lambda, \mu \in \bC$. Then we have the following relations in $\End(\f z_\n{fin})$:
\boxedgather{
[a_\lambda ,a_\mu ] = [b^-_{\lambda},b^+_\mu ] \, = \, [c^-_{\lambda},c^+_{\mu}] \, = \, [d_{\lambda},d_{\mu}] \, = \, 0 \nonumber \\
\begin{aligned}
{[a_\lambda ,b^+_\mu ] } &= - \frac{\ii \gamma}{\lambda-\mu} \left( b^+_\lambda a_\mu  - b^+_\mu a_\lambda  \right), \qquad &
[b^+_\lambda ,a_\mu ] &= - \frac{\ii \gamma}{\lambda-\mu} \left( a_\lambda b^+_\mu  - a_\mu b^+_\lambda  \right), \\
[d_{\lambda},b^-_{\mu}] &= \frac{\ii \gamma}{\lambda-\mu} \left( b^-_{\lambda}d_{\mu} - b^-_{\mu}d_{\lambda} \right), &
[b^-_{\lambda},d_{\mu}] &=  \frac{\ii \gamma}{\lambda-\mu} \left(d_{\lambda}b^-_{\mu} - d_{\mu}b^-_{\lambda}\right), \\
[a_\lambda ,c^+_{\mu}] &= \frac{\ii \gamma}{\lambda-\mu} \left( c^+_{\lambda}a_\mu  - c^+_{\mu}a_\lambda  \right), &
[c^+_{\lambda},a_\mu ] &=  \frac{\ii \gamma}{\lambda-\mu} \left( a_\lambda c^+_{\mu} - a_\mu c^+_{\lambda}\right), \\
[d_{\lambda},c^-_{\mu}] &= - \frac{\ii \gamma}{\lambda-\mu} \left( c^-_{\lambda}d_{\mu} - c^-_{\mu}d_{\lambda}\right), &
[c^-_{\lambda},d_{\mu}] &= - \frac{\ii \gamma}{\lambda-\mu} \left( d_{\lambda}c^-_{\mu} - d_{\mu}c^-_{\lambda}\right).
\end{aligned} \nonumber}
\end{thm}

It can also be checked that $[b^\pm_\lambda ,b^\pm_\mu] \ne 0$ (as expected) by applying both $b^\pm_\lambda b^\pm_\mu$ and $b^\pm_\mu b^\pm_\lambda $ to $\Vac$. This is equivalent to saying that $\psi_{\lambda,\mu} \ne \psi_{\mu,\lambda}$, as can be immediately checked from the expression for $\psi_{\lambda_1,\lambda_2}$ in \rfex{prewavefnexample}. 
By taking adjoints it follows that $[c^\pm_\lambda ,c^\pm_\mu ] \ne 0$.\\

We actually claim that a stronger result holds:
\begin{conj} \label{abcdcommrelsconj}
The relations listed in \rft{abcdcommrels} hold in $\End(\f h_\n{fin})$, and hence in $\End(\f h)$.
\end{conj}
This could be proven by showing that the set of all $\psi_{\bm \lambda}$, where $\bm \lambda$ runs through $\bC^N$, is complete in $\f h_N$ e.g. by establishing invertibility of the propagation operator $P^{(N)}_\gamma$ on a dense subset of  $\f h_N$ containing the pre-wavefunctions and using the completeness of the plane waves $\e^{\ii \bm \lambda}$ in $\f h_N$.
Alternatively, it may be possible to establish commutation relations among the $b^\pm_{\lambda;n} \in \End(\f h_N)$ and use the expansions $b^\pm_\lambda = \sum_n \gamma^n b^\pm_{\lambda;n}$ to establish \rfp{bbcommrel1} and the first two statements in \rft{bbcommrel2} as identities in $\Hom(\f h_N,\f h_{N+2})$.

\begin{rem}
Note that all presented relations in \rft{abcdcommrels} can be ``symmetrized'' to established relations for the symmetric integral operators as presented in \rfc{Tcommrelcor}.
For example, from 
\[ [a_\lambda ,b^+_\mu ] = \frac{-\ii \gamma}{\lambda-\mu} \left( b^+_\lambda a_\mu  - b^+_\mu  a_\lambda \right)  \]
one obtains
\[ \ca S^{(N+1)} [a_\lambda ,b^+_\mu ]|_{\ca H_N} = \frac{-\ii \gamma}{\lambda-\mu} \ca S^{(N+1)} \left( b^+_\lambda a_\mu |_{\ca H_N} - b^+_\mu  a_\lambda |_{\ca H_N} \right), \]
which yields
\[ [a_\lambda ,\ca S^{(N+1)} b^+_\mu ]|_{\ca H_N} = \frac{-\ii \gamma}{\lambda-\mu} \left( \ca S^{(N+1)} b^+_\lambda a_\mu |_{\ca H_N} - \ca S^{(N+1)} b^+_\mu  a_\lambda |_{\ca H_N} \right) \]
by virtue of \rfp{abcdproperties} \ref{abcdpermutation}. Then applying \rfp{abcdrestrABCD} gives \rfeqn{ABcommrel}.
This way all relations in \rfc{Tcommrelcor} are obtained, except the $AD$-, $BC$-, $CB$- and $DA$-relations.
\end{rem}

Another set of commutation relations can be obtained from \rfl{abcdproducts} \ref{cintermsofb}.
These cannot be symmetrized to relations in \rfc{Tcommrelcor}.
\begin{lem}
Let $L \in \bR_{>0}$, $\gamma \in \bR$ and $\lambda,\mu \in \bC$ such that $\lambda \ne \mu$. Then
\begin{align} 
[a_\lambda ,d_\mu] &= \gamma \left( c^-_{\mu}b^+_\lambda  - c^+_{\lambda}b^-_{\mu} \right) \in \End(\f z_N); \label{adcommrel} \\
[d_{\lambda},a_\mu ] &= \gamma \left( c^+_{\mu}b^-_{\lambda} - c^-_{\lambda}b^+_\mu  \right) \in \End(\f z_N). \label{dacommrel}
\end{align}
In particular, $[a_\lambda ,d_{\mu}]$ is not invariant under $\lambda \leftrightarrow \mu$, in contrast to $[A_\lambda,D_\mu ]$ cf. \rfeqn{ADcommrel}.
\end{lem}

\begin{proof}
Focusing on the right-hand side, we have
\[ \gamma \left( c^-_{\mu}b^+_\lambda  - c^+_{\lambda}b^-_{\mu} \right) = \check \phi^+(-L/2) d_{\mu} b^+_\lambda  - d_{\mu} \check \phi^+(-L/2) b^+_\lambda  - \check \phi^-(L/2) a_\lambda  b^-_{\mu} + a_\lambda  \check \phi^-(L/2) b^-_{\mu}, \]
by virtue of \rfl{abcdproducts} \ref{cintermsofb}. Now applying \rfl{abcdproducts} \ref{adintermsofb} and subsequently \rfp{bbcommrel1} and \rfeqn{eqn42} we obtain that
\[ \gamma \left( c^-_{\mu}b^+_\lambda  - c^+_{\lambda}b^-_{\mu} \right) = 
\check \phi^+(-L/2) \check \phi^-(L/2) b^-_{\mu} b^+_\lambda  - \check \phi^-(L/2) \check \phi^+(-L/2) b^+_\lambda  b^-_{\mu} + a_\lambda  d_{\mu} - d_{\mu} a_\lambda = [a_\lambda ,d_{\mu}]. \]
This establishes \rfeqn{adcommrel}; \rfeqn{dacommrel} follows by applying $\lambda \leftrightarrow \mu$.
\end{proof}

\begin{rem}[Integrability in the non-symmetric case] \label{psitransfermatrix}
We claim that the only eigenfunctions of $t_\mu:=a_\mu +d_{\mu}$\nc{rtl}{$t_\lambda$}{Shorthand for $a_\lambda+d_\lambda$} are the $\Psi_{\bm \lambda}$ where $\bm \lambda$ satisfies the BAEs \rf{BAE}.
In connection with this, we recall that $\psi_{\bm \lambda}$ cannot be made $L$-periodic by imposing conditions on $\bm \lambda$ as has been demonstrated for $N=2$ in Subsection \ref{subsectperiodicity}.
The claim that $\Psi_{\bm \lambda}$ is indeed an eigenfunction of $t_\mu$ follows immediately from the remark that $a_\mu $ and $d_{\mu}$ restrict to $A_\mu $ and $D_\mu $ on the domain of symmetric functions. Hence $t_\mu$ restricts to $T_\mu$ on that domain, as well, and we have $[t_{\mu},t_{\nu}]|_{\ca H_N}=0$.
However, in general, $[t_{\mu},t_{\nu}] \ne 0$, since $[a_{\mu},d_{\nu}] \ne [a_{\nu},d_{\mu}] $.
\end{rem}

\section{The limiting case $J=\bR$}

This section generalizes Section \ref{infiniteJ} to the non-symmetric case.
Analogously to the behaviour of $A_\mu $ and $D_\mu $ for $L \to \infty$ as discussed in \rfl{Linfty}, we present
\begin{prop}\label{Linftypsi}
Let $\gamma \in \bR$, $\bm \lambda \in \bR^N_\n{reg}$ and $\mu \in \bC \setminus \{ \lambda_1,\ldots,\lambda_N\}$.
Then 
\[ b^-_{\mu} \psi_{\bm \lambda} = \psi_{\bm \lambda,\mu} \in \ca C(\bR^N), 
\qquad b^+_\mu \psi_{\bm \lambda} = \psi_{\mu,\bm \lambda} \in \ca C(\bR^N) \]
and, for $\bm x \in \bR^N$,
\begin{align*}
\lim_{L \to \infty} \e^{\ii \mu L/2} \left(a_\mu \psi_{\bm \lambda} \right)(\bm x) &= \tau^+_\mu(\bm \lambda) \psi_{\bm \lambda}(\bm x), && \n{if } \Im \mu>0, \\
\lim_{L \to \infty} \e^{-\ii \mu L/2} \left(d_\mu \psi_{\bm \lambda} \right)(\bm x) &= \tau^-_\mu(\bm \lambda) \psi_{\bm \lambda}(\bm x), && \n{if } \Im \mu<0.
\end{align*}
\end{prop}

\begin{proof}
Write $\lambda_{N+1} = \mu$.
The expressions for $\left( \lim_{L \to \infty} b^\pm_\mu \right) \psi_{\bm \lambda}$ are obtained by induction, noting that $\psi_{\bm \lambda}$ does not depend on $L$ and hence 
\[ \lim_{L \to \infty} b^-_{\lambda_{N+1}} \ldots b^-_{\lambda_1} \Vac = b^-_{\lambda_{N+1}} \ldots b^-_{\lambda_1} \Vac  = \psi_{\lambda_1,\ldots,\lambda_{N+1}},\] 
and similarly for $b^+_\mu $.\\

With respect to $\left( \lim_{L \to \infty} \e^{-\ii \mu L/2} d_{\mu} \right) \psi_{\bm \lambda}$, note that $1+\ii \gamma \tilde \Delta_{j \, N+1}  = \frac{\lambda_j-\lambda_{N+1}+\ii \gamma}{\lambda_j - \lambda_{N+1}} - \frac{\ii \gamma}{\lambda_j - \lambda_{N+1}} \tilde s_{j \, N+1}$.
Expanding the product $\left(1+\ii \gamma \tilde \Delta_{N \, N+1} \right) \ldots \left(1+\ii \gamma \tilde \Delta_{1 \, N+1} \right)$ in \rfeqn{dpsi} yields linear combinations of products of $\tilde s_{j \, N+1}$. 
The only term proportional to $1 \in S_N$ is obtained by choosing the term with 1 in each factor $1+\ii \gamma \tilde \Delta_{j \, N+1} $. 
This produces a term $ \frac{\lambda_N-\mu+\ii \gamma}{\lambda_N - \mu} \ldots  \frac{\lambda_1-\mu+\ii \gamma}{\lambda_1 - \mu} = \tau^-_\mu(\bm \lambda)$. 
We see that $\e^{-\ii \mu L/2} d_{\mu} \psi_{\bm \lambda}$ will be a linear combination of $\tau^-_\mu(\bm \lambda) \psi_{\lambda_1,\ldots,\lambda_N}$ and terms proportional to $\e^{\ii (\lambda_j-\mu) L/2} \psi_{\lambda_1,\ldots,\lambda_{j-1},\mu,\lambda_{j+1},\ldots,\lambda_N}$ (with coefficients independent of $L$). Note that the exponent in $\e^{\ii (\lambda_j-\mu)L/2}$ has negative real part, provided $\Im \mu<0$, causing all terms but $\tau^-_\mu(\bm \lambda) \psi_{\lambda_1,\ldots,\lambda_N}$ to vanish in the limit $L \to 0$. The expression $ \e^{\ii \mu L/2} a_\mu  \psi_{\bm \lambda}$ is analysed in a similar manner.
\end{proof}

If we can prove a completeness theorem of the $\psi_{\bm \lambda}$ in $\f h_N$ (in the sense alluded to in Section \ref{infiniteJ}), we would obtain from \rfp{Linftypsi} and the estimate $|\tau^\pm_\mu(\bm \lambda)| \leq \left( 1+ \frac{|\gamma|}{|\Im \mu|} \right)^N$ the following statement.
\begin{conj}
We can extend $\lim_{L \to \infty} \e^{\ii \mu L/2} a_\mu $ and $\lim_{L \to \infty} \e^{-\ii \mu L/2} d_{\mu}$ to operators on $\f h_N$. Furthermore, they are bounded on $\f h_N$ provided that $\Im \mu \ne 0$.
\end{conj}

\newpage

\chapter{Conclusions and some open problems} \label{chSummary}

We have highlighted the important theoretical role the pre-wavefunction plays for the QNLS model, in particular for the connection between the dAHA approach and the QISM. 
The key points of this thesis are best summarized by comparing some important properties of the pre-wavefunction $\psi_{\bm \lambda}$ and the Bethe wavefunction $\Psi_{\bm \lambda}$. This will also allow us to pinpoint some possible future avenues of research.

\begin{description}
\item[Relation to the symmetric group\hspace{15cm}] Although the pre-wavefunction is not $S_N$-invariant, its definition by means of the non-symmetric propagation operator involves the symmetric group in an essential way: $\psi_{\bm \lambda} = \sum_{w \in S_N}  \chi_{w^{-1}\bR^N_+} w^{-1} w_\gamma \e^{\ii \bm \lambda}$. 
For instance, antisymmetrizing $\psi_{\bm \lambda}$ will not result in a fermionic wavefunction, i.e. one that transforms as $\Psi_{w \bm \lambda} = \sgn(w) \Psi_{\bm \lambda}$ for $w \in S_N$. 
This leaves us with a question whether a similar propagation operator formalism can be set up for a one-dimensional fermionic system with pairwise contact interaction.

\item[Periodicity\hspace{15cm}]  We have seen in Subsection \ref{subsectperiodicity} that the Bethe wavefunctions $\Psi_{\bm \lambda}$ can be made periodic by imposing the Bethe ansatz equations \rf{BAE} on the $\bm \lambda$ but that this is not possible for the pre-wavefunctions $\psi_{\bm \lambda}$. Equivalently, on a bounded interval, the $\Psi_{\bm \lambda}$ become eigenfunctions of the transfer matrix (\rft{ABA}), but not the $\psi_{\bm \lambda}$ (\rfr{psitransfermatrix}). When taking the limit $L \to \infty$ however, both $\Psi_{\bm \lambda}$ and $\psi_{\bm \lambda}$ are eigenfunctions of $\e^{\pm \ii \mu L/2}T_\mu$ and $\e^{\pm \ii \mu L/2}t_\mu$, respectively, if $\Im \mu \gtrless 0$, as per \rfp{Linfty} and \rfp{Linftypsi}. A natural problem to consider would be the variant where the periodicity condition is replaced with an open or reflecting boundary condition and to investigate in how far the pre-wavefunction formalism carries through in those cases.

\item[Physical interpretation\hspace{15cm}] Both $\psi_{\bm \lambda}$ (by virtue of \rfc{prewavefnQNLS}) and $\Psi_{\bm \lambda}$ solve the QNLS eigenvalue problem \rfeqnser{QNLS1}{QNLS2}, although, of course, $\psi_{\bm \lambda}$ is not a physically acceptable solution, as it is not $S_N$-invariant and hence does not represent a bosonic state. 
A priori it is conceivable that the pre-wavefunction represents a state of $N$ non-identical particles of the same mass and with the same pairwise contact interaction. 
However, as discussed at the end of Subsection \ref{periodicity}, physically we are only interested in the subspace of $\L^2$ functions which vanish at the boundary of their domain. 
This leads to periodicity conditions $\psi_{\bm \lambda}|_{x_j=-L/2}=\psi_{\bm \lambda}|_{x_j=L/2} \, (=0)$, which we know cannot be satisfied already in the case $N=2$, cf. \rfex{prewavefnperiodicity}.

\item[Completeness, orthogonality and norm formulae\hspace{15cm}] As alluded to in the Introduction, the set $\set{\Psi_{\bm \lambda}}{\bm \lambda \in \bR^N_+ \n{ satisfies the BAEs \rf{BAE}}}$ is complete in $\ca H_N([-L/2,L/2])$, as follows from the work of Dorlas \cite[Thm.~3.1]{Dorlas}; in fact, it is an orthogonal basis for a dense subspace. 
As highlighted in the discussion following \rfcn{abcdcommrelsconj}, it would be helpful if we could establish a similar property for the set $\set{\psi_{\bm \lambda}}{\bm \lambda \in \bR^N_+ \n{ satisfies the BAEs \rf{BAE}}}$ or even the larger set $\set{\psi_{\bm \lambda}}{\bm \lambda \in \bC^N}$ in $\f h_N([-L/2,L/2])$. 
It would also be worthwhile to obtain formulae for the $\L^2$-norms $\| \psi_{\bm \lambda} \|$ and $\L^2$-inner products $\innerrnd{\psi_{\bm \lambda}}{\psi_{\bm \mu}}$. 
Initial calculations for the case $N=2$ would seem to suggest that if both $\bm \lambda$ and $\bm \mu$ satisfy the BAEs \rf{BAE} then $\innerrnd{\psi_{\bm \lambda}}{\psi_{\bm \mu}} = 0$ only if $\| \bm \lambda\| = \| \bm \mu \|$.

\item[Recursive construction\hspace{15cm}] Both the Bethe wavefunction (\rfd{Psi}) and the pre-wave-function (\rft{brecursion}) can be generated by a product of operators $B_\mu $ and $b^\pm_\mu$ acting on the reference state $\Vac$, respectively, although for the pre-wavefunction care must be taken with the ordering of the $b^\pm_\mu$. 
An interesting question is whether there are such recursions in terms of (explicit) integral formulae for Bethe wavefunctions and pre-wavefunctions corresponding to other boundary conditions, which ties in with Sklyanin's work \cite{Sklyanin1988} on the \emph{boundary Yang-Baxter equation}.

\item[Yang-Baxter algebra\hspace{15cm}] There is the notion of a Yang-Baxter algebra in both the symmetric and the non-symmetric context, and we have similar expressions for $a_\mu \psi_{\bm \lambda}$ and $A_\mu \Psi_{\bm \lambda}$, and $d_{\mu}\psi_{\bm \lambda}$ and $D_\mu \Psi_{\bm \lambda}$, using the regular representation of the dAHA. However, the $AD$-, $BC$-, $CB$- and $DA$-relations from \rfc{Tcommrelcor} cannot be directly generalized to relations for the non-symmetric integral operators $a,b^\pm,c^\pm,d$. 
Instead, we have the relations \rfeqnser{adcommrel}{dacommrel}. 
An open problem is the precise relation of the operators $a,b^\pm,c^\pm,d$ to the Yangian.
\end{description}

\clearpage

\appendix
\noappendicestocpagenum
\addappheadtotoc

\chapter[The dAHA: calculations]{The degenerate affine Hecke algebra: calculations}

\section{The regular representation} \label{dAHAregrepprops}

\subsection{Properties of the divided difference operators $\tilde \Delta_{j \, k}$} \label{divdiffopsprops}

Here we present several technical lemmas involving the operators $\tilde \Delta_{j \, k}$ which are used to give a proof of \rfp{dAHAregrep}, i.e. to demonstrate that the regular representation of the dAHA in momentum space is indeed a representation.
We believe it is important that this thesis is self-contained and therefore that these technical steps are included.

\begin{lem} \label{regrep4lem}
Let $1 \leq j \ne k \leq N$ and let $l=1,\ldots,N$.
Writing $\bar l = s_{j \, k} l$, we have 
\[ \tilde \Delta_{j \, k} \lambda_l - \lambda_{\bar l} \tilde \Delta_{j \, k} = \delta_{j \, l} - \delta_{k \, l} \in \End(\ca C^\omega(\bC^N)).\]
\end{lem}

\begin{proof} 
This follows from
\[ \tilde \Delta_{j \, k} \lambda_l - \lambda_{\bar l} \tilde \Delta_{j \, k} 
= \frac{\lambda_l - \lambda_{\bar l}  \tilde s_{j \, k} - \lambda_{\bar l} + \lambda_{\bar l}  \tilde s_{j \, k}}{\lambda_j - \lambda_k} 
= \frac{\lambda_l - \lambda_{\bar l}}{\lambda_j-\lambda_k} 
= \delta_{j \, l} - \delta_{k \, l}. \qedhere
\]
\end{proof}

\begin{lem} \label{regrep5lem}
If $j,k,l,m \in \{1,\ldots,N\}$ are distinct, then $\left[  \tilde s_{j \, k}, \tilde \Delta_{l \, m} \right] =  \left[ \tilde \Delta_{j \, k}, \tilde \Delta_{l \, m} \right] = 0$.
\end{lem}

\begin{proof}
This follows directly.
\end{proof}

\begin{lem} \label{regrep7lem}
Let $j,k,l \in \{1,\ldots,N\}$ be distinct.
Then $\left[\tilde \Delta_{j \, k},\tilde \Delta_{k \, l} \right] = \tilde s_{k \, l} \tilde \Delta_{j \, k} \tilde \Delta_{k \, l}  \tilde s_{j \, k}$.
\end{lem}

\begin{proof}
First note that
\begin{equation} 
\tilde \Delta_{j \, k} \tilde \Delta_{k \, l} = \frac{1 -  \tilde s_{j \, k}}{\lambda_j-\lambda_k} \frac{1 -  \tilde s_{k \, l}}{\lambda_k-\lambda_l)} = \frac{1}{\lambda_j-\lambda_k} \left( \frac{1 -  \tilde s_{k \, l}}{\lambda_k-\lambda_l} + \frac{ \tilde s_{j \, k}  \tilde s_{k \, l} -  \tilde s_{j \, k}}{\lambda_j-\lambda_l} \right) \label{regrepprod}.
\end{equation}
By swapping $j$ and $l$ and using $\tilde \Delta_{k \, j} = -\tilde \Delta_{j \, k}$, we have
\begin{equation} 
\tilde \Delta_{k \, l} \tilde \Delta_{j \, k} = \frac{1}{\lambda_k-\lambda_l} \left( \frac{1 -  \tilde s_{j \, k}}{\lambda_j-\lambda_k} + \frac{\tilde s_{k \, l} \tilde s_{j \, k} -  \tilde s_{k \, l}}{\lambda_j-\lambda_l} \right) \label{regrepprodrev}.
\end{equation}

When subtracting \rfeqn{regrepprodrev} from \rfeqn{regrepprod} we notice that the terms proportional to $1 \in \bC S_N$ cancel.
In particular, we find the following expression for the commutator:
\begin{align*}
[\tilde \Delta_{j \, k},\tilde \Delta_{k \, l}]  &= \frac{(\lambda_k-\lambda_l)^{-1} - (\lambda_j-\lambda_l)^{-1}}{\lambda_j-\lambda_k } \tilde s_{j \, k} - \frac{\tilde s_{k \, l}  \tilde s_{j \, k}}{(\lambda_k-\lambda_l)(\lambda_j-\lambda_l)} + \\
& \qquad + \frac{\tilde s_{j \, k}  \tilde s_{k \, l}}{(\lambda_j-\lambda_k) (\lambda_j-\lambda_l)}  + \frac{ (\lambda_j-\lambda_l)^{-1} - (\lambda_j-\lambda_k)^{-1}  }{\lambda_k-\lambda_l}\tilde s_{k \, l}.
\end{align*}
Making use of the fact that $a^{-1}-b^{-1} = (b-a)a^{-1}b^{-1}$ we see that
\begin{align*}
[\tilde \Delta_{j \, k},\tilde \Delta_{k \, l}] &= \frac{(\lambda_k-\lambda_l)^{-1} ( \tilde s_{j \, k} -  \tilde s_{k \, l}  \tilde s_{j \, k})+(\lambda_j-\lambda_k)^{-1} ( \tilde s_{j \, k}  \tilde s_{k \, l} -  \tilde s_{k \, l} )}{\lambda_j-\lambda_l} \\
&= \frac{ (\lambda_k-\lambda_l)^{-1}  \tilde s_{k \, l} ( \tilde s_{k \, l}  \tilde s_{j \, k} -  \tilde s_{j \, k}) + (\lambda_j-\lambda_k)^{-1}  \tilde s_{k \, l} ( \tilde s_{j \, k}  \tilde s_{k \, l}  \tilde s_{j \, k} - 1)}{\lambda_j-\lambda_l} \\
&= \tilde s_{k \, l} \frac{(\lambda_k-\lambda_l)^{-1}  (1 -  \tilde s_{k \, l})  + (\lambda_j-\lambda_l)^{-1}  ( \tilde s_{j \, k}  \tilde s_{k \, l} -  \tilde s_{j \, k} )}{\lambda_j-\lambda_k}  \tilde s_{j \, k}. 
\end{align*}
By comparing with \rfeqn{regrepprod}, we see that this equals $ \tilde s_{k \, l} \tilde \Delta_{j \, k} \tilde \Delta_{k \, l}  \tilde s_{j \, k}$.
\end{proof}
For the following lemma we no longer need to write out $\tilde \Delta_{j \, k}$ as a divided difference. 
We can simply make use of the previous lemmas and properties of the $\tilde s_{j \, k}$.
\begin{lem}
Let $j,k,l \in \{1,\ldots,N\}$ be distinct.
Then
\begin{align} 
 \tilde s_{j \, k} \tilde \Delta_{k \, l} \tilde s_{j \, k} &= \tilde s_{k \, l} \tilde \Delta_{j \, k} \tilde s_{k \, l},  \label{regrep6} \\
 \tilde s_{j \, k}  \tilde s_{k \, l} \tilde \Delta_{j \, k} &= \tilde \Delta_{k \, l}  \tilde s_{j \, k}  \tilde s_{k \, l}, \label{regrep6a} \\
\tilde \Delta_{k \, l}  \tilde s_{j \, k} \tilde \Delta_{k \, l} &= \tilde \Delta_{j \, k} \tilde \Delta_{k \, l}  \tilde s_{j \, k} +  \tilde s_{j \, k} \tilde \Delta_{k \, l} \tilde \Delta_{j \, k},  \label{regrep9}\\
\tilde \Delta_{j\, k} \tilde \Delta_{k \, l} \tilde \Delta_{j \, k} &= \tilde \Delta_{k \, l} \tilde \Delta_{j \, k} \tilde \Delta_{k \, l}. \label{regrep10}
\end{align}
\end{lem}

\begin{proof}
\rfeqn{regrep6} follows immediately by moving the leftmost transposition through the divided difference operator (both the left- and right-hand side equal $\tilde \Delta_{j \, l}$).
\rfeqn{regrep6a} is obtained from \rfeqn{regrep6} by left-multiplying by $\tilde s_{j \, k}$ and right-multiplying by $\tilde s_{k \, l}$. \\

Note that from \rfeqn{regrep6} and \rfl{regrep7lem} we obtain
\[ \tilde s_{j \, k} \tilde \Delta_{k \, l} \tilde s_{j \, k} \tilde \Delta_{k \, l} \tilde s_{j \, k} = \tilde s_{k \, l} \tilde \Delta_{j \, k} \tilde s_{k \, l} \tilde \Delta_{k \, l} \tilde s_{j \, k} = [\tilde \Delta_{j \, k},\tilde \Delta_{k \, l}]. \]
Now conjugate this by $\tilde s_{j\, k}$ to establish \rfeqn{regrep9}.\\

Finally, we use \rfl{regrep7lem} to obtain that the left-hand side of \rfeqn{regrep10} produces polynomials invariant in $k \leftrightarrow l$ and maps polynomials invariant in $k \leftrightarrow l$ to 0:
\begin{align*}
\tilde \Delta_{j \, k} \tilde \Delta_{k \, l} \tilde \Delta_{j \, k} &= [\tilde \Delta_{j \, k},\tilde \Delta_{k \, l}]\tilde \Delta_{j \, k} + \tilde \Delta_{k \, l} \tilde \Delta_{j \, k}^2 &&=  \tilde s_{k \, l} \tilde \Delta_{j \, k} \tilde \Delta_{k \, l}  \tilde s_{j \, k} \tilde \Delta_{j \, k} && =  \tilde s_{k \, l} \tilde \Delta_{j \, k} \tilde \Delta_{k \, l} \tilde \Delta_{j \, k}; \\
\tilde \Delta_{j \, k} \tilde \Delta_{k \, l} \tilde \Delta_{j \, k} &= \tilde \Delta_{j \, k} [\tilde \Delta_{k \, l},\tilde \Delta_{j \, k}] + \tilde \Delta_{j \, k}^2 \tilde \Delta_{k \, l} && =  \tilde \Delta_{j \, k}  \tilde s_{j \, k} \tilde \Delta_{k \, l} \tilde \Delta_{j \, k}  \tilde s_{k \, l} && = - \tilde \Delta_{j \, k} \tilde \Delta_{k \, l} \tilde \Delta_{j \, k}  \tilde s_{k \, l}.
\end{align*}
We obtain $\tilde \Delta_{j \, k} \tilde \Delta_{k \, l} \tilde \Delta_{j \, k} =$ $-\tilde s_{k \, l} \tilde \Delta_{j \, k} \tilde \Delta_{k \, l} \tilde \Delta_{j \, k}  \tilde s_{k \, l}$ and $\tilde \Delta_{k \, l} \tilde \Delta_{j \, k} \tilde \Delta_{k \, l} = $ $-  \tilde s_{j \, k}  \tilde \Delta_{k \, l} \tilde \Delta_{j \, k} \tilde \Delta_{k \, l}  \tilde s_{j \, k}$, by swapping $j$ and $l$. 
Hence in order to prove \rfeqn{regrep10} it suffices to show $ \tilde s_{k \, l} \tilde \Delta_{j \, k} \tilde \Delta_{k \, l} \tilde \Delta_{j \, k}  \tilde s_{k \, l}$ $=$ $ \tilde s_{j \, k}  \tilde \Delta_{k \, l} \tilde \Delta_{j \, k} \tilde \Delta_{k \, l}  \tilde s_{j \, k}$.
This is done by repeatedly applying \rfeqn{regrep6}:
\begin{align*}
\tilde s_{k \, l} \tilde \Delta_{j \, k} \tilde \Delta_{k \, l} \tilde \Delta_{j \, k}  \tilde s_{k \, l} &=  
-\tilde s_{k \, l} \tilde \Delta_{j \, k} \tilde s_{k \, l} \tilde \Delta_{k \, l} \tilde s_{k \, l} \tilde \Delta_{j \, k}  \tilde s_{k \, l} &&=  
-\tilde s_{j \, k} \tilde \Delta_{k \, l}  \tilde s_{j \, k} \tilde \Delta_{k \, l} \tilde s_{j \, k} \tilde \Delta_{k \, l}  \tilde s_{j \, k} \\
& =  -\tilde s_{j \, k} \tilde \Delta_{k \, l}  \tilde s_{k \, l} \tilde \Delta_{j \, k} \tilde s_{k \, l} \tilde \Delta_{k \, l}  \tilde s_{j \, k} 
&&= \tilde s_{j \, k} \tilde \Delta_{k \, l}  \tilde \Delta_{j \, k}  \tilde \Delta_{k \, l}  \tilde s_{j \, k}.  \qedhere
\end{align*}
\end{proof}

\subsection{Recursive relations involving the deformed permutations}

We present two lemmas that establish recursive relations involving the regular representation used in \rfp{regrepsymmetrizers} and \rfp{Psiinduction}.

\begin{lem} \label{regreplem1}
Let $\gamma \in \bR$, $\bm \lambda \in \bC^N$ and $\mu \in \bC \setminus \{ \lambda_1,\ldots,\lambda_N \}$.
We have
\[ \sum_{m=1}^N \tilde s_{m,\gamma} \ldots \tilde s_{N-1,\gamma} \frac{\ii \gamma}{\lambda_N-\mu} \tilde{\ca S}^{(N)} = \left( 1 - \tau^+_\mu(\bm \lambda) \right) \tilde{\ca S}^{(N)} \quad \in \End(\ca C^\omega(\bC^N_\n{reg})). \]
\end{lem}

\begin{proof}
By induction; the $N=1$ case follows from $\frac{\ii \gamma}{\lambda_1-\mu} =1- \frac{\lambda_1-\mu-\ii \gamma}{\lambda_1-\mu}$.
To establish the induction step we first remark that
\begin{align*} 
\tilde s_{N,\gamma} \frac{\ii \gamma}{\lambda_{N+1}-\mu} \tilde{\ca S}^{(N+1)} &= \left( \frac{-\ii \gamma}{\lambda_N-\lambda_{N+1}}  \frac{\ii \gamma}{\lambda_{N+1}-\mu}+\frac{\lambda_N-\lambda_{N+1}+\ii \gamma}{\lambda_N-\lambda_{N+1}} \frac{\ii \gamma}{\lambda_N-\mu} \tilde s_N \right) \tilde{\ca S}^{(N+1)} \\
&= 
\frac{\ii \gamma}{\lambda_N-\mu} \frac{\lambda_{N+1}-\mu-\ii \gamma}{\lambda_{N+1}-\mu} \tilde{\ca S}^{(N+1)}.
\end{align*}
Using this we see that
\begin{align*}
\lefteqn{\sum_{m=1}^{N+1} \tilde s_{m,\gamma} \ldots \tilde s_{N,\gamma} \frac{\ii \gamma}{\lambda_{N+1}-\mu} \tilde{\ca S}^{(N+1)} =}\\
&= \left( \sum_{m=1}^N \tilde s_{m,\gamma} \ldots \tilde s_{N,\gamma}  \frac{\ii \gamma}{\lambda_{N+1}-\mu} + \frac{\ii \gamma}{\lambda_{N+1}-\mu} \right) \tilde{\ca S}^{(N+1)} \\
&= \left( \sum_{m=1}^N \tilde s_{m,\gamma} \ldots \tilde s_{N-1,\gamma} \frac{\ii \gamma}{\lambda_N-\mu} \frac{\lambda_{N+1}-\mu-\ii \gamma}{\lambda_{N+1}-\mu} + \frac{\ii \gamma}{\lambda_{N+1}-\mu}  \right) \tilde{\ca S}^{(N+1)} \\
&= \left( \sum_{m=1}^N \tilde s_{m,\gamma} \ldots \tilde s_{N-1,\gamma}  \frac{\ii \gamma}{\lambda_N-\mu} \tilde{\ca S}^{(N)} \frac{\lambda_{N+1}-\mu-\ii \gamma}{\lambda_{N+1}-\mu} + \frac{\ii \gamma}{\lambda_{N+1}-\mu} \right) \tilde{\ca S}^{(N+1)} . 
\end{align*}
By virtue of the induction hypothesis this yields
\begin{align*} \sum_{m=1}^{N+1} \tilde s_{m,\gamma} \ldots \tilde s_{N,\gamma}  \frac{\ii \gamma}{\lambda_{N+1}-\mu} \tilde{\ca S}^{(N+1)} 
&= \left(  \left( 1-\tau^+_\mu(\bm \lambda) \right) \tilde{\ca S}^{(N)} \frac{\lambda_{N+1}-\mu-\ii \gamma}{\lambda_{N+1}-\mu} + \frac{\ii \gamma}{\lambda_{N+1}-\mu} \right) \tilde{\ca S}^{(N+1)} \\
&= \left(  \left( 1-\tau^+_\mu(\bm \lambda)  \right)  \frac{\lambda_{N+1}-\mu-\ii \gamma}{\lambda_{N+1}-\mu} + \frac{\ii \gamma}{\lambda_{N+1}-\mu}  \right) \tilde{\ca S}^{(N+1)} \\
&= \left( 1-\tau^+_{\mu}(\bm \lambda,\lambda_{N+1}) \right) \tilde{\ca S}^{(N+1)}. \qedhere
\end{align*}
\end{proof}

\begin{lem} \label{regreplem2}
Let $\gamma \in \bR$, $(\bm \lambda,\lambda_{N+1}) \in \bC^{N+1}$.
In $\End(\ca C^\omega(\bC^{N+1}))$ we have
\[ \sum_{m=1}^{N+1} \tilde s_m \ldots \tilde s_N  \tau^+_{\lambda_{N+1}}(\bm \lambda) \tilde{\ca S}^{(N)} = \sum_{m=1}^{N+1}  \tilde s_{m,\gamma} \ldots \tilde s_{N,\gamma}  \tilde{\ca S}^{(N)}  . \]
\end{lem}

\begin{proof}
We proceed by induction; note that in $\End(\ca C^\omega(\bC^N))$ we have
\begin{equation} \label{eqn35} \tilde s_{j,\gamma} = \tilde s_j - \frac{\ii \gamma }{\lambda_j-\lambda_{j+1}} (1-\tilde s_j) = -\frac{\ii \gamma}{\lambda_j-\lambda_{j+1}} + \tilde s_j \frac{\lambda_j-\lambda_{j+1}-\ii \gamma}{\lambda_j-\lambda_{j+1}}, \end{equation}
so that
\[ (1+\tilde s_1) \frac{\lambda_1-\lambda_2-\ii \gamma}{\lambda_1-\lambda_2} =  1+\tilde s_{1,\gamma} , \]
which proves the case $N=1$.
For the induction step write $\bm \lambda' = (\lambda_1,\ldots,\lambda_{N-1})$.
We have
\begin{align*}
\sum_{m=1}^{N+1} \tilde s_m \ldots \tilde s_N \tau^+_{\lambda_{N+1}}(\bm \lambda) \tilde{\ca S}^{(N)} 
&= \left( \left(\sum_{m=1}^{N} \tilde s_m \ldots \tilde s_{N-1} \right) \tilde s_N+1 \right) \tau^+_{\lambda_{N+1}}(\bm \lambda) \tilde{\ca S}^{(N)}  \\
&= \left( \sum_{m=1}^{N} \tilde s_m \ldots \tilde s_{N-1} \tau^+_{\lambda_N}(\bm \lambda') \tilde s_N \frac{\lambda_N-\lambda_{N+1}-\ii \gamma}{\lambda_N-\lambda_{N+1}} + \tau^+_{\lambda_{N+1}}(\bm \lambda)  \right) \tilde{\ca S}^{(N)}.
\end{align*}
Focusing on the first term we have $\sum_{m=1}^{N} \tilde s_m \ldots \tilde s_{N-1} \tau^+_{\lambda_N}(\bm \lambda') \tilde s_N \frac{\lambda_N-\lambda_{N+1}-\ii \gamma}{\lambda_N-\lambda_{N+1}} \tilde{\ca S}^{(N)} =$
\begin{align*}
&= \sum_{m=1}^{N} \tilde s_m \ldots \tilde s_{N-1} \tau^+_{\lambda_N}(\bm \lambda') \tilde{\ca S}^{(N-1)} \left( \tilde s_{N,\gamma} + \frac{\ii \gamma}{\lambda_N-\lambda_{N+1}} \right) \tilde{\ca S}^{(N)} \\
&= \sum_{m=1}^{N} \tilde s_{m,\gamma} \ldots \tilde s_{N-1,\gamma}  \tilde{\ca S}^{(N-1)} \left( \tilde s_{N,\gamma} + \frac{\ii \gamma}{\lambda_N-\lambda_{N+1}} \right) \tilde{\ca S}^{(N)} \\
&= \sum_{m=1}^{N}  \tilde s_{m,\gamma} \ldots \tilde s_{N-1,\gamma} \left( \tilde s_{N,\gamma} + \frac{\ii \gamma}{\lambda_N-\lambda_{N+1}} \right) \tilde{\ca S}^{(N)},
\end{align*}
where we have used \rfeqn{eqn35} and \rfl{symmetrizerrecursion}, as well as the induction hypothesis.
Assume for now that $\lambda_{N+1} \ne \lambda_j$ for $j=1,\ldots,N$
It follows that 
\begin{align*}
\sum_{m=1}^{N+1} \tilde s_m \ldots \tilde s_N \tau^+_{\lambda_{N\!+\!1}}(\bm \lambda) \tilde{\ca S}^{(N)} &= \left( \sum_{m=1}^{N} \tilde s_{m,\gamma} \ldots \tilde s_{N,\gamma}  + \tau^+_{\lambda_{N\!+\!1}}(\bm \lambda)  +\sum_{m=1}^{N} \tilde s_{m,\gamma} \ldots \tilde s_{N\!-\!1,\gamma}  \frac{\ii \gamma}{\lambda_N \! - \! \lambda_{N\!+\!1}} \right) \tilde{\ca S}^{(N)} \\
&= \left( \sum_{m=1}^{N}  \tilde s_{m,\gamma} \ldots \tilde s_{N,\gamma}  + 1 \right) \tilde{\ca S}^{(N)} \, = \, \sum_{m=1}^{N+1} \tilde s_{m,\gamma} \ldots \tilde s_{N,\gamma}  \tilde{\ca S}^{(N)},
\end{align*}
by virtue of \rfl{regreplem1} with $\mu = \lambda_{N+1}$, as required. 
To obtain this statement for $\lambda_{N+1} = \lambda_j$, take the limit $\lambda_{N+1} \to \lambda_j$, apply the De l'H\^opital's rule and use that differentiation preserves the space $\ca C^\omega(\bC^N)$.
\end{proof}

\subsection{The regular representation and the Yang-Baxter algebra}

The following two lemmas are used in Sections \ref{YBalgebradAHAsec} and \ref{abcdcommrelssec}.
\begin{lem} \label{regrep20}
We have 
\[ \left( 1+ \ii \gamma \tilde \Delta_{N \, N+1}\right) \ldots \left( 1+ \ii \gamma \tilde \Delta_{1 \, N+1}\right) = \tilde s_{N,\gamma} \ldots \tilde s_{1,\gamma} \tilde s_1 \ldots \tilde s_N. \]
\end{lem}

\begin{proof}
By induction on $N$; the case $N=1$ is trivially true. 
For the induction step we find
\begin{align*} 
\left( 1+ \ii \gamma \tilde \Delta_{N \, N \! + \! 1}\right) \ldots \left( 1+ \ii \gamma \tilde \Delta_{1 \, N+1}\right) 
&= \left(1+\ii \gamma \tilde \Delta_N \right) \tilde s_N \tilde s_N \left( 1+ \ii \gamma \tilde \Delta_{N-1 \, N+1}\right) \ldots \left( 1+ \ii \gamma \tilde \Delta_{1 \, N+1}\right) , \\
& =\left( \tilde s_N - \ii \gamma \Delta_N \right) \left( 1+ \ii \gamma \tilde \Delta_{N-1 \, N}\right) \ldots \left( 1+ \ii \gamma \tilde \Delta_{1 \, N}\right) \tilde s_N,  \\
&= \tilde s_{N,\gamma}  \tilde s_{N-1,\gamma} \ldots \tilde s_{1,\gamma} \tilde s_1 \ldots \tilde s_{N-1} \tilde s_N,
\end{align*}
where we have used \rfl{divdiffprops} \ref{divdiffequivariance}-\ref{divdiffnilpotency}, the definition of $\tilde s_{j,\gamma}$, and the induction hypothesis.
%
\end{proof}

\begin{lem} \label{regrep22}
Let $j,k,l \in \{1,\ldots,N\}$ be distinct. Then $\left[ \tilde \Delta_{j \,k}, \tilde \Delta_{j \, l} \right] (1+\tilde s_{k \, l}) =0$. 
\end{lem}

\begin{proof}
Note that
\begin{align*}
\tilde \Delta_{j \,k}\tilde \Delta_{j \, l} (1+\tilde s_{k \, l}) &= \frac{1}{\lambda_j-\lambda_k} \left( 1-\tilde s_{j \, k}\right)\frac{1}{\lambda_j-\lambda_l} \left( 1-\tilde s_{j \, l}\right)(1+\tilde s_{k \, l}) \\
&= \frac{1}{\lambda_j-\lambda_k} \left( \frac{1}{\lambda_j-\lambda_l} \left( 1-\tilde s_{j \, l}\right) -\frac{1}{\lambda_k-\lambda_l} \left( \tilde s_{j \, k}-\tilde s_{j \, l}\right) \right) (1+\tilde s_{k \, l}) \\ 
&= \left( \frac{1}{\lambda_j-\lambda_k} \frac{1}{\lambda_j-\lambda_l} -\frac{1}{\lambda_k-\lambda_l} \left( \frac{1}{\lambda_j-\lambda_k} \tilde s_{j \, k} - \frac{1}{\lambda_j-\lambda_l}\tilde s_{j \, l} \right) \right) (1+\tilde s_{k \, l}),
\end{align*}
which is clearly invariant under the swap $k \leftrightarrow l$.
\end{proof}

\section{Plane waves and the degenerate affine Hecke algebra}

Here we will show that the regular and integral representations are intimately related.

\begin{lem} \label{planewaveregintrep}
Let $\bm \lambda = (\lambda_1,\ldots,\lambda_N) \in \bC^N$. Then we have, for $1 \leq j \ne k \leq N$,
\begin{align}
s_{j \, k} \e^{\ii \bm \lambda} &= \tilde s_{j \, k} \e^{\ii \bm \lambda}, \label{planewaveregintrep0} \\
I_{j \, k} \e^{\ii \bm \lambda} &= - \ii \tilde \Delta_{j \, k} \e^{\ii \bm \lambda}. \label{planewaveregintrep1}
\end{align}
Furthermore, for $j=1,\ldots,N-1$ we have
\begin{equation} s_{j,\gamma} \e^{\ii \bm \lambda} = \tilde s_{j,\gamma} \e^{\ii \bm \lambda}. \label{planewaveregintrep2} \end{equation}
\end{lem}

\begin{proof}
\rfeqn{planewaveregintrep0} follows immediately.
Let $\bm x \in \bR^N$ be arbitrary.
For \rfeqn{planewaveregintrep1}, note that for $\bm \lambda \in \bC^N$ such that $\lambda_j \ne \lambda_k$,
\begin{align*}
\left(I_{j \, k} \e^{\ii \bm \lambda}\right)(\bm x) &= \e^{\ii \inner{\bm \lambda}{\bm x}} \int_0^{x_j-x_k} \dd y \e^{-\ii (\lambda_j-\lambda_k) y} &&= \frac{-\ii}{\lambda_j-\lambda_k} \left( \e^{\ii \inner{\bm \lambda}{\bm x}}-  \e^{\ii \inner{\bm \lambda}{\bm x-(x_j-x_k)(\bm e_j-\bm e_k)}}\right) \\
&= \frac{-\ii}{\lambda_j-\lambda_k} \left( \e^{\ii \inner{\bm \lambda}{\bm x}}-  \e^{\ii \inner{\bm \lambda}{s_{j \, k} \bm x}}\right)
&&= -\ii \frac{\e^{\ii \bm \lambda}(\bm x)-  \e^{\ii \tilde s_{j \, k} \bm \lambda}(\bm x)}{\lambda_j-\lambda_k} \; =  \; -\ii \tilde \Delta_{j \, k} \e^{\ii \bm \lambda}(\bm x), \end{align*}
as required.
In the case that $\lambda_j = \lambda_k$, we have
\[ \left( I_{j\, k} \e^{\ii \bm \lambda} \right) (\bm x) = \e^{\ii \inner{\bm \lambda}{\bm x}} \int_0^{x_j-x_k} \dd y = (x_j-x_k) \e^{\ii \inner{\bm \lambda}{\bm x}}; \]
on the other hand, $(-\ii \tilde \Delta_{j \, k} \e^{\ii \bm \lambda})(\bm x)$ is to be interpreted as
\begin{align*} 
\lim_{\mu \to 0} -\ii \tilde \Delta_{j \, k} \e^{\ii \inner{\bm \lambda + \mu(\bm e_j-\bm e_k)}{\bm x}} 
&= -\ii \lim_{\mu \to 0} \frac{\e^{\ii \inner{\bm \lambda+\mu(\bm e_j-\bm e_k)}{\bm x}} - \e^{\ii \inner{\bm \lambda-\mu(\bm e_j-\bm e_k)}{\bm x}} }{2 \mu} \\
&= -\ii \lim_{\mu \to 0} \frac{\e^{\ii \mu(x_j-x_k)} - \e^{-\ii \mu(x_j-x_k)} }{2 \mu} \e^{\ii \inner{\bm \lambda}{\bm x}} \\
&= \lim_{\mu \to 0} (x_j-x_k) \e^{\ii \mu(x_j-x_k)} \e^{\ii \inner{\bm \lambda}{\bm x}} = (x_j-x_k) \e^{\ii \inner{\bm \lambda}{\bm x}}, \end{align*}
by virtue of De l'H\^opital's rule, whence we obtain \rfeqn{planewaveregintrep1} for all $\bm \lambda \in \bC^N$. 
Finally, \rfeqn{planewaveregintrep2} follows from combining \rfeqn{planewaveregintrep0} and \rfeqn{planewaveregintrep1} with $k=j+1$.
\end{proof}

As a result, we have
\begin{cor} \label{planewaveregintrepcor}
Let $\bm \lambda = (\lambda_1,\ldots,\lambda_N) \in \bC^N$ and $w \in S_N$. Then 
\begin{align}
w \e^{\ii \bm \lambda} &= \tilde w^{-1} \e^{\ii \bm \lambda}, \label{planewaveregintrep4} \\
w_\gamma \e^{\ii \bm \lambda} &= \tilde w_\gamma^{-1} \e^{\ii \bm \lambda}. \label{planewaveregintrep5}
\end{align}
\end{cor}

\begin{proof}
We decompose $w$ as a product of $l$, say, simple reflections: $w = s_{i_1} \ldots s_{i_l}$.
Using \rfeqn{planewaveregintrep0} we find that
\[ w \e^{\ii \bm \lambda} = s_{i_1} \ldots s_{i_l} \e^{\ii \bm \lambda} = s_{i_1} \ldots s_{i_{l-1}} \tilde s_{i_l} \e^{\ii \bm \lambda} = \tilde s_{i_l} s_{i_1} \ldots s_{i_{l-1}} \e^{\ii \bm \lambda}, \]
since $[\tilde s_j,s_k]=0$ for all $i,j$.
Continuing this way, we obtain \rfeqn{planewaveregintrep4}:
\[ w \e^{\ii \bm \lambda} =\tilde s_{i_l} \ldots \tilde s_{i_1} \e^{\ii \bm \lambda} = \tilde w^{-1} \e^{\ii \bm \lambda}. \]
The proof of \rfeqn{planewaveregintrep5} is entirely analogous, relying on \rfeqn{planewaveregintrep2}.
\end{proof}

\section{The Dunkl-type representation} \label{Dunklrepprops}

Similar to Subsection \ref{divdiffopsprops}, we provide technical results needed to settle \rfp{Dunklrep}, i.e. to show that the Dunkl-type representation of the dAHA is in fact a representation.

\begin{lem} \label{Dunklrep1} 
Let $j =1, \ldots, N-1$ and $k = 1, \ldots, N$.
Write $\bar k = s_j k$.
Then 
\[ s_j \Lambda_k - \Lambda_{\bar k} s_j = - \left( \delta_{j \, k} - \delta_{j+1 \, k} \right).\]
\end{lem}

\begin{proof}
We have
\[ s_j \Lambda_k s_j - \Lambda_{\bar k} =  \sum_{l<k} \theta_{\bar k \, \bar l} s_{\bar k \, \bar l} -  \sum_{l<k} \theta_{\bar l \, \bar k}  s_{\bar k \, \bar l} - \sum_{l<\bar k} \theta_{\bar k \, l} s_{\bar k \, l} + \sum_{l>\bar k} \theta_{l \, \bar k}  s_{\bar k \, l} . \]
The terms with $l\ne j, j+1, k,\bar k$ have $\bar l=l$ and all cancel out. 
In the case that $j+1<k=\bar k$ the remaining terms cancel each other as well:
\[ \theta_{k \, j+1} s_{k \, j+1} + \theta_{k \, j} s_{k \, j} - \theta_{k \, j} s_{k \, j} - \theta_{k \, j+1} s_{k \, j+1}  = 0, \]
and similarly in the case that $k = \bar k <j$. Hence, only in the case $k = j,j+1$ is there any contribution; this is easily seen to equal $-\theta_{j+1 \, j} s_{j \, j+1}-\theta_{j \, j+1} s_{j \, j+1} = -s_j$ in the case that $k=j$; in the case $k=j+1$ it equals $\theta_{j+1 \, j} s_{j \, j+1}+\theta_{j \, j+1} s_{j \, j+1} = s_j$.
\end{proof}

From \rfl{Dunklrep1} we immediately obtain
\begin{cor} \label{Dunklrep2}
Let $j=1, \ldots, N-1$ and $k = 1, \ldots, N$.
Then $s_j \partial_{k,\gamma} - \partial_{s_j k,\gamma} s_j = \gamma \left( \delta_{j \, k} - \delta_{j + 1 \, k} \right)$. 
\end{cor}

\begin{lem} \label{Dunklrep3}
Let $j=1,\ldots,N-1$ and $k,l \in 1,\ldots,N$.
Then
\[ s_j \left[\partial_{k,\gamma},\partial_{l,\gamma}\right] = \left[\partial_{s_j(k),\gamma},\partial_{s_j(l),\gamma}\right] s_j. \]
\end{lem}

\begin{proof}
For all $k = 1,\ldots,N$ denote $\bar k:=s_j(k)$.
Repeatedly using \rfc{Dunklrep2}, we have
\begin{align*}
s_j \partial_{k,\gamma} \partial_{l,\gamma} &= \left( \partial_{\bar k,\gamma} s_j + \left( \delta_{j+1 \, k} - \delta_{j \, k} \right) \ii \gamma \right) \partial_{l,\gamma} \\
&= \partial_{\bar k,\gamma} \left( \partial_{\bar l,\gamma} s_j + \left( \delta_{j+1 \, l} - \delta_{j \, l} \right) \ii \gamma \right) + \left( \delta_{j+1 \, k} - \delta_{j \, k} \right) \ii \gamma \partial_{l,\gamma} \\
&= \partial_{\bar k,\gamma} \partial_{\bar l,\gamma} s_j + \ii \gamma \left( \left( \delta_{j+1 \, k} - \delta_{j \, k} \right) \partial_{l,\gamma} + \left( \delta_{j+1 \, l} - \delta_{j \, l} \right) \partial_{\bar k,\gamma} \right).
\end{align*}
By reversing the role of $k$ and $l$ we obtain
\[ s_j \partial_{l,\gamma} \partial_{k,\gamma} =\partial_{\bar l,\gamma} \partial_{\bar k,\gamma} s_j + \ii \gamma \left( \left( \delta_{j+1 \, l} - \delta_{j \, l} \right) \partial_{k,\gamma} + \left( \delta_{j+1 \, k} - \delta_{j \, k} \right) \partial_{\bar l,\gamma} \right), \]
so that
\[ s_j [\partial_{k,\gamma},\partial_{l,\gamma}] = [\partial_{\bar k,\gamma},\partial_{\bar l,\gamma}]s_j + \ii \gamma Y_{j;k,l}, \]
where
\[ Y_{j;k,l} = \left( \delta_{j+1 \, l} - \delta_{j \, l} \right) \left( \partial_{\bar k,\gamma} - \partial_{k,\gamma} \right) - \left( \delta_{j+1 \, k} - \delta_{j \, k} \right) \left(\partial_{\bar l,\gamma} - \partial_{l,\gamma} \right) . \]
It suffices to show that $Y_{j;k,l}=0$.
We can assume without loss of generality that $k<l$.
Hence we must be in either of three cases: $k \ne j,j+1$, or $l \ne j,j+1$, or $k=j,l=j+1$.
If $k \ne j,j+1$, $\partial_{\bar k,\gamma}=\partial_{k,\gamma}$ and $\delta_{j \, k}=\delta_{j+1 \, k}=0$, so that $Y_{j;k,l}=0$.
The argument for $l \ne j,j+1$ is identical.
Finally, if $k=j,l=j+1$ we have
\[ Y_{j;k,l} = (1-0)(\partial_{j+1,\gamma}-\partial_{j,\gamma})-(0-1)(\partial_{j,\gamma}-\partial_{j+1,\gamma}) = \partial_{j+1,\gamma}-\partial_{j,\gamma}+\partial_{j,\gamma}-\partial_{j+1,\gamma} = 0.  \qedhere \]
\end{proof}

\begin{lem} \label{Dunklrep4}
Let $w \in S_N$ and $k,l = 1, \ldots, N$.
Then 
\[ w [\partial_{k,\gamma},\partial_{l,\gamma}] = [\partial_{w(k),\gamma},\partial_{w(l),\gamma}] w. \]
\end{lem}

\begin{proof}
Write $w = s_{i_1} \ldots s_{i_l}$ for some positive integer $l$.
We have by induction
\begin{align*} 
w [\partial_{k,\gamma},\partial_{l,\gamma}] &= s_{i_1} \ldots s_{i_l} [\partial_{k,\gamma},\partial_{l,\gamma}] \; = \; s_{i_1} \ldots s_{i_{l-1}} [\partial_{s_{i_l}(k),\gamma},\partial_{s_{i_l}(l),\gamma}] s_{i_l} \; = \; \ldots \; = \\
&= [\partial_{s_{i_1} \ldots s_{i_l}(k),\gamma},\partial_{s_{i_1} \ldots s_{i_l}(l),\gamma}] s_{i_1} \ldots s_{i_l} \; = \; [\partial_{w(k),\gamma},\partial_{w(l),\gamma}] w. \qedhere
\end{align*}
\end{proof}

\begin{lem} \label{Dunklrep5}
Let $k,l=1,\ldots,N$. 
Then $[\partial_{k,\gamma},\partial_{l,\gamma}]=0$.
\end{lem}

\begin{proof}
Let $w \in S_N$.
It is sufficient to prove that $\left( [\partial_{k,\gamma},\partial_{l,\gamma}]f \right)(\bm x)=0$ for all $f \in \ca C^\infty(\bR^N_\n{reg})$ and all $\bm x \in w^{-1} \bR^N_+$.
By virtue of \rfl{Dunklrep4} and \rfeqn{Dunklreprestr}
\begin{align*} 
\left([\partial_{k,\gamma},\partial_{l,\gamma}]f\right)(\bm x) &= \left( w^{-1} [\partial_{w(k),\gamma},\partial_{w(l),\gamma}]  w f\right)(\bm x) && = \left([\partial_{w(k),\gamma},\partial_{w(l),\gamma}] w f\right)(w \bm x) \\
&=  \left([\partial_{w(k)},\partial_{w(l)}] w f\right)(w \bm x) && = 0. \qedhere
\end{align*}
\end{proof}

\newpage

\chapter[The non-symmetric YBA: calculations]{The non-symmetric Yang-Baxter algebra: calculations} \label{nonsymmintopscalcs}

This part of the appendix lists useful properties of the non-symmetric integral operators $a_\lambda $, $b^\pm_\lambda $, $c^\pm_\lambda $ and $d_{\lambda}$. Throughout it, we assume that $\gamma \in \bR$ and $L\in \bR_{>0}$. As usual we denote $J=[-L/2,L/2]$.

\section{Properties of the elementary integral operators}

\begin{lem}\label{eopsadjointness}
Let $n=0,\ldots,N$, $\bm i \in \f i^n_{N}$ and $\lambda \in \bC$. Then
$\hat e^\pm_{\lambda;\bm i}$ is the formal adjoint of $\check e^\mp_{\bar \lambda;\bm i}$ and $\bar e^+_{\lambda;\bm i}$ is the formal adjoint of $\bar e^-_{\bar \lambda;\bm i}$.
\end{lem}

\begin{proof}
These statements are proven by changing the order of integration in $\innerrnd{\cdot}{\cdot}$.
For example, to show that $\hat e^+_{\lambda;\bm i}$ is the formal adjoint of $\check e^-_{\bar \lambda;\bm i}$ it is sufficient to prove, for arbitrary $f \in \f h_N$, $g \in \f h_{N+1}$,
\begin{equation} \label{eqn36} 
\int_{J^{N+1}} \dd^{N+1} \bm x \left(\hat e^+_{\lambda;\bm i} f \right)(\bm x) \overline{g(\bm x)} = 
\int_{J^N} \dd^N \bm x f(\bm x) \overline{\left( \check e^-_{\bar \lambda;\bm i} g \right)(\bm x)}. 
\end{equation}
Write $i_{n+1}=1$.
The left-hand side of \rfeqn{eqn36} is given by
\begin{align*} 
\lefteqn{\int_{J^{N+1}} \dd^{N+1} \bm x \theta(x_{i_1+1}>\ldots>x_{i_n+1}>x_1) \left( \prod_{m=1}^n \int_{x_{i_{m+1}+1}}^{x_{i_m}+1} \dd y_m \right) \cdot} \\
& \hspace{40mm} \cdot \e^{\ii \lambda (y_1+\ldots+y_n-x_{i_1+1}-\ldots-x_{i_n+1}-x_1)}  (\phi_{\bm i_+ \to \bm y} \hat \phi_1  f)(\bm x) \overline{g(\bm x)} \displaybreak[2] \\
&= \left( \prod_{j=2 \atop \forall m\, j \ne i_m+1}^{N+1} \int_J \dd x_j \right)  \left( \prod_{m=1}^n \int_J \dd y_m \right)\left( \prod_{m=1}^{n+1} \int_J \dd z_m \right)  \theta(z_1>y_1>\ldots>z_n>y_n>z_{n+1})  \cdot\\
& \hspace{40mm} \cdot \e^{\ii \lambda (y_1+\ldots+y_n-z_1-\ldots-z_n-z_{n+1})} (\phi_{\bm i_+ \to \bm y} \hat \phi_1  f)(\bm x) \overline{(\phi_{1 \to z_{n+1}, \bm i_+ \to \bm z}g)(\bm x)},
\end{align*} 
where we have relabelled $x_{i_1+1},\ldots,x_{i_n+1}$ as $z_1,\ldots,z_n$ and $x_1$ as $z_{n+1}$. We will now relabel $y_1,\ldots,y_n$ as $x_{i_1+1},\ldots,x_{i_n+1}$ to obtain the expression
\begin{align*} 
\lefteqn{\hspace{-5mm} \left( \prod_{j=2}^{N+1} \int_J \dd x_j \right) \overline{(\hat \phi_1 f)(\bm x)} \theta(x_{i_1+1}>\ldots>x_{i_n+1}) \int_{x_{i_1+1}}^{L/2} \! \! \dd z_1 \left( \prod_{m=2}^{n} \int_{x_{i_m+1}}^{x_{i_{m-1}+1}} \! \! \dd z_m \right) \int_{-L/2}^{x_{i_n+1}} \! \! \dd z_{n+1} \cdot} \\
& \hspace{40mm} \cdot \overline{\e^{\ii \bar \lambda (x_{i_1}+\ldots+x_{i_n}-z_1-\ldots-z_n-z_{n+1})} (\phi_{\bm i_+ \to \bm z}g)(z_{n+1},x_2,\ldots,x_{N+1})} \\
&= \left( \prod_{j=1}^N \int_J \dd x_j \right) \overline{f(x_1,\ldots,x_N)} \theta(x_{i_1}>\ldots>x_{i_n}) \left( \prod_{m=1}^{n+1} \int_{x_{i_m}}^{x_{i_{m-1}}} \dd z_m \right) \cdot \\
& \hspace{40mm} \cdot \overline{\e^{\ii \bar\lambda (x_{i_1}+\ldots+x_{i_n}-z_1-\ldots-z_n-z_{n+1})}  (\phi_{\bm i \to \bm z} \check \phi^+_{z_{n+1}} g)(x_1,\ldots,x_N)},
\end{align*}
where we have relabelled $(x_2,\ldots,x_{N+1})\to(x_1,\ldots,x_N)$.
This equals the right-hand side of \rfeqn{eqn36}. The other statements are proven in a similar way (without the last relabelling of the $x_j$).
\end{proof}

\begin{lem} \label{eopsbounded}
Let $n=0,\ldots,N$, $\bm i \in \f i^n_{N}$ and $\lambda \in \bC$.
Then $\hat e^\pm_{\lambda;\bm i}$, $\bar e^\pm_{\lambda;\bm i}$ and $\check e^-_{\lambda; \bm i}$ are all bounded on $\f h_\n{fin}$.
\end{lem}

\begin{proof}
We prove the statement for $\hat e^-_{\lambda; \bm i}$; the other proofs are along the same lines (and by virtue of \rfl{eopsadjointness} only proofs for three out of the six operators need to be given). 
Let $\bm x \in J^{N+1}$. For $\bm y \in J^n$ we have the estimate
\begin{align} 
\hspace{-3mm}  \left| \e^{\ii \lambda \left( x_{N+1} + \sum_{m=1}^n \left( x_{i_m}-y_m\right)\right)} \right|^2 &=  \e^{-\Im \lambda \left( x_{N+1} + \sum_{m=1}^n \left( x_{i_m}-y_m\right)\right)} \nonumber \\
& \leq  \e^{|\Im \lambda| \left( |x_{N+1}| + \sum_{m=1}^n \left( |x_{i_m}|+|y_m|\right)\right)}  
\leq  \e^{(2n+1)|\Im \lambda|L}  \label{eqnB1}.
\end{align}
Using this for $f \in \f h_N$ we obtain   
\begin{align*}
\lefteqn{ \left| \left( \hat e^-_{\lambda; \bm i}f \right) (\bm x)\right|^2 =}\\
&=
\left| \int_{J^n} \dd^n \bm y \e^{\ii \lambda \left( x_{N+1} + \sum_{m=1}^n \left( x_{i_m}-y_m\right)\right)} \theta(x_{N+1}>y_1>x_{i_1}>\ldots>y_n>x_{i_n}) \left( \rho_{\bm i \to \bm y} f \right) (\bm x') \right|^2 
\\
&\leq \left( \int_{J^n} \dd^n \bm y \left| \e^{\ii \lambda \left( x_{N+1} + \sum_{m=1}^n \left( x_{i_m}-y_m\right)\right)} \right|^2 \theta(x_{N+1}>y_1>x_{i_1}>\ldots>y_n>x_{i_n}) \right)  \cdot \\
& \qquad \cdot
\left( \int_{J^n} \dd^n \bm y \left| \left( \rho_{\bm i \to \bm y} f \right) (\bm x')  \right|^2 \right) \\
&\leq  \e^{(2n+1)|\Im \lambda|L}  \frac{L^n}{n!}
\int_{J^n} \dd^n \bm y \left| \left( \rho_{\bm i \to \bm y} f \right) (\bm x')  \right|^2,
\end{align*}
where $\bm x' = (x_1,\ldots,x_N)$ and we have used standard inequalities for absolute values of integrals and integrals of products with Lebesgue-integrable integrand.\\

For the norm of $\hat e^-_{\lambda; \bm i} f$ we have
\begin{align*}
\| \hat e^-_{\lambda; \bm i} f \|_{N+1}^2 &= \int_{J^{N+1}} \dd^{N+1} \bm x \left| \left( \hat e^-_{\lambda; \bm i}f \right) (\bm x)\right|^2 \\
&\leq \frac{L^n}{n!} \e^{(2n+1)|\Im \lambda|L} \left( \int_{J^{n+1}} \dd^{n+1} \bm y \right) \int_{J^N} \dd^N \bm x \left|f(\bm x)  \right|^2 \; = \; \frac{L^{2n+1}}{n!} \e^{(2n+1)|\Im \lambda|L} \| f\|^2_N,
\end{align*}
i.e.
\[ \| \hat e^-_{\lambda; \bm i} f \|_{N+1} \leq \frac{L^{n+\frac{1}{2}}}{\sqrt{n!}} \e^{(n+1)|\Im \lambda|L} \| f\|_N. \]
It follows that $\hat e^-_{\lambda; \bm i}$ is bounded. 
\end{proof}

Given $w \in S_N$, denote by $w_+$ the element of $S_{N+1}$ determined by
\begin{equation} w_+(1) =1, \qquad w_+(j+1)=w(j)+1 \n{ for } j=1,\ldots, N. \end{equation}%
\nc{rwl}{$w_+$}{Shifted permutation \nomrefeqpage}%
\vspace{-10mm}
\begin{lem}\label{eopspermutation}
Let $n=0,\ldots,N$, $\bm i \in \f i^n_{N}$, $\lambda \in \bC$ and $w \in S_N$.
Then
\begin{align}
w \hat e^-_{\lambda; \bm i} &= \hat e^-_{\lambda;w \bm i}w \hspace{-20mm} && \in \Hom(\f h_N,\f h_{N+1})
\label{ehatmin} \\
w_+ \hat e^+_{\lambda;\bm i} &= \hat e^+_{\lambda;w \bm i}w \hspace{-20mm} && \in \Hom(\f h_N,\f h_{N+1})
\label{ehatplus} \\
w \check e^+_{\lambda;\bm i} &= \check e^+_{\lambda;w \bm i}w \hspace{-20mm} && \in \Hom(\f h_{N+1},\f h_N) \label{echeckplus} \\
w \check e^-_{\lambda; \bm i} &= \check e^-_{\lambda;w \bm i}w_+ \hspace{-20mm} && \in \Hom(\f h_{N+1},\f h_N) \label{echeckmin} \\
w \bar e^\pm_{\lambda;\bm i} &= \bar e^\pm_{\lambda;w \bm i}w \hspace{-20mm} && \in \End(\f h_N) \label{ebar}.
\end{align}
\end{lem}

\begin{proof}
Let $\bm x \in J^{N+1}$. By virtue of \rfeqn{eqn43} we have
\begin{align*}
w \hat e^-_{\lambda; \bm i}  
&= \e^{\ii \lambda x_{N+1}} w \left( \prod_{m=1}^n \theta(x_{i_{m-1}}-x_{i_m}) \int_{x_{i_m}}^{x_{i_{m-1}}} \dd y_m \e^{\ii \lambda (x_{i_m} -y_m)} \right) \phi_{\bm i \to \bm y} \hat \phi_{N+1} \\
&= \e^{\ii \lambda x_{N+1}} \left( \prod_{m=1}^n \theta(x_{w i_{m-1}}-x_{w i_m}) \int_{x_{w i_m}}^{x_{w i_{m-1}}} \dd y_m \e^{\ii \lambda (x_{w i_m} -y_m)} \right) \phi_{w \bm i \to \bm y} \hat \phi_{N+1} w \; = \; \hat e^-_{\lambda;w \bm i} w,
\end{align*}
where $x_{i_0} = x_{N+1}$, which proves \rfeqn{ehatmin}. 
\rfeqn{echeckplus} can be proven by taking adjoints; indeed, from
\begin{align*} 
\innerrnd{f}{\check e^+_{\lambda;\bm i} w^{-1} g}_{N} &= \innerrnd{\hat e^-_{\bar \lambda;\bm i} f}{w^{-1} g}_{N+1}
&&=  \innerrnd{w \hat e^-_{\bar \lambda;\bm i} f}{g}_{N+1} \\
&= \innerrnd{\hat e^-_{\bar \lambda;w \bm i} w f}{g}_{N+1} 
&&= \innerrnd{w f}{\check e^+_{\lambda;w \bm i} g}_{N} \; = \; \innerrnd{f}{w^{-1} \check e^+_{\lambda;w \bm i} g}_{N}, 
\end{align*}
where $f \in \f h_N$, $g \in \f h_{N+1}$, 
we infer that $\check e^+_{\lambda;\bm i} w^{-1} = w^{-1} \check e^+_{\lambda;w \bm i}$, i.e. $w \check e^+_{\lambda;\bm i}= \check e^+_{\lambda;w \bm i} w$.
The other equations are proven analogously.
\end{proof}

\begin{lem} \label{eopssymmetric}
Let $n=0,\ldots,N-1$, $\bm i \in \f i^n_{N-1}$ and $\lambda \in \bC$. We have
\[ \check e^+_{\lambda;\bm i}|_{\ca H_N} = \check e^-_{\lambda; \bm i}|_{\ca H_N}. \]
\end{lem}

\begin{proof}
The desired statement follows from the observation
that for $F \in \ca H_N$ and $\bm x \in \bR^{N-1}$ both $\left( \check e^\pm_{\lambda;\bm i} F \right)(\bm x)$ are equal to
\[ \e^{\ii \lambda \sum_{m=1}^{n-1} x_{i_m}} \left( \prod_{m=1}^{n+1} \theta(x_{i_{m-1}}-x_{i_m}) \int_{x_{i_m}}^{x_{i_{m-1}}} \dd y_m \e^{-\ii \lambda y_m} \right) F(\bm x_{\hat{\bm \imath}},\bm y). \qedhere \]
\end{proof}

\begin{lem} \label{eopsrestr}
Let $n=0,\ldots,N$, $\bm i \in \f i^n_{N}$ and $\lambda \in \bC$.
Then
\begin{align*}
\hat e^\pm_{\lambda;\bm i}F|_{J^{N+1}_+} & = \begin{cases} \hat E_{\lambda;\bm i}F|_{J^{N+1}_+}, & \bm i \in \f I^n_N, \\ 0, & \n{otherwise}, \end{cases} && \n{for } F \in \ca H_N, \\
\bar e^\pm_{\lambda;\bm i}F|_{J^N_+} &= \begin{cases} \bar E^\pm_{\lambda;\bm i}F|_{J^N_+}, & \bm i \in \f I^n_{N}, \\ 0, & \n{otherwise}, \end{cases} && \n{for } F \in \ca H_N, \\
\check e^\pm_{\lambda;\bm i}F|_{J^N_+} &= \begin{cases} \check E_{\lambda;\bm i}F|_{J^N_+}, & \bm i \in \f I^n_N, \\ 0, & \n{otherwise}, \end{cases}  && \n{for } F \in \ca H_{N+1}.
\end{align*}
\end{lem}

\begin{proof}
Concerning the statement for $\bar e^\pm_{\lambda;\bm i}$, note that for $\bm x \in J^N_+$, 
\[ \prod_{m=1}^n \theta(x_{i_m}-x_{i_{m+1}}) =\prod_{m=1}^n \theta(x_{i_{m-1}}-x_{i_m})= \begin{cases} 1, & \n{if } i_1<\ldots<i_n, \\ 0, & \n{otherwise}; \end{cases} \]
this yields that $\bar e^\pm_{\lambda;\bm i}F$ restricted to the alcove vanishes unless $i_1>\ldots>i_n$, in which case one obtains the equality with the $\bar E^\pm_{\lambda;\bm i}$ by applying \rfd{eopsdefn} and \rfd{Eopsdefn} straightforwardly. The statements for $\hat e^\pm_{\lambda;\bm i}$ and $\check e^\pm_{\lambda;\bm i}$ follow in a similar manner.
\end{proof}

\section[The non-symmetric particle creation operators and the Dunkl-type operators]{The non-symmetric particle creation operators and the \\ Dunkl-type operators} 
\sectionmark{The non-symm. creation operators and Dunkl-type operators}
\label{nonsymmintopscommrels}

Here we aim to provide auxiliary results from which \rft{bDunkl} can be proven.
It is apparent that we need to study commutation relations of the operators $\hat e^\pm_{\lambda; \bm i}$ (which can be seen as linear operators: $\ca C^\infty(\bR^N_\n{reg}) \to \ca C^\infty(\bR^{N+1}_\n{reg})$) with the partial differential operators $\partial_j$ and the step functions $\theta_{k \, l}$. 
Before we look at these commutators, we establish some useful results.

\begin{lem} \label{intrule1}
Let $h: \bR \to \bC$ be integrable; let $x_0,x \in \bR$ and let $\lambda \in \bC$.
We have
\[ [\partial_1, \int_{x_1}^{x_0} \dd y_1 \e^{\ii \lambda (x_0+x_1-y_1)} \phi_{1 \to y_1}] h(x_1) = -\e^{\ii \lambda x_1}h(x_0). \]
\end{lem}

\begin{proof}
Apply the Leibniz integral rule and integrate by parts.
\end{proof}

\begin{lem} \label{intrule2}
Let $h: \bR^2  \to \bC$ be integrable; let $x_0,x_1,x_2 \in \bR$ and let $\lambda \in \bC$.
We have
\begin{align*} 
\lefteqn{[ \partial_1, \int_{x_1}^{x_0} \dd y_1 \int_{x_2}^{x_1} \dd y_2 \e^{\ii \lambda (x_0+x_1+x_2-y_1-y_2)} \phi_{1 \to y_1}] h(x_1,y_2) }\\
 & \qquad = \int_{x_1}^{x_0} \dd y_1 \e^{\ii \lambda(x_0+x_2-y_1)} h(y_1,x_1)  - \int_{x_2}^{x_1} \dd y_2 \e^{\ii \lambda(x_1+x_2-y_2)} h(x_0,y_2).
\end{align*}
\end{lem}
 
\begin{proof}
The desired expression equals
\begin{align*}
\lefteqn{ \int_{x_1}^{x_0} \dd y_1 \e^{\ii \lambda(x_0+x_1-y_1)} \partial_{x_1} \int_{x_2}^{x_1} \dd y_2 \e^{\ii \lambda (x_2-y_2)} h(y_1,y_2) +} \\
& \qquad \qquad +  \int_{x_2}^{x_1} \dd y_2 \e^{\ii \lambda(x_2-y_2)} \left[ \partial_{x_1}, \int_{x_1}^{x_0} \dd y_1 \e^{\ii \lambda (x_0+x_1-y_1)} \phi_{1 \to y_1} \right] h(x_1,y_2).
\end{align*}
By virtue of \rfl{intrule1} and the Leibniz integral rule once more we obtain the result.
\end{proof}

\subsection{The operators $\hat e^\pm_{\lambda;\bm i}$}

In the rest of this section, we assume that $\lambda \in \bC$ and we suppress it in the notation: $\hat e^\pm_{\bm i}=\hat e^\pm_{\lambda;\bm i}$ and $b^\pm_{\lambda;n}=b^\pm_n$. We will make statements for both $\hat e^+_{\bm i}$ and $\hat e^-_{\bm i}$; in general we will provide detailed proofs for $\hat e^-_{\bm i}$ and indicate how the proof is modified for $\hat e^+_{\bm i}$.

\begin{lem} \label{hatelem0}
Let $\bm i \in \f i^n_N$ and $j=1,\ldots, N$ such that none of the $i_m$ equals $j$.
Then
\[ \partial_j^{(N\! +\! 1)} \hat e^-_{\bm i} = \hat e^-_{\bm i}  \partial_j^{(N)} , \qquad  \partial_{j\! +\! 1}^{(N\! +\! 1)} \hat e^+_{\bm i} = \hat e^+_{\bm i}  \partial_j^{(N)}. \]
\end{lem}

\begin{proof}
This follows immediately from the definitions of $\hat e^\pm_{\bm i}$.
\end{proof}

Conversely, if one of the $i_m$ equals $j$, the commutators $\partial_j^{(N\! +\! 1)} \hat e^-_{\bm i} - \hat e^-_{\bm i}  \partial_j^{(N)} $ and $\partial_{j\! +\! 1}^{(N\! +\! 1)} \hat e^+_{\bm i} - \hat e^+_{\bm i}  \partial_j^{(N)}$ are nonzero.
First of all, we deal with the case that $n=1$.
\begin{lem} \label{hatelem1}
Let $j=1,\ldots,N$.
Then
\[ \partial_j^{(N\! +\! 1)} \hat e^-_j - \hat e^-_j \partial_j^{(N)} = -\theta_{N+1 \, j}^{(N\! +\! 1)} s_{j \, N+1}^{(N\! +\! 1)} \hat e^-, \qquad \partial_{j\! +\! 1}^{(N\! +\! 1)} \hat e^+_{j} - \hat e^+_{j} \partial_j^{(N)} = \theta_{j \! +\! 1 \, 1}^{(N\! +\! 1)} s_{1 \, j \! +\! 1}^{(N\! +\! 1)} \hat e^+. \]
\end{lem}

\begin{proof} This follows immediately upon applying \rfl{intrule1}. \end{proof}

The following lemmas deal with the case $n>1$.

\begin{lem} \label{hatelem2}
Let $n=1,\ldots,N$.
\begin{itemize}
\item Let $(j,\bm i) \in \f i^n_N$. Then 
\[ \partial_j^{(N\! +\! 1)} \hat e^-_{j \, \bm i} - \hat e^-_{j \, \bm i} \partial_j^{(N)} = \theta_{j \, i_1}^{(N\! +\! 1)} \hat e^-_{\bm i} \theta_{i_1 \, j}^{(N)} s_{i_1 \, j}^{(N)} -  \theta_{N+1 \, j}^{(N\! +\! 1)} s_{j \, N+1}^{(N\! +\! 1)}  \hat e^-_{\bm i} .\]
\item Let $(\bm i,j) \in \f i^n_N$. Then
\[ \partial_{j \! +\! 1}^{(N\! +\! 1)} \hat e^+_{\bm i \, j} - \hat e^+_{\bm i \, j}  \partial_j^{(N)} = - \theta_{i_{n-1}+1 \, j \! +\! 1}^{(N\! +\! 1)} \hat e^+_{\bm i} \theta_{j \, i_{n-1}}^{(N)} s_{i_{n-1} \, j}^{(N)} + \theta_{j \! +\! 1 \, 1}^{(N\! +\! 1)} s_{1 \, j \! +\! 1}^{(N\! +\! 1)} \hat e^+_{\bm i}. \]
\end{itemize}
\end{lem}

\begin{proof}
We prove the statement for $\hat e^-$. The proof for $\hat e^+$ goes completely analogously.
Let $\bm x = (x_1,\ldots,x_{N+1}) \in \bR^{N+1}_\n{reg}$. Write $\bm x' = (x_1,\ldots,x_N)$, $\bm y = (y_1,\ldots,y_{n-1})$, and $i_0 = j$. assume $x_{N+1}>x_j>x_{i_1}>\ldots>x_{i_{n-1}}$.
Using \rfl{intrule2} we obtain
\begin{align*} 
\lefteqn{\partial_j^{(N\! +\! 1)} \hat e^-_{j \, \bm i} - \hat e^-_{j \, \bm i} \partial_j^{(N)} =} \displaybreak[2] \\
&=  \left[ \partial_j, \int_{x_j}^{x_{N\!+\!1}} \dd Y \int_{x_{i_1}}^{x_j} \dd y_1\e^{\ii \lambda(x_{N\!+\!1}+x_j+x_{i_1}-Y-y_1)} \phi_{j \to Y} 
\right]  \left( \prod_{m=2}^{n-1} \int_{x_{i_m}}^{x_{i_{m\!-\!1}}} \dd y_m \e^{\ii \lambda (x_{i_m}-y_m)} \right) \phi_{\bm i \to \bm y} \hat \phi_{N\!+\!1}  \displaybreak[2]  \\
&= \left( \int_{x_j}^{x_{N+1}} \dd Y \e^{\ii \lambda (x_{N+1}+x_{i_1}-Y)} \left( \prod_{m=2}^{n-1} \int_{x_{i_m}}^{x_{i_{m-1}}} \dd y_m \e^{\ii \lambda (x_{i_m}-y_m)} \right)  \phi_{\bm i \to (x_j,y_2,\ldots,y_{n-1})} \hat \phi_{N+1} + \right. \\
& \left. \qquad - \int_{x_{i_1}}^{x_j} \dd y_1 \e^{\ii \lambda(x_j+x_{i_1}-y_1)} \left( \prod_{m=2}^{n-1} \int_{x_{i_m}}^{x_{i_{m-1}}} \dd y_m \e^{\ii \lambda (x_{i_m}-y_m)} \right) 
\phi_{\bm i \to \bm y, j \to x_{N+1}} \hat \phi_{N+1} \right) \\
&=    \left( \int_{x_j}^{x_{N+1}} \dd y_1 \e^{\ii \lambda (x_{N+1}+x_{i_1}-y_1)} \left( \prod_{m=2}^{n-1} \int_{x_{i_m}}^{x_{i_{m-1}}} \dd y_m \e^{\ii \lambda (x_{i_m}-y_m)} \right) s_{i_1 \, j} \phi_{ (j,i_2,\ldots,i_{n-1}) \to \bm y} \hat \phi_{N+1} + \right. \\
& \left. \qquad - s_{j \, N+1} \e^{\ii \lambda x_{N+1}} \int_{x_{i_1}}^{x_{N+1}} \dd y_1 \e^{\ii \lambda(x_{i_1}-y_1)} \left( \prod_{m=2}^{n-1} \int_{x_{i_m}}^{x_{i_{m-1}}} \dd y_m \e^{\ii \lambda (x_{i_m}-y_m)} \right) \phi_{\bm i \to \bm y} \hat \phi_{N+1} \right).
\end{align*}
We can re-write the second term as $-s_{j \, N+1}\hat e^-_{\bm i}$; hence it suffices to show
\begin{align*}
\hat e^-_{\bm i} \theta_{i_1 \, j} s_{i_1 \, j} &= \int_{x_j}^{x_{N+1}} \dd y_1 \e^{\ii \lambda (x_{N+1}+x_{i_1}-y_1)} \left( \prod_{m=2}^{n-1} \int_{x_{i_m}}^{x_{i_{m-1}}} \dd y_m \e^{\ii \lambda (x_{i_m}-y_m)} \right)  \phi_{\bm i \to \bm y} \hat \phi_{N+1} s_{i_1 \, j}.
\end{align*}
This can be established by writing $\int_{x_j}^{x_{N+1}} \dd y_1 = \int_{x_{i_1}}^{x_{N+1}} \dd y_1 \theta(y-x_j)$.
\end{proof}

\begin{lem} \label{hatelem3}
Let $n=1,\ldots,N$.
\begin{itemize}
\item Let $(\bm i,j) \in \f i^n_N$ and $\bm i' = (i_1,\ldots,i_{n-2})$. Then
\[ \partial_j^{(N\! +\! 1)} \hat e^-_{\bm i \, j} - \hat e^-_{\bm i \, j} \partial_j^{(N)} = 
- \theta_{i_{n-1} \, j}^{(N\! +\! 1)} \hat e^-_{\bm i' \, j} \theta_{j \, i_{n-1}}^{(N)} s_{i_{n-1} j}^{(N)}. \]
\item Let $(j, \bm i) \in \f i^n_N$ and $\bm i' = (i_2,\ldots,i_{n-1})$. Then
\[ \partial_{j \! +\! 1}^{(N\! +\! 1)} \hat e^+_{j \, \bm i} - \hat e^+_{j \, \bm i} \partial_j^{(N)} = \theta_{j \! +\! 1 \, i_1+1}^{(N\! +\! 1)} \hat e^+_{j \, \bm i'} \theta_{i_1 \, j}^{(N)} s_{i_1 \, j}^{(N)}. \]
\end{itemize}
\end{lem}

\begin{proof}
Again, we note that the proof for the statement for $\hat e^+$ is analogous to the following proof for $\hat e^-$.
Let $\bm x = (x_1,\ldots,x_{N+1}) \in \bR^{N+1}_\n{reg}$. Write $\bm y = (y_1,\ldots,y_{n-1})$ and $i_0 = N+1$, and assume that $x_{N+1}>x_{i_1}>\ldots>x_{i_{n-1}}>x_j$.
Note that $\left[\partial_j,\hat e^-_{\bm i \, j}\right] =$
\begin{align*}
&= \e^{\ii \lambda x_{N+1}} \left( \prod_{m=1}^{n-1} \int_{x_{i_m}}^{x_{i_{m-1}}} \dd y_m \e^{\ii \lambda (x_{i_m}-y_m)} \right) \left[ \partial_j, \int_{x_j}^{x_{i_{n-1}}} \dd Y \e^{\ii \lambda (x_j-Y)} \phi_{j \to Y} \right] \phi_{\bm i \to \bm y} \hat \phi_{N+1} \\
&= - \e^{\ii \lambda x_{N+1}} \left( \prod_{m=1}^{n-1} \int_{x_{i_m}}^{x_{i_{m-1}}} \dd y_m \e^{\ii \lambda (x_{i_m}-y_m)} \right) \e^{\ii \lambda(x_j-x_{i_{n-1}})}  \phi_{(\bm i',j) \to \bm y} \hat \phi_{N+1}s_{i_{n-1},j},
\end{align*}
by virtue of \rfl{intrule1}. Using $\int_{x_{i_{n-1}}}^{x_{i_{n-2}}} \dd y_{n-1}= \int_{x_j}^{x_{i_{n-2}}} \dd y_{n-1} \theta(y_{n-1}-x_{i_{n-1}})$ and $\theta_{}$ we obtain the result.
\end{proof}

\begin{lem} \label{hatelem4}
Let $n=1,\ldots,N$, $l=2,\ldots,n-1$ and $(\bm i,j,\bm k) \in \f i^n_N$ such that $\bm i \in \f i^{l-1}_N$.
\begin{itemize}
\item Write $\bm i' = (i_1,\ldots,i_{l-2})$.
Then
\[ \partial_j^{(N\! +\! 1)} \hat e^-_{\bm i \, j \, \bm k} - \hat e^-_{\bm i \, j \, \bm k} \partial_j^{(N)} = \theta_{j \, k_1}^{(N\! +\! 1)} \hat e^-_{\bm i \, \bm k} \theta_{k_1 \, j}^{(N)}  s_{j \, k_1}^{(N)} - \theta_{i_{l-1} \, j}^{(N\! +\! 1)} \hat e^-_{\bm i' \, j \bm k} \theta_{j \, i_{l-1}}^{(N)} s_{i_{l-1} \, j}^{(N)}. \]
\item Write $\bm k' = (k_2,\ldots,k_{n-l})$.
Then
\[  \partial_{j \! +\! 1}^{(N\! +\! 1)} \hat e^+_{\bm i \, j \, \bm k} - \hat e^+_{\bm i \, j \, \bm k} \partial_j^{(N)}
 = -\theta_{i_{l-1}+1 \, j \! +\! 1}^{(N\! +\! 1)} \hat e^+_{\bm i \, \bm k} \theta_{j \, i_{l-1}}^{(N)}  s_{i_{l_1} \, j}^{(N)} + \theta_{j \! +\! 1 \, k_1 \! +\! 1}^{(N\! +\! 1)} \hat e^+_{\bm i \, j \bm k'} \theta_{k_1 \,j}^{(N)} s_{j \, k_1}^{(N)}.  \]
\end{itemize}
\end{lem}

\begin{proof}
Let $\bm x = (x_1,\ldots,x_{N+1}) \in \bR^{N+1}_\n{reg}$. Write $\bm y = (y_1,\ldots,y_{l-1})$, $\bm y' = (y_1,\ldots,y_{l-2})$, $\bm z = (z_1,\ldots,z_{n-l})$ and $\bm z' = (z_2,\ldots,z_{n-l})$. Also write $i_0 = N+1$, $k_0=j$, $\bm i' = (i_1,\ldots,i_{l-2})$, $\bm k' = (k_2,\ldots,k_{n-l})$ and assume that $x_{N+1}>x_{i_1}>\ldots>x_{i_{l-1}}>x_j>x_{k_1}>\ldots>x_{k_{n-l}}$.
We apply \rfl{intrule2} to obtain
\begin{align*}
\left[\partial_j,\hat e^-_{\bm i \, j \, \bm k}\right]  &= \e^{\ii \lambda x_{N+1} }\left( \prod_{m=1}^{l-1} \int_{x_{i_m}}^{x_{i_{m-1}}} \dd y_m \e^{\ii \lambda (x_{i_m}-y_m)}\right) \left( \prod_{m=2}^{n-l} \int_{x_{k_m}}^{x_{k_{m-1}}} \dd z_m \e^{\ii \lambda (x_{k_m}-z_m)}\right) \cdot \\
& \qquad \cdot \left[ \partial_j,\int_{x_j}^{x_{i_{l-1}}} \dd Y  \int_{x_{k_1}}^{x_j} \dd z_1 \e^{\ii \lambda(x_j+x_{k_1}-Y-z_1)} \phi_{j \to Y} \right] \phi_{\bm i \to \bm y,\bm k \to \bm z} \hat \phi_{N+1} \displaybreak[2] \\
&= \e^{\ii \lambda x_{N+1}} \left( \prod_{m=1}^{l-1} \int_{x_{i_m}}^{x_{i_{m-1}}} \dd y_m \e^{\ii \lambda (x_{i_m}-y_m)}\right) \left( \prod_{m=2}^{n-l} \int_{x_{k_m}}^{x_{k_{m-1}}} \dd z_m \e^{\ii \lambda (x_{k_m}-z_m)}\right) \cdot \\
& \qquad \cdot \left( \int_{x_j}^{x_{i_{l-1}}} \dd z_1 \e^{\ii \lambda (x_{k_1}-z_1)} 
\phi_{\bm i \to \bm y, \, j \to z_1, \, k_1 \to x_j, \, \bm k' \to \bm z'} \hat \phi_{N+1} + \right. \\
& \qquad \qquad \left.  - \int_{x_{k_1}}^{x_j} \dd z_1 \e^{\ii \lambda(x_j+x_{k_1}-x_{i_{l-1}}-z_1)} \phi_{\bm i \to \bm y, \, j \to x_{i_{l-1}}, \bm k \to \bm z} \hat \phi_{N+1} \right)  \displaybreak[2]  \\
&= \e^{\ii \lambda x_{N+1}} \left( \left( \prod_{m=1}^{l-1} \int_{x_{i_m}}^{x_{i_{m-1}}} \dd y_m \e^{\ii \lambda (x_{i_m}-y_m)}\right) \left( \prod_{m=2}^{n-l} \int_{x_{k_m}}^{x_{k_{m-1}}} \dd z_m \e^{\ii \lambda (x_{k_m}-z_m)}\right) \cdot \right. \\
& \qquad \qquad \cdot  \int_{x_j}^{x_{i_{l-1}}} \dd z_1 \e^{\ii \lambda (x_{k_1}-z_1)} \phi_{\bm i \to \bm y, \, \bm k \to \bm z} \hat \phi_{N+1} s_{j \, k_1} + \\
& \qquad -  \left( \prod_{m=1}^{l-2} \int_{x_{i_m}}^{x_{i_{m-1}}} \dd y_m \e^{\ii \lambda (x_{i_m}-y_m)}\right) \left( \prod_{m=1}^{n-l} \int_{x_{k_m}}^{x_{k_{m-1}}} \dd z_m \e^{\ii \lambda (x_{k_m}-z_m)}\right) \cdot \\
& \left. \qquad \qquad \cdot  \int_{x_{i_l}}^{x_{i_{l-2}}} \dd y_l \e^{\ii \lambda(x_j-y_l)} \phi_{\bm i' \to \bm y', \, j \to y_{l-1}, \bm k \to \bm z} \hat \phi_{N+1} s_{i_{l-1} \, j} \right).
\end{align*}
For the first term in square brackets, write $\int_{x_j}^{x_{i_{l-1}}} \dd z_1 = \int_{x_{k_1}}^{x_{i_{l-1}}} \dd z_1 \theta(z_1-x_j)$, and for the second, write $\int_{x_{i_{l-1}}}^{x_{i_{l-2}}} \dd y_{l-1} = \int_{x_j}^{x_{i_{l-2}}} \dd y_{l-1} \theta(y_l-x_{i_l})$. 
This gives the desired result for $\hat e^-$; the proof for $\hat e^+$ is along the same lines.
\end{proof}

\begin{lem} \label{hatelem5}
Let $n=0,\ldots,N$ and $\bm i \in \f i^n_N$.
For distinct positive integers $j,k$ not exceeding $N$ we have
\begin{align} 
\theta_{j \, k}^{(N\! +\! 1)} \hat e^-_{\bm i} -  \hat e^-_{\bm i} \theta_{j \, k}^{(N)} &= \theta_{j \, k}^{(N\! +\! 1)} \hat e^-_{\bm i} \theta_{k \, j}^{(N)}, \label{eqnB4} \\
\theta_{j \! +\! 1 \, k \! +\! 1}^{(N\! +\! 1)} \hat e^+_{\bm i}-\hat e^+_{\bm i} \theta_{j \, k}^{(N)} &= \theta_{j \! +\! 1 \, k \! +\! 1}^{(N\! +\! 1)} \hat e^+_{\bm i} \theta_{k \, j}^{(N)}. \label{eqnB5} 
\end{align}
In particular $\theta_{j \, k}^{(N\! +\! 1)} \hat e^-_{\bm i} -  \hat e^-_{\bm i} \theta_{j \, k}^{(N)}=\theta_{j \! +\! 1 \, k \! +\! 1}^{(N\! +\! 1)} \hat e^+_{\bm i}-\hat e^+_{\bm i} \theta_{j \, k}^{(N)} =0$ if no $i_l$ equals $j$ or $k$, or if $i_l=j$ and $i_m=k$ for some $l,m$.
\end{lem}

\begin{proof}
We prove the statements for $\hat e^-_{\bm i}$;  the proof for the statements for $\hat e^+_{\bm i}$ goes entirely analogously.\\

In the case that $j \ne i_l \ne k$ for all $l$, we immediately have $\theta_{j \, k}^{(N\! +\! 1)} \hat e^-_{\bm i} = \hat e^-_{\bm i} \theta_{j \, k}^{(N)}$.
Also, if $j=i_l, k=i_m$, say, then both $\theta_{j \, k}^{(N\! +\! 1)} \hat e^-_{\bm i}$ and $\hat e^-_{\bm i} \theta_{j \, k}^{(N)}$ vanish if $l>m$ and are equal to $\hat e^-_{\bm i}$ if $l<m$.
This is obvious for $\theta_{j \, k}^{(N\! +\! 1)} \hat e^-_{\bm i}$, and for $\hat e^-_{\bm i} \theta_{j \, k}^{(N)}$ it follows from the definition of $\hat e^-_{\bm i}$ where $y_l<y_m$ precisely if $x_{i_l}<x_{i_m}$. In particular, we have $\theta_{j \, k}^{(N\! +\! 1)} \hat e^-_{\bm i} = \hat e^-_{\bm i} \theta_{j \, k}^{(N)}$ in this case, as well. \\

In the remaining case, $j \ne i_m$ for all $m$ and $k=i_l$, say. The situation with $j$ and $k$ swapped goes analogously.
\rfeqn{eqnB4} is equivalent to $\hat e^-_{\bm i} \theta_{j \, i_l}^{(N)} = \theta_{j \, i_l}^{(N\! +\! 1)} \hat e^-_{\bm i} - \theta_{j \, i_l}^{(N\! +\! 1)} \hat e^-_{\bm i} \theta_{i_l \, j}^{(N)} = \theta_{j \, i_l}^{(N\! +\! 1)} \hat e^-_{\bm i} \theta_{j \, i_l}^{(N)}$, which is true since none of the $i_m$ equals $j$ and $x_j > x_{i_l}$ is implied by $x_j > y_l$ in the integration in the definition of $\hat e^-_{\bm i}$.
\end{proof}

\subsection{The operators $b^\pm_{\lambda;n}$}

The next step in building the operators $b^\pm_\lambda $ are the $b^\pm_{\lambda;n}$, and we list some useful results about these operators here.
For $N \in \bZ_{\geq 0}$, $j=1,\ldots,N$, and $n = 0, \ldots, N$ we introduce the notation
\begin{equation} \f i^n_N(j) = \set{\bm i \in \f i^n_N}{\forall m \, i_m \ne j}. \end{equation}%
\nc{rilwb}{$\f i^n_N(j)$}{Shorthand for $\set{\bm i \in \f i^n_N}{\forall m \, i_m \ne j}$ \nomrefeqpage}%
\vspace{-10mm}
\begin{lem} \label{bnlem1}
Let $n=0,\ldots,N-1$.
Then 
\begin{align}
(\partial_{N+1}^{(N\! +\! 1)}-\ii \lambda)b^-_{(n+1)} &= \Lambda^{(N\! +\! 1)}_{N+1} b^-_{(n)}, \label{eqnB11}\\
(\partial_{1}^{(N\! +\! 1)}-\ii \lambda)b^+_{(n+1)} &= \Lambda^{(N\! +\! 1)}_{1} b^+_{(n)}. \label{eqnB12} 
\end{align}
\end{lem}

\begin{proof}
Let $\bm i \in \f i^{n+1}_N$, $\bm x = (x_1,\ldots,x_{N+1}) \in \bR^{N+1}_\n{reg}$ and $\bm y = (y_1,\ldots,y_{n+1})$.
To prove \rfeqn{eqnB11}, write $\bm i' = (i_2,\ldots,i_{n+1})$, $\bm x' = (x_1,\ldots,x_N)$ and $\bm y' = (y_2,\ldots,y_{n+1})$ and note that
\begin{align*} 
\partial_{N+1}^{(N\! +\! 1)} \hat e^-_{\bm i}  &= 
\ii \lambda \hat e^-_{\bm i}  + \e^{\ii \lambda x_{i_1}} \theta_{N+1 \, \bm i}^{(N\! +\! 1)} s_{i_1 \, N+1}^{(N\! +\! 1)} 
\left( \prod_{m=2}^{n+1} \int_{x_{i_{m}}}^{x_{i_{m-1}}} \dd y_m \e^{\ii \lambda(x_{i_{m}}-y_m)}\right) \phi_{\bm i' \to \bm y'} \hat \phi_{N+1} \\
&= \ii \lambda \hat e^-_{\bm i}  + \theta_{N+1 \, i_1}^{(N\! +\! 1)} s_{i_1 \, N+1}^{(N\! +\! 1)} \hat e^-_{\bm i'} .
\end{align*}
Summing over $\bm i$ gives
\[ (\partial_{N+1}^{(N\! +\! 1)}-\ii \lambda) b^-_{(n+1)} = \sum_{\bm i \in \f i_N^{n+1}} \theta_{N+1 \, i_1}^{(N\! +\! 1)}s_{i_1 \, N+1}^{(N\! +\! 1)} \hat e^-_{\bm i'} = \sum_{j=1}^N \theta_{N+1 \, j}^{(N\! +\! 1)} s_{j \, N+1}^{(N\! +\! 1)} \sum_{\bm i \in \f i^n_N(j)}  \hat e^-_{\bm i}  =  \Lambda^{(N\! +\! 1)}_{N+1} b^-_{(n)}, \]
where we have used that $\theta_{N+1 \, j}^{(N\! +\! 1)} s_{j \, N+1}^{(N\! +\! 1)} \hat e^-_{\bm i} =0$ if one of the $i_m$ equals $j$.
As for \rfeqn{eqnB12}, a similar argument applies, where we write $\bm i' = (i_1,\ldots,i_n)$, $\bm x' = (x_2,\ldots,x_{N+1})$ and $\bm y' = (y_1,\ldots,y_n)$.
We obtain
\[\partial_1^{(N\! +\! 1)} \hat e^+_{\bm i} = \ii \lambda \hat e^+_{\bm i} - \theta_{i_{n+1} \, 1}^{(N\! +\! 1)} s_{1 \, i_{n+1}}^{(N\! +\! 1)} \hat e^+_{\bm i'} ; \]
hence summing over $\bm i$ yields, as required,
\[ (\partial_1^{(N\! +\! 1)}-\ii \lambda) b^+_{(n+1)} = - \sum_{j=2}^{N+1} \theta^{(N\! +\! 1)}_{j \, 1} s^{(N\! +\! 1)}_{1 \, j} \sum_{\bm i \in \f i^n_N(j)}  \hat e^+_{\bm i} = \Lambda^{(N\! +\! 1)}_1 b^+_{(n)}. \qedhere \]
\end{proof}

\begin{lem} \label{bnlem2}
We have $\Lambda^{(N\! +\! 1)}_{N+1} b^-_{(N)} = \Lambda^{(N\! +\! 1)}_{1} b^+_{(N)}  = 0$.
\end{lem}

\begin{proof}
Writing $\bm x = (x_1,\ldots,x_{N+1})$, we note that 
\[ b^-_{(N)} =\sum_{\bm i \in \f i_N^N} \hat e^-_{\bm i} = \sum_{w \in S_N} \hat e^-_{w(1) \ldots w(N)}. \]
Note that $\hat e^-_{\bm i}$ is nonzero only if $x_{N+1}>x_m$ for all $m=1,\ldots,N$.
Therefore, for any $\bm i \in \f i_N^n$ and any $j=1,\ldots,N$ we have $\theta_{N+1 \, j}^{(N\! +\! 1)} s_{j \, N+1}^{(N\! +\! 1)} \hat e^-_{\bm i} = 0$.
Summing over $j$ and $\bm i$ then proves the lemma. A similar argument may be made for $\Lambda^{(N\! +\! 1)}_{1} b^+_{(N)}$.
\end{proof}

\begin{lem} \label{bnlem3}
Let $n=0,\ldots,N-1$ and $j=1,\ldots,N$.
Then
\begin{align*} 
\partial_j^{(N\! +\! 1)} b^-_{(n+1)} - b^-_{(n+1)} \partial_j^{(N)}  &= 
- \theta_{N+1 \, j}^{(N\! +\! 1)} s_{j \, N+1}^{(N\! +\! 1)} b^-_{(n)} + \sum_{l=1}^n \sum_{\bm i \in \f i^n_N(j)}\left( \theta_{j \, i_l}^{(N\! +\! 1)} \hat e^-_{\bm i} \theta_{i_l \, j}^{(N)} s_{i_l \, j}^{(N)} - s_{i_l \, j}^{(N\! +\! 1)} \theta_{j \, i_l}^{(N\! +\! 1)} \hat e^-_{\bm i} \theta_{i_l \, j}^{(N)} \right), \\
\partial_{j \! +\! 1}^{(N\! +\! 1)} b^+_{(n+1)} - b^+_{(n+1)} \partial_j^{(N)} &= 
\theta_{j \! +\! 1 \, 1}^{(N\! +\! 1)} s_{1 \, j \! +\! 1}^{(N\! +\! 1)} b^+_{(n)}  - \sum_{l=1}^{n} \sum_{\bm i \in \f i^n_N(j)}  \left( \theta^{(N\! +\! 1)}_{i_l \! +\! 1 \, j \! +\! 1} \hat e^+_{\bm i} \theta^{(N)}_{j \, i_l} s^{(N)}_{i_l \, j} - s^{(N\! +\! 1)}_{i_l \! +\! 1 \, j \! +\! 1} \theta^{(N\! +\! 1)}_{i_l \! +\! 1 \, j \! +\! 1} \hat e^+_{\bm i} \theta^{(N)}_{j \, i_l} \right).\end{align*}
\end{lem}

\begin{proof}
Let $j=1,\ldots,N$. First we deal with the case $n=0$. \rflser{hatelem0}{hatelem1} give us
\begin{align*} 
\partial_j^{(N\! +\! 1)} b^-_{(1)} - b^-_{(1)} \partial_j^{(N)} &= \sum_{k=1}^N \left( \partial_j^{(N\! +\! 1)} \hat e^-_k - \hat e^-_k \partial_j^{(N)} \right) && = \partial_j^{(N\! +\! 1)} \hat e^-_j - \hat e^-_j \partial_j^{(N)} \\
& = -\theta^{(N\! +\! 1)}_{N+1 \, j} s^{(N\! +\! 1)}_{j \, N+1}\hat e^- &&= -\theta_{N+1\, j}^{(N\! +\! 1)}s_{j \, N+1}^{(N\! +\! 1)} b^-_{(0)}, 
\end{align*}
and similarly
\[ \partial_{j \! +\! 1}^{(N\! +\! 1)} b^+_{(1)} - b^+_{(1)} \partial_j^{(N)} = \partial_{j \! +\! 1}^{(N\! +\! 1)} \hat e^+_j - \hat e^+_j \partial_j^{(N)} = \theta_{j \! +\! 1 \, 1}^{(N\! +\! 1)} s_{1 \, j \! +\! 1}^{(N\! +\! 1)} b^+_{(0)}. \]
Since the summations over $l$ in the equations in the lemma vanish for $n=0$ the results follow.\\

For $n>0$ we have
\begin{align*} 
\partial_j^{(N\! +\! 1)} b^-_{(n+1)} - b^-_{(n+1)} \partial_j^{(N)} &= \sum_{\bm i \in \f i_N^{n+1}} \left( \partial_j^{(N\! +\! 1)} \hat e^-_{\bm i} - \hat e^-_{\bm i} \partial_j^{(N)} \right) \quad = \sum_{\bm i \in \f i_N^{n+1} \atop \exists l: \, i_l=j} \left( 
\partial_j^{(N\! +\! 1)} \hat e^-_{\bm i} - \hat e^-_{\bm i} \partial_j^{(N)} \right) \\
&= \sum_{\bm i \in \f i^n_N(j)} \sum_{l=1}^{n+1} \left( \partial_j^{(N\! +\! 1)} \hat e^-_{i_1 \ldots i_{l-1} \, j \, i_l \ldots i_n} - \hat e^-_{i_1 \ldots i_{l-1} \, j \, i_l \ldots i_n} \partial_j^{(N)} \right). 
\end{align*}
Let $\bm i \in \f i_N^{n}(j)$.
Using \rflser{hatelem2}{hatelem4} we have
\begin{align*}
\lefteqn{\sum_{l=1}^{n+1} \left( \partial_j^{(N\! +\! 1)} \hat e^-_{i_1 \ldots i_{l-1} \, j \, i_l \ldots i_n} - \hat e^-_{i_1 \ldots i_{l-1} \, j \, i_l \ldots i_n} \partial_j^{(N)} \right) = } \\
&= -\theta_{N+1,j}^{(N\! +\! 1)} s_{j \, N+1}^{(N\! +\! 1)} \hat e^-_{\bm i}+  \theta_{j \, i_1}^{(N\! +\! 1)} \hat e^-_{\bm i} \theta_{i_1 \, j}^{(N)} s_{i_1 \,j}^{(N)}+ \\
& \quad + \sum_{l=2}^n  \left( \theta_{j \, i_l}^{(N\! +\! 1)} \hat e^-_{\bm i} \theta_{i_l \, j}^{(N)} s_{i_l \, j}^{(N)}
 - \theta_{i_{l-1} \, j}^{(N\! +\! 1)} \hat e^-_{i_1 \ldots i_{l-2} \, j \, i_l \, \ldots \, i_n} \theta_{j \, i_{l-1}}^{(N)} s_{i_{l-1} \, j}^{(N)} \right)  -  \theta_{i_n \, j}^{(N\! +\! 1)} \hat e^-_{i_1 \ldots i_{n-1} j} \theta_{j \, i_n}^{(N)} s_{i_n \, j}^{(N)}\\
&= -\theta_{N+1 \, j}^{(N\! +\! 1)} s_{j \, N+1}^{(N\! +\! 1)} \hat e^-_{\bm i} + \sum_{l=1}^n \theta_{j \, i_l}^{(N\! +\! 1)} \hat e^-_{\bm i} \theta_{i_l \, j}^{(N)} s_{i_l \, j}^{(N)} - \sum_{l=2}^{n+1} \theta_{i_{l-1} \, j}^{(N\! +\! 1)} \hat e^-_{i_1 \ldots i_{l-2} \, j \, i_l \, \ldots \, i_n} \theta_{j \, i_{l-1}}^{(N)} s_{i_{l-1} \, j}^{(N)}.
\end{align*}
Note that by virtue of \rfeqn{ehatmin} the second summation over $l$ can be written as
\[ - \sum_{l=1}^n \theta_{i_l \, j}^{(N\! +\! 1)} \hat e^-_{i_1 \ldots i_{l-1} \, j \, i_{l+1} \, \ldots \, i_n} \theta_{j \, i_l}^{(N)} s_{i_l \, j}^{(N)} =  
- \sum_{l=1}^n s_{i_l \, j}^{(N\! +\! 1)} \theta_{j \,  i_l}^{(N\! +\! 1)} \hat e^-_{\bm i} \theta_{i_l \, j}^{(N)} . \]
Combining this, we obtain that
\begin{gather*} \sum_{l=1}^{n+1} \left( \partial_j^{(N\! +\! 1)} \hat e^-_{i_1 \ldots i_{l-1} \, j \, i_l \ldots i_n} - \hat e^-_{i_1 \ldots i_{l-1} \, j \, i_l \ldots i_n} \partial_j^{(N)} \right)  = \hspace{40mm} \\
\hspace{25mm} = -\theta_{N+1 \, j}^{(N\! +\! 1)} s_{j \, N+1}^{(N\! +\! 1)} \hat e^-_{\bm i} + \sum_{l=1}^n \left( \theta_{j \, i_l}^{(N\! +\! 1)} \hat e^-_{\bm i} \theta_{i_l \, j}^{(N)} s_{i_l \, j}^{(N)} - s_{i_l \, j}^{(N\! +\! 1)}  \theta_{j \, i_l}^{(N\! +\! 1)} \hat e^-_{\bm i} \theta_{i_l \, j}^{(N)} \right). \end{gather*}
Summing over all $\bm i \in \f i^n_N(j)$ we find that
\begin{align*} 
\partial_j^{(N\! +\! 1)} b^-_{(n+1)} - b^-_{(n+1)} \partial_j^{(N)} &= -\theta_{N+1 \, j}^{(N\! +\! 1)} s_{j \, N+1}^{(N\! +\! 1)} \sum_{\bm i \in \f i^n_N(j)} \hat e^-_{\bm i}  + \\
& \qquad + \sum_{l=1}^n \sum_{\bm i \in \f i^n_N(j)} \left( \theta_{j \, i_l}^{(N\! +\! 1)} \hat e^-_{\bm i} \theta_{i_l \, j}^{(N)} s_{i_l \, j}^{(N)} - s_{i_l \, j}^{(N\! +\! 1)}  \theta_{j \, i_l}^{(N\! +\! 1)} \hat e^-_{\bm i} \theta_{i_l \, j}^{(N)} \right). \end{align*}
Now finally note that for $\bm i \in \f i^n_N$ such that $i_m = j$ for some $m$, we have
\begin{equation} \label{eqn453} \theta_{N+1 \, i_m}^{(N\! +\! 1)}s_{j \, N+1}^{(N\! +\! 1)} \hat e^-_{\bm i} = \theta_{i_m+1 \, 1}^{(N\! +\! 1)}s_{1 \, j \! +\! 1}^{(N\! +\! 1)} \hat e^+_{\bm i} = 0, \end{equation}
since $\theta_{N+1 \, i_m}^{(N\! +\! 1)} s_{i_m \, N+1}^{(N\! +\! 1)} \theta_{N+1 \, \bm i}^{(N\! +\! 1)} = \theta_{i_m+1 \, 1}^{(N\! +\! 1)} s_{1 \, i_m+1}^{(N\! +\! 1)} \theta_{\bm i_+ \, 1}^{(N\! +\! 1)} = 0$. 
This means that
\[ \theta_{N+1 \, j}^{(N\! +\! 1)} s_{j \, N+1}^{(N\! +\! 1)} b^-_{(n)} = \theta_{N+1 \, j}^{(N\! +\! 1)} s_{j \, N+1}^{(N\! +\! 1)} \sum_{\bm i \in \f i^n_N(j) }  \hat e^-_{\bm i}, \]
which completes the proof for $n>0$ for the formula for $\hat e^-$. The proof for the formula for $\hat e_+$ follows the same arguments:
\begin{align*}
\partial_{j \! +\! 1}^{(N\! +\! 1)} b^+_{(n+1)} - b^+_{(n+1)} \partial_j^{(N)} &= \sum_{\bm i \in \f i^n_N(j)} \sum_{l=1}^{n+1} \left( \partial_{j \! +\! 1}^{(N\! +\! 1)} \hat e^+_{i_1 \ldots i_{l-1} \, j \, i_l \ldots i_n} - \hat e^+_{i_1 \ldots i_{l-1} \, j \, i_l \ldots i_n} \partial_j^{(N)}\right) \\
&= \sum_{\bm i \in \f i^n_N(j)} \theta^{(N\! +\! 1)}_{j \! +\! 1 \, 1} s^{(N\! +\! 1)}_{1 \, j \! +\! 1} \hat e^+_{\bm i} + \\
& \quad - \sum_{\bm i \in \f i^n_N(j)} \sum_{l=1}^{n} \left( s^{(N\! +\! 1)}_{i_l \! +\! 1 \, j \! +\! 1} \theta^{(N\! +\! 1)}_{i_l \! +\! 1 \, j \! +\! 1} \hat e^+_{\bm i} \theta^{(N)}_{j \, i_l} - \theta^{(N\! +\! 1)}_{i_l \! +\! 1 \, j \! +\! 1} \hat e^+_{\bm i} \theta^{(N)}_{j \, i_l} s^{(N)}_{i_l \, j} \right). \qedhere
\end{align*} 
\end{proof}

For $j=1,\ldots,N$ denote
\begin{equation}
\begin{aligned}
\left( \Lambda^-_j \right)^{(N\! +\! 1)} &= \sum_{k=1}^{j-1} \theta_{j \, k}^{(N\! +\! 1)} s_{j \, k}^{(N\! +\! 1)} - \sum_{k=j \! +\! 1}^N \theta_{k \, j}^{(N\! +\! 1)} s_{j \, k}^{(N\! +\! 1)} \\
\left( \Lambda^+_j \right)^{(N\! +\! 1)} &= \sum_{k=1}^{j-1} \theta_{j \! +\! 1 \, k \! +\! 1}^{(N\! +\! 1)} s_{j \! +\! 1 \, k \! +\! 1}^{(N\! +\! 1)} - \sum_{k=j \! +\! 1}^N \theta_{k \! +\! 1 \, j \! +\! 1}^{(N\! +\! 1)} s_{j \! +\! 1 \, k \! +\! 1}^{(N\! +\! 1)},
\end{aligned}
\end{equation}%
\nc{glcz}{$( \Lambda^\pm_j )^{(N + 1)}$}{Shifted version of $\Lambda_j^{(N+1)}$ \nomrefeqpage}%
so that
\[ \Lambda^{(N\! +\! 1)}_j  = \left( \Lambda^-_j \right)^{(N\! +\! 1)}-\theta_{N+1 \, j}^{(N\! +\! 1)} s_{j \, N+1}^{(N\! +\! 1)}, \qquad 
\Lambda^{(N\! +\! 1)}_{j \! +\! 1}  = \left( \Lambda^+_j \right)^{(N\! +\! 1)} + \theta_{j \! +\! 1 \, 1}^{(N\! +\! 1)} s_{1 \, j \! +\! 1}^{(N\! +\! 1)}. \]

\begin{lem} \label{bnlem4}
Let $j=1,\ldots,N$ and $n=0,\ldots,N$.
Then
\begin{align*}
\left( \Lambda^-_j \right)^{(N\! +\! 1)} b^-_{(n)} - b^-_{(n)} \Lambda^{(N)}_{j} &= \sum_{l=1}^n  \sum_{\bm i \in \f i^n_N(j)} \left( \theta_{j \, i_l}^{(N\! +\! 1)} \hat e^-_{\bm i} \theta_{i_l \, j}^{(N)} s_{i_l \, j}^{(N)} - s_{i_l \, j}^{(N\! +\! 1)} \theta_{j \, i_l}^{(N\! +\! 1)} \hat e^-_{\bm i} \theta_{i_l \, j}^{(N)} \right), \\ 
\left( \Lambda^+_j \right)^{(N\! +\! 1)} b^+_{(n)} - b^+_{(n)} \Lambda^{(N)}_{j} &= - \sum_{l=1}^{n} \sum_{\bm i \in \f i^n_N(j)}  \left( \theta^{(N\! +\! 1)}_{i_l \! +\! 1 \, j \! +\! 1} \hat e^+_{\bm i} \theta^{(N)}_{j \, i_l} s^{(N)}_{i_l \, j} - s^{(N\! +\! 1)}_{i_l \! +\! 1 \, j \! +\! 1} \theta^{(N\! +\! 1)}_{i_l \! +\! 1 \, j \! +\! 1} \hat e^+_{\bm i} \theta^{(N)}_{j \, i_l} \right).
\end{align*}
\end{lem}

\begin{proof}
Let $1 \leq j \ne k \leq N$.
Note that $\hat e^- = \e^{\ii \lambda x_{N+1}} \hat \phi_{N+1}$ and hence $\theta_{j\,k}^{(N\! +\! 1)} s_{j \, k}^{(N\! +\! 1)} \hat e^- - \hat e^- \theta_{j\,k}^{(N)} s_{j \, k}^{(N)} = \theta_{k\,j}^{(N\! +\! 1)} s_{j \, k}^{(N\! +\! 1)} \hat e^- - \hat e^- \theta_{k\,j}^{(N)} s_{j \, k}^{(N)}=0$, so that the statements for $n=0$ follow.. \\

Now suppose $n \geq 1$, $1 \leq j \ne k \leq N$ and let $\bm i \in \f i_N^n$.
Then
\[ \theta_{j\,k}^{(N\! +\! 1)} s_{j \, k}^{(N\! +\! 1)} b^-_{(n)} - b^-_{(n)} \theta_{j\,k}^{(N)} s_{j \, k}^{(N)} = \sum_{\bm i \in \f i^n_N}  \left( \theta_{j\,k}^{(N\! +\! 1)} \hat e^-_{\bm i} - \hat e^-_{\bm i} \theta_{j\,k}^{(N)} \right) s_{j\,k}^{(N)},\]
because of \rfeqn{ehatmin}.
Similarly we obtain
\[ -\theta_{k\,j}^{(N\! +\! 1)} s_{j \, k}^{(N\! +\! 1)} b^-_{(n)} + b^-_{(n)} \theta_{k\,j}^{(N)} s_{j \, k}^{(N)} 
=\sum_{\bm i \in \f i^n_N} \left( \theta_{k\,j}^{(N\! +\! 1)} \hat e^-_{\bm i} - \hat e^-_{\bm i} \theta_{k\,j}^{(N)} \right)s_{j\,k}^{(N)}  =\sum_{\bm i \in \f i^n_N} \left( \theta_{j\,k}^{(N\! +\! 1)} \hat e^-_{\bm i} - \hat e^-_{\bm i} \theta_{j\,k}^{(N)} \right)  s_{j\,k}^{(N)}, \]
by virtue of $\theta_{j \, k}+\theta_{k\, j} = 1$.
It follows that
\begin{align*}
\left( \Lambda^-_j \right) ^{(N\! +\! 1)} b^-_{(n)} - b^-_{(n)} \Lambda^{(N)}_{j} 
&= \sum_{k \ne j}  \sum_{\bm i \in \f i^n_N} \left( \theta_{j\,k}^{(N\! +\! 1)} \hat e^-_{\bm i} - \hat e^-_{\bm i} \theta_{j\,k}^{(N)} \right)  s_{j\,k}^{(N)} \; = \; \sum_{k \ne j} \sum_{\bm i \in \f i^n_N}  \theta_{j \, k}^{(N\! +\! 1)} \hat e^-_{\bm i} \theta_{k \, j}^{(N)} s_{j\,k}^{(N)} \\
&= \sum_{k \ne j} \left( \sum_{\bm i \in \f i^n_N(j) \atop \exists l: \, i_l= k} \theta_{j \, k}^{(N\! +\! 1)} \hat e^-_{\bm i} \theta_{k \, j}^{(N)} + \sum_{\bm i \in \f i^n_N(k) \atop \exists l: \, i_{l}=j}  \theta_{j \, k}^{(N\! +\! 1)} \hat e^-_{\bm i} \theta_{k \, j}^{(N)} \right) s_{j\,k}^{(N)}, 
\end{align*}
where we have applied \rfl{hatelem5}.
\rfl{hatelem5} and $\theta_{j \, k}+ \theta_{k \, j} = 1$ yield
$\theta_{k \, j}^{(N\! +\! 1)} \hat e^-_{\bm i} \theta_{j \, k}^{(N)} = - \theta_{j \, k}^{(N\! +\! 1)} \hat e^-_{\bm i} \theta_{k \, j}^{(N)}$ so 
that this equals
\begin{align*}
\left( \Lambda^-_j \right) ^{(N\! +\! 1)} b^-_{(n)} - b^-_{(n)} \Lambda^{(N)}_{j} &=  \sum_{k \ne j} \left( \sum_{\bm i \in \f i^n_N(j) \atop \exists l: \, i_l= k} \theta_{j \, k}^{(N\! +\! 1)} \hat e^-_{\bm i} \theta_{k \, j}^{(N)} s_{j\,k}^{(N)} +  \sum_{\bm i \in \f i^n_N(k) \atop \exists l: \, i_{l}=j} s_{j\,k}^{(N\! +\! 1)} \theta_{k \, j}^{(N\! +\! 1)} \hat e^-_{s_{j\,k} \bm i} \theta_{j \, k}^{(N)}  \right) \displaybreak[2] \\
&= \sum_{k \ne j} \sum_{\bm i \in \f i^n_N(j) \atop \exists l: \, i_l= k}  \left( \theta_{j \, k}^{(N\! +\! 1)} \hat e^-_{\bm i} \theta_{k \, j}^{(N)} s_{j\,k}^{(N)} + s_{j\,k}^{(N\! +\! 1)} \theta_{k \, j}^{(N\! +\! 1)} \hat e^-_{\bm i} \theta_{j \, k}^{(N)}  \right) \displaybreak[2] \\
&=  \sum_{k \ne j} \sum_{\bm i \in \f i^n_N(j) \atop \exists l: \, i_l= k}  \left( \theta_{j \, k}^{(N\! +\! 1)} \hat e^-_{\bm i} \theta_{k \, j}^{(N)} s_{j\,k}^{(N)} - s_{j\,k}^{(N\! +\! 1)} \theta_{j \, k}^{(N\! +\! 1)} \hat e^-_{\bm i} \theta_{k \, j}^{(N)}  \right) \displaybreak[2] \\
&= \sum_{l=1}^n \sum_{\bm i \in \f i^n_N(j)} \left( \theta_{j \, i_l}^{(N\! +\! 1)} \hat e^-_{\bm i} \theta_{i_l \, j}^{(N)} s_{i_l \, j}^{(N)} - s_{i_l \, j}^{(N\! +\! 1)} \theta_{j \, i_l}^{(N\! +\! 1)} \hat e^-_{\bm i} \theta_{i_l \, j}^{(N)}  \right). 
\end{align*}
Again, the statement for $b^+_{(n)}$ is proved in a similar way.
\end{proof}

\begin{lem} \label{bnlem5}
Let $j=1,\ldots,N$.
Then 
\[ \theta_{N+1 \, j}^{(N\! +\! 1)} s_{j \, N+1}^{(N\! +\! 1)} b^-_{(N)} = \left( \Lambda_j^- \right)^{(N\!+\!1)} b_{(N)}^- - b_{(N)}^- \Lambda_j^{(N)}, \quad
-\theta_{1 \, j \! +\! 1}^{(N\! +\! 1)} s_{1 \, j \! +\! 1}^{(N\! +\! 1)} b^+_{(N)} = \left( \Lambda_j^+ \right)^{(N \!+ \!1)} b_{(N)}^+ - b_{(N)}^+  \Lambda_j^{(N)}. \]
\end{lem}
\begin{proof}
Both left- and right-hand sides vanish; the left-hand sides because of \rfeqn{eqn453}, and the right-hand sides because of \rfl{bnlem4} for $n=N$ (in which case $\f i^n_N$ is the $S_N$-orbit of $(1,2,\ldots,N)$).
\end{proof}

\cleardoublepage
\addcontentsline{toc}{chapter}{List of symbols}
\printnomenclature[30mm] \label{listofsymbols}

\bibliographystyle{plain}
\bibliography{H:/bibliography}

\begin{thebibliography}{10}

\bibitem{Baxter1972}
R.~J. Baxter.
\newblock Partition function of the eight-vertex lattice model.
\newblock {\em Ann. Physics}, 70:193--228, 1972.

\bibitem{Baxter1989}
R.~J. Baxter.
\newblock {\em Exactly solved models in statistical mechanics}.
\newblock Academic Press Inc. [Harcourt Brace Jovanovich Publishers], London,
  1989.
\newblock Reprint of the 1982 original.

\bibitem{BazhanovLZ}
Vladimir~V. Bazhanov, Sergei~L. Lukyanov, and Alexander~B. Zamolodchikov.
\newblock Integrable structure of conformal field theory. {II}. {${\rm
  Q}$}-operator and {DDV} equation.
\newblock {\em Comm. Math. Phys.}, 190(2):247--278, 1997.

\bibitem{BernardGHP}
D.~Bernard, M.~Gaudin, F.~D.~M. Haldane, and V.~Pasquier.
\newblock Yang-{B}axter equation in long-range interacting systems.
\newblock {\em J. Phys. A}, 26(20):5219--5236, 1993.

\bibitem{BernsteinGG}
I.~N. Bern{\v{s}}te{\v\i}n, I.~M. Gel{\cprime}fand, and S.~I. Gel{\cprime}fand.
\newblock Schubert cells, and the cohomology of the spaces {$G/P$}.
\newblock {\em Uspehi Mat. Nauk}, 28(3(171)):3--26, 1973.

\bibitem{Bethe}
H.~Bethe.
\newblock Zur {T}heorie der {M}etalle. i. {E}igenwerte und {E}igenfunktionen
  der linearen {A}tomkette.
\newblock {\em Zeitschrift f{\"u}r Physik A}, 71:205--226, 1931.

\bibitem{BustamanteVDDlM}
M.~D. Bustamante, J.~F. van Diejen, and A.~C. de~la Maza.
\newblock Norm formulae for the {B}ethe ansatz on root systems of small rank.
\newblock {\em J. Phys. A}, 41(2):025202, 13, 2008.

\bibitem{ChariPressley}
V.~Chari and A.~Pressley.
\newblock {\em A guide to quantum groups}.
\newblock Cambridge University Press, Cambridge, 1994.

\bibitem{Cherednik2}
I.~Cherednik.
\newblock A unification of {K}nizhnik-{Z}amolodchikov and {D}unkl operators via
  affine {H}ecke algebras.
\newblock {\em Invent. Math.}, 106(2):411--431, 1991.

\bibitem{Cherednik}
I.~Cherednik.
\newblock {\em Double affine {H}ecke algebras}, volume 319 of {\em London
  Mathematical Society Lecture Note Series}.
\newblock Cambridge University Press, Cambridge, 2005.

\bibitem{Demazure}
M.~Demazure.
\newblock Invariants sym\'etriques entiers des groupes de {W}eyl et torsion.
\newblock {\em Invent. Math.}, 21:287--301, 1973.

\bibitem{Dorlas}
T.~C. Dorlas.
\newblock Orthogonality and completeness of the {B}ethe ansatz eigenstates of
  the nonlinear {S}chroedinger model.
\newblock {\em Comm. Math. Phys.}, 154(2):347--376, 1993.

\bibitem{Drinfeld1985}
V.~G. Drinfel{\cprime}d.
\newblock Hopf algebras and the quantum {Y}ang-{B}axter equation.
\newblock {\em Dokl. Akad. Nauk SSSR}, 283(5):1060--1064, 1985.

\bibitem{Drinfeld1986}
V.~G. Drinfel{\cprime}d.
\newblock Degenerate affine {H}ecke algebras and {Y}angians.
\newblock {\em Funktsional. Anal. i Prilozhen.}, 20(1):69--70, 1986.

\bibitem{Emsiz}
E.~Emsiz.
\newblock {\em Affine {W}eyl groups and integrable systems with
  delta-potentials}.
\newblock PhD thesis, University of Amsterdam, Faculty of Science, 2006.

\bibitem{EmsizOS}
E.~Emsiz, E.~M. Opdam, and J.~V. Stokman.
\newblock Periodic integrable systems with delta-potentials.
\newblock {\em Comm. Math. Phys.}, 264(1):191--225, 2006.

\bibitem{Faddeev1995}
L.~Faddeev.
\newblock Instructive history of the quantum inverse scattering method.
\newblock {\em Acta Appl. Math.}, 39(1-3):69--84, 1995.
\newblock KdV '95 (Amsterdam, 1995).

\bibitem{Gaudin1971-1}
M.~Gaudin.
\newblock Bose gas in one dimension, {I}. {T}he closure property of the
  scattering wavefunctions.
\newblock {\em J. Math. Phys.}, 12:1674--1676, 1971.

\bibitem{Gaudin1971-2}
M.~Gaudin.
\newblock Bose gas in one dimension, {II}. orthogonality of the scattering
  states.
\newblock {\em J. Math. Phys.}, 12:1677--1680, 1971.

\bibitem{Gaudin1971-3}
M.~Gaudin.
\newblock Boundary energy of a {B}ose gas in one dimension.
\newblock {\em Phys. Rev. A}, 4(1):386--394, 1971.

\bibitem{Gaudin1983}
M.~Gaudin.
\newblock {\em La fonction d'onde de {B}ethe}.
\newblock Collection du Commissariat \`a l'\'Energie Atomique: S\'erie
  Scientifique. [Collection of the Atomic Energy Commission: Science Series].
  Masson, Paris, 1983.

\bibitem{Girardeau}
M.~Girardeau.
\newblock Relationship between systems of impenetrable bosons and fermions in
  one dimension.
\newblock {\em J. Mathematical Phys.}, 1:516--523, 1960.

\bibitem{Gutkin1982}
E.~Gutkin.
\newblock Integrable systems with delta-potential.
\newblock {\em Duke Math. J.}, 49(1):1--21, 1982.

\bibitem{Gutkin1985}
E.~Gutkin.
\newblock Conservation laws for the nonlinear {S}chr\"odinger equation.
\newblock {\em Ann. Inst. H. Poincar\'e Anal. Non Lin\'eaire}, 2(1):67--74,
  1985.

\bibitem{Gutkin1987}
E.~Gutkin.
\newblock Operator calculi associated with reflection groups.
\newblock {\em Duke Math. J.}, 55(1):1--18, 1987.

\bibitem{Gutkin1988}
E.~Gutkin.
\newblock Quantum nonlinear {S}chr\"odinger equation: two solutions.
\newblock {\em Phys. Rep.}, 167(1-2):1--131, 1988.

\bibitem{GutkinSutherland}
E.~Gutkin and B.~Sutherland.
\newblock Completely integrable systems and groups generated by reflections.
\newblock {\em Proc. Nat. Acad. Sci. U.S.A.}, 76(12):6057--6059, 1979.

\bibitem{HeckmanOpdam1996}
G.~J. Heckman and E.~M. Opdam.
\newblock Harmonic analysis for affine {H}ecke algebras.
\newblock In {\em Current developments in mathematics, 1996 ({C}ambridge,
  {MA})}, pages 37--60. Int. Press, Boston, MA, 1997.

\bibitem{HeckmanOpdam1997}
G.~J. Heckman and E.~M. Opdam.
\newblock {Y}ang's system of particles and {H}ecke algebras.
\newblock {\em Ann. of Math. (2)}, 145(1):139--173, 1997.

\bibitem{Hikami1996}
K.~Hikami.
\newblock Boundary {$K$}-matrix, elliptic {D}unkl operator and quantum
  many-body systems.
\newblock {\em J. Phys. A}, 29(9):2135--2147, 1996.

\bibitem{Hikami}
K.~Hikami.
\newblock Notes on the structure of the {$\delta$}-function interacting gas.
  {I}ntertwining operator in the degenerate affine {H}ecke algebra.
\newblock {\em J. Phys. A}, 31(4):L85--L91, 1998.

\bibitem{IzerginKorepin}
A.~G. Izergin and V.~E. Korepin.
\newblock The {P}auli principle for one-dimensional bosons and the algebraic
  {B}ethe ansatz.
\newblock {\em Journal of Math. Sci.}, 34(5):1933--1937, 1986.

\bibitem{KazhdanLusztig}
D.~A. Kazhdan and G.~Lusztig.
\newblock Proof of the {D}eligne-{L}anglands conjecture for {H}ecke algebras.
\newblock {\em Invent. Math.}, 87(1):153--215, 1987.

\bibitem{KomoriHikami}
Y.~Komori and K.~Hikami.
\newblock Nonlinear {S}chr\"odinger model with boundary, integrability and
  scattering matrix based on the degenerate affine {H}ecke algebra.
\newblock {\em Internat. J. Modern Phys. A}, 12(30):5397--5410, 1997.

\bibitem{Korepin}
V.~E. Korepin.
\newblock Calculation of norms of {B}ethe wave functions.
\newblock {\em Comm. Math. Phys.}, 86(3):391--418, 1982.

\bibitem{KorepinBI}
V.~E. Korepin, N.~M. Bogoliubov, and A.~G. Izergin.
\newblock {\em Quantum inverse scattering method and correlation functions}.
\newblock Cambridge Monographs on Mathematical Physics. Cambridge University
  Press, Cambridge, 1993.

\bibitem{Korff2006}
C.~Korff.
\newblock A {$Q$}-operator for the twisted {$XXX$} model.
\newblock {\em J. Phys. A}, 39(13):3203--3219, 2006.

\bibitem{KricheverLWZ}
I.~Krichever, O.~Lipan, P.~Wiegmann, and A.~Zabrodin.
\newblock Quantum integrable models and discrete classical {H}irota equations.
\newblock {\em Comm. Math. Phys.}, 188(2):267--304, 1997.

\bibitem{KuznetsovS1998}
V.~B. Kuznetsov and E.~K. Sklyanin.
\newblock On {B}\"acklund transformations for many-body systems.
\newblock {\em J. Phys. A}, 31(9):2241--2251, 1998.

\bibitem{LiebLi}
E.~H. Lieb and W.~Liniger.
\newblock Exact analysis of an interacting {B}ose gas. {I}. {T}he general
  solution and the ground state.
\newblock {\em Phys. Rev. (2)}, 130:1605--1616, 1963.

\bibitem{Lusztig}
G.~Lusztig.
\newblock Affine {H}ecke algebras and their graded version.
\newblock {\em J. Amer. Math. Soc.}, 2(3):599--635, 1989.

\bibitem{Macdonald1}
I.~G. Macdonald.
\newblock {\em Symmetric functions and {H}all polynomials}.
\newblock Oxford Mathematical Monographs. The Clarendon Press Oxford University
  Press, New York, second edition, 1995.
\newblock With contributions by A. Zelevinsky, Oxford Science Publications.

\bibitem{Macdonald2}
I.~G. Macdonald.
\newblock {\em Affine {H}ecke algebras and orthogonal polynomials}, volume 157
  of {\em Cambridge Tracts in Mathematics}.
\newblock Cambridge University Press, Cambridge, 2003.

\bibitem{Molev}
A.~I. Molev.
\newblock {\em Yangians and classical {L}ie algebras}, volume 143 of {\em
  Mathematical Surveys and Monographs}.
\newblock American Mathematical Society, Providence, RI, 2007.

\bibitem{MurakamiWadati}
S.~Murakami and M.~Wadati.
\newblock Connection between {Y}angian symmetry and the quantum inverse
  scattering method.
\newblock {\em J. Phys. A}, 29(24):7903--7915, 1996.

\bibitem{Opdam}
E.~M. Opdam.
\newblock Harmonic analysis for certain representations of graded {H}ecke
  algebras.
\newblock {\em Acta Math.}, 175(1):75--121, 1995.

\bibitem{Oxford}
S.~C. Oxford.
\newblock {\em T{he} {Hamiltonian} {of} {the} {Quantized} {Nonlinear}
  {Schroedinger} {Equation}}.
\newblock PhD thesis, University of California, Los Angeles, 1979.
\newblock Thesis (Ph.D.).

\bibitem{PasquierGaudin}
V.~Pasquier and M.~Gaudin.
\newblock The periodic {T}oda chain and a matrix generalization of the {B}essel
  function recursion relations.
\newblock {\em J. Phys. A}, 25(20):5243--5252, 1992.

\bibitem{Polychronakos}
A.~P. Polychronakos.
\newblock Exchange operator formalism for integrable systems of particles.
\newblock {\em Phys. Rev. Lett.}, 69(5):703--705, 1992.

\bibitem{ReedSimon}
M.~Reed and B.~Simon.
\newblock {\em Methods of modern mathematical physics. {I}. {F}unctional
  analysis}.
\newblock Academic Press, New York, 1972.

\bibitem{ReedSimon2}
M.~Reed and B.~Simon.
\newblock {\em Methods of modern mathematical physics. {II}. {F}ourier
  analysis, self-adjointness}.
\newblock Academic Press, New York, 1975.

\bibitem{Sklyanin1982}
E.~K. Sklyanin.
\newblock Quantum variant of the method of the inverse scattering problem.
\newblock {\em J. Sov. Math.}, 19(5):1546--1596, 1982.

\bibitem{Sklyanin1988}
E.~K. Sklyanin.
\newblock Boundary conditions for integrable quantum systems.
\newblock {\em J. Phys. A: Math. Gen.}, 21:2375--2389, 1988.

\bibitem{Sklyanin1989}
E.~K. Sklyanin.
\newblock New approach to the quantum nonlinear {S}chr\"odinger equation.
\newblock {\em J. Phys. A}, 22(17):3551--3560, 1989.

\bibitem{Sklyanin2000}
E.~K. Sklyanin.
\newblock B\"acklund transformations and {B}axter's {$Q$}-operator.
\newblock In {\em Integrable systems: from classical to quantum ({M}ontr\'eal,
  {QC}, 1999)}, volume~26 of {\em CRM Proc. Lecture Notes}, pages 227--250.
  Amer. Math. Soc., Providence, RI, 2000.

\bibitem{SklyaninFaddeev}
E.~K. Sklyanin and L.~D. Faddeev.
\newblock Quantum-mechanical approach to completely integrable models of field
  theory.
\newblock {\em Dokl. Akad. Nauk SSSR}, 243:1430--1433, 1978.

\bibitem{SklyaninTakhtajanFaddeev}
E.~K. Sklyanin, L.~A. Takhtajan, and L.~D. Faddeev.
\newblock Quantum inverse problem method. i.
\newblock {\em Teoret. Mat. Fiz.}, 40(2):194--220, 1979.

\bibitem{Toda}
M.~Toda.
\newblock {\em Theory of nonlinear lattices}, volume~20 of {\em Springer Series
  in Solid State Sciences}.
\newblock Springer-Verlag, Berlin, 1981.

\bibitem{Amerongen}
A.~H. van Amerongen.
\newblock {\em One-dimensional {B}ose gas on an atom chip}.
\newblock PhD thesis, University of Amsterdam, 2008.

\bibitem{AmerongenEWKD}
A.~H. van Amerongen, J.~J.~P. van Es, P.~Wicke, K.~V. Kheruntsyan, and N.~J.
  van Druten.
\newblock Yang-{Y}ang thermodynamics on an atom chip.
\newblock {\em Phys. Rev. Lett.}, 100(9):090402, 2008.

\bibitem{YangYang}
C.~N. Yang and C.~P. Yang.
\newblock Thermodynamics of a one-dimensional system of bosons with repulsive
  delta-function interaction.
\newblock {\em J. Mathematical Phys.}, 10:1115--1122, 1969.

\end{thebibliography}
\label{references}

\end{document}